\documentclass[10pt,letterpaper]{article}

\usepackage[letterpaper,margin=1in]{geometry}
\usepackage[final]{pdfpages}

\usepackage{amssymb}
\usepackage{amsmath}
\usepackage{amsfonts}
\usepackage{amsthm}
\usepackage{bm} % bold math fonts, e.g. $\bm \Sigma$ will give a bold-face $\Sigma$
\usepackage{url}
\usepackage{hyperref}
\usepackage[colorinlistoftodos]{todonotes}
\usepackage{relsize}

\usepackage{xspace}
\usepackage{algorithmicx}
\usepackage{algorithm}
\usepackage{algpseudocode}
\usepackage[shortlabels]{enumitem}

\usepackage{colortbl}

\usepackage[para,online,flushleft]{threeparttable}
\usepackage{multirow}

\usepackage{thmtools}
\usepackage{thm-restate}

\usepackage{tikz}
\usetikzlibrary{calc}
\usetikzlibrary{shapes,decorations.markings,arrows.meta,fit}

\definecolor{lgray}{rgb}{.6, .6, .6}
\definecolor{llgray}{rgb}{.8, .8, .8}

\newcommand{\norm}[1]{\|#1\|}
\newcommand{\set}[1]{\{#1\}}                    % Set (as in \set{1,2,3}).
\newcommand{\setof}[2]{\{{#1}\mid{#2}\}}        % Set (as in \setof{x}{x>0}).
\newcommand{\dom}{\textsf{Dom}}

\newcommand{\degree}{\text{\sf deg}}

\newcommand{\OUT}{\textsf{OUT}}

\newcommand{\xiao}[1]{\mycomment{\todo[inline,color=cyan]{\textsf{#1} \hfill \textsc{--Xiao.}}}}

\newcommand{\mycomment}[1]{#1}
\newcommand{\shortorfull}[2]{#2}

\newcommand{\calS}{\mathcal S}

\newcommand{\calA}{\mathcal A}
\newcommand{\calB}{\mathcal B}
\newcommand{\calH}{\mathcal H}
\newcommand{\calV}{\mathcal V}
\newcommand{\calE}{\mathcal E}

\newcommand{\calF}{\mathcal F}
\newcommand{\calG}{\mathcal G}
\newcommand{\calT}{\mathcal T}

\newcommand{\calI}{\mathcal I}
\newcommand{\calP}{\mathcal P}

\newtheorem{theorem}{Theorem}[section]
\newtheorem{lemma}[theorem]{Lemma}
\newtheorem{proposition}[theorem]{Proposition}

\newtheorem{claim}[theorem]{Claim}

\theoremstyle{definition}              % Examples and all
\newtheorem{definition}[theorem]{Definition}
\newtheorem{example}[theorem]{Example}
\newtheorem{remark}[theorem]{Remark}

\newcommand{\defeq}{\stackrel{\mathrm{def}}{=}}

\newcommand{\R}{\mathbb R} % the real numbers
\newcommand{\cd}{\text{ :- }}

\newcommand{\bigjoin}{\mathlarger{\mathlarger{\mathlarger{\Join}}}}

\newcommand{\inn}{\text{\sf in}}
\newcommand{\outt}{\text{\sf out}}
\newcommand{\out}{\textsf{OUT}}
\newcommand{\vars}{\text{\sf vars}}
\newcommand{\nodes}{\text{\sf nodes}}

\newcommand{\atoms}{\mathsf{atoms}}

\newcommand{\opt}{\mathsf{opt}}

\newcommand{\true}{\text{\sf true}\xspace}
\newcommand{\false}{\text{\sf false}\xspace}

\newcommand{\panda}{\mathsf{PANDA}\xspace}
\newcommand{\pandaexpress}{\textsf{PANDAExpress}\xspace}
\newcommand{\opandaexpress}{\textsf{$\omega$-PANDAExpress}\xspace}
\newcommand{\obcqsolver}{\textsf{$\omega$-BCQSolver}\xspace}

\newcommand{\fhtw}{\textsf{fhtw}}

\newcommand{\subw}{\textsf{subw}}

\newcommand{\osubw}{\textsf{$\omega$-subw}}

\newcommand{\ed}{\mathsf{ED}}

\newcommand{\mm}{\mathsf{MM}}

\newcommand{\emm}{\mathsf{EMM}}

\newcommand{\args}{\mathsf{args}}

\newcommand{\ov}{\overline}

\newcommand{\squareOmega}{
    \omega^{\square}
}
\newcommand{\rectOmega}{
    \ov{\omega}
}
\newcommand{\squareC}{
    c^{\square}
}
\newcommand{\rectC}{
    \ov{c}
}

\newcommand{\rectP}{
    \ov{P}
}

\newcommand{\applystep}{\textsf{apply-step}\xspace}
\newcommand{\resetineq}{\textsf{reset}\xspace}
\newcommand{\polylog}{\text{\sf polylog}}
\newcommand{\fourcycle}{Q_{\square}}
\newcommand{\fourcycleH}{\calH_{\square}}

\newcommand{\positiveterm}[1]{\fcolorbox{black}{red!35}{$#1$}}
\newcommand{\negativeterm}[1]{\colorbox{blue!35}{$#1$}}

\colorlet{lightergray}{gray!20!white}

\colorlet{darkgreen}{green!40!black}

\newcommand{\algocomment}[1]{\algorithmiccomment{{\color{blue!70}\em #1}}}

\allowdisplaybreaks

\begin{document}

\title{Fast Matrix Multiplication meets the Submodular Width}

\newcommand{\myemail}[1]{\small\texttt{#1}}

\author{
  {Mahmoud Abo Khamis}\\
  University of Oxford\\RelationalAI\\
  \myemail{mahmoud.abokhamis@relational.ai}
  % Location: Berkeley, CA, US
  % ORCID: 0000-0003-3894-6494
  \and
  {Xiao Hu}\\
  University of Waterloo\\
  \myemail{xiaohu@uwaterloo.ca}
  % Location: Waterloo, ON, Canada
  % ORCID: 0000-0002-7890-665X
  \and
  {Dan Suciu}\\
  University of Washington\\
  \myemail{suciu@cs.washington.edu}
  % Location: Seattle, WA, US
  % ORCID: 0000-0002-4144-0868
  \thanks{This work was partially supported by NSF-BSF 2109922, NSF-IIS 2314527, NSF-SHF 2312195,
    and by the Natural Sciences and Engineering Research Council of Canada -- Discovery Grant,
    and was initiated while the authors participated in the
  Fall 2023 {\em Simons Program on Logic and Algorithms in Databases and AI.}}
}

\date{}

\maketitle

\begin{abstract}
One fundamental question in database theory is the following:
Given a Boolean Conjunctive Query (BCQ) $Q$,
what is the best complexity for computing the answer to $Q$ in terms of the input database size $N$?
When restricted to the class of combinatorial algorithms, the best known complexity for any query
$Q$ is captured by the {\em submodular width} of $Q$~\cite{DBLP:journals/jacm/Marx13,DBLP:conf/pods/Khamis0S17,theoretics:13722}.
However, beyond combinatorial algorithms, certain queries are known to admit faster algorithms
that often involve a clever combination of fast matrix multiplication and data partitioning.
Nevertheless, there is no systematic way to derive and analyze the complexity of such algorithms
for arbitrary queries $Q$.

In this work, we introduce a general framework that captures the best complexity for answering
any BCQ $Q$ using matrix multiplication.
Our framework unifies both combinatorial and non-combinatorial techniques under the umbrella of information theory.
It generalizes the notion of submodular width to a new stronger notion called the {\em $\omega$-submodular width} that naturally incorporates the power of fast matrix multiplication.
We describe a matching algorithm that computes the answer to any query $Q$ in time corresponding to the $\omega$-submodular width of $Q$.
We show that our framework recovers the best known complexities for Boolean queries that have been studied in the literature, to the best of our knowledge, and also discovers new algorithms for some classes of queries that improve upon the best known complexities.
\end{abstract}

\section{Introduction}

\label{sec:intro}

We focus on the problem of answering a {\em Boolean Conjunctive Query} (BCQ) $Q$. In particular,
we have a set $\vars(Q)$ of variables (or attributes)  and a set $\atoms(Q)$ of relations where each relation $R(\bm X)\in\atoms(Q)$ is over a variable set $\bm X \subseteq \vars(Q)$.
In particular, each relation $R(\bm X)$ is a list of satisfying assignments to the variables $\bm X$,
and the query $Q$ asks whether there exists an assignment to all variables $\vars(Q)$ that simultaneously satisfies all relations in $\atoms(Q)$:
\begin{align}
    Q() \cd \bigwedge_{R(\bm X) \in \atoms(Q)} R(\bm X)
    \label{eq:bcq}
\end{align}
We assume the query to be fixed, hence its size is a constant, and we measure the runtime in terms of the total size of the input relations, denoted by $N$, i.e., we use~{\em data complexity}.
For brevity, throughout the paper, we refer to a Boolean conjunctive query as just {\em ``query''}.

Using only combinatorial algorithms, the best known complexity for any query $Q$ is given by the {\em submodular width}~\cite{DBLP:journals/jacm/Marx13,DBLP:conf/pods/Khamis0S17,theoretics:13722}.
However, when fast matrix multiplication is allowed, some isolated queries admit faster algorithms,
but there is no general framework to derive such algorithms for any query. In this paper, we introduce such a framework that naturally unifies both combinatorial and non-combinatorial techniques under the umbrella of {\em information theory}. In particular, we generalize the submodular width to incorporate matrix multiplication
and develop a matching algorithm. We show that our general algorithm matches or improves upon the best known custom algorithms for queries that have been studied in the literature.

% \dan{We need to add the paragraph below, or some variation.  We could either leave it here, or move to the end of Sec.~\ref{sec:intro} (I don't have a preference).}

A preliminary, short version of this paper appeared in~\cite{DBLP:journals/pacmmod/KhamisHS25}.  That version described the algorithm only at an intuitive level.  The current paper contains the full details of the algorithm, proves its runtime, and removes a polylogarithmic factor from the runtime by incorporating new techniques described in~\cite{panda-express}.

\subsection{Background}
\subsubsection{Combinatorial Join Algorithms}
We start with some background on combinatorial algorithms for queries.
When restricted to the class of combinatorial algorithms, there are only three basic techniques that are sufficient to recover the best-known complexity over this class of algorithms for any query $Q$:
\begin{itemize}[leftmargin=*]
    \item {\em For-loops}: Worst-case optimal join (WCOJ) algorithms~\cite{10.1145/3196959.3196990,10.1145/3180143}, like GenericJoin~\cite{DBLP:journals/sigmod/NgoRR13}
    or LeapFrog-TrieJoin~\cite{DBLP:conf/icdt/Veldhuizen14}, can be viewed as a sequence of nested for-loops, each of which iterates over possible assignments of one variable. For example, consider the Boolean triangle query:
    \begin{align}
        Q_\triangle() \cd R(X, Y) \wedge S(Y, Z) \wedge T(X, Z)
        \label{eq:intro:triangle}
    \end{align}
    One possible WCOJ algorithm consists of a for-loop over the intersection of $X$-values
    from $R$ and $T$, and for each such assignment, a for-loop over the intersection of $Y$-values from $R$ and $S$, and for each such assignment, a for-loop over the intersection of $Z$-values from $S$ and $T$.
    This simple algorithm gives a runtime of $O(N^{3/2})$ for this query
    and $O(N^{\rho^*(Q)})$ in general, where
    $\rho^*(Q)$ is the {\em fractional edge cover number} of $Q$, which is also an upper bound on the join size~\cite{DBLP:conf/soda/GroheM06,DBLP:journals/siamcomp/AtseriasGM13,DBLP:conf/focs/AtseriasGM08}.

    \item {\em Tree Decompositions (TDs)}: Sometimes, two nested loops can be (conditionally) independent of one other. For example,
    consider the query:
    \begin{align}
        \label{eq:intro:triangletriangle}
        Q_{\triangle\triangle}() \cd R(X, Y) \wedge S(Y, Z) \wedge T(X, Z) \wedge S'(Y, Z') \wedge T'(X, Z')
    \end{align}
    We could solve it in time $O(N^2)$ using 4 nested for-loops over $X$, $Y$, $Z$, and $Z'$ in order.
    However, note that once we fix the values of $X$ and $Y$, the two inner loops over $Z$ and $Z'$ become independent, hence can be unnested. One way to capture and utilize such conditional independence is using the framework of {\em tree decompositions}, which are a form of query plans. In this example, we could break down the query using a tree decomposition, or {\em TD} for short, consisting of two ``bags'' (i.e.~two subqueries in the query plan) where one bag corresponds to a triangle query over $\{X, Y, Z\}$
    while the other bag is a triangle query over $\{X, Y, Z'\}$, thus leading to a runtime
    of $O(N^{3/2})$. Using TDs (alongside for-loops), we can answer any
    query $Q$ in time $O(N^{\fhtw(Q)})$ where $\fhtw(Q)$ is the {\em fractional hypertree width} of $Q$~\cite{DBLP:conf/soda/GroheM06,DBLP:journals/talg/GroheM14,DBLP:conf/pods/KhamisNR16}.

    \item {\em Data Partitioning}:
    For-loops and TDs alone are not sufficient to unleash the full power
    of combinatorial join algorithms for all queries.
    Consider the following {\em 4-cycle query}:
    \begin{align}
        Q_\square() \cd R(X, Y) \wedge S(Y, Z) \wedge T(Z, W) \wedge U(W, X)
        \label{eq:intro:4cycle}
    \end{align}
    $Q_\square$ admits two TDs: one with two bags $\{X, Y, Z\}$ and $\{Z, W, X\}$, while the other with two bags $\{Y, Z, W\}$ and $\{W, X, Y\}$.
    However, using either TD alone, we cannot achieve a runtime better than $O(N^2)$. On the other hand, if we partition the input relations carefully into multiple parts, and select a proper TD for each part, we could achieve
    a runtime of $O(N^{3/2})$~\cite{DBLP:journals/algorithmica/AlonYZ97}.
    Partitioning is done based on the ``degrees'' of relations where we think of a (binary) relation like $R(X, Y)$ as a bipartite graph, and compute degrees of vertices
    accordingly.
    Taking the partitioning approach to the extreme\footnote{By that, we mean partitioning not just the input relations but also intermediate relations that result from the join of (input or intermediate) relations.
    This ``multi-level'' partitioning can indeed lower the complexity further than one-level partitioning, for some classes of queries.}, the $\panda$ algorithm~\cite{DBLP:conf/pods/Khamis0S17,theoretics:13722} can solve any query $Q$
    in time $O(N^{\subw(Q)}\cdot \polylog(N))$ where $\subw(Q)$ is the {\em submodular width} of $Q$~\cite{DBLP:journals/jacm/Marx13}.
    An improved algorithm, called $\pandaexpress$~\cite{panda-express}, removes the polylogarithmic factor, achieving a runtime of $O(N^{\subw(Q)}\cdot \log N)$.
\end{itemize}

The submodular width is a single definition that combines the above three techniques, and the corresponding
$\panda$ algorithm achieves (up to a polylogarithmic factor) the best-known complexity for any query $Q$ over combinatorial algorithms.
The submodular width and $\panda$ have a deep connection to information theory.
At a very high level, the submodular width of a given query $Q$ can be thought of as aiming to capture the best complexity
of answering $Q$ using the following meta-algorithm: Think of (binary) input relations, like $R(X, Y)$
above, as bipartite graphs, and partition each of them into (a polylogarithmic number of) parts that are almost ``uniform'',~i.e. where
all vertices within the same part have roughly the same degree.
Now each part of the data can be described by its combination of degrees, called ``degree configuration''.
For each degree configuration, we pick the best TD, go over its bags, and solve
the corresponding subqueries, using for-loops.\footnote{Note that this algorithm uses
the three techniques mentioned above in reverse order.}
Instead of reasoning about degree configurations directly, we model them as {\em edge-dominated
polymatroids}, or {\em ED-polymatroids} for short, which is an information-theoretic concept.
In particular, given a degree configuration, the corresponding polymatroid is roughly the entropy
of a certain probability distribution over the join of input relations
having the given degree configuration. (Formal Definition will be given in Sec.~\ref{sec:prelims}.)
Using polymatroids as a proxy to degree configurations allows us to transform a database problem
into an information-theoretic problem. To sum up, the submodular width has the following skeleton (the formal definition will be given later):
\begin{align}
    \subw(Q) \quad\defeq\quad
        \underbrace{\max_{\text{ED-polymatroid $\bm h$}}}_{\substack{\text{worst part}\\\text{of the data}}} \quad
        \underbrace{\min_{\text{tree decomposition $T$}}}_{\substack{\text{best query plan}\\\text{for this part}}} \quad
        \underbrace{\max_{\text{bag $B\in T$}}}_{\substack{\text{worst subquery}\\\text{in the plan}}} \quad
        \underbrace{h(B)}_{\substack{\text{subquery cost}\\\text{(using for-loops)}}}
    \label{eq:intro:subw}
\end{align}

\subsubsection{Beyond Combinatorial Join Algorithms}
For certain queries, there are known non-combinatorial algorithms with lower complexity
than the best known combinatorial algorithms.
In addition to the three techniques mentioned above,
these non-combinatorial algorithms typically involve a fourth technique, which is {\em matrix multiplication}, or MM for short.
For background, given two $n\times n$ matrices, there are algebraic algorithms
that can multiply them in time $o(n^3)$. The {\em matrix multiplication exponent} $\omega$ is the smallest exponent $\alpha$
where this multiplication can be done in time $O(n^{\alpha})$.
It was first discovered by Strassen~\cite{Strassen1969GaussianEI} that $\omega < 3$,
and to date, the best known upper bound for $\omega$ is $2.371552$~\cite{doi:10.1137/1.9781611977912.134}.
For certain queries, incorporating MM can lead to faster algorithms by first partitioning the data based on degrees, and then for each part,  choosing to either perform an MM or to use a traditional combinatorial algorithm (consisting of a TD and for-loops). Which choice is better depends on the degree configuration of the part.
For parts with low degrees, combinatorial algorithms are typically better, while parts with high degrees
tend to benefit from MM. The complexities of such algorithms typically involve
$\omega$. For example, for the triangle query $Q_\triangle$, there is a non-combinatorial algorithm
with complexity $O(N^{\frac{2\omega}{\omega + 1}})$~\cite{DBLP:journals/algorithmica/AlonYZ97}.\footnote{Note that for $\omega = 3$, this complexity collapses back to $O(N^{3/2})$ which is the same as the combinatorial algorithm.}
However, such non-combinatorial algorithms are only known for some isolated queries;
see Table~\ref{tab:intro:comparison} for a summary of known results.
There is no general framework %to infer the best complexity 
for answering any query $Q$ using
MM.

\begin{table}[t]
        \centering
        \begin{tabular}{|c|c|c|}
        \hline
            \bf Query \cellcolor{lgray} & \bf \cellcolor{lgray} Best Prior Algorithm & \cellcolor{lgray} \bf Our Algorithm \\
            \hline\hline
            \cellcolor{llgray}{\em Arbitrary} query $Q$& \cellcolor{llgray}$O(N^{\subw(Q)}\cdot\log N)$~\cite{DBLP:conf/pods/Khamis0S17,theoretics:13722,panda-express} &\cellcolor{llgray}
            $O(N^{\osubw(Q)}\cdot \log^2 N)$\\
            \hline
            Triangle $Q_\triangle$~(Eq.~\eqref{eq:intro:triangle}) & $O\left(N^{\frac{2\omega}{\omega+1}}\right)$~\cite{DBLP:journals/algorithmica/AlonYZ97}  & same \\ \hline
            4-Clique & $O\left(N^\frac{\omega+1}{2}\right)$ \cite{dalirrooyfard2024towards} & same  \\
            \hline
            5-Clique & $O\left(N^{\frac{\omega}{2} + 1}\right)$ \cite{dalirrooyfard2024towards} & same \\
            \hline
             \multirow{3}{*}{$k$-Clique ($k\ge 6$)} & \multirow{3}{*}{$O\left(N^{\rectOmega(\frac{1}{2} \cdot \lceil \frac{k}{3}\rceil, \frac{1}{2} \cdot \lceil \frac{k-1}{3}\rceil, \frac{1}{2} \cdot \lfloor \frac{k}{3}\rfloor)}\right)$ \cite{eisenbrand2004complexity}} & \multirow{3}{*}{$O\left(N^{\frac{1}{2} \cdot \lceil \frac{k}{3}\rceil +  \frac{1}{2} \cdot \lceil \frac{k-1}{3}\rceil + \frac{1}{2} \cdot \lfloor \frac{k}{3}\rfloor \cdot (\omega-2)}\right)$} \\
             & & \multirow{3}{*}{(same for $\omega = 2$)} \\ 
             & & \\ \hline 
             4-Cycle $Q_\square$ (Eq.~\eqref{eq:intro:4cycle}) & $O\left(N^{\frac{4\omega-1}{2\omega+1}}\right)$~\cite{yuster2004detecting, dalirrooyfard2019graph} & $O\left(N^{\frac{4\omega-1}{2\omega+1}}\cdot\log^2 N\right)$ \\ \hline
             $k$-Cycle & $O\left(N^{\rectC_k}\cdot{\color{red}\polylog(N)}\right)$ \cite{yuster2004detecting, dalirrooyfard2019graph} & \cellcolor{yellow}$O\left(N^{\squareC_k}\cdot {\color{darkgreen}\log^2 N}\right)$ \\[-\arrayrulewidth]
             & & \cellcolor{yellow}($\rectC_k = \squareC_k$ for $\omega = 2$)\\\hline
             $k$-Pyramid (Eq.~\eqref{eq:body:k-pyramid}) & $O\left(N^{\color{red}2-\frac{1}{k}}\cdot \log N\right)$ \cite{DBLP:conf/pods/Khamis0S17,theoretics:13722,panda-express} & \cellcolor{yellow}$O\left(N^{\color{darkgreen}2-\frac{2}{\omega(k-1)-k+3}}\cdot\log^2 N\right)$ \\ \hline
        \end{tabular}
        \caption{A summary of prior results and the corresponding results obtained by our framework.
        \colorbox{yellow}{Highlighted cells} show {\em improvements} of our results over the best known prior results.
        Our results assume $\omega$ to be rational, but still apply to irrational $\omega$ by using any rational upper bound on $\omega$ instead.
        The extra $\log^2 N$ factor in our general bound $O(N^{\osubw(Q)}\cdot \log^2 N)$ is not needed for every $Q$, e.g., $k$-Clique.
        The notation $\rectOmega(a,b,c)$ denotes the smallest exponent for multiplying two rectangular matrices of dimensions $n^a \times n^b$ and $n^b \times n^c$ within $O\left(n^{\rectOmega(a,b,c)}\right)$ time.
        In contrast, $\squareOmega(a,b,c)$ is the smallest upper bound on $\rectOmega(a,b,c)$
        that is obtained through {\em square} matrix multiplication. In particular, $\squareOmega(a, b, c) \defeq \max\{a + b+(\omega-2)c, a+(\omega-2)b + c, (\omega-2)a + b + c\}$; see Sec.~\ref{sec:prelims}.
        Obviously, $\rectOmega(a,b,c) \leq \squareOmega(a, b, c)$. Moreover, this becomes an equality when $\omega = 2$ or when $a=b=c$; see Sec.~\ref{sec:prelims}.
        The symbol $\rectC_k$ is the best-known exponent for detecting cycles using rectangular matrix multiplication~\cite[Theorem 1.3]{dalirrooyfard2019graph},
        while $\squareC_k$ is the smallest upper bound on $\rectC_k$ that is obtained through {\em square} matrix multiplication; see \shortorfull{\cite{full-version}}{Eq.~\eqref{eq:rect-ck} and~\eqref{eq:square-ck}}.
        By definition, $\rectC_k \leq \squareC_k$, and this becomes an equality when $\omega = 2$.
        Moreover, this is an equality when $k$ is odd as well as $k = 4$ or $6$; see~\cite{dalirrooyfard2019graph}.}
        \label{tab:intro:comparison}
    \end{table}

\subsection{Our Contributions}
In this paper, we make the following contributions: (Recall that a ``query'' refers to a Boolean conjunctive query.)
\begin{itemize}[leftmargin=*]
    \item We introduce a generalization of the submodular width of a query $Q$, called the {\em $\omega$-submodular width of $Q$}, and denoted by $\osubw(Q)$, that
    naturally incorporates the power of matrix multiplication.
    The $\omega$-submodular width is always
    upper bounded by the submodular width, and becomes identical when $\omega = 3$.
    \item We introduce a general framework to compute any query $Q$ in time $O(N^{\osubw(Q)}\cdot \log^2 N)$ for any rational\footnote{If $\omega$ is irrational, we can use any rational upper bound on $\omega$ instead.} value of $\omega$.
    Our framework unifies known combinatorial and non-combinatorial techniques under the umbrella of {\em information theory}. (The extra $\log^2 N$ factor is not needed for every query, e.g. $k$-Clique in Table\ref{tab:intro:comparison}.)
  \item We show that for any query $Q$, our framework recovers the best known complexity for $Q$
    over {\em both} combinatorial {\em and} non-combinatorial algorithms. See Table~\ref{tab:intro:comparison}.
    \item We show that there are classes of queries where our framework discovers {\em new}
    algorithms with strictly lower complexity than the best-known ones. See \colorbox{yellow}{highlighted cells} in Table~\ref{tab:intro:comparison}.
\end{itemize}

The intuition behind the connection between our framework and information theory is as follows:
Information inequalities can be used to prove an upper bound on the size of every intermediate result in the query plan. 
At the same time, any algorithm, with a proven runtime, also leads to an upper bound on the size of all intermediate results, since these cannot be larger than the runtime of the algorithm.
Hence, there is a connection between information inequalities and algorithms, namely both imply upper bounds on all intermediate results of the algorithm.  What is non-obvious is how to convert the information inequalities into an algorithm: this is what the $\panda$ algorithm did~\cite{DBLP:conf/pods/Khamis0S17,theoretics:13722}, and what we extend in our paper to handle matrix multiplication. We give an overview of how to convert information inequalities into an algorithm in Section~\ref{sec:overview}.

\subsection{Paper Outline}
In Section~\ref{sec:overview}, we present a high-level overview of our framework and illustrate it
using the triangle query $Q_\triangle$.
In Section~\ref{sec:prelims}, we give formal definitions for background concepts.
We formally define the $\omega$-submodular width in Section~\ref{sec:definition}.
We show in Section~\ref{sec:upper} the $\omega$-submodular width for several classes of queries, with proofs deferred to Appendix~\ref{app:upper}.
In Sections~\ref{sec:computing-osubw}, we give an algorithm for computing the $\omega$-submodular for any query $Q$.
Section~\ref{sec:algo} is the technical core of the paper where we present our algorithm for computing the answer to any query $Q$ in $\omega$-subdmodular width time. We conclude in Section~\ref{sec:conclusion} with some extensions and open problems.

\section{Overview}
\label{sec:overview}
We give here a simplified overview of our framework.
We start with how to generalize the subdmodular width to incorporate
MM, and to that end, we need to answer two basic questions:
\begin{itemize}[leftmargin=*]
    \item Q1: How can we express the complexity of MM in terms of polymatroids?
    \item Q2: How can we develop a notion of query plans that naturally reconciles TDs with MM?
\end{itemize}

\subsection{Q1: Translating MM complexity into polymatroids}
Given a query $Q$ of the form~\eqref{eq:bcq}, a polymatroid is a function
$\bm h:2^{\vars(Q)}\to \R_+$ that satisfies the  {\em basic Shannon inequalities}.
By that, we mean
$\bm h$ is {\em monotone} (i.e.~$h(\bm X)\leq h(\bm Y)$ for all $\bm X\subseteq \bm Y\subseteq \vars(Q)$),
{\em submodular} (i.e.~$h(\bm X)+h(\bm Y) \geq h(\bm X \cup \bm Y) + h(\bm X \cap \bm Y)$ for all $\bm X, \bm Y\subseteq\vars(Q)$), and satisfies $h(\emptyset) = 0$.
Given a query $Q$, a polymatroid $\bm h$ is {\em edge-dominated} if $h(\bm X)\leq 1$ for all $R(\bm X)\in\atoms(Q)$.
The subdmodular width, given by Eq.~\eqref{eq:intro:subw}, uses edge-dominated polymatroids as a proxy for the degree configurations of different parts of the data.

Throughout the paper, we assume that $\omega$ is a fixed constant within the range $[2, 3]$.
Given two rectangular matrices of dimensions $n^a\times n^b$ and $n^b\times n^c$, we can multiply them by partitioning them into square blocks of dimensions $n^d\times n^d$ where $d \defeq \min(a, b, c)$
and then multiplying each pair of blocks using square matrix multiplication in time $n^{d\cdot \omega}$.
This is a folklore idea, e.g.~\cite{Pan1984HowToMultiplyMatricesFaster}.
It leads to an overall runtime of $n^{\squareOmega(a, b, c)}$ where $\squareOmega(a, b, c)$ is defined below and $\gamma \defeq \omega - 2$: (Proof is in Sec.~\ref{sec:prelims}.)
\begin{align}
    \squareOmega(a, b, c) \defeq \max\{a + b + \gamma \cdot c,\quad a + \gamma \cdot b + c,\quad \gamma \cdot a + b + c\}
    \label{eq:rect-mat-mult}
\end{align}

Now suppose we have two relations $R(X, Y)$ and $S(Y, Z)$, and we want to compute
$P(X, Z) \cd$ $R(X, Y) \wedge S(Y, Z)$, by viewing $R$ and $S$ as two matrices of dimensions
$n^a\times n^b$ and $n^b \times n^c$ respectively and multiplying them.
In the polymatroid world, we can think of $h(X)$, $h(Y)$, and $h(Z)$
as representing $a$, $b$, and $c$, respectively. Motivated by this, we define the following new information measure:
\begin{align}
    \mm(X; Y; Z) \defeq \max(h(X)+h(Y)+\gamma h(Z), \;\; h(X)+\gamma h(Y)+h(Z), \;\; \gamma h(X)+h(Y)+h(Z))
    \label{eq:intro:mm}
\end{align}
And now, we can use $\mm(X; Y; Z)$ to capture the complexity of the above MM, on log-scale.
We also extend the $\mm$ notation to allow treating multiple variables as a single dimension.
For example, given the query $Q_{\triangle\triangle}$ from Eq.~\eqref{eq:intro:triangletriangle}, we use $\mm(X;Y;ZZ')$
to refer to the cost of MM where we treat $Z$ and $Z'$ as a single dimension.
In particular, we view $R(X, Y)$ as one matrix and $S(Y, Z) \wedge S'(Y, Z')$ as another matrix, and multiply them to get $P(X, Z, Z')$.

\subsection{Q2: Query plans for MM}
\label{subsubsec:intro:emm}
{\em Variable Elimination}~\cite{10.1016/S0004-3702_99_00059-4,rina-dechter-pgms,Zhang1994ASA,10.5555/1622756.1622765} is a language for expressing query plans that is known to be equivalent to TDs~\cite{DBLP:conf/pods/KhamisNR16} for combinatorial join algorithms.
We show that variable elimination can be naturally extended to incorporate MM
in its query plans. For background, given a query $Q$, {\em variable elimination} refers to the process of picking an order $\bm\sigma$ of the variables $\vars(Q)$, known as {\em variable elimination order}, or {\em VEO} for short, and then going through the variables in order and ``eliminating'' them one-by-one.
Eliminating a variable $X$ from a query $Q$ means transforming $Q$ (along with the associated input data) into an equivalent query $Q'$ (with new input data)
% \dan{The phrasing above is a bit misleading, since the transformation affects the data, not just the query}
that doesn't contain $X$, and this is done by removing all relations that contain $X$ and creating a new relation with all variables that co-occurred with $X$.
(Formal Definition will be given in Sec.~\ref{sec:prelims}.)
For example, given the 4-cycle query $Q_\square$ from Eq.~\eqref{eq:intro:4cycle}, we can eliminate $Y$ by computing a new relation
$P(X, Z) \cd R(X, Y) \wedge S(Y, Z)$ and now the remaining query becomes a triangle query:
\begin{align*}
    Q_\square'() \cd P(X, Z) \wedge T(Z, W) \wedge U(W, X)
\end{align*}
We use $U_Y^{\bm\sigma}$ to refer to the set of variables involved in the subquery that eliminates $Y$, which is $\{X, Y, Z\}$ in this example.
Every VEO is equivalent to a TD, and vice versa~\cite{DBLP:conf/pods/KhamisNR16}.
In the 4-cycle example, any VEO that eliminates either $X$ or $Z$ first is equivalent to the TD with bags $\set{Y,Z,W}$ and $\set{W, X, Y}$, while
any VEO that eliminates either $Y$ or $W$ first
is equivalent to the TD with bags $\{X, Y, Z\}$ and $\{Z, W, X\}$.

In the presence of MM, VEOs become more expressive. For example, the triangle
query $Q_\triangle$ has only one trivial TD with a single bag $\{X, Y, Z\}$.
Now, suppose we have a VEO that eliminates $Y$ first.
There are two different ways to eliminate $Y$:
\begin{itemize}[leftmargin=*]
    \item Either compute the full join combinatorially using for-loops, and then project $Y$ away. This computation
    costs $h(XYZ)$ on log-scale.
    \item Or view $R(X, Y)$ and $S(Y, Z)$ as matrices and multiply them to get $P(X, Z)$.
    This costs $\mm(X; Y; Z)$.
\end{itemize}
Alternatively, we could have eliminated either $X$ or $Z$ first.
In this simple example, there is only a single way to eliminate a variable, say $Y$, using MM, but in general, there could be several, and we can choose the best of them.
For example, in the query $Q_{\triangle\triangle}$ from Eq.~\eqref{eq:intro:triangletriangle},
$Y$ occurs in three relations, and we can arrange them into two matrices
in different ways. One way is to join $S(Y,Z)$ and $S'(Y,Z')$ into a single matrix $S''(Y,ZZ')$ and multiply
it with the matrix $R(X,Y)$ leading to a cost of $\mm(X;Y;ZZ')$.
Alternatively, we could have obtained costs $\mm(XZ;Y;Z')$ or $\mm(XZ';Y;Z)$,
and we will see later that there are even more options!~\footnote{In particular, later we will extend the notion of $\mm$ to allow for {\em group-by variables}, and we will also generalize the concept of VEOs to allow eliminating {\em multiple variables} at once.}
Given a VEO $\bm\sigma$ and a variable $Y$,
we use $\emm_Y^{\bm\sigma}$ to denote the minimum cost of eliminating $Y$ using MM.
In contrast, $h(U_Y^{\bm\sigma})$ is the cost of eliminating $Y$ using for-loops.
For example, in $Q_\triangle$, $\emm_Y^{\bm\sigma}=\mm(X; Y; Z)$ and
$h(U_Y^{\bm\sigma}) = h(XYZ)$,
whereas in $Q_{\triangle\triangle}$,
$\emm_Y^{\bm\sigma}=$ $\min(\mm(X;Y;ZZ'),\mm(XZ;Y;Z'),\mm(XZ';Y;Z),\ldots)$
and $h(U_Y^{\bm\sigma}) = h(XYZZ')$.

\subsection{Defining the $\omega$-submodular width}
Putting pieces together, we are now ready to define our notion of $\omega$-subdmodular width
of a query $Q$. We take the maximum over all ED-polymatroids $\bm h$,
and for each polymatroid, we take the minimum over all VEOs $\bm\sigma$.
For each $\bm\sigma$, we take the maximum elimination cost over all variables $X$,
where the elimination cost of $X$ is the minimum over all possible ways to eliminate $X$ using
either for-loops or MM:
\begin{align}
    \osubw(Q) \defeq
        \underbrace{\max_{\text{ED-polymatroid $\bm h$}}}_{\substack{\text{worst part}\\\text{of the data}}} \quad
        \underbrace{\min_{\text{VEO $\bm\sigma$}}}_{\substack{\text{best query plan}\\\text{for this part}}} \quad
        \underbrace{\max_{\text{variable $X$}}}_{\substack{\text{worst variable}\\\text{elimination cost}}}
        {\color{red}\min\bigl(}
        \underbrace{h(U_X^{\bm\sigma})}_{\substack{\text{cost of eliminating}\\\text{$X$ using for-loops}}}{\color{red},
        \quad
        \underbrace{\emm_X^{\bm\sigma}}_{\substack{\text{cost of eliminating}\\\text{$X$ using MM}}}
        \bigr)}
    \label{eq:intro:osubw}
\end{align}
To compare the above to the submodular width, we include below an alternative
definition of the submodular width that is equivalent to Eq.~\eqref{eq:intro:subw}.
The equivalence follows from the equivalence of VEOs and TDs~\cite{DBLP:conf/pods/KhamisNR16}:
\begin{align}
    \subw(Q) \defeq
        \underbrace{\max_{\text{ED-polymatroid $\bm h$}}}_{\substack{\text{worst part}\\\text{of the data}}} \quad
        \underbrace{\min_{\text{VEO $\bm\sigma$}}}_{\substack{\text{best query plan}\\\text{for this part}}} \quad
        \underbrace{\max_{\text{variable $X$}}}_{\substack{\text{worst variable}\\\text{elimination cost}}}\quad
        \underbrace{h(U_X^{\bm\sigma})}_{\substack{\text{cost of eliminating}\\\text{$X$ using for-loops}}}
    \label{eq:intro:subw:ve}
\end{align}
The only difference between Eq.~\eqref{eq:intro:osubw} and Eq.~\eqref{eq:intro:subw:ve} is the inclusion of $\color{red}\emm_X^{\bm\sigma}$ in the former.
This shows that $\osubw(Q)$ is always upper bounded by $\subw(Q)$. We show later that they become
identical when $\omega = 3$.

For example, $Q_\triangle$ has 6 different VEOs. For a fixed VEO, the cost of eliminating
the first variable dominates the other two, thus we ignore the two.
We have seen before that the cost of eliminating $Y$ is $\min(h(XYZ), \mm(X;Y;Z))$,
which also happens to be the cost of eliminating either $X$ or $Z$.\footnote{Note that $\mm(X;Y;Z)$ is symmetric.}
Hence, the $\omega$-subdmodular width becomes:
\begin{align}
    \osubw(Q_\triangle) = \max_{\text{ED-polymatroid $\bm h$}} {\color{red}\min\bigl(}h(XYZ){\color{red}, \mm(X;Y;Z)\bigr)}
    \label{eq:intro:osubw:triangle}
\end{align}

\subsection{Computing the $\omega$-submodular width}
\label{subsubsec:intro:computing-osubw}
Now that we have defined the $\omega$-submodular width, our next concern is how to compute it
for a given query $Q$. The $\omega$-submodular width is a deeply nested expression of min and max.
(Recall that $\emm_X^{\bm\sigma}$ is a minimum of potentially many terms of the form $\mm(X; Y; Z)$, each of which is a maximum of three terms in Eq.~\eqref{eq:intro:mm}.)
To compute $\osubw(Q)$, we first pull all $\max$ operators outside by distributing $\min$ over $\max$\footnote{Note that $\min(a, \max(b, c)) = \max(\min(a, b), \min(a, c))$}, and then swap the order of the $\max$ operators so that the max over $\bm h$ is the inner most max. Applying this to Eq.~\eqref{eq:intro:osubw:triangle}, we get:
\begin{align}
    \osubw(Q_\triangle) = \max\bigl(
        & \max_{\text{ED-polymatroid $\bm h$}}\min(h(XYZ), h(X) + h(Y) + \gamma h(Z)),\nonumber\\
        & \max_{\text{ED-polymatroid $\bm h$}}\min(h(XYZ), h(X) + \gamma h(Y) + h(Z)),\nonumber\\
        & \max_{\text{ED-polymatroid $\bm h$}}\min(h(XYZ), \gamma h(X) + h(Y) + h(Z))
        \bigr)
    \label{eq:intro:osubw:triangle:distributed}
\end{align}
%
% \dan{My own preference would be to prove the upper bound by showing Eq.~\eqref{eq:intro:shannon:triangle} here, then complementing with the lower bound below.  Then, mention the LP at the very end of the section, and remove details, just say that both the bound and the proof sequence can be derived using a linear program, in a standard way.  But I realize this is subjective, since many people want to see quickly the path to an algorithm, which is what the LP provides.}
% \mak{I don't have a strong preference. The current outline is meant to make it easier to explain the algorithm later. However, please keep in mind that inequality~\eqref{eq:intro:shannon:triangle}
% is only one of three Shannon inequalities that are needed to prove the upper bound on the $\omega$-submodular width for this query.
% To me, this was another reason to emphasize that we have three LPs.}
%
Assume that $\gamma\defeq \omega - 2$ is fixed, and
consider the first term inside the outermost max above. We can turn this term into a linear program (LP)
by introducing a new variable $t$ and replacing the $\min$ operator with a max of $t$
subject to some upper bounds on $t$:
\begin{align}
    \max_{\substack{t \in \R\\ \text{ED-polymatroid $\bm h$}}}
    \left\{
        t \quad\mid\quad
        t \leq h(XYZ), \quad t \leq h(X) + h(Y) + \gamma h(Z)
    \right\}
    \label{eq:intro:inner-lp:triangle}
\end{align}
Recall that the polymatroid $\bm h$ is a function $\bm h:2^{\set{X,Y,Z}}\to \R_+$ that satisfies the basic Shannon inequalities, or, equivalently, a vector in $\R_+^{2^3}$ subject to certain linear constraints.  Thus,  Eq.~\eqref{eq:intro:inner-lp:triangle} is linear program, and we will denote by
$\opt$ its  optimal objective value.
We will show that $\opt = \frac{2\omega}{\omega+1}$.
Since the other two LPs are similar, it follows that $\osubw(Q_\triangle)=\frac{2\omega}{\omega+1}$.

First, we show that $\opt \geq \frac{2\omega}{\omega+1}$.
To that end, consider the following polymatroid (where $h(\emptyset) = 0)$:
\begin{align*}
    h(X) = h(Y) = h(Z) = \frac{2}{\omega+1},\quad\quad
    h(XY) = h(YZ) = h(XZ) = 1,\quad\quad
    h(XYZ) = \frac{2\omega}{\omega+1}.
\end{align*}
It can be verified that this is a valid polymatroid (for any $\omega \in [2, 3]$), it is edge-dominated, and forms a feasible (primal) solution to the LP (alongside $t = \frac{2\omega}{\omega+1}$), thus proving that $\opt \geq \frac{2\omega}{\omega+1}$.
In the next section, we will show a feasible dual solution that proves $\opt \leq \frac{2\omega}{\omega+1}$.

\subsection{Computing query answers in $\omega$-subdmodular width time}
\label{subsec:intro:algo}
We now give a simplified overview of our algorithm for computing the answer to a query $Q$ in time $O(N^{\osubw(Q)}\cdot \log^2 N)$.
The extra factor of $\log^2 N$ does not manifest in every query, and it doesn't manifest in the triangle query $Q_\triangle$, which we will use as an example.
It should be noted however that this simple example of $Q_\triangle$ is not sufficient to reveal the major technical challenges of
designing the general algorithm. Many of these challenges are unique to matrix multiplication,
and are not encountered in the original $\panda$ algorithm~\cite{DBLP:conf/pods/Khamis0S17,theoretics:13722}, or its improved version~\cite{panda-express}.
We explain our general algorithm in detail in Section~\ref{sec:algo}.

A {\em Shannon inequality} is an inequality that holds over all polymatroids $\bm h$.
The following Shannon inequality corresponds to a feasible dual solution to the LP from Eq.~\eqref{eq:intro:inner-lp:triangle}:
\begin{align}
    \omega \underbrace{h(XYZ)}_{\substack{\rotatebox{90}{$\leq$}\\t \\\text{\color{red}(for-loop cost)}}} +
    \underbrace{h(X) + h(Y) + \gamma h(Z)}_{\substack{\rotatebox{90}{$\leq$}\\t \\\text{\color{red}(one term of MM cost)}}}
    \quad\leq\quad
    2 \underbrace{h(XY)}_{\substack{\rotatebox{-90}{$\leq$}\\1\\\color{red}(R(X, Y))}} +
    (\omega-1) \underbrace{h(YZ)}_{\substack{\rotatebox{-90}{$\leq$}\\1\\\color{red}(S(Y, Z))}} +
    (\omega - 1) \underbrace{h(XZ)}_{\substack{\rotatebox{-90}{$\leq$}\\1\\\color{red}(T(X, Z))}}
    \label{eq:intro:shannon:triangle}
\end{align}
In particular, the above is a Shannon inequality because it is a sum of the following submodularities:
\begin{align*}
    h(XYZ) + h(X) &\leq h(XY) + h(XZ)\\
    h(XYZ) + h(Y) &\leq h(XY) + h(YZ)\\
    \gamma h(XYZ) + \gamma h(Z) &\leq \gamma h(XZ) + \gamma h(YZ)
\end{align*}
Since $\bm h$ is edge-dominated, each term on the RHS of Eq.~\eqref{eq:intro:shannon:triangle}
is upper bounded by 1, hence the RHS is $\leq 2 \omega$.
On the other hand, by Eq.~\eqref{eq:intro:inner-lp:triangle}, the LHS is at least
$(\omega + 1)t$. Therefore, Inequality~\eqref{eq:intro:shannon:triangle} implies $t \leq \frac{2\omega}{\omega + 1}$, hence $\opt \leq\frac{2\omega}{\omega + 1}$.
Each term on the RHS of Eq.~\eqref{eq:intro:shannon:triangle} corresponds to one input relation,
whereas each {\em group of terms} on the LHS corresponds to the cost of solving a subquery in the plan. In particular, the group $h(XYZ)$ corresponds to the cost of solving the query
using for-loops, whereas the group $h(X) + h(Y) + \gamma h(Z)$ corresponds to one of three terms
that capture the cost of solving the query using MM.

First, the algorithm constructs a {\em proof sequence} of the Shannon inequality~\eqref{eq:intro:shannon:triangle}, which is a step-by-step proof of the inequality that transforms the
RHS into the LHS.
Figure~\ref{fig:ps-algo:triangle} (left) shows the proof sequence for Eq.~\eqref{eq:intro:shannon:triangle}.
Then, the algorithm translates each proof step into a corresponding database operation.
In particular, initially each term on the RHS of Eq.~\eqref{eq:intro:shannon:triangle} corresponds to an input relation.
Each time we apply a proof step replacing some terms on the RHS with some other terms,
we simultaneously apply a database operation replacing the corresponding relations with
some new relations.
Figure~\ref{fig:ps-algo:triangle} (right) shows the corresponding database operations, which
together make the algorithm for answering $Q_\triangle$.
In this example, there are only two types of proof steps:
\begin{itemize}[leftmargin=*]
    \item {\em Decomposition Step} of the form $h(XY) \to h(X) + h(Y|X)$.
    Let $R(X, Y)$ be the relation corresponding to $h(XY)$.
    The corresponding database operation is to {\em partition} $R(X, Y)$ into two parts
    based on the {\em degree} of $X$, i.e.~the number of matching $Y$-values for a given $X$.
    In particular, $X$-values with degree $> \Delta \defeq N^{\frac{\omega-1}{\omega+1}}$
    go into in the ``heavy'' part $R_h(X)$, whereas the remaining $X$-values (along with their matching $Y$-values) go into the ``light'' part $R_\ell(X, Y)$.
    Note that $|R_h|$ cannot exceed $N/\Delta = N^{\frac{2}{\omega+1}}$.
    \item {\em Submodularity Step}\footnote{Note that $h(XZ) + h(Y|X) \geq h(XYZ)$ is just another form of the submodularity $h(XZ) + h(XY) \geq h(X) + h(XYZ)$, hence the name.} of the form $h(XZ) + h(Y|X) \to h(XYZ)$. The corresponding database operation is to {\em join} the two
    corresponding relations, $T(X, Z)\Join R_\ell(X, Y)$.
    Since $R_\ell$ is the light part, this join takes time $N\cdot N^{\frac{\omega-1}{\omega+1}}=$
    $N^{\frac{2\omega}{\omega+1}}$, as desired. The same goes for the other submodularity steps.
\end{itemize}
The three submodularity steps compute three relations $Q_{\ell,1}, Q_{\ell,2}, Q_{\ell,3}$
covering triangles $(X, Y, Z)$ where either $X$, $Y$, or $Z$ is light, and three unary relations $R_h(X), S_h(Y), T_h(Z)$ containing all heavy elements.
We use the latter  to compute the triangles where all three nodes are heavy.
For that, we use $R_h(X)$, $S_h(Y)$, and $R(X,Y)$ to form a dense matrix $M_1(X,Y)$, use $S_h(Y), T_h(Z)$, and $S(Y,Z)$ to form a matrix 
$M_2(Y, Z)$, then multiply them to get $M(X, Z)$.
Since $|R_h|, |S_h|, |T_h| \leq N^{\frac{2}{\omega+1}}$, this multiplication takes time
$N^{\frac{2\omega}{\omega+1}}$, as desired.
Finally, we join $M(X, Z)$ with $T(X, Z)$ to get $Q_h(X, Z)$.
There exists a triangle if and only if either one of $Q_{\ell,1}, Q_{\ell,2}, Q_{\ell,3}$ or $Q_h$ is non-empty.
% \footnote{Recall that inequality~\eqref{eq:intro:shannon:triangle}
% comes from the optimal dual solution of the LP in Eq.~\eqref{eq:intro:inner-lp:triangle},
% which comes from the first term (out of three) inside the outer max in Eq.~\eqref{eq:intro:osubw:triangle:distributed}.
% There are two other Shannon inequalities that come from the other two terms.
% In this simple example, we are lucky enough since they lead to the same algorithm
% described above.}
%\dan{we should add a short comment on the difference between this example and PANDA}
In Section~\ref{subsubsec:algo:panda:ddr:example}, we explain the general principle that leads
to the above simple algorithm.

\begin{figure}
    \begin{tabular}{c c c | c c c}
        \multicolumn{3}{c}{Proof Sequence} & \multicolumn{3}{c}{Algorithm}\\
        $h(XY)$ & $\to$ & $h(X) + h(Y|X)$ &
            $R(X, Y)$ & $\xrightarrow{\text{partition}}$ & $R_h(X), R_\ell(X, Y)$\\
        $h(XZ) + h(Y|X)$ & $\to$ & $h(XYZ)$ &
            $T(X,Z) \Join R_\ell(X, Y)$ & $\to$ & $Q_{\ell,1}(X, Y, Z)$\\
        $h(YZ)$ & $\to$ & $h(Y) + h(Z|Y)$ &
            $S(Y, Z)$ & $\xrightarrow{\text{partition}}$ & $S_h(Y), S_\ell(Y, Z)$\\
        $h(XY) + h(Z|Y)$ & $\to$ & $h(XYZ)$ &
            $R(X,Y) \Join S_\ell(Y, Z)$ & $\to$ & $Q_{\ell,2}(X, Y, Z)$\\
        $\gamma h(XZ)$ & $\to$ & $\gamma h(Z) + \gamma h(X|Z)$ &
            $T(X, Z)$ & $\xrightarrow{\text{partition}}$ & $T_h(Z), T_\ell(Z, X)$\\
        $\gamma h(YZ) + \gamma h(X|Z)$ & $\to$ & $\gamma h(XYZ)$ &
            $S(Y, Z) \Join T_\ell(Z, X)$ & $\to$ & $Q_{\ell,3}(X, Y, Z)$\\
        \hline
        & & &  \multicolumn{3}{c}{$R_h(X) \Join S_h(Y) \Join R(X, Y) \to M_1(X, Y)$}\\
        & & &  \multicolumn{3}{c}{$S_h(Y) \Join T_h(Z) \Join S(Y, Z) \to M_2(Y, Z)$}\\
        & & &  \multicolumn{3}{c}{$M_1(X, Y) \times M_2(Y, Z) \xrightarrow{\text{MM}} M(X, Z)$}\\
        & & &  \multicolumn{3}{c}{$M(X, Z) \Join T(X, Z) \to Q_h(X, Z)$}
    \end{tabular}
    \caption{The proof sequence for the Shannon inequality~\eqref{eq:intro:shannon:triangle}
    along with the corresponding algorithm for $Q_\triangle$.}
    \label{fig:ps-algo:triangle}
\end{figure}

\section{Preliminaries}
\label{sec:prelims}

In this section, we present the formal definitions and notations for various background concepts
that were introduced informally in the introduction.
Given a number $k$, we use $[k]$ to denote the set $\{1, \ldots, k\}$.

\subsection{Hypergraphs}
A hypergraph $\calH$ is a pair $\calH = (\calV, \calE)$, where $\calV$ is a set of vertices, and $\calE \subseteq 2^{\calV}$ is a set of hyperedges. Each hyperedge $\bm Z \in \calE$
is a subset of $\calV$. We typically use $k$ to denote the number of vertices in $\calV$.
Given a hypergraph $\calH = (\calV, \calE)$ and a vertex $X \in \calV$, we define:
\begin{itemize}[leftmargin=*]
    \item $\partial_\calH(X)$ is the set of hyperedges that contain $X$, i.e.~
    $\partial_\calH(X) \defeq\setof{\bm Z \in \calE}{X \in \bm Z}$.
    \item $U_\calH(X)$ is the union of all hyperedges that contain $X$, i.e.~
    $U_\calH(X) \defeq\bigcup_{\bm Z \in \partial_\calH(X)} \bm Z$.
    \item $N_\calH(X)$ is the set of neighbors of $X$ (excluding $X$), i.e.~
    $N_\calH(X) \defeq U_\calH(X) \setminus \set{X}$.
\end{itemize}
When $\calH$ is clear from the context, we drop the subscript and simply write $\partial(X)$, $U(X)$, and $N(X)$.

\begin{example}%[$\partial_\calH(X)$, $U_\calH(X)$ and $N_\calH(X)$]
    Consider a hypergraph $\calH = (\calV, \calE)$ with vertices $\calV = \{A, B, C, D, E\}$
    and hyperedges $\calE = \{\{A, B, C\}, \{A, B, D\}, \{C, D, E\}\}$. Then:
    \begin{align*}
        \partial(A) = \{\{A, B, C\}, \{A, B, D\}\}, \quad\quad
        U(A) = \{A, B, C, D\}, \quad\quad
        N(A) = \{B, C, D\}.
        %\\ S(A) &= \{C, D\}.
    \end{align*}
\end{example}

Given a query $Q$ of the form~\eqref{eq:bcq}, the {\em hypergraph of $Q$} is a hypergraph
$\calH=(\calV, \calE)$ where $\calV \defeq\vars(Q)$ and $\calE \defeq \{\bm Z \mid R(\bm Z)\in\atoms(Q)\}$.
\underline{We often use a query $Q$ and its hypergraph $\calH$ interchangeably,} e.g. in the contexts
of tree decompositions, submodular width, etc.

%%%%%%%%%%%%%%%%%%%%%%%%%%%%%%%%%%%%%%%%%%%%%%%%%%%%%%%%%%%%%%%%%%%%%%%%%%%%%%%%%%%%%%%%%%%%

\subsection{Tree Decompositions}
    Given a hypergraph $\calH = (\calV, \calE)$ (or a query $Q$ whose hypergraph is $\calH$), a {\em tree decomposition}, or {\em TD} for short, is a pair $(T, \chi)$, where $T$ is a tree, and $\chi: \nodes(T) \to 2^\calV$ is a map
    from the nodes of $T$ to subsets of $\calV$, that satisfies the following properties:
    \begin{itemize}[leftmargin=*]
        \item For every hyperedge $\bm Z \in \calE$, there is a node $t \in \nodes(T)$ such that
        $\bm Z \subseteq \chi(t)$.
        \item For every vertex $X \in \calV$, the set $\setof{t \in \nodes(T)}{X \in \chi(t)}$ forms a connected sub-tree of $T$.
    \end{itemize}
    Each set $\chi(t)$ is called a {\em bag} of the tree decomposition.
    We use $\calT(\calH)$ to denote the set of all tree decompositions of $\calH$.

    % \dan{I made some changes below, and to Proposition~\ref{prop:td=ve}.  Please check.  The old text is commented.}

    Given two sets of sets $\calA, \calB \subseteq 2^\calV$, we write $\calA \sqsubseteq \calB$ if $\forall \bm A \in \calA, \exists \bm B \in \calB: \bm A \subseteq \bm B$.  The relation $\sqsubseteq$ is a preorder, and we write $\calA \equiv \calB$ when $\calA \sqsubseteq \calB$ and $\calB \sqsubseteq \calA$.  We identify a tree decomposition with the set of its bags, and extend $\sqsubseteq$ to a preorder on tree decompositions: $(T_1,\chi_1) \sqsubseteq (T_2,\chi_2)$ if $\setof{\chi_1(t)}{t \in \nodes(T_1)}\sqsubseteq \setof{\chi_2(t)}{t \in \nodes(T_2)}$.  The {\em trivial} tree decomposition, which consists of a single bag containing all vertices, is a maximal element of the preorder.
    A tree decomposition $(T,\chi)$ is {\em redundant} if it contains two different bags $\chi(t_1) \subsetneq \chi(t_2)$.  It is well-known that every tree decomposition is equivalent to a unique non-redundant one, obtained by removing bags that are contained in other bags.

% \dan{Do we need this definition of ``redundant''?  If we do, then I wonder whether we should also talk about  \emph{minimal} TDs, which is a TD $T$ that is minimal under $\sqsubseteq$, i.e. there is no $T_0$ s.t. $T_0 \sqsubset T$.  Obviously, if $T$ is not minimal, then we can replace it with a better $T_0$}
% \mak{Non-redundant is a tighter notion than minimal, and it is also a standard notion.
% For example, if we have a TD $(T_1, \chi_1)$ with two bags $\{A, B\}$ and $\{A\}$
% and another TD $(T_2, \chi_2)$ with two bags $\{A, B\}$, then both are minimal (note that $(T_1, \chi_1) \equiv (T_2, \chi_2)$), but only $(T_2, \chi_2)$ is non-redundant.}
    
%     A tree decomposition is called {\em trivial} if it consists of a single bag containing all vertices. A tree decomposition $(T, \chi)$ is called {\em redundant} if it contains two different bags
% $t_1\neq t_2\in\nodes(T)$ where $\chi(t_1) \subseteq \chi(t_2)$.
% It is well-known that every redundant tree decomposition can be converted into a non-redundant one by removing bags that are contained in other bags.

\begin{example}%[Tree Decompositions of a Hypergraph]
    \label{ex:4cycle:td}
    Consider the following hypergraph $\calH = (\calV, \calE)$ that represents a 4-cycle:
    \begin{align}
        \calV = \{A, B, C, D\}, \quad\quad
        \calE = \{\{A, B\}, \{B, C\}, \{C, D\}, \{D, A\}\} \label{eq:4-cycle}
    \end{align}
    This hypergraph has the following two (non-trivial and non-redundant) tree decompositions:
    \begin{itemize}[leftmargin=*]
        \item A tree decomposition $(T_1, \chi_1)$ with two nodes $t_{11}$ and $t_{12}$ corresponding to two bags,
        $\chi_1(t_{11}) = \{A, B, C\}$ and $\chi_1(t_{12}) = \{C, D, A\}$.
        \item A tree decomposition $(T_2, \chi_2)$ with two nodes $t_{21}$ and $t_{22}$ corresponding to two bags,
        $\chi_2(t_{21}) = \{B, C, D\}$ and $\chi_2(t_{22}) = \{D, A, B\}$.
    \end{itemize}
\end{example}

%%%%%%%%%%%%%%%%%%%%%%%%%%%%%%%%%%%%%%%%%%%%%%%%%%%%%%%%%%%%%%%%%%%%%%%%%%%%%%%%%%%%%%%%%%%%

\subsection{Variable Elimination}
    Let $\calH = (\calV, \calE)$ be a hypergraph, $k\defeq |\calV|$, and $\pi(\calV)$ denote the set of all permutations of $\calV$. 
%\dan{Nitpicking: the standard notation for the set of permutations is $\text{Sym}(\calV)$, or $S_n$ if $n=|\calV|$.}
Given a fixed permutation $\bm \sigma = (X_1, \ldots, X_k) \in \pi(\calV)$, we define a sequence of hypergraphs $\calH_1^{\bm\sigma}, \ldots, \calH_{k+1}^{\bm\sigma}$, called an {\em elimination hypergraph sequence}, as follows:
    $\calH_1^{\bm\sigma} \defeq \calH$, and for $i = 1, ... , k$, the hypergraph
        $\calH_{i+1}^{\bm\sigma} = (\calV_{i+1}^{\bm\sigma},\calE_{i+1}^{\bm\sigma})$
        is defined recursively in terms of the previous hypergraph
        $\calH_{i}^{\bm\sigma} = (\calV_{i}^{\bm\sigma},\calE_{i}^{\bm\sigma})$
        using:
        \begin{align*}
            \calV_{i+1}^{\bm\sigma} \defeq \calV_{i}^{\bm\sigma} \setminus \set{X_{i}},\quad\quad\quad
            \calE_{i+1}^{\bm\sigma} \defeq \calE_{i}^{\bm\sigma}
                \setminus \partial_{\calH_{i}^{\bm\sigma}}(X_{i})
                \cup \set{N_{\calH_{i}^{\bm\sigma}}(X_i)}.
        \end{align*}
In words, $\calH_{i+1}^{\bm\sigma}$ results from $\calH_{i}^{\bm\sigma}$
by removing the vertex $X_i$ and replacing all hyperedges that contain $X_i$
with a single hyperedge which is their union {\em minus} $X_i$.
We refer to the permutation $\bm\sigma$ as a {\em variable elimination order}, or {\em VEO} for short.
For convenience, for any $i\in[k]$, we define
$\partial^{\bm\sigma}_i \defeq \partial_{\calH_{i}^{\bm\sigma}}(X_i)$
and also define $U^{\bm\sigma}_i$, and $N^{\bm\sigma}_i$
analogously.

\begin{example}%[Elimination Hypergraph Sequence]
    \label{ex:4cycle:ve}
    Consider the 4-cycle hypergraph $\calH$ from Example~\ref{ex:4cycle:td} and the variable elimination order $\bm\sigma_1 = (B, C, D, A)$. This elimination order results in the following sequence of hypergraphs:
    \begin{align*}
        \calH_1^{\bm\sigma_1} &= (\{A, B, C, D\}, \{\{A, B\}, \{B, C\}, \{C, D\}, \{D, A\}\}) \\
        \calH_2^{\bm\sigma_1} &= (\{A, C, D\}, \{\{A, C\}, \{C, D\}, \{D, A\}\}) \\
        \calH_3^{\bm\sigma_1} &= (\{A, D\}, \{\{D, A\}\}) \\
        \calH_4^{\bm\sigma_1} &= (\{A\}, \{\{A\}\}) \\
        \calH_5^{\bm\sigma_1} &= (\{\}, \{\{\}\})
    \end{align*}
    In contrast, the order $\bm\sigma_2 = (A, B, C, D)$ results in the following sequence:
    \begin{align*}
        \calH_1^{\bm\sigma_2} &= (\{A, B, C, D\}, \{\{A, B\}, \{B, C\}, \{C, D\}, \{D, A\}\}) \\
        \calH_2^{\bm\sigma_2} &= (\{B, C, D\}, \{\{D, B\}, \{B, C\}, \{C, D\}\}) \\
        \calH_3^{\bm\sigma_2} &= (\{C, D\}, \{\{C, D\}\}) \\
        \calH_4^{\bm\sigma_2} &= (\{D\}, \{\{D\}\}) \\
        \calH_5^{\bm\sigma_2} &= (\{\}, \{\{\}\})
    \end{align*}
\end{example}

% \subsubsection*{Equivalence of Tree Decompositions and Variable Elimination Orders}

% It is well-known that variable elimination is an equivalent language to describe tree decompositions of a give hypergraph; see e.g.~\cite{DBLP:conf/pods/KhamisNR16}. In particular, there is a mapping between
% variable elimination orders and tree decompositions of the same hypergraph, as the following
% Proposition shows.

As with tree decompositions, we identify a VEO $\sigma$ with the set $\setof{U^\sigma_i}{i\in[k]}$
where $k \defeq |\calV|$, and we use the preorder $\sqsubseteq$ to compare tree decompositions and VEOs:

\begin{proposition}[Equivalence of TDs and VEOs~\cite{DBLP:conf/pods/KhamisNR16}]
    Given $\calH = (\calV, \calE)$:
    \begin{enumerate}
        \item For every TD $(T, \chi) \in \calT(\calH)$, there exists a VEO $\bm \sigma \in \pi(\calV)$ such that $\bm \sigma \sqsubseteq (T,\chi)$.
        \label{prop:td=>ve}
        \item For every VEO $\bm \sigma \in \pi(\calV)$, there exists a TD $(T, \chi) \in \calT(\calH)$ such that $(T,\chi) \sqsubseteq \bm \sigma$.
        \label{prop:ve=>td}
    \end{enumerate}
    \label{prop:td=ve}
\end{proposition}

% \begin{proposition}[Equivalence of TDs and VEOs~\cite{DBLP:conf/pods/KhamisNR16}]
%     Given $\calH = (\calV, \calE)$ where $k \defeq |\calV|$:
%     \begin{enumerate}
%         \item For every TD $(T, \chi) \in \calT(\calH)$, there exists a VEO $\bm \sigma \in \pi(\calV)$ that satisfies: for every
%         $i \in [k]$, there exists a node $t \in \nodes(T)$ such that
%         $U^{\bm\sigma}_i \subseteq \chi(t)$.
%         \label{prop:td=>ve}
%         \item For every VEO $\bm \sigma \in \pi(\calV)$, there exists a TD $(T, \chi) \in \calT(\calH)$ that satisfies:
%         For every node $t \in \nodes(T)$, there exists $i \in [k]$ such that
%         $\chi(t) \subseteq U^{\bm\sigma}_i$.
%         \label{prop:ve=>td}
%     \end{enumerate}
%     \label{prop:td=ve}
% \end{proposition}

\begin{example}%[Equivalence of Tree Decompositions and Variable Elimination Orders]
    \label{ex:4cycle:td=ve}
    Continuing with Example~\ref{ex:4cycle:ve}, consider the 4-cycle hypergraph in Eq.~\eqref{eq:4-cycle}. The variable elimination order $\bm\sigma_1 = (B, C, D, A)$ results in:
    \begin{align*}
        U^{\bm\sigma_1}_1 = \{A, B, C\}, \quad U^{\bm\sigma_1}_2 = \{C, D, A\},\quad
        U^{\bm\sigma_1}_3 = \{D, A\}, \quad U^{\bm\sigma_1}_4 = \{A\}.
    \end{align*}
    Note that $\bm\sigma_1 \sqsubseteq (T_1, \chi_1)$ and $(T_1, \chi_1) \sqsubseteq \bm\sigma_1$, hence $\bm\sigma_1\equiv(T_1, \chi_1)$.

    In contrast, the variable order $\bm\sigma_2 = (A, B, C, D)$ results in:
    \begin{align*}
        U^{\bm\sigma_2}_1 = \{D, A, B\}, \quad U^{\bm\sigma_2}_2 = \{B, C, D\},\quad
        U^{\bm\sigma_2}_3 = \{C, D\}, \quad U^{\bm\sigma_2}_4 = \{D\}.
    \end{align*}
    In particular, $\bm\sigma_2\equiv(T_2, \chi_2)$.
\end{example}

%%%%%%%%%%%%%%%%%%%%%%%%%%%%%%%%%%%%%%%%%%%%%%%%%%%%%%%%%%%%%%%%%%%%%%%%%%%%%%%%%%%%%%%%%%%%

\subsection{The Submodular Width}
Given a set $\calV$, a function $\bm h:2^{\calV} \to \R_+$ is called a {\em polymatroid}
if it satisfies the following properties:
\begin{align}
    h(\bm X) + h(\bm Y) &\geq h(\bm X \cup \bm Y) + h(\bm X \cap \bm Y)
        &\forall \bm X, \bm Y \subseteq \calV\quad
        &\text{(submodularity)}\label{eq:submod} \\
    h(\bm X) &\leq h(\bm Y)
        &\forall \bm X \subset \bm Y \subseteq \calV\quad
        &\text{(monotonicity)}\label{eq:monotone} \\
    h(\emptyset) &= 0
    & &\text{(strictness)}\label{eq:emptyset}
\end{align}
The above properties are also known as {\em basic Shannon inequalities}.
We use $\Gamma_{\calV}$ to denote the set of all polymatroids over $\calV$.
When $\calV$ is clear from the context, we drop $\calV$ and simply write $\Gamma$.
Given a polymatroid $\bm h$ and sets $\bm X, \bm Y , \bm Z\subseteq \calV$,
we use $\bm X\bm Y$ as a shorthand for $\bm X \cup \bm Y$, and we define:
\begin{align}
    h(\bm Y|\bm X) &\defeq h(\bm X \bm Y) - h(\bm X) \label{eq:conditional}\\
    h(\bm Y ; \bm Z |\bm X) &\defeq h(\bm X\bm Y) + h(\bm X \bm Z) - h(\bm X) - h(\bm X \bm Y\bm Z)\label{eq:sub-measure}
\end{align}
Using the above notation, we can rewrite Eq.~\eqref{eq:monotone} as $h(\bm Y|\bm X) \geq 0$,
and Eq.~\eqref{eq:submod} as $h(\bm X ; \bm Y |\bm X\cap\bm Y) \geq 0$.
A term $h(\bm Y|\bm X)$ is called {\em unconditional} if $\bm X = \emptyset$.

Given a hypergraph $\calH = (\calV, \calE)$, a function $\bm h:2^{\calV} \to \R_+$ is called {\em edge-dominated} if it satisfies $h(\bm X) \leq 1$ for all $\bm X \in \calE$.
We use $\ed_\calH$ to denote the set of all edge-dominated functions over $\calH$.
When $\calH$ is clear from the context, we drop $\calH$ and simply write $\ed$.

Given a hypergraph $\calH = (\calV, \calE)$, the {\em submodular width}~\cite{DBLP:journals/jacm/Marx13} of $\calH$ is defined as follows:
\begin{align}
    \subw(\calH) \quad\defeq\quad
        \max_{\bm h \in \Gamma \cap \ed}\quad
        \min_{(T,\chi) \in \calT(\calH)}\quad
        \max_{t \in \nodes(T)}\quad
        h(\chi(t))
        \label{eq:subw}
\end{align}
Based on Proposition~\ref{prop:td=ve} along with the fact that a polymatroid $\bm h:2^{\calV}\to\R_+$ is monotone (Eq.~\eqref{eq:monotone}), we can equivalently define the submodular width using  Eq.~\eqref{eq:intro:subw:ve}, which we repeat here:
\begin{align}
    \subw(\calH) \quad\defeq\quad
        \max_{\bm h \in \Gamma \cap \ed}\quad
        \min_{\bm\sigma \in \pi(\calV)}\quad
        \max_{i \in [k]}\quad
        h(U^{\bm\sigma}_i)
        \label{eq:subw:vo}
\end{align}
Appendix~\ref{app:computing-subw} summarizes how to compute the submodular width of a given hypergraph $\calH$.

%%%%%%%%%%%%%%%%%%%%%%%%%%%%%%%%%%%%%%%%%%%%%%%%%%%%%%%%%%%%%%%%%%%%%%%%%%%%%%%%%%%%%%%%%%%%

\subsection{Fast Matrix Multiplication}
We show here how to multiply two rectangular matrices $A$ and $B$ of dimensions $n^a\times n^b$ and $n^b\times n^c$ in the runtime $n^{\squareOmega(a, b, c)}$, as defined by Eq.~\eqref{eq:rect-mat-mult}.
% \dan{We need to cite for the source of this algorithm.}
% \mak{Done}
This is a folklore result; see e.g.~\cite{Pan1984HowToMultiplyMatricesFaster}.
Let $d = \min(a, b, c)$.
We partition $A$ into $\frac{n^a}{n^d}\times \frac{n^b}{n^d}$ blocks
and $B$ into $\frac{n^b}{n^d}\times \frac{n^c}{n^d}$ blocks.
Therefore, we have to perform $n^{a+b+c-3d}$ block multiplications, each of which takes time $O(n^{d\cdot\omega})$. The overall complexity, on $(\log n)$-scale, is
$a + b + c -(3 - \omega) \min(a, b, c)$, which gives Eq.~\eqref{eq:rect-mat-mult}.
When $\omega = 2$, this algorithm for rectangular matrix multiplication is already optimal
because the complexity from Eq.~\eqref{eq:rect-mat-mult} becomes linear in the sizes of the input and output
matrices.
However, when $\omega > 2$, there could be faster algorithms
that are not based on square matrix multiplication; see e.g.~\cite{doi:10.1137/1.9781611975031.67}.
We define $\rectOmega(a, b, c)$ as the smallest exponent for multiplying two rectangular matrices of sizes $n^a\times n^b$ and $n^b\times n^c$ within $O(n^{\rectOmega(a, b, c)})$ time.
Based on the above discussion, we have $\rectOmega(a, b, c) \leq \squareOmega(a, b, c)$
and this becomes an equality when $\omega = 2$ or when $a = b = c$.

% \mak{TODO: Find what the standard symbol for $\gamma$ is. ChatGPT suggests that it is $\alpha$
% but I am not sure about that.}
% \xiao{I think $\alpha$ has been used in other scenario. Let $\omega(a,b,c)$ to denote the exponent such that computing two rectangular
% matrices of size $n^{a} \times n^b$ and $n^b \times n^c$ can be done in $O(n^{\omega(a,b,c)})$ time. There are some important
% constants related to rectangular matrix multiplication, such as, $\alpha \le 1$ defined as the largest constant such that $\omega(1,\alpha,1) =2$, and $\mu$ is the (unique) solution to the equation $\omega(\mu, 1,1)=2\mu + 1$. Note that $\alpha=1$ if and only if $\omega =2$. The current best bound on $\alpha$ is $0.321334 < \alpha \le 1$. Note that $\mu = \frac{1}{2}$ if $\omega=2$.  The current best bound on $\mu$ is $\frac{1}{2} \le \mu < 0.527661$.}

%%%%%%%%%%%%%%%%%%%%%%%%%%%%%%%%%%%%%%%%%%%%%%%%%%%%%%%%%%%%%%%%%%%%%%%%%%%%%%%%%%%%%%%%%%%%
\subsection{Disjunctive Datalog Rules}
\label{subsec:prelims:ddr}

In order to evaluate a BCQ $Q$ (Eq.~\eqref{eq:bcq})
in submodular width time $O(N^{\subw(Q)})$,
the original $\panda$ algorithm~\cite{DBLP:conf/pods/Khamis0S17,theoretics:13722,panda-express}
reduces the problem of evaluating $Q$ into evaluating a collection of
{\em Disjunctive Datalog Rules}, which we review in this section.
To that end, we need some preliminaries.

Let $\calV$ be a set of variables. An {\em atom} is an expression of the form $R(\bm X)$,
where $R$ is a relation symbol, and $\bm X \subseteq \calV$ is a set of variables.
A {\em schema} $\Sigma$ is a set of atoms.

Given a finite domain $\dom$ and variables $\bm X \subseteq \calV$, let $\dom^{\bm X}$
denote the set of all tuples over the variables in $\bm X$ with values from $\dom$. Given a
tuple $\bm t \in \dom^{\calV}$, we use $\bm t_{\bm X}$ to denote the projection of $\bm t$
onto the variables in $\bm X$. A {\em database instance} $D$ over schema $\Sigma$ is a
mapping that assigns each atom $R(\bm X) \in \Sigma$ to a finite relation $R^{D} \subseteq
\dom^{\bm X}$. Usually the database instance is clear from the context, so we simply write
$R$ instead of $R^D$.
Given a relation $R\subseteq \dom^{\bm X}$, we use $\vars(R)$ to denote the set of variables $\bm X$.
Given a tuple $\bm t\in\dom^{\calV}$ and a relation $R(\bm X)$ for $\bm X\subseteq \calV$,
we say that $\bm t$ is {\em covered} by $R$ if $\bm t_{\bm X} \in R$.

Given a database instance over schema $\Sigma$,
the {\em full natural join} of the instance is the set of tuples
$\bm t \in \dom^{\calV}$ that satisfy all atoms in $\Sigma$:
\begin{align}
  \bigjoin \Sigma & \defeq \setof{\bm t \in \dom^{\calV}}{
    \bm t_{\bm X} \in R,
    \text{ for all } R(\bm X) \in \Sigma
  }
  \label{eq:full:join}
\end{align}
In other words, $\bigjoin\Sigma$ is the set of tuples $\bm t \in \dom^{\calV}$
that are covered by all relations in $\Sigma$.

\begin{definition}[DDR and its model]
Given two schemas $\Sigma_{\inn}$ and $\Sigma_{\outt}$, a {\em Disjunctive Datalog Rule} (DDR)
is the expression
\begin{align}
    \bigvee_{Q(\bm Z) \in \Sigma_{\outt}} Q(\bm Z)
        &\cd  \bigwedge_{R(\bm X)\in \Sigma_{\inn}} R(\bm X)
    \label{eq:ddr}
\end{align}
Given a database instance over the input schema $\Sigma_{\inn}$,
an {\em output instance} (or {\em model}) of the DDR~\eqref{eq:ddr}
is a database instance over the output schema $\Sigma_{\outt}$ such that
for every tuple $\bm t \in \bigjoin \Sigma_{\inn}$, there exists at least one atom
$Q(\bm Z) \in \Sigma_{\outt}$ such that $\bm t_{\bm Z} \in Q(\bm Z)$.
In other words, every tuple $\bm t \in \bigjoin\Sigma_\inn$ is {\em covered} by at least one relation in the output instance.
\label{defn:ddr}
\end{definition}

Notice that the output instance is not unique.  For example, consider the DDR:
\begin{align*}
  Q_1(X) \vee Q_2(Y) &\cd R(X,Y)
\end{align*}
One output instance is $Q_1 = \emptyset$, $Q_2 = \Pi_Y(R)$ (the projection of $R$ on $Y$), another instance is $Q_1 = \Pi_X(R), Q_2 = \emptyset$.  Many other instances exists: it suffices to ensure that, for every $(x,y) \in R$, at least one of $x \in Q_1$ or $y \in Q_2$ holds.

The {\em size} of an output instance to a DDR is the total number of tuples in all its relations:
\begin{align}
\norm{\Sigma_{\outt}} \defeq \sum_{Q(\bm Z) \in \Sigma_{\outt}} |Q|. \label{eqn:ddr:size}
\end{align}
% \dan{I remember that using sum instead of max leads to a lot of pain.  Should we use max instead?}
% \mak{For this paper, it doesn't matter because we didn't explicitly define an output size notion for $\omega$-DDRs.}
An output instance of a DDR is {\em minimal} if no proper subset of it is also an output instance.
Note that a {\em conjunctive query} (CQ) is a special case of DDR where the output schema
$\Sigma_{\outt}$ contains only one atom.
The answer to a CQ is its unique minimal output instance.

%%%%%%%%%%%%%%%%%%%%%%%%%%%%%%%%%%%%%%%%%%%%%%%%%%%%%%%%%%%%%%%%%%%%%%%%%%%%%%%%%%%%%%%%%%%%
\subsection{Degrees in a Relation}
\label{subsec:prelims:degrees}
The $\panda$ algorithm~\cite{DBLP:conf/pods/Khamis0S17,theoretics:13722,panda-express} relies heavily on the concept of {\em degrees} in a relation, which generalizes degrees in a graph. We review the definition below.

Consider a database instance $D$ over schema $\Sigma$.  If $R \in \Sigma$ is a relation name, and $\bm Z$ is a set of variables, then we denote by\footnote{$S \ltimes T$ denotes the {\em semi-join reduce} operator defined by $S \ltimes T \defeq \pi_{\vars(S)}(S \Join T)$.} $R^D_{|\bm Z} \defeq \dom^{\bm Z} \ltimes R^D$ the \emph{restriction} of $R^D$ to the variables $\bm Z$.  Notice that $\bm Z$ is not necessarily a subset of $\vars(R)$, in which case $R^D_{|\bm Z}$ may be an infinite relation (assuming the domain $\dom$ is infinite), but, when $\bm Z \subseteq \vars(R)$, then $R^D_{|\bm Z}=R^D$.
% Let $R \in \Sigma$ be a relation in this instance, and
Let $\bm X, \bm Y$ be two disjoint subsets of variables, and $\bm x \in \dom^{\bm X}$ be a data tuple.  We associate three quantities: the {\em degree} $\degree_{R^D}(\bm Y | \bm X = \bm x)$ in $R^D$ of $\bm x$ w.r.t.  $\bm Y$, the degree $\degree_{R^D}(\bm Y|\bm X)$ in $R^D$ of the pair $(\bm Y|\bm X)$, and the degree $\degree_{\Sigma^D}(\bm Y|\bm X)$ of $(\bm Y|\bm X)$ in the $\Sigma$-instance $\Sigma^D$ ($=D$):
\begin{align}
    \degree_{R^D}(\bm Y | \bm X = \bm x) &\defeq  |\setof{\bm y \in \dom^{\bm Y}}{ (\bm x, \bm y) \in R^D_{|\bm X \cup \bm Y} }| \label{eq:degree:x} \\
    \degree_{R^D}(\bm Y|\bm X) &\defeq \max_{\bm x \in \dom^{\bm X}} \degree_{R^D}(\bm Y|\bm X = \bm x) \label{eq:degree:delta} \\
  \degree_{\Sigma^D}(\bm Y|\bm X) &\defeq \min_{R \in \Sigma}\degree_{R^D}(\bm Y|\bm X).
  \label{eqn:degree:sigma}
\end{align}
As usual, we drop the superscript $D$ when clear from the context, and write simply $\degree_R(\bm Y|\bm X)$, $\degree_{\Sigma}(\bm Y|\bm X)$.

%%%%%%%%%%%%%%%%%%%%%%%%%%%%%%%%%%%%%%%%%%%%%%%%%%%%%%%%%%%%%%%%%%%%%%%%%%%%%%%%%%%%%%%%%%%%
\subsection{Sub-probability measures}
\label{subsec:prelims:sub-probs}
The original $\panda$ algorithm~\cite{DBLP:conf/pods/Khamis0S17,theoretics:13722}
had an extra polylog factor in the runtime, which was later removed in~\cite{panda-express} by a more refined algorithm called $\pandaexpress$.
This refined algorithm relies on the concept of sub-probability measures, which we review in
this section.
The goal of this paper is to  generalize $\pandaexpress$
to incorporate  fast matrix multiplication and achieve
the $\omega$-submodular width runtime.

Let $\Omega$ be a finite set. A function $p : 2^\Omega \to R_+$ is called a {\em measure}
on $\Omega$ if $p(\emptyset)=0$ and $p(A) = \sum_{x \in A} p(x)$
for any $A \subseteq \Omega$.
In particular, to define a measure on $\Omega$, it is sufficient to define the values
of $p(x)$ for all $x \in \Omega$.
It is called a {\em sub-probability measure} if $p(\Omega) \leq 1$,
and a {\em probability measure} if $p(\Omega) = 1$.

Let $p$ be a measure on $\Omega$. If $\Omega = \Omega_1 \times \Omega_2$ is a Cartesian
product of two domains, then the measure $p_1(x_1) \defeq \sum_{x_2 \in \Omega_2} p(x_1, x_2),
\forall x_1 \in \Omega_1$ is called the {\em marginal measure} of $p$ on $\Omega_1$. Let $A
\subseteq \Omega$ with positive measure $p(A) > 0$, then the measure $p_{|A}(B) \defeq \frac{p(B
\cap A)}{p(A)}$ for any $B \subseteq \Omega$ is called the {\em conditional measure} of $p$ on
$A$. Typically we write $p(B|A)$ instead of $p_{|A}(B)$. It is easy to see that, if $p$ is a
sub-probability measure then any of its marginals is also a sub-probability measure, and any
of its conditional measures is a probability measure.

In this paper, we will deal with domains $\Omega = \dom^{\bm X}$ where $\bm X$ is some
set of variables, and $\dom$ is a finite set. A sub-probability measure $p$ on $\dom^{\bm X}$
will be denoted by $p_{\bm X}$. The aforementioned facts are specialized as follows.
\begin{proposition}
Given a sub-probability measure $p_{\bm X \bm Y}$, the following hold:
\begin{itemize}
\item The following marginal measure is a sub-probability measure:
\begin{align}
    p_{\bm X}(\bm x) \defeq \sum_{\bm y \in \dom^{\bm Y}} p_{\bm X \bm Y}(\bm x, \bm y)
\label{eqn:marginal}
\end{align}
\item Given $\bm x \in \dom^{\bm X}$
the following conditional measure is a sub-probability measure:
\begin{align}
    p_{\bm Y | \bm X = \bm x}(\bm y) &\defeq
        \begin{cases}
            \frac{p_{\bm X \bm Y}(\bm x, \bm y)}{p_{\bm X}(\bm x)} & \text{if } p_{\bm X}(\bm x) > 0 \\
            0 & \text{otherwise}
        \end{cases}
\label{eqn:conditional}
\end{align}
\end{itemize}
\label{prop:marginal:conditional}
\end{proposition}
As is customary in probability theory, we write $p_{\bm Y | \bm X}(\bm y | \bm x)$ instead
of $p_{\bm Y | \bm X = \bm x}(\bm y)$, with the implicit understanding that there is a
measure for each $\bm x \in \dom^{\bm X}$.

\begin{definition}
Given $k$ sub-probability measures $p_1, p_2, \ldots, p_k$ on the same domain $\Omega$,
their {\em geometric mean} is a sub-probability measure $p$ defined by
$p(x) \defeq \left(\prod_{i=1}^k p_i(x)\right)^{\frac{1}{k}}, \forall x \in \Omega$.
\end{definition}

\begin{proposition}
The geometric mean of sub-probability measures is also a sub-probability measure.
\label{prop:geometric:mean}
\end{proposition}
\begin{proof}
This follows trivially from the AM-GM inequality:
\begin{align*}
    p(\Omega)
    = \sum_{x \in \Omega} \left(\prod_{i=1}^k p_i(x)\right)^{\frac{1}{k}}
    \leq \sum_{x \in \Omega} \frac{1}{k} \sum_{i=1}^k p_i(x)
    = \frac{1}{k} \sum_{i=1}^k \sum_{x \in \Omega} p_i(x)
    \leq 1
\end{align*}
\end{proof}

Given a sub-probability measure $p_{\bm Y|\bm X}$ and a tuple $\bm t \in \dom^{\calV}$
where $\bm X\cup\bm Y\subseteq \calV$, we define
\begin{align}
    p_{\bm Y|\bm X}(\bm t) \defeq p_{\bm Y|\bm X}(\bm t_{\bm Y}|\bm t_{\bm X})
\end{align}
\section{The $\omega$-Submodular Width: Formal Definition}
\label{sec:definition}

While Eq.~\eqref{eq:intro:osubw} gives a high-level sketch of our definition of the $\omega$-submodular width,
we aim here to provide the full formal definition.
To that end, we start with some auxiliary definitions.

%%%%%%%%%%%%%%%%%%%%%%%%%%%%%%%%%%%%%%%%%%%%%%%%%%%%%%%%%%%%%%%%%%%%%%%%%%%%%%%%%%%%%%%%%%%
\subsection{Generalizing variable elimination orders}
In the traditional submodular width (Eq.~\eqref{eq:intro:subw:ve} or~\eqref{eq:subw:vo}), it was sufficient to consider variable elimination orders that
eliminate only one variable at a time. However, once we allow using MM
to eliminate variables, there could be situations where eliminating multiple variables at once
using a single MM might be cheaper than eliminating the same set of variables one at a time
using multiple MMs.
For example, suppose we want to compute the query:
\[
Q(X, Z) \cd R(X, Y_1, Y_2) \wedge S(Y_1, Y_2, Z)
\]
There are at least two options to evaluate this query using MM:
\begin{itemize}[leftmargin=*]
    \item Option 1: For each value $y_1$ of $Y_1$, treat $\sigma_{Y_1 = y_1}R$ and
    $\sigma_{Y_1 = y_1}S$ as two matrices and multiply them.
    Finally, put together the resulting matrices for all $y_1$ and project $Y_1$ away.
    \item Option 2: Combine $Y_1$ and $Y_2$ into a single attribute $(Y_1Y_2)$, view $R$ and $S$
    as two matrices with dimensions $X \times (Y_1Y_2)$ and $(Y_1Y_2) \times Z$ respectively, and multiply them.
\end{itemize}
Depending on the relations $R$ and $S$, either option could be cheaper.
Motivated by this observation, our first task is to generalize
the notion of a VEO and an elimination hypergraph sequence to allow eliminating
multiple variables at once.

First, we lift the definitions of $\partial_\calH(X)$, $U_\calH(X)$, and $N_\calH(X)$ from a single variable $X$ to a set of variables $\bm X$.
Given a hypergraph $\calH = (\calV, \calE)$ and a non-empty set of variables $\bm X\subseteq\calV$, we define
$\partial_\calH(\bm X)$ as the set of hyperedges in $\calH$ that overlap with $\bm X$:
\begin{align*}
    \partial_\calH(\bm X) \defeq \setof{\bm Z \in \calE}{\bm X \cap \bm Z \neq \emptyset},\quad\quad
    U_\calH(\bm X) \defeq \bigcup_{\bm Z \in \partial_\calH(\bm X)} \bm Z,\quad\quad
    N_\calH(\bm X) \defeq U_\calH(\bm X) \setminus \bm X.
\end{align*}

\begin{definition}[Generalized Variable Elimination Order (GVEO)]
    Given a hypergraph $\calH=(\calV,\calE)$, let $\ov{\pi}(\calV)$
    denote the set of all {\em ordered partitions} of $\calV$. Namely, each element
    $\ov{\bm\sigma} \in \ov{\pi}(\calV)$ is a tuple
    $(\bm X_1, \bm X_2, \ldots, \bm X_{|\ov{\bm\sigma}|})$
    of non-empty and pairwise disjoint sets $\bm X_1, \ldots, \bm X_{|\ov{\bm\sigma}|}$
    whose union is $\calV$.
    We refer to each $\ov{\bm\sigma} \in \ov{\pi}(\calV)$ as a {\em generalized variable
    elimination order}, or {\em GVEO} for short.
    Given $\ov{\bm\sigma}\in \ov{\pi}(\calV)$, we define a {\em generalized elimination hypergraph sequence}
    $\calH^{\ov{\bm\sigma}}_1, \ldots, \calH^{\ov{\bm\sigma}}_{|\ov{\bm\sigma}|+1}$ as follows:
    $\calH_1^{\ov{\bm\sigma}} \defeq \calH$, and for $i = 1, ... , |\ov{\bm\sigma}|$:
        \[
            \calV_{i+1}^{\ov{\bm\sigma}} \defeq \calV_{i}^{\ov{\bm\sigma}} \setminus \bm X_i,
            \quad\quad\quad
            \calE_{i+1}^{\ov{\bm\sigma}} \defeq \calE_{i}^{\ov{\bm\sigma}}
                \setminus \partial_{\calH_{i}^{\ov{\bm\sigma}}}(\bm X_{i})
                \cup \set{N_{\calH_{i}^{\ov{\bm\sigma}}}(\bm X_i)}.
        \]
    For any $i \in [|\ov{\bm\sigma}|]$, we define $\partial^{\ov{\bm\sigma}}_i\defeq \partial_{\calH^{\ov{\bm\sigma}}_i}(\bm X_i)$,
    just like before, and the same goes for
    $U^{\ov{\bm\sigma}}_i$ and $N^{\ov{\bm\sigma}}_i$.
    \label{defn:gve}
\end{definition}

%%%%%%%%%%%%%%%%%%%%%%%%%%%%%%%%%%%%%%%%%%%%%%%%%%%%%%%%%%%%%%%%%%%%%%%%%%%%%%%%%%%%%%%%%%%

\subsection{Expressing MM runtime using polymatroids}
We now give the formal definition of the MM expression that generalizes
the special case $\mm(X; Y; Z)$ that was given earlier in Eq.~\eqref{eq:intro:mm}.
Recall that $\omega$ is a constant in the range $[2, 3]$ and $\gamma \defeq \omega - 2$.
\begin{definition}[Matrix multiplication expression, $\mm$]
Let $\bm h:2^{\calV}\to\R_+$ be a polymatroid.
Given four pairwise disjoint subsets, $\bm X, \bm Y, \bm Z, \bm G\subseteq \calV$, we define
the {\em matrix multiplication expression} $\mm(\bm X;\bm Y;\bm Z |\bm G)$ as follows:
\begin{align}
    \mm(\bm X;\bm Y;\bm Z |\bm G) \defeq \max\bigl(&h(\bm X | \bm G) + h(\bm Y | \bm G) + \gamma h(\bm Z | \bm G) + h(\bm G),\nonumber\\
        &h(\bm X | \bm G) + \gamma h(\bm Y | \bm G) + h(\bm Z | \bm G) + h(\bm G),\nonumber\\
        &\gamma h(\bm X | \bm G) + h(\bm Y | \bm G) + h(\bm Z | \bm G) + h(\bm G)\bigr)
        \label{eq:mm}
\end{align}
When $\bm G$ is empty, we write $\mm(\bm X;\bm Y;\bm Z)$ as a shorthand for $\mm(\bm X;\bm Y;\bm Z |\emptyset)$.
\label{defn:mm}
\end{definition}

% Note that the expression in Eq.~\eqref{eq:mm} is symmetric in $\bm X$, $\bm Y$, and $\bm Z$.
% Moreover, it mirrors the expression for rectangular matrix multiplication from Eq.~\eqref{eq:rect-mat-mult}.
% In particular, if we take $\bm G = \emptyset$, and $h(\bm X), h(\bm Y), h(\bm Z)$ to be $\{\log m, \log n$, and $\log p$, then $\mm(\bm X;\bm Y;\bm Z |\bm G)$ becomes the log of the quantity in Eq.~\eqref{eq:rect-mat-mult}.

The following proposition intuitively says that the MM runtime is at least linear
in the sizes of the two input matrices and the output matrix.
\begin{proposition}
    $\max(h(\bm X \bm Y \bm G),
        h(\bm Y \bm Z \bm G), h(\bm X \bm Z \bm G)) \quad\leq\quad
        \mm(\bm X ;\bm Y; \bm Z | \bm G)$
    \label{prop:mm:output:size}
\end{proposition}
\begin{proof}
    $
        h(\bm X \bm Y \bm G) \leq
        h(\bm X | \bm G) + h(\bm Y | \bm G) + h(\bm G)
        \leq
        h(\bm X | \bm G) +  h(\bm Y | \bm G) + \gamma h(\bm Z | \bm G) + h(\bm G).
    $
\end{proof}

\begin{proposition}
    \label{prop:mm:omega=3}
    If $\omega = 3$, then $h(\bm X\bm Y\bm Z\bm G) \leq \mm(\bm X; \bm Y; \bm Z | \bm G)$.
\end{proposition}
\begin{proof}
    When $\omega = 3$,
    $
        h(\bm X \bm Y \bm Z\bm G) \leq h(\bm X | \bm G) + h(\bm Y | \bm G) + h (\bm Z|\bm G) + h(\bm G) = \mm(\bm X; \bm Y; \bm Z | \bm G).
    $
\end{proof}

%%%%%%%%%%%%%%%%%%%%%%%%%%%%%%%%%%%%%%%%%%%%%%%%%%%%%%%%%%%%%%%%%%%%%%%%%%%%%%%%%%%%%%%%%%%

\subsection{Variable elimination via matrix multiplication}
Here we present the formal definition of the quantity $\emm$ introduced earlier in Section~\ref{subsubsec:intro:emm}, which captures the cost of eliminating a variable (or set of variables) using matrix multiplication.
Given a hypergraph $\calH=(\calV,\calE)$ and a set of vertices $\bm X \subseteq \calV$,
in order to eliminate $\bm X$ using a multiplication of two matrices,
we take the neighboring hyperedges $\partial_\calH(\bm X)$ of $\bm X$
and assign them to two (potentially overlapping) sets of hyperedges $\calA$ and $\calB$,
each of which will form a matrix.
Since every neighboring hyperedge needs to participate in this multiplication,
we need $\calA$ and $\calB$ to cover $\partial_\calH(\bm X)$, i.e. $\calA\cup\calB = \partial_\calH(\bm X)$.
Let $\bm A$ and $\bm B$ be the sets of vertices in $\calA$ and $\calB$ respectively, i.e. $\bm A = \cup \calA$ and $\bm B = \cup\calB$.
In order to  eliminate $\bm X$ using matrix multiplication, we require that $\bm X$ is a subset of $\bm A \cap \bm B$.
Additionally, we are allowed to choose a set $\bm G \subseteq \bm A \cup \bm B$ of ``group-by variables'':
These are variables that do not participate in the matrix multiplication.
Instead, for every assignment $\bm g \in \dom^{\bm G}$ of $\bm G$, we perform a matrix multiplication over the remaining variables (hence the name ``group-by variables'').
%
%\dan{I updated the text below, please check.  The old text is commented out in Latex.}
%
Given $\bm G$, denote by $\bm A^{\bm G} \defeq \bm A \setminus \bm G$, $\bm B^{\bm G} \defeq \bm B\setminus \bm G$; for the choice of $\bm G$ to be meaningful we require that $\bm X = \bm A^{\bm G} \cap \bm B^{\bm G}$.  Then, for every assignment of the variables in $\bm G$, we perform a multiplication of a matrix with dimensions $(\bm A^{\bm G} \setminus \bm X) \times \bm X$ with a matrix with dimensions $\bm X \times (\bm B^{\bm G} \setminus \bm X)$.  The cost of this multiplication is given by $\mm(\bm A^G\setminus \bm X, \bm B^G\setminus \bm X, \bm X \mid  \bm G)$; see Figure~\ref{fig:venn}.  Finally, $\emm_\calH(\bm X)$ is the minimum cost over all valid assignments of $\calA$, $\calB$, and $\bm G$.

\begin{figure}
  \centering
  \begin{tikzpicture}

    % G ellipse outline — drawn after A and B so it sits visually behind
    \draw[thick, color=gray, fill = blue!20] (0,-0.34) ellipse (3cm and 1cm);
    
    % A ellipse (top-left)
    \draw[thick, color=gray] (-1,0.3) ellipse (2.4cm and 1.8cm);
    
    % B ellipse (top-right)
    \draw[thick, color=gray] (1,0.3) ellipse (2.4cm and 1.8cm);
    
    % Labels
    % A and B: top-level set labels outside the ellipses
    \node[font=\LARGE\bfseries] at (-3, 2)  {$\bm A$};
    \node[font=\LARGE\bfseries] at ( 3, 2)  {$\bm B$};
    
    % A^G: left region where A and G overlap but not B
    \node[font=\large]          at (-1.25,1.1)  {$\bm A^{\bm G}$};
    
    % B^G: right region where B and G overlap but not A
    \node[font=\large]          at ( 1.25,1.1)  {$\bm B^{\bm G}$};
    
    % X: centre intersection of A and B (above G)
    \node[font=\large\bfseries] at (0, 1.1)     {$\bm X$};
    
    % G: label inside the shaded ellipse, lower-left area
    \node[font=\LARGE\bfseries] at (0,-0.4)  {\color{blue}$\bm G$};
    
\end{tikzpicture}
  \caption{The Venn diagram of the sets $\bm A, \bm B, \bm G, \bm X$ from Definition~\ref{defn:emm}.}
  \label{fig:venn}
\end{figure}

% For our choice of $\bm G$ to be meaningful, it has to minimally include all variables in $\bm A\cap \bm B$ that are {\em not} in $\bm X$.
% For every assignment of the variables in $\bm G$, we perform a multiplication of a matrix with dimensions $\left((\bm A\setminus \bm B) \setminus \bm G\right)\times \bm X$ with another matrix with  dimensions $\bm X\times \left((\bm B\setminus \bm A) \setminus \bm G\right)$ (thus the multiplication eliminates the common dimension $\bm X$).
% The cost of this multiplication is given by $\mm((\bm A\setminus \bm B) \setminus \bm G\;;\; (\bm B\setminus \bm A) \setminus \bm G\;;\; \bm X \;|\; \bm G)$.
% Finally, $\emm_\calH(\bm X)$ is the minimum cost over all valid assignments of $\calA$, $\calB$, and $\bm G$.
% 

\begin{definition}[Variable elimination expression via matrix multiplication, $\emm$]
    Let $\calH=(\calV,\calE)$ be a hypergraph and $\bm h:2^{\calV}\to\R_+$ be a polymatroid.
    Given a non-empty set of vertices $\bm X \in \calV$, we define the {\em variable elimination expression via matrix multiplication}, denoted by $\emm_\calH(\bm X)$, as:
    %(Recall the definition of $\partial_\calH(\bm X)$ from Eq.~\eqref{eq:partial(X):set}.)
    \begin{align}
        \emm_\calH(\bm X) \defeq \min\biggl\{&
            \mm\bigl(
                \bm A^{\bm G}\setminus \bm X\quad;\quad
                \bm B^{\bm G}\setminus \bm X\quad;\quad
                \bm X\quad|\quad
                \bm G
            \bigr) \quad\mid\nonumber\\
            &\quad \exists \calA, \calB \subseteq \partial_\calH(\bm X),\quad
            \calA\cup \calB = \partial_\calH(\bm X),\quad
            \bm A = \cup \calA,\quad \bm B = \cup \calB,\nonumber\\
            &\quad\bm G\subseteq N_{\calH}(\bm X),\quad \bm A^{\bm G}=\bm A\setminus \bm G,\quad \bm B^{\bm G} = \bm B\setminus \bm G, \quad \bm X = \bm A^{\bm G} \cap \bm B^{\bm G}
        \biggr\}
        \label{eq:emm}
    \end{align}
%     \begin{align}
%         \emm_\calH(\bm X) \defeq \min\biggl\{&
%             \mm\bigl(
%                 (\bm A\setminus \bm B) \setminus \bm G\quad;\quad
%                 (\bm B\setminus \bm A) \setminus \bm G\quad;\quad
%                 \bm X\quad|\quad
%                 \bm G
%             \bigr) \quad\mid\nonumber\\
%             &\quad \exists \calA, \calB \subseteq \partial_\calH(\bm X),\quad
%             \calA\cup \calB = \partial_\calH(\bm X),\quad
%             \bm A = \cup \calA,\quad \bm B = \cup \calB,\nonumber\\
%             &\quad \bm X \subseteq \bm A \cap \bm B,\quad
%             (\bm A \cap \bm B)\setminus \bm X \quad\subseteq\quad
%             \bm G \quad\subseteq\quad (\bm A \cup \bm B)\setminus \bm X
%         \biggr\}
%         \label{eq:emm}
%     \end{align}
    \label{defn:emm}
\end{definition}
\vspace{-.4cm}
%
%\xiao{These four parts are disjoint. This ensures that for every hypergraph $e \in \partial_H(X)$, either $e \subseteq \{(\bm A \setminus \bm B) \setminus \bm G\} \cup \{(\bm A \cap \bm B) \setminus \bm G\} \cup \bm G$ or $e \subseteq \{(\bm B \setminus \bm A) \setminus \bm G\} \cup \{(\bm A \cap \bm B) \setminus \bm G\} \cup \bm G$.}
%\xiao{The rationale behind $\left(\{\bm A \setminus \bm B) \setminus \bm G\}\right) \cap \left(\{\bm B \setminus \bm A) \setminus \bm G\}\right) = \emptyset$. 
%\begin{lemma}
%    $\mm(ab;ad;c) \ge \mm(b;d;c|a) + h(a)$.
%\end{lemma}
%\begin{proof}
%    \begin{align*}
%        \textsf{LHS}=& \max\left\{h(ab) + h(ad) +  \gamma  h(c),  h(ab) + \gamma  h(ad) +  h(c), \gamma  h(ab) + h(ad) +  h(c)\right\}\\
%        \ge & h(a) + \max\left\{h(b|a) + h(d|a) +  \gamma  h(c|a),  h(b|a) + \gamma  h(d|a) +  h(c|a), \gamma  h(b|a) + h(d|a) +  h(c|a)\right\} = \textsf{RHS} 
%    \end{align*} 
%\end{proof}
%}
We can simplify the above expression further by excluding trivial combinations where
either one of the two sets $\bm A^{\bm G}\setminus \bm X$, or
$\bm B^{\bm G}\setminus \bm X$ is empty.
% We can simplify the above expression further by excluding trivial combinations where
% either one of the two sets $(\bm A\setminus \bm B) \setminus \bm G$, or
% $(\bm B\setminus \bm A) \setminus \bm G$ is empty.

%\xiao{I think $(\bm A\cap \bm B)\setminus \bm G \neq \emptyset$ since $X \in (\bm A\cap \bm B)\setminus \bm G$. The remaining case when $(\bm A\setminus \bm B) \setminus \bm G = \emptyset$ or $(\bm B\setminus \bm A) \setminus \bm G = \emptyset$ degenerates to materializing $\bm A \cup \bm B \setminus \bm G$ into one relation. Does this already capture the join case?}

% \begin{example}
%     \label{ex:emm:triangle}
%     Consider the following hypergraph representing a triangle:
%     \begin{align}
%         \calH = (\{X, Y, Z\},\quad \{\{X, Y\}, \{X, Z\}, \{Y, Z\}\})
%         \label{eq:H:triangle}
%     \end{align}
%     In this hypergraph, $\emm_\calH(X)$ has only one combination:
%     \begin{align*}
%         \emm_{\calH}(X) = \mm(Y; Z; X)
%     \end{align*}
% \end{example}

\begin{example}
    \label{ex:emm:4clique}
    Consider the following hypergraph representing a 4-clique:
    \begin{align}
        \calH = (\{X, Y, Z, W\},\quad \{\{X, Y\}, \{X, Z\}, \{X, W\}, \{Y, Z\}, \{Y, W\}, \{Z, W\}\})
        \label{eq:H:4clique}
    \end{align}
    Here we have $\partial_\calH(X) = \{\{X, Y\}, \{X, Z\}, \{X, W\}\}$, and there are 6 different
    and non-trivial ways to assign them to $\calA$ and $\calB$, resulting in 6 different MM
    expressions:
    \begin{align}
        \emm_{\calH}(X) = \min\bigl(
        &\mm(YZ; W; X),\quad
        \mm(YW; Z; X),\quad
        \mm(ZW; Y; X), \nonumber\\
        &\mm(Y; Z; X | W),\quad
        \mm(Y; W; X | Z),\quad
        \mm(Z; W; X | Y) 
        \bigr)
        \label{eq:emm_X:4clique}
    \end{align}
    The first three expressions above correspond to cases where $\calA$ and $\calB$ are disjoint,
    while the last three correspond to cases where $\calA$ and $\calB$ overlap.
    For example, the fourth expression $\mm(Y; Z; X | W)$ results from taking $\calA =\{\{X, Y\}, \{X, W\}\}$, 
    $\calB =\{\{X, Z\}, \{X, W\}\}$, and $\bm G = \{W\}$.
\end{example}

Let $\ov{\bm\sigma}=(\bm X_1, \ldots, \bm X_{|\ov{\bm\sigma}|})$ be a generalized variable elimination order
and $\calH^{\ov{\bm\sigma}}_1, \ldots, \calH^{\ov{\bm\sigma}}_{|\ov{\bm\sigma}|+1}$ be the corresponding generalized hypergraph sequence (Definition~\ref{defn:gve}).
For any $i \in [|\ov{\bm\sigma}|]$, we use $\emm^{\ov{\bm\sigma}}_i$ to denote $\emm_{\calH^{\ov{\bm\sigma}}_i}(\bm X_i)$.

%%%%%%%%%%%%%%%%%%%%%%%%%%%%%%%%%%%%%%%%%%%%%%%%%%%%%%%%%%%%%%%%%%%%%%%%%%%%%%%%%%%%%%%%%%%

\subsection{Putting pieces together}

We now have all the components in place to formally define the $\omega$-submodular width.
\begin{definition}[$\omega$-submodular width]
    Given a hypergraph $\calH=(\calV,\calE)$, the {\em $\omega$-submodular width} of $\calH$,
    denoted by $\osubw(\calH)$, is defined as follows:
    \begin{align}
        {\color{red}\osubw}(\calH) \quad\defeq\quad
            \max_{\bm h \in \Gamma \cap \ed}\quad
            \min_{\ov{\bm\sigma} \in \ov{\pi}(\calV)}\quad
            \max_{i \in [|\ov{\bm\sigma}|]}\quad
            {\color{red}\min(}h(U^{\ov{\bm\sigma}}_i){\color{red}, \emm^{\ov{\bm\sigma}}_i)}
            \label{eq:osubw}
    \end{align}
    \label{defn:osubw}
\end{definition}
\vspace{-.4cm}

To compare the $\omega$-submodular width with the traditional submodular width,
we include below a definition of the submodular width that is equivalent to Eq.~\eqref{eq:subw:vo}, except that it uses GVEOs:

\begin{proposition}
    The following is an equivalent definition of the submodular width of $\calH$:
\begin{align}
    \subw(\calH) \quad\defeq\quad
        \max_{\bm h \in \Gamma \cap \ed}\quad
        \min_{\ov{\bm\sigma} \in \ov{\pi}(\calV)}\quad
        \max_{i \in [|\ov{\bm\sigma}|]}\quad
        h(U^{\ov{\bm\sigma}}_i)
    \label{eq:subw:gvo}
\end{align}
    \label{prop:subw:vo}
\end{proposition}
\begin{proof}
    The above definition is obviously upper bounded by Eq.~\eqref{eq:subw:vo} since
    $\ov{\pi}(\calV)$ is a superset of $\pi(\calV)$.
    In order to show the opposite direction, consider an arbitrary $\ov{\bm\sigma}=(\bm X_1, \ldots, \bm X_{|\ov{\bm\sigma}|}) \in \ov{\pi}(\calV)$.
    Consider the non-negative quantity
    $\varphi(|\ov{\bm\sigma}|) \defeq \sum_{i \in [|\ov{\bm\sigma}|]}(|\bm X_i| - 1)$.
    As long as $\ov{\bm\sigma}$ contains some $\bm X_i$ whose size it at least 2,
    this quantity will be positive, and we will construct another generalized variable elimination order $\ov{\bm\sigma'}$ that reduces this quantity and where 
    every $U_j^{\ov{\bm\sigma'}}$ is a subset of some $U_i^{\ov{\bm\sigma}}$.
    WLOG assume that $|\bm X_1| > 1$. Let $Y$ be an arbitrary element of $\bm X_1$, and define $\ov{\bm\sigma'}=(\{Y\}, \bm X\setminus \{Y\}, \bm X_2, \ldots, \bm X_{|\ov{\bm\sigma}|})$. Clearly, both $U_1^{\ov{\bm\sigma'}}$ and $U_2^{\ov{\bm\sigma'}}$
    are subsets of $U_1^{\ov{\bm\sigma}}$. Moreover,
    $\calH_3^{\ov{\bm\sigma'}}$ is identical to $\calH_2^{\ov{\bm\sigma}}$,
    hence by induction, every $U_j^{\ov{\bm\sigma'}}$ is a subset of some $U_i^{\ov{\bm\sigma}}$.
\end{proof}
%The following propositions are straightforward. Recall that $\omega \in [2, 3]$.
\begin{proposition}
    For any hypergraph $\calH$, $\osubw(\calH) \leq \subw(\calH)$.
    \label{prop:osubw<=subw}
\end{proposition}

\begin{proposition}
    If $\omega = 3$, then for any hypergraph $\calH$, $\osubw(\calH) = \subw(\calH)$.
    \label{prop:osubw=subw}
\end{proposition}
Proposition~\ref{prop:osubw<=subw} holds by comparing Eq.~\eqref{eq:osubw} to~\eqref{eq:subw:gvo}.
Proposition~\ref{prop:osubw=subw} holds because Proposition~\ref{prop:mm:omega=3} implies that
when $\omega = 3$, $h(U^{\ov{\bm\sigma}}_i) \leq \emm^{\ov{\bm\sigma}}_i$.

%%%%%%%%%%%%%%%%%%%%%%%%%%%%%%%%%%%%%%%%%%%%%%%%%%%%%%%%%%%%%%%%%%%%%%%%%%%%%%%%%%%%%%%%%%%

For the purpose of computing the $\omega$-submodular width,
we propose an equivalent form of Eq.~\eqref{eq:osubw} that is easier to compute.

\begin{proposition}
    The $\omega$-submodular width of a hypergraph can be equivalently defined as follows:
    \begin{align}
        {\osubw}(\calH) \quad\defeq\quad
            \max_{\bm h \in \Gamma \cap \ed}\quad
            \min_{\ov{\bm\sigma} \in \ov{\pi}(\calV)}\quad
            \max_{\substack{i \in [|\ov{\bm\sigma}|]\\
            {\color{red}\forall j < i,  U_i^{\ov{\bm\sigma}} \not\subseteq U_j^{\ov{\bm\sigma}}}}}\quad
            {\min(}h(U^{\ov{\bm\sigma}}_i){, \emm^{\ov{\bm\sigma}}_i)}
            \label{eq:osubw:trimmed}
    \end{align}
    \label{prop:osubw:trimmed}
\end{proposition}
\vspace{-.4cm}
The only difference between Eq.~\eqref{eq:osubw:trimmed} and Eq.~\eqref{eq:osubw} is that
$\max_{i \in [|\ov{\bm\sigma}|]}$ is more restricted in Eq.~\eqref{eq:osubw:trimmed}. In particular,
we only need to consider $i$ where $U_i^{\ov{\bm\sigma}}$ is not contained in any previous $U_j^{\ov{\bm\sigma}}$.
For example, given the hypergraph of the triangle query $Q_\triangle$
from Eq.~\eqref{eq:intro:triangle}, Eq.~\eqref{eq:osubw:trimmed} allows us to only consider
the first variable elimination step, thus simplifying the $\omega$-submodular width to become Eq.~\eqref{eq:intro:osubw:triangle}.

\begin{proof}[Proof of Proposition~\ref{prop:osubw:trimmed}]
    Consider an arbitrary polymatroid $\bm h \in \Gamma \cap \ed$ and a generalized elimination order $\ov{\bm\sigma} \in \ov{\pi}(\calV)$.
    It suffices to show that for any pair $(j,i)$ where $j < i$ and $U^{\ov{\bm \sigma}}_i \subseteq U^{\ov{\bm \sigma}}_j$,
    we must have
    $h(U^{\ov{\bm \sigma}}_i) \le \min\left\{h(U^{\ov{\bm \sigma}}_j), \emm^{\ov{\bm \sigma}}_j\right\}$. In this case, we further notice that $\bm X_j \cap U^{\ov{\bm \sigma}}_i = \emptyset$ since $j < i$ hence the variables $\bm X_j$ must have been eliminated before $U^{\ov{\bm \sigma}}_i$ was created.
    This implies that $U^{\ov{\bm \sigma}}_i \subseteq U^{\ov{\bm \sigma}}_j \setminus \bm X_j$. By the monotonicity of $\bm h$ (Eq.~\eqref{eq:monotone}), $h(U^{\ov{\bm \sigma}}_i) \le h(U^{\ov{\bm \sigma}}_j \setminus \bm X_j) \le h(U^{\ov{\bm \sigma}}_j)$.
    Moreover, by Proposition~\ref{prop:mm:output:size}, we have
    $h(U_j^{\ov{\bm\sigma}}\setminus \bm X_j) \leq \emm^{\ov{\bm\sigma}}_j$.
\end{proof}

%%%%%%%%%%%%%%%%%%%%%%%%%%%%%%%%%%%%%%%%%%%%%%%%%%%%%%%%%%%%%%%%%%%%%%%%%%%%%%%%%%%%%%%%%%%

% \begin{example}
%     Consider the triangle hypergraph from Eq.~\eqref{eq:H:triangle}.
%     We show the expression for the $\omega$-submodular width of this hypergraph using Eq.~\eqref{eq:osubw:trimmed}.
%     Note that no matter which $\ov{\bm\sigma}\in \ov{\pi}(\calV)$ we choose, we have
%     $U^{\ov{\bm\sigma}}_1 =\{X, Y, Z\} = \calV$, hence we don't need to consider
%     $U^{\ov{\bm\sigma}}_i$ for $i > 1$. Using this observation, we conclude that:
%     \begin{align}
%         \osubw(\calH) = \max_{\bm h \in \Gamma \cap \ed}
%         \min\bigl(h(XYZ),\quad \mm(X; Y; Z)\bigr)
%         \label{eq:osubw:triangle}
%     \end{align}
%     \label{ex:osubw:triangle}
% \end{example}

\begin{example}
    Let's take the 4-clique hypergraph from Eq.~\eqref{eq:H:4clique}.
    By Proposition~\ref{prop:osubw:trimmed}, for any $\ov{\bm\sigma}\in \ov{\pi}(\calV)$,
    we only need to consider
    $U^{\ov{\bm\sigma}}_1$.
    Hence, the $\omega$-submodular width becomes:
    \begin{align}
        \osubw(\calH) = \max_{\bm h \in \Gamma \cap \ed} \min\bigl(
            &h(XYZW),\quad
            \mm(XY;Z;W),\quad
            \mm(XZ;Y;W),\nonumber\\
            &\mm(XW;Y;Z),\quad
            \mm(YZ;X;W),\quad
            \mm(YW;X;Z),\nonumber\\
            &\mm(ZW;X;Y),\quad
            \mm(Y;Z;W|X),\quad
            \mm(X;Z;W|Y),\nonumber\\
            &\mm(X;Y;W|Z),\quad
            \mm(X;Y;Z|W)
            \bigr)
        \label{eq:osubw:4clique}
    \end{align}
    \label{ex:osubw:4clique}
\end{example}

\section{\bf The $\omega$-Submodular Width for Example Queries.}
\label{sec:upper}
We show 
in Table~\ref{tab:comparison}
the $\omega$-submodular width for several classes of queries
and compare it with the submodular width for each of them.
Proofs can be found in \shortorfull{the full version~\cite{full-version}}{Appendix~\ref{app:upper}}.
They use a variety of techniques to obtain both upper and lower bounds on the $\omega$-submodular width. Below are the hypergraphs for these query classes:
\begin{align}
    \text{$k$-Clique:}\quad &\calH = \left(\{X_1,X_2,\cdots,X_k\},\ \{\{X_i,X_j\}: i,j \in [k], i \neq j\}\right)\label{eq:body:k-clique}\\
    \text{$k$-Cycle:}\quad &\calH = \left(\{X_1,X_2,\cdots,X_k\},\ \{\{X_i,X_{i+1}\}: i \in [k-1]\}\cup\{\{X_k, X_1\}\}\right)\label{eq:body:k-cycle}\\
    \text{$k$-Pyramid:}\quad &\calH = \left(\{Y, X_1, X_2,\cdots, X_k\}, \ \{\{Y,X_1,\}, \{Y,X_2\}, \cdots, \{Y, X_k\}, \{X_1,X_2,\cdots,X_k\}\}\right)\label{eq:body:k-pyramid}
\end{align}
Our main result, Theorem~\ref{thm:panda:osubw}, says that any query $Q$ can be evaluated in time $O(N^{\osubw(Q)}\cdot \log^2 N)$.
Combined with Table~\ref{tab:comparison}, this implies the results in Table~\ref{tab:intro:comparison} for the corresponding queries.

       \begin{table}[ht]
        \centering
        \begin{tabular}{|c|c|c|}
        \hline
            \bf Queries \cellcolor{lgray} & \bf \cellcolor{lgray} Submodular Width & \cellcolor{lgray} \bf $\omega$-Submodular Width \\
            \hline %\hline
            \multirow{2}{*}{Triangle $Q_\triangle$~(Eq.\eqref{eq:intro:triangle})} & \multirow{2}{*}{$1.5$}  & \multirow{2}{*}{$\frac{2\omega}{\omega+1}$  \shortorfull{}{(Lemma~\ref{lem:clique-3})}}\\
            &&\\\hline
            \multirow{2}{*}{4-Clique} & \multirow{2}{*}{$2$} & \multirow{2}{*}{$\frac{\omega+1}{2}$ \shortorfull{}{(Lemma~\ref{lem:clique-4})}}\\
            &&\\
            \hline
            \multirow{2}{*}{5-Clique} & \multirow{2}{*}{$2.5$} & \multirow{2}{*}{$\frac{\omega}{2} + 1$ \shortorfull{}{(Lemma~\ref{lem:clique-5})}}\\
            &&\\
            \hline
            \multirow{2}{*}{$k$-Clique (Eq.~\eqref{eq:body:k-clique})} & \multirow{2}{*}{$\frac{k}{2}$} & $\frac{1}{2} \cdot \lceil \frac{k}{3}\rceil +  \frac{1}{2} \cdot \lceil \frac{k-1}{3}\rceil + \frac{1}{2} \cdot \lfloor \frac{k}{3}\rfloor \cdot (\omega-2)$ \\
            && \shortorfull{}{(Lemma~\ref{lem:clique})}
            \\ \hline 
            \multirow{2}{*}{4-Cycle $Q_\square$ (Eq.\eqref{eq:intro:4cycle})} & \multirow{2}{*}{$1.5$} & \multirow{2}{*}{$2 - \frac{3}{2\cdot \min\{\omega,\frac{5}{2}\}+1}$ \shortorfull{}{(Lemma~\ref{lem:4cycle})}} \\ &&\\
            \hline
            \multirow{2}{*}{$k$-Cycle (Eq.~\eqref{eq:body:k-cycle})} & \multirow{2}{*}{$2 - \frac{1}{\lceil k/2 \rceil}$} & \multirow{2}{*}{$\le \squareC_k$ \shortorfull{}{(Lemma~\ref{lem:cycle})}}\\
            &&\\\hline
            \multirow{2}{*}{$3$-Pyramid} & \multirow{2}{*}{$\frac{5}{3}$}  & \multirow{2}{*}{$2 - \frac{1}{\omega}$ \shortorfull{}{(Lemma~\ref{lmm:3-pyramid})}} \\ 
            &&\\ \hline
            \multirow{2}{*}{$k$-Pyramid (Eq.~\eqref{eq:body:k-pyramid})} & \multirow{2}{*}{$2-\frac{1}{k}$}  & \multirow{2}{*}{$\le 2 - \frac{2}{\omega \cdot (k-1) -k + 3}$ \shortorfull{}{(Lemma~\ref{lmm:k-pyramid})}}\\
            &&\\ \hline
        \end{tabular}
        \caption{Comparison of the submodular with and $\omega$-submodular width on example queries.
        Entries marked with ``$\leq$'' are upper bounds, while the rest are exact values.
        The symbol $\rectC_k$ is the best-known exponent for detecting cycles using rectangular matrix multiplication~\cite[Theorem 1.3]{dalirrooyfard2019graph},
        while $\squareC_k$ is the smallest upper bound on $\rectC_k$ that is obtained through {\em square} matrix multiplication; see \shortorfull{\cite{full-version}}{Eq.~\eqref{eq:rect-ck} and~\eqref{eq:square-ck}}.
        }
        \label{tab:comparison}
    \end{table}
       
\section{Algorithm for Computing the $\omega$-Submodular Width}
\label{sec:computing-osubw}
We present in this section an algorithm for computing the value of the $\omega$-submodular width for any given hypergraph.
This algorithm is a crucial component in our algorithm for computing the answer to a query $Q$
in $\omega$-submodular width time, as we will see in Section~\ref{sec:algo}.
The algorithm in this section can also be viewed as a generalization of the algorithm for computing the submodular width of a hypergraph, which is summarized in Appendix~\ref{app:computing-subw}.

Using Eq.~\eqref{eq:osubw:trimmed} alone, it is not immediately
clear how to compute the $\omega$-submodular width of a hypergraph $\calH$
because the outer $\max_{\bm h \in \Gamma \cap \ed}$ ranges over an infinite set.
Nevertheless, every other $\max$ and $\min$ in Eq.~\eqref{eq:osubw:trimmed} ranges over a finite set
whose cardinality only depends on $\calH$.
Note that by Eq.~\eqref{eq:emm}, $\emm^{\ov{\bm\sigma}}_i$ is a minimum of a collection of terms,
each of which has the form $\mm(\bm X; \bm Y; \bm Z | \bm G)$.
Note also that each term $\mm(\bm X; \bm Y; \bm Z | \bm G)$ is by itself a maximum of 3 terms,
as described by Eq.~\eqref{eq:mm}.

In order to compute the $\omega$-submodular width, the first step is to distribute every $\min$
over $\max$ in Eq.~\eqref{eq:osubw:trimmed}, thus pulling all $\max$ operators 
to the top level.
To formally describe this distribution, we need some notation.
Given an expression $e$ which is either a minimum or a maximum of a collection of terms,
we use $\args(e)$ to denote the collection of terms over which the minimum or maximum is taken.
Let $k$ be the number of vertices in $\calH$.
Let $f:\ov{\pi}(\calV)\to [k]$ be a function that maps every generalized
variable elimination order $\ov{\bm\sigma}\in\ov{\pi}(\calV)$ to a number $i$ between 1 and $|\ov{\bm\sigma}|$
that satisfies $U^{\ov{\bm\sigma}}_i \not\subseteq U^{\ov{\bm\sigma}}_j$ for all $j \in [i-1]$.
(Note that $|\ov{\bm\sigma}|\leq k$ for every $\ov{\bm\sigma}$.)
Let $\calF$ be the set of all such functions $f$.
For a fixed function $f \in\calF$, let
$g_f$ be a function 
that maps every pair $(\ov{\bm\sigma}, \mm(\bm X;\bm Y;\bm Z|\bm G))$
where $\ov{\bm\sigma}\in\ov{\pi}(\calV)$ and
$\mm(\bm X;\bm Y;\bm Z|\bm G)$ is a term in $\args(\emm^{\ov{\bm\sigma}}_{f(\ov{\bm\sigma})})$ to one
of the three terms in $\args(\mm(\bm X;\bm Y;\bm Z|\bm G))$ from Eq.~\eqref{eq:mm}.
Let $\calG_f$ be the set of all such functions $g_f$ for a fixed $f$.
Then, by distributing the min over the max in Eq.~\eqref{eq:osubw:trimmed}, we get:
\begin{align}
    &{\osubw}(\calH) \quad\nonumber\\
    &\defeq
        \max_{\bm h \in \Gamma \cap \ed}\quad
        \min_{\ov{\bm\sigma} \in \ov{\pi}(\calV)}\quad
        \max_{\substack{i \in [|\ov{\bm\sigma}|]\\
        {\forall j < i,  U_i^{\ov{\bm\sigma}} \not\subseteq U_j^{\ov{\bm\sigma}}}}}\quad
        \min(h(U^{\ov{\bm\sigma}}_i), \emm^{\ov{\bm\sigma}}_i)
        \nonumber\\&=
    \max_{\bm h \in \Gamma \cap \ed}\quad
        \max_{f \in \calF}\quad
        \min_{\ov{\bm\sigma} \in \ov{\pi}(\calV)}\quad
        \min\left(h(U^{\ov{\bm\sigma}}_{f(\ov{\bm\sigma})}), \emm^{\ov{\bm\sigma}}_{f(\ov{\bm\sigma})}\right)
        \nonumber\\&=
    \max_{\bm h \in \Gamma \cap \ed}\quad
        \max_{f \in \calF}\quad
        \min_{\ov{\bm\sigma} \in \ov{\pi}(\calV)}\quad
        \min\left(
            h(U^{\ov{\bm\sigma}}_{f(\ov{\bm\sigma})}),
            \min_{\mm(\bm X;\bm Y;\bm Z|\bm G) \in \args\left(\emm^{\ov{\bm\sigma}}_{f(\ov{\bm\sigma})}\right)} \mm(\bm X;\bm Y;\bm Z|\bm G)
            \right)
            \nonumber\\&=
    \max_{\bm h \in \Gamma \cap \ed}\quad
        \max_{f \in \calF}\quad
        \min_{\ov{\bm\sigma} \in \ov{\pi}(\calV)}\quad
        \min\left(
            h(U^{\ov{\bm\sigma}}_{f(\ov{\bm\sigma})}),
            \min_{\mm(\bm X;\bm Y;\bm Z|\bm G) \in \args\left(\emm^{\ov{\bm\sigma}}_{f(\ov{\bm\sigma})}\right)}
            \quad
            \max_{e \in \args(\mm(\bm X;\bm Y;\bm Z|\bm G))} e
            \right)
            \nonumber\\&=
    \max_{\bm h \in \Gamma \cap \ed}\quad
        \max_{f \in \calF}\quad
        \max_{g_f \in \calG_f}\quad
        \min_{\ov{\bm\sigma} \in \ov{\pi}(\calV)}
        %\quad\nonumber\\&\hspace{4.1cm}
        \min\left(
            h(U^{\ov{\bm\sigma}}_{f(\ov{\bm\sigma})}),
            \min_{\mm(\bm X;\bm Y;\bm Z|\bm G) \in \args\left(\emm^{\ov{\bm\sigma}}_{f(\ov{\bm\sigma})}\right)}
            g_f(\ov{\bm\sigma}, \mm(\bm X;\bm Y;\bm Z|\bm G))
            \right)
            \nonumber\\&=
    \max_{f \in \calF}\quad
    \max_{g_f \in \calG_f}\quad\nonumber\\
    &\quad\quad\underbrace{\max_{\bm h \in \Gamma \cap \ed}\quad
        \min_{\ov{\bm\sigma} \in \ov{\pi}(\calV)}\quad
        \min\left(
            h(U^{\ov{\bm\sigma}}_{f(\ov{\bm\sigma})}),
            \min_{\mm(\bm X;\bm Y;\bm Z|\bm G) \in \args\left(\emm^{\ov{\bm\sigma}}_{f(\ov{\bm\sigma})}\right)}
            g_f(\ov{\bm\sigma}, \mm(\bm X;\bm Y;\bm Z|\bm G))
            \right)}_{\text{Equivalent to an LP}}
    \label{eq:osubw:distributed:swapped:raw}
\end{align}
The above expression can be rewritten into the following form:
\begin{align}
    &\osubw(\calH) =
    \max_{i \in [\calI]}\nonumber\\
    &\quad\quad\quad\underbrace{
    \max_{\bm h \in \Gamma \cap \ed}\quad
    \min\bigl(
    \min_{\ell \in [L_i]} h(\bm U_{i \ell}),\quad
    \min_{j \in [J_i]}
    h(\bm X_{ij}|\bm G_{ij}) + h(\bm Y_{ij}|\bm G_{ij}) + \gamma h(\bm Z_{ij}|\bm G_{ij}) + h(\bm G_{ij})
    \bigr)
    }_{\text{Equivalent to LP~\eqref{eq:osubw:inner-lp} {\em assuming} $\gamma$ is fixed}}
    \label{eq:osubw:distributed:swapped}
\end{align}
In the above expression, $\calI$ is the number of $\max$ terms at the top level. For every $i \in [\calI]$,
the $i$-th term is a maximum over $\bm h \in \Gamma\cap\ed$ of a minimum of $L_i + J_i$ terms that are divided into two sets:
\begin{itemize}[leftmargin=*]
    \item $L_i$ terms of the form $h(\bm U)$ corresponding to $h(U_i^{\ov{\bm\sigma}})$ from Eq.~\eqref{eq:osubw:trimmed}.
    \item $J_i$ terms of the form $h(\bm X | \bm G) + h(\bm Y | \bm G) + \gamma h(\bm Z | \bm G) + h(\bm G)$ corresponding to terms in Eq.~\eqref{eq:mm}.
\end{itemize}

Suppose that $\omega$ (and by extension $\gamma$) is a fixed constant.
For a fixed $i \in [\calI]$, the inner expression in Eq.~\eqref{eq:osubw:distributed:swapped}
is equivalent to an LP, as we now show.
The condition $\bm h \in \Gamma\cap\ed$ is a finite collection of linear constraints.
The objective function is a minimum of $L_i+J_i$ linear functions of $\bm h$.
To convert it to a linear objective function, we introduce a fresh variable $t$
that is lower bounded by each of these $L_i+J_i$ linear functions, and we set the objective to
maximize $t$ as follows:
\begin{samepage}
\begin{align}
    \max_{t, \bm h \in \Gamma \cap \ed}
    \bigl\{
        t \quad\mid\quad &\forall \ell \in [L_i],\quad
        t\leq h(\bm U_{i\ell}),\nonumber\\
        &\forall j \in [J_i],\quad
        t \leq h(\bm X_{ij}|\bm G_{ij})+
    h(\bm Y_{ij}|\bm G_{ij}) +
    \gamma h(\bm Z_{ij}|\bm G_{ij})
    +h(\bm G_{ij})
    \bigr\}
    \label{eq:osubw:inner-lp}
\end{align}
\end{samepage}
Hence, computing the $\omega$-submodular width of a hypergraph $\calH$ reduces
to solving $\calI$ linear programs of the form in Eq.~\eqref{eq:osubw:inner-lp},
and taking the maximum of their optimal values.\footnote{If $\omega$ is not fixed, then Eq.~\eqref{eq:osubw:inner-lp} is not an LP because the coefficient $\gamma$ depends on $\omega$. Hence, we end up having constraints
that are quadratic in terms of $\omega$ and $\bm h$.}
Section~\ref{subsubsec:intro:computing-osubw} exemplifies the above using the hypergraph
of the triangle query $Q_\triangle$ (Eq.~\eqref{eq:intro:triangle}).
\label{app:computing-osubw}

\begin{example}
    Consider the 4-clique hypergraph from Eq.~\eqref{eq:H:4clique}.
    The $\omega$-submodular width of this hypergraph is given by Eq.~\eqref{eq:osubw:4clique}.
    Inside the $\min$, we have 10 different terms $\mm$, each of which is a maximum of 3 terms.
    By distributing the $\min$ over the $\max$, we get $\calI = 3^{10} = 59049$ terms.
    Our mechanical algorithm for computing the $\omega$-submodular width of this hypergraph
    consists of exhaustively solving an LP
    for each one of these $\calI$ terms and takeing their maximum optimal objective value,
    which turns out to be $\frac{\omega+1}{2}$ in this example.
    There are smarter but non-algorithmic ways to reach the same conclusion, as shown in Lemma~\ref{lem:clique-4}.
\end{example}
\section{Algorithm for Answering Queries in $\omega$-Submodular Width Time}
\label{sec:algo}

We show below our main result about evaluating a BCQ of the form~\eqref{eq:bcq} in $\omega$-submodular width time.
The runtime is measured in terms of $N$, which is the size of the input database instance, i.e.,~$N\defeq \sum_{R(\bm Z)\in\atoms(Q)} |R|$.
\begin{restatable}{theorem}{ThmMainResult}
    Assume $\omega$ is a rational number.\footnote{If $\omega$ is irrational, we can use any rational upper bound on $\omega$ instead.}
    Given a Boolean Conjunctive Query (BCQ) $Q$ and a corresponding input database instance $D$,
    \obcqsolver (Algorithm~\ref{alg:omega-panda-express:outer}) computes the answer to $Q$ over $D$ in time
    $O(N^{\osubw(Q)}\cdot \log^2 N)$ in data complexity, where
    $N$ is the size of $D$.
    \label{thm:panda:osubw}
\end{restatable}
We will see later that the extra $\log^2 N$ factor in the runtime of \obcqsolver comes from Lemma~\ref{lem:omega-join-matrix-mult}, which will be described in Section~\ref{subsec:omega-ddrs}.
We note, however, that this extra $\log^2 N$ factor does not manifest in every query. See Table~\ref{tab:intro:comparison} for examples of some queries, e.g. $k$-Clique, where
the runtime of \obcqsolver becomes tighter.

\begin{algorithm}[ht!]
    \caption{$\obcqsolver(Q, D)$}
    \begin{algorithmic}[1]
        \Statex{{\bf Inputs:} A BCQ $Q$ (Eq.~\eqref{eq:bcq}) and an input database instance $D$ for $Q$}
        \Statex{{\bf Output:} The answer to $Q$ over $D$}
        \vspace{-.25cm}\Statex\hrulefill
        \State $\ov B \gets N^{\osubw(Q)}$
        \For{$f \in \calF, g \in \calG_f$} \algocomment{Recall definitions from Section~\ref{sec:computing-osubw} and Eq.~\eqref{eq:osubw:distributed:swapped:raw}}
        \label{alg:omega-panda-express:outer:loop1}
            \State Construct an {\em $\omega$-DDR} of the form~\eqref{eq:omega-ddr:subw:raw}
            \algocomment{$\omega$-DDRs generalize DDRs. See Section~\ref{subsec:omega-ddrs}}
            \State Rename the head of the $\omega$-DDR to match the form~\eqref{eq:omega-ddr:subw}
            \State Construct the LP~\eqref{eq:osubw:inner-lp}
            \State Let $\opt$ be its optimal objective value\algocomment{$\opt \leq \osubw(Q)$ by Eq.~\eqref{eq:osubw:distributed:swapped}}
            \State Construct an {\em $\omega$-Shannon-flow inequality~\eqref{eq:omega-shannon}} from the LP~\eqref{eq:osubw:inner-lp} using Lemma~\ref{lmm:lp-to-omega-shannon}\\
            \algocomment{$\omega$-Shannon-flow inequalities generalize Shannon-flow inequalities~\eqref{eq:shannon-flow}. See Section~\ref{subsec:omega-shannon}}
            \State $B \gets N^\opt$\algocomment{$B\leq \ov B$ since $\opt\leq\osubw(Q)$}
            \State Compute a {\em $B$-probabilistic model} for the $\omega$-DDR using $\opandaexpress$ (Algorithm~\ref{alg:omega-panda-express})
            \\\algocomment{A $B$-probabilistic model is also a $\ov B$-probabilistic model
            (Definition~\ref{defn:prob:model:omega-ddr} and Proposition~\ref{prop:model:omega-ddr:monotonicity})}
            \\\algocomment{$\opandaexpress$ generalizes $\pandaexpress$~\cite{panda-express}. See Section~\ref{subsec:algo:panda:ddr}}
        \EndFor
        \For{each {\em target selector $\calT$}}\algocomment{Definition~\ref{defn:target-selector}}
        \label{alg:omega-panda-express:outer:loop2}
            \State{Use Proposition~\ref{prop:target} to find an {\em $\omega$-query plan} $(\ov{\bm\sigma}, e)$ supported by $\calT$}\algocomment{Definition~\ref{defn:omega-query-plan}}
            \\\algocomment{$\omega$-query plans generalize tree decompositions}
            \State Let $\calP_\calT$ be the {\em probabilistic interpretation} of $\calT$ \algocomment{Definition~\ref{defn:target-selector:probabilistic-interpretation}}\\
            \algocomment{$\calP_\calT$ is computed
            by collecting the probabilistic interpretations of the targets in $\calT$}
            \State Evaluate $Q_\calT$ from~\eqref{eq:bcq:target-selector} using Proposition~\ref{prop:bcq:target-selector}
            \If {$Q_\calT$ evaluates to \true}
                \State Return \true
                \algocomment{Completeness follows from Proposition~\ref{prop:target-selector:completeness}:}\\
                \algocomment{Every $\bm t\in\bigjoin\Sigma_\inn$ is covered by $\calS_\calT \subseteq \ov\calS_\calT$ for some $\calT$}
            \EndIf
        \EndFor
        \State Return \false 
    \end{algorithmic}
    \label{alg:omega-panda-express:outer}
\end{algorithm}

Our $\obcqsolver$ algorithm is given in Algorithm~\ref{alg:omega-panda-express:outer}.
We give a high-level overview of the algorithm here, and then we spend the rest of this section zooming in on one technical component at a time.

Our algorithm is a huge generalization of the original $\panda$ algorithm~\cite{DBLP:conf/pods/Khamis0S17,theoretics:13722,panda-express}, where we have to deal with many technical challenges that are not encountered in the original $\panda$.
Nevertheless, there are still some similarities between the structures of the two algorithms in hind sight.
The original $\panda$ algorithm~\cite{DBLP:conf/pods/Khamis0S17,theoretics:13722,panda-express} evaluates a BCQ $Q$ in two phases:
\begin{itemize}
    \item In the first phase, we construct and solve a collection of Disjunctive Datalog Rules, or DRRs (Section~\ref{subsec:prelims:ddr}).
    In particular, there is one DDR for every {\em ``bag selector''},
    where a bag selector is a collection of bags that is obtained by picking one bag from each tree decomposition of $Q$.
    \item In the second phase, we take the outputs of these DDRs, and we use them to construct a collection of acyclic queries. We then use the Yannakakis algorithm~\cite{DBLP:conf/vldb/Yannakakis81} to evaluate each one of these queries,
    and return the union of the answers.
\end{itemize}
In contrast, our $\obcqsolver$ algorithm also has two phases:
\begin{itemize}
    \item In the first phase (the loop in line~\ref{alg:omega-panda-express:outer:loop1} of Algorithm~\ref{alg:omega-panda-express:outer}),
    we construct a collection of {\em $\omega$-Disjunctive Datalog Rules}, or {\em $\omega$-DDRs}. These are generalizations of DDRs that are specifically tailored to deal
    with fast matrix multiplication. 
    The formal definition will be given in Section~\ref{subsec:omega-ddrs}.
    In particular, mirroring the formulation of $\osubw(Q)$ from Eq.~\eqref{eq:osubw:distributed:swapped:raw}, where there is an LP~\eqref{eq:osubw:inner-lp} for every $f\in\calF$ and $g\in\calG_f$, we construct an $\omega$-DDR for every such pair $(f,g)$.
    We will see shortly how to evaluate these $\omega$-DDRs.
    \item In the second phase (the loop in line~\ref{alg:omega-panda-express:outer:loop2} of Algorithm~\ref{alg:omega-panda-express:outer}),
    assuming we already evaluated the $\omega$-DDRs,
    we take their outputs and use them to construct the inputs to
    a collection of {\em $\omega$-query plans}, which generalize tree decompositions.
    In particular, an {\em $\omega$-query plan} is a generalized variable elimination order
    (Definition~\ref{defn:gve}) that is also amended with a specification of how each variable elimination step is executed, i.e.,~whether it is executed using a join or a fast matrix multiplication, and if so, specifically which matrix multiplication to use (Definition~\ref{defn:omega-query-plan}).
    We execute these $\omega$-query plans one by one, and return the union of the answers.
    This process is explained in detail in Section~\ref{subsec:algo:panda:osubw}.
\end{itemize}

We now explain at high level how we evaluate an $\omega$-DDR in the first phase of our $\obcqsolver$ algorithm.
To that end, we first go back to the original $\panda$ algorithm~\cite{DBLP:conf/pods/Khamis0S17,theoretics:13722,panda-express} and recap how it evaluates a DDR within the desired runtime bound.
The original $\panda$ algorithm starts from a {\em Shannon-flow inequality}~\cite{DBLP:conf/pods/Khamis0S17,theoretics:13722,panda-express}, which is a Shannon inequality of the following form:
\begin{align}
    \sum_{\ell \in [L]} \lambda_\ell h(\bm U_{\ell})
        \quad\leq\quad \sum_{i \in [I]} w_i h(\bm Y_i |\bm X_i),
    \label{eq:shannon-flow}
\end{align}
where the coefficients $\lambda_\ell$ and $w_i$ are non-negative.
$\panda$ then constructs a {\em proof sequence} for inequality~\eqref{eq:shannon-flow},
which is a sequence of steps that proves the inequality by transforming the RHS into the LHS.
Finally, $\panda$ translates each proof step
into an algorithmic operation.
Mirroring this process, in order to evaluate an $\omega$-DDR in our $\obcqsolver$ algorithm, we use the following steps:
\begin{itemize}
    \item We construct an {\em $\omega$-Shannon-flow inequality}, which is a generalization
of inequality~\eqref{eq:shannon-flow} that we will formally define in Section~\ref{subsec:omega-shannon}.
This $\omega$-Shannon-flow inequality is constructed from an optimal dual solution to the LP~\eqref{eq:osubw:inner-lp}, as we will see in Section~\ref{subsec:algo:panda:osubw}.
\item We construct a {\em proof sequence} for this $\omega$-Shannon-flow inequality, as we will describe in Section~\ref{subsec:proof-sequence}.
\item We evaluate the $\omega$-DDR by {\em translating} each step of the proof sequence into an algorithmic operation. This translation is described in Section~\ref{subsec:algo:panda:ddr}.
Some operations might exceed our time budget, in which case, we will have to skip them,
and modify the inequality and the proof sequence accordingly,
in order to obtain a new inequality/sequence {\em without} the expensive operation.
This is achieved by what we call the {\em reset lemma}, which will be described in Section~\ref{subsec:reset}.
\end{itemize}
\paragraph*{Section Overview}
This section is organized as follows.
Section~\ref{subsec:mm-upper-bounds} introduces some upper bounds on the matrix multiplication
expression, which are needed to define the class of $\omega$-Shannon-flow inequalities in Section~\ref{subsec:omega-shannon}.
Sections~\ref{subsec:reset} and~\ref{subsec:proof-sequence} present, respectively, the reset lemma
and the proof sequence construction for $\omega$-Shannon-flow inequalities.
In Section~\ref{subsec:omega-ddrs}, we introduce $\omega$-DDRs,
and we explain how to evaluate them in Section~\ref{subsec:algo:panda:ddr}.
Finally, in Section~\ref{subsec:algo:panda:osubw},
we explain how to evaluate a BCQ $Q$ in $\osubw(Q)$ time, thus proving Theorem~\ref{thm:panda:osubw}.
While most of the section uses the triangle query $Q_\triangle$ from Eq.~\eqref{eq:intro:triangle} as a running example,
Section~\ref{subsec:algo:4cycle} presents a more advanced example of the 4-cycle query $Q_\square$ from Eq.~\eqref{eq:intro:4cycle}.
% We now explain in detail the above technical components, one at a time, in the following subsections.

%%%%%%%%%%%%%%%%%%%%%%%%%%%%%%%%%%%%%%%%%%%%%%%%%%%%%%%%%%%%%%%%%%%%%%%%%%%%%%%%%%%%%%%%%%%
\subsection{Upper Bounds on the Matrix Multiplication Expression}
\label{subsec:mm-upper-bounds}
We start by introducing some upper bounds on the matrix multiplication expression from Definition~\ref{defn:mm}.
We will use these upper bounds in Section~\ref{subsec:omega-shannon}
to define our $\omega$-Shannon-flow inequalities.

As a warmup example, note that the following upper bound on $\mm(X; Y; Z)$ always holds:
\begin{align}
    \mm(X; Y; Z) \quad\leq\quad \omega \cdot\max(h(X), h(Y), h(Z))
    \label{eq:mm:upper-bound:example}
\end{align}
In particular, to prove the above inequality, assume WLOG that $h(X) \geq h(Y) \geq h(Z)$.
Then, $\mm(X; Y; Z) = h(X) + h(Y) + \gamma h(Z)$, but the latter is upper bounded by $\omega h(X)$.
Inequality~\eqref{eq:mm:upper-bound:example} is just one example of an upper bound on $\mm(X; Y; Z)$.
Here is another example that can be proved analogously:
\begin{align}
    \mm(X; Y; Z) \quad\leq\quad \frac{\omega}{2}\cdot\max( h(X) + h(Y), h(Y) + h(Z), h(X) + h(Z))
\end{align}
We introduce below a general way to derive such upper bounds.
\begin{definition}[$\omega$-dominant triple $(\alpha, \beta, \zeta)$]
    A triple of numbers $(\alpha, \beta, \zeta)$ is called {\em $\omega$-dominant} if it satisfies the following conditions:
    \begin{align}
        \alpha, \beta &\geq 1, \label{eq:omega-dominant:1}\\
        \zeta &\geq 0, \label{eq:omega-dominant:2}\\
        \alpha+\beta +\zeta &\geq \omega \label{eq:omega-dominant:3}
    \end{align}
    \label{defn:omega-dominant}
\end{definition}

\begin{proposition}[Upper bound on $\mm(\bm X;\bm Y;\bm Z |\bm G)$]
    Let the triples $(\alpha_1, \beta_1, \zeta_1)$,
    $(\alpha_2, \beta_2, \zeta_2)$, and $(\alpha_3, \beta_3, \zeta_3)$ be $\omega$-dominant.
    For any polymatroid $\bm h:2^\calV \to \R_+$ and pairwise-disjoint subsets $\bm X, \bm Y, \bm Z, \bm G \subseteq \calV$, the following inequality holds:
    \begin{align}
        \mm(\bm X;\bm Y;\bm Z |\bm G) \leq \max\bigl(
        &\alpha_1 h(\bm X | \bm G) + \beta_1 h(\bm Y | \bm G) + \zeta_1 h(\bm Z | \bm G) + h(\bm G),\nonumber\\
        &\alpha_2 h(\bm X | \bm G) + \zeta_2 h(\bm Y | \bm G) + \beta_2 h(\bm Z | \bm G) + h(\bm G),\nonumber\\
        &\zeta_3 h(\bm X | \bm G) + \alpha_3 h(\bm Y | \bm G) + \beta_3 h(\bm Z | \bm G) + h(\bm G)\bigr)
        \label{eq:mm:upper-bound}
    \end{align}
    \label{prop:mm:upper-bound}
\end{proposition}

\begin{proof}
    WLOG suppose that $h(\bm Z|\bm G)$ is the minimum among the three terms
    $h(\bm X|\bm G)$, $h(\bm Y|\bm G)$ and $h(\bm Z|\bm G)$.
    Then, $\mm(\bm X;\bm Y;\bm Z |\bm G)$ becomes identical
    to $h(\bm X|\bm G) + h(\bm Y|\bm G) + \gamma h(\bm Z|\bm G) + h(\bm G)$,
    and the following inequality holds:
    \begin{align*}
        h(\bm X|\bm G) + h(\bm Y|\bm G) + \gamma h(\bm Z|\bm G) + h(\bm G) \leq
        \alpha_1 h(\bm X | \bm G) + \beta_1 h(\bm Y | \bm G) + \zeta_1 h(\bm Z | \bm G) + h(\bm G)
    \end{align*}
\end{proof}

%\xiao{Why do we need three triples in the proposition above? From the proof, it seems that you have proved $\mm(\bm X, \bm Y , \bm Z| \bm G) \le \zeta_1 h(\bm X | \bm G) + \alpha_1 h(\bm Y | \bm G) + \beta_1 h(\bm Z | \bm G) + h(\bm G)$ holds for every $\omega$-dominant triple $(\alpha_1, \beta_1, \eta_1)$.}
%\mak{I simplified the proof following Dan's suggestion. I hope it is clearer now.}

%%%%%%%%%%%%%%%%%%%%%%%%%%%%%%%%%%%%%%%%%%%%%%%%%%%%%%%%%%%%%%%%%%%%%%%%%%%%%%%%%%%%%%%%%%%
\subsection{$\omega$-Shannon-flow Inequalities}
\label{subsec:omega-shannon}
As explained earlier, our query evaluation algorithm relies on the idea of deriving a Shannon inequality of a certain form, called {\em $\omega$-Shannon-flow inequality}, constructing a proof sequence of this inequality, and then translating each proof step
into an algorithmic operation.
In this section, we formally define the class of $\omega$-Shannon-flow inequalities.
\begin{definition}[$\omega$-Shannon-flow inequality]
    A Shannon inequality is called an {\em $\omega$-Shannon-flow inequality} if it has the following form:
    \begin{align}
        \sum_{\ell \in [L]} \lambda_\ell h(\bm U_{\ell}) +
        \sum_{j \in [J]} \left(\alpha_j h(\bm X_j|\bm G_j) +
            \beta_j h(\bm Y_j|\bm G_j) +
            \zeta_j h(\bm Z_j|\bm G_j) +
            \kappa_j h(\bm G_j)\right)
            \quad\leq\quad \sum_{i \in [I]} w_i h(\bm Y_i |\bm X_i)
        \label{eq:omega-shannon}
    \end{align}
    where
    \begin{itemize}[leftmargin=*]
        \item all sets $\bm U_\ell, \bm X_j, \bm Y_j, \bm Z_j$ are non-empty (whereas $\bm G_j$ can be empty),
        \item all coefficients $\lambda_\ell, \alpha_j, \beta_j, \zeta_j$ and $w_i$
    are non-negative,
        \item all coefficients $\kappa_j$ are positive,
        \item and all triples $\left(\frac{\alpha_j}{\kappa_j}, \frac{\beta_j}{\kappa_j}, \frac{\zeta_j}{\kappa_j}\right)$ are $\omega$-dominant.
    \end{itemize}
    \label{defn:omega-shannon}
\end{definition}
Note that $\omega$-Shannon-flow inequalities generalize Shannon-flow inequalities~\cite{DBLP:conf/pods/Khamis0S17,theoretics:13722,panda-express}, which were given in Eq.~\eqref{eq:shannon-flow}.
\begin{example}
Eq.~\eqref{eq:intro:shannon:triangle}
is an $\omega$-Shannon-flow inequality.
To be fully consistent with Eq.~\eqref{eq:omega-shannon}, we can rewrite Eq.~\eqref{eq:intro:shannon:triangle} as follows:
\begin{align}
\omega h(XYZ) + (h(X|\emptyset) + h(Y|\emptyset) + (\omega-2)h(Z|\emptyset) + h(\emptyset)) \leq 2h(XY) + (\omega-1) h(YZ) + (\omega-1)h(XZ)
\label{eq:intro:shannon:triangle:explicit}
\end{align}
\end{example}

Following~\cite{theoretics:13722}, the following proposition follows immediately from Farkas' Lemma. Recall the notation for $h(\bm Y; \bm Z|\bm X)$ from Eq.~\eqref{eq:sub-measure}.

\begin{proposition}
    Given an $\omega$-Shannon-flow inequality of the form~\eqref{eq:omega-shannon}, there must exist
    non-negative vectors $\bm m \defeq (m_p)_{p \in [P]}$ and $\bm s \defeq (s_q)_{q \in [Q]}$ such that
    the following equality is an identity over symbolic variables $h(\bm X)$ for $\bm X \subseteq \calV$ where $h(\emptyset) = 0$:
    \begin{align}
        \sum_{\ell \in [L]} \lambda_\ell h(\bm U_{\ell}) +
        \sum_{j \in [J]} \left(\alpha_j h(\bm X_j|\bm G_j) +
            \beta_j h(\bm Y_j|\bm G_j) +
            \zeta_j h(\bm Z_j|\bm G_j) +
            \kappa_j h(\bm G_j)\right)
            \quad=\quad \sum_{i \in [I]} w_i h(\bm Y_i |\bm X_i)\nonumber\\
            - \sum_{p \in [P]} m_p h(\bm Y_p |\bm X_p)
            - \sum_{q \in [Q]} s_q h(\bm Y_q;\bm Z_q |\bm X_q)
        \label{eq:omega-shannon:identity}
    \end{align}
    \label{prop:omega-shannon:identity}
\end{proposition}

\begin{example}
As an example of the above proposition, the following is the identity form~\eqref{eq:omega-shannon:identity} of the $\omega$-Shannon-flow inequality in Eq.~\eqref{eq:intro:shannon:triangle}:
\begin{align}
    \omega h(XYZ) + h(X) + h(Y) + (\omega-2)h(Z) \quad=\quad& 2h(XY) + (\omega-1) h(YZ) + (\omega-1)h(XZ)\nonumber\\
    &-h(Y;Z|X) - h(X;Z|Y) - (\omega-2)h(X;Y|Z)
    \label{eq:intro:shannon:triangle:identity}
\end{align}
By substituting $h(\bm Y;\bm Z|\bm X)$ from Eq.~\eqref{eq:sub-measure}, the above identity takes
the following form:
\begin{align}
    \omega h(XYZ) + h(X) + h(Y) + (\omega-2)h(Z) \quad=\quad& 2h(XY) + (\omega-1) h(YZ) + (\omega-1)h(XZ)\nonumber\\
    &-[h(XY) + h(XZ) - h(X) - h(XYZ)]\nonumber\\
    &-[h(XY) + h(YZ) - h(Y) - h(XYZ)]\nonumber\\
    &-(\omega-2)[h(XZ) + h(YZ) - h(Z) - h(XYZ)]
    \label{eq:intro:shannon:triangle:identity:explicit}
\end{align}
Note that Eq.~\eqref{eq:intro:shannon:triangle:identity:explicit} is indeed an identity
in the sense that every term $h(\bm W)$ for $\bm W \subseteq \{X, Y, Z\}$ cancels out on both sides of the equation.
Hence, the above equation holds {\em no matter what} values we assign to the symbolic variables $h(X), h(Y), h(Z), h(XY), h(XZ), h(YZ)$ and $h(XYZ)$ (even if these values don't represent a valid polymatroid).
\end{example}

We use the standard notation $\norm{\bm \lambda}_1$ to denote $\sum_{\ell \in [L]} |\lambda_\ell|$.
\begin{proposition}
    Given any $\omega$-Shannon-flow inequality of the form~\eqref{eq:omega-shannon}, the following inequality holds:
    \begin{align}
        \norm{\bm \lambda}_1 + \norm{\bm\kappa}_1 \leq \sum_{i \in [I] \mid \bm X_i = \emptyset} w_i
        \label{eq:omega-shannon:unconditional}
    \end{align}
    \label{prop:omega-shannon:unconditional}
\end{proposition}
\begin{example}
The $\omega$-Shannon-flow inequality in Eq.~\eqref{eq:intro:shannon:triangle:explicit} has $\norm{\bm \lambda}_1 + \norm{\bm\kappa}_1 = \omega + 1$,
while the RHS of Eq.~\eqref{eq:omega-shannon:unconditional} is $2\omega$.
Hence, inequality~\eqref{eq:omega-shannon:unconditional} becomes $\omega+1 \leq 2\omega$,
which holds since $\omega \in [2, 3]$.
\end{example}
\begin{proof}[Proof of Proposition~\ref{prop:omega-shannon:unconditional}]
Every $\omega$-Shannon-flow inequality~\eqref{eq:omega-shannon} must hold for the following polymatroid:
\begin{align*}
    h(\bm W) = \begin{cases}
        0 &\text{if $\bm W = \emptyset$}\\
        1 &\text{otherwise}
    \end{cases}
\end{align*}
When evaluated over the above polymatroid, the RHS of identity~\eqref{eq:omega-shannon:identity}
is at most $\sum_{i \in [I] \mid \bm X_i = \emptyset} w_i$,
while the LHS is at least $\norm{\bm \lambda}_1 + \norm{\bm\kappa}_1$.
Note that for every $j \in [J]$, we must have $\alpha_j \geq \kappa_j$
because the triple $\left(\frac{\alpha_j}{\kappa_j}, \frac{\beta_j}{\kappa_j}, \frac{\zeta_j}{\kappa_j}\right)$ is $\omega$-dominant.
Hence, when $\bm G_j = \emptyset$,
the term $\alpha_j h(\bm X_j|\bm G_j)$ contributes at least $\kappa_j$ to the LHS of~\eqref{eq:omega-shannon:identity}.
\end{proof}

\begin{definition}[An {\em integral} $\omega$-Shannon-flow inequality]
    An $\omega$-Shannon-flow inequality of the form~\eqref{eq:omega-shannon} is called {\em integral} if
    \begin{itemize}[leftmargin=*]
        \item all coefficients $\lambda_\ell, \alpha_j, \beta_j, \zeta_j, \kappa_j$ and $w_i$
        are integers,
        \item and there exist non-negative
    {\em integers} $(m_p)_{p \in [P]}$ and $(s_q)_{q \in [Q]}$ such that the identity in Eq.~\eqref{eq:omega-shannon:identity} holds.
    \end{itemize}
    \label{defn:integral-omega-shannon}
\end{definition}
The $\omega$-Shannon-flow inequality in Eq.~\eqref{eq:intro:shannon:triangle} is an example of an integral $\omega$-Shannon-flow inequality.

%%%%%%%%%%%%%%%%%%%%%%%%%%%%%%%%%%%%%%%%%%%%%%%%%%%%%%%%%%%%%%%%%%%%%%%%%%%%%%%%%%%%%%%%%%%

\subsection{The Reset Lemma}
\label{subsec:reset}

Our query evaluation algorithm relies on translating a proof sequence of an $\omega$-Shannon-flow inequality
into a sequence of algorithmic operations, as exemplified in Figure~\ref{fig:ps-algo:triangle}.
However, some operations might be too expensive to fit within our time budget, e.g.,~performing the join $T(X, Z)\Join R(X, Y)$ over {\em heavy} values $X$, which takes time $O(N^2)$ exceeding our budget of $O(N^{3/2})$ in this example.
Whenever we encounter such an expensive operation, we have to {\em give up} on performing it {\em while} guaranteeing that the skipped tuples will still be produced later by some future operation, e.g., by the matrix multiplication $M_1(X, Y)\times M_2(Y, Z)$ in Figure~\ref{fig:ps-algo:triangle}.
In the language of $\omega$-Shannon-flow inequalities, this ``giving up'' corresponds to dropping some terms $h(XZ) + h(Y|X)$ on the RHS of the inequality corresponding to the expensive operation
$T(X, Z)\Join R(X, Y)$.
However, when we drop those terms from the RHS, we have to fix the inequality {\em without} losing too many terms on the LHS, in order to {\em guarantee}
that the skipped tuples will still be produced later by some future operation.
This idea of dropping terms and repairing the resulting Shannon-flow inequality goes back to what was known as the {\em reset lemma} in the original $\panda$ algorithm~\cite{DBLP:conf/pods/Khamis0S17,theoretics:13722}.

We present here a highly non-trivial generalization of the reset lemma 
from Shannon-flow inequalities~\eqref{eq:shannon-flow} to $\omega$-Shannon-flow inequalities~\eqref{eq:omega-shannon}.
In particular, given an integral $\omega$-Shannon-flow inequality~\eqref{eq:omega-shannon},
this lemma allows us to sacrifice any unconditional term $h(\bm Y_i|\emptyset)$ on the RHS
while losing at most one unit from the quantity $\norm{\bm \lambda}_1 + \norm{\bm\kappa}_1$ on the LHS.
Unlike the original Reset Lemma~\cite{DBLP:conf/pods/Khamis0S17,theoretics:13722} which was restricted to Shannon-flow inequalities~\eqref{eq:shannon-flow}, 
our new lemma has to support proper conditionals $h(\bm X_j|\bm G_j)$ on the LHS of the $\omega$-Shannon-flow inequality.
Moreover, this lemma has to maintain the $\omega$-dominance property of the triples
$\left(\frac{\alpha_j}{\kappa_j}, \frac{\beta_j}{\kappa_j}, \frac{\zeta_j}{\kappa_j}\right)$,
which will be crucial for our query evaluation algorithm.

\begin{lemma}[Generalized Reset Lemma]
Given an integral $\omega$-Shannon-flow inequality of the form~\eqref{eq:omega-shannon}, suppose that
for some $i_0 \in [I]$, we have $\bm X_{i_0} = \emptyset$ and $w_{i_0}>0$. Then, there exist
non-negative integers $\lambda_\ell', \alpha_j', \beta_j', \zeta_j', \kappa_j', w_i'$
for which the following is also an integral $\omega$-Shannon-flow inequality:
\begin{align*}
    \sum_{\ell \in [L]} \lambda_\ell' h(\bm U_{\ell}) +
    \sum_{j \in [J]} \left(\alpha_j' h(\bm X_j|\bm G_j) +
        \beta_j' h(\bm Y_j|\bm G_j) +
        \zeta_j' h(\bm Z_j|\bm G_j) +
        \kappa_j' h(\bm G_j)\right)
        \quad\leq\quad \sum_{i \in [I]} w_i' h(\bm Y_i |\bm X_i)
\end{align*}
and the coefficients satisfy the following conditions:
\begin{align}
    w_{i_0}' &\leq w_{i_0}-1,\nonumber\\
    w_{i}' &\leq w_{i} &\forall i \in [I]\setminus\{i_0\},\nonumber\\
    \norm{\bm \lambda'}_1 + \norm{\bm\kappa'}_1 &\geq 
    \norm{\bm \lambda}_1 + \norm{\bm\kappa}_1-1\label{eq:reset:mass}
\end{align}
\label{lmm:reset}
\end{lemma}
Before proving the lemma, we give some examples.
\begin{example}
    Consider the $\omega$-Shannon-flow inequality in Eq.~\eqref{eq:intro:shannon:triangle} (or equivalently Eq.~\eqref{eq:intro:shannon:triangle:explicit}).
    Suppose we want to drop one unit of $h(XY)$ on the RHS of the inequality.
    The lemma says that it suffices to sacrifice at most one unit from the quantity $\norm{\bm \lambda}_1 + \norm{\bm\kappa}_1$ on the LHS.
    In this case, we can sacrifice one unit of $h(XYZ)$ on the LHS, to obtain a valid $\omega$-Shannon-flow inequality:
    \begin{align}
        (\omega-1) h(XYZ) + h(X) + h(Y) + (\omega-2)h(Z) \quad\leq\quad& h(XY) + (\omega-1) h(YZ) + (\omega-1)h(XZ)
    \end{align}
    The above is indeed a valid $\omega$-Shannon-flow inequality since it can be rewritten as the following identity of the form~\eqref{eq:omega-shannon:identity}:
    \begin{align}
        (\omega-1) h(XYZ) + h(X) + h(Y) + (\omega-2)h(Z) \quad=\quad& h(XY) + (\omega-1) h(YZ) + (\omega-1)h(XZ)\nonumber\\
        &-h(Z|X)-h(X;Z|Y) - (\omega-2)h(X;Y|Z)
    \end{align}
\end{example}
\begin{example}
    Consider the following $\omega$-Shannon-flow inequality:
    \begin{align}
        h(X|\emptyset) + h(Y|\emptyset) + \gamma h(Z|\emptyset) + h(\emptyset) \quad\leq\quad h(X) + h(Y) + \gamma h(Z)
    \end{align}
    Recall that $\gamma\defeq \omega - 2$, hence $\gamma \in [0, 1]$.
    Suppose we want to drop $\gamma h(Z)$ from the RHS.
    We can only afford to reduce $\norm{\bm \lambda}_1 + \norm{\bm\kappa}_1$ by at most $\gamma$.
    If $\gamma = 1$, then we can simply replace the original inequality with the trivial inequality $0\leq 0$.
    Otherwise, we replace the original inequality with the following one:
    \begin{align}
        h(X|\emptyset) + h(Y|\emptyset) + 0\cdot h(Z|\emptyset) + (1-\gamma) h(\emptyset)
        \quad\leq\quad h(X) + h(Y)
    \end{align}
    The above is indeed a valid $\omega$-Shannon-flow inequality.
    In particular, note that the triple $(\frac{1}{1-\gamma}, \frac{1}{1-\gamma}, 0)$ is $\omega$-dominant, namely satisfies Eq.~\eqref{eq:omega-dominant:1}-\eqref{eq:omega-dominant:3}.
    Eq.~\eqref{eq:omega-dominant:3} holds because \[\frac{2}{1-\gamma} = 2 + \frac{2\gamma}{1-\gamma}\geq 2 + 2\gamma \geq 2 + \gamma = \omega.\]
\end{example}
\begin{proof}[Proof of Lemma~\ref{lmm:reset}]
    Let $\bm W \defeq \bm Y_{i_0}$.
    We recognize the following cases:
    \begin{enumerate}[leftmargin=*]
        \item If $\bm W = \bm U_\ell$ for some $\ell \in [L]$ where $\lambda_\ell > 0$, we set $\lambda_\ell' \defeq \lambda_\ell - 1$ and $w_{i_0}'\defeq w_{i_0}-1$. All other coefficients $\lambda_\ell', \alpha_j', \beta_j', \zeta_j', \kappa_j', w_i'$ remain the same as the original $\lambda_\ell, \alpha_j, \beta_j, \zeta_j, \kappa_j, w_i$.
        The new inequality remains a valid integral $\omega$-Shannon-flow inequality since we dropped
        $h(\bm W|\emptyset)$ from both sides.
        
        \item If $\bm W = \bm G_j\bm X_j$ for some $j \in [J]$ where $\alpha_j, \kappa_j>0$, then we set
        $\alpha_j' \defeq \alpha_j-1$, $\kappa_j' \defeq \kappa_j - 1$,
        and $w_{i_0}'\defeq w_{i_0}-1$.
        The inequality still holds since we dropped $h(\bm W|\emptyset)$ from both sides.
        If $\kappa_j' = 0$, then we drop $j$ from $[J]$.
        Otherwise, the triple $\left(\frac{\alpha_j'}{\kappa_j'}, \frac{\beta_j}{\kappa_j'}, \frac{\zeta_j}{\kappa_j'}\right)$ is still $\omega$-dominant (Definition~\ref{defn:omega-dominant}).
        In particular, inequality~\eqref{eq:omega-dominant:1} continues to hold because:
        \begin{align*}
            \frac{\alpha_j'}{\kappa_j'} &= \frac{\alpha_j-1}{\kappa_j-1} \geq
            \frac{\kappa_j-1}{\kappa_j-1}=1.
        \end{align*}
        Moreover, inequality~\eqref{eq:omega-dominant:3} continues to hold
        because of the following:
        \begin{align}
            \frac{\alpha_j'+\beta_j+\zeta_j}{\kappa_j'} &=
            \frac{\alpha_j+\beta_j+\zeta_j-1}{\kappa_j-1}\nonumber\\
            &\geq \frac{\kappa_j\omega-1}{\kappa_j-1} &\text{(Inductively by Eq.~\eqref{eq:omega-dominant:3})}\nonumber\\
            &= \omega + \underbrace{\frac{\omega-1}{\kappa_j-1}}_{\geq 0}\label{eq:omega-dominant:preserve}
        \end{align}
    
        \item If $\bm W = \bm G_j \bm Y_j$ for some $j \in [J]$ where $\beta_j, \kappa_j>0$, then this is symmetric to the previous case. Namely, we set $\beta_j' \defeq \beta_j-1$, $\kappa_j' \defeq \kappa_j - 1$,
        and $w_{i_0}'\defeq w_{i_0}-1$.
        
        \item If $\bm W = \bm G_j \bm Z_j$ for some $j \in [J]$ where $\zeta_j, \kappa_j>0$,
        then we set $\zeta_j' \defeq \zeta_j-1$,
         $\kappa_j' \defeq \kappa_j - 1$, and $w_{i_0}'\defeq w_{i_0}-1$.
         If $\kappa_j' = 0$, then we drop $j$ from $[J]$.
        Otherwise, the triple $\left(\frac{\alpha_j}{\kappa_j'}, \frac{\beta_j}{\kappa_j'}, \frac{\zeta_j'}{\kappa_j'}\right)$ is still $\omega$-dominant.
        In particular, similar to above, inequality~\eqref{eq:omega-dominant:3} continues to hold as follows:
        \begin{align*}
            \frac{\alpha_j+\beta_j+\zeta_j'}{\kappa_j'} &=
            \frac{\alpha_j+\beta_j+\zeta_j-1}{\kappa_j-1}\geq \omega
        \end{align*}
        The last inequality follows from Eq.~\eqref{eq:omega-dominant:preserve}.

        \item If $\bm W = \bm G_j$ for some $j \in [J]$ where $\kappa_j > 0$, note that by Eq.~\eqref{eq:omega-dominant:1}, $\alpha_j\geq\kappa_j$, hence the term $\kappa_j h(\bm G_j)$ on the LHS is already getting canceled with the term $\alpha_j h(\bm X_j|\bm G_j)$ on the same side. Hence, the term $h(\bm W)$ on the RHS must be getting canceled with something else. Therefore, this case falls under the next one.

        \item In all other cases, $h(\bm W|\emptyset)$ must cancel out with some term
        on the RHS of the identity~\eqref{eq:omega-shannon:identity}, just like
        in the original Reset Lemma~\cite{DBLP:conf/pods/Khamis0S17,theoretics:13722}.
        Here we recognize the same three cases from the original Reset Lemma, which we
        repeat below for self-containment:
        \begin{enumerate}[leftmargin=*]
            \item
            $h(\bm W)$ cancels out with some other term $h(\bm Y_i|\bm X_i)$ for some $i \in [I]\setminus \{i_0\}$ where $\bm X_i = \bm W$. In this case, we apply the following
            rewrite to the RHS of the identity~\eqref{eq:omega-shannon:identity}:
            \begin{align}
                h(\bm W) + h(\bm Y_i|\bm W) = h(\bm Y_i\bm W)
                \label{eq:reset:composition}
            \end{align}
            And now inductively, our target is to cancel out the newly added term
            $h(\bm Y_i\bm W)$.
            \label{case:reset:composition}
            \item
            $h(\bm W)$ cancels out with some term $-h(\bm Y_p|\bm X_p)$ from some $p \in [P]$
            where $\bm W = \bm X_p\bm Y_p$. In this case, we apply the following rewrite to the RHS of the identity~\eqref{eq:omega-shannon:identity}:
            \begin{align}
                h(\bm W)-h(\bm Y_p|\bm X_p) = h(\bm W)-h(\bm X_p\bm Y_p) + h(\bm X_p)
                = h(\bm X_p)
                \label{eq:reset:monotonicity}
            \end{align}
            And now, we proceed inductively to eliminate the new term $h(\bm X_p)$.
            \label{case:reset:monotonicity}
            \item
            $h(\bm W)$ cancels out with some term $-h(\bm Y_q;\bm Z_q|\bm X_q)$ for some
            $q \in [Q]$ where $\bm W = \bm X_q\bm Y_q$. In this case, we apply the rewrite:
            \begin{align}
                h(\bm W) -h(\bm Y_q;\bm Z_q|\bm X_q) &=
                h(\bm W) -h(\bm X_q\bm Y_q) -h(\bm X_q\bm Z_q)+h(\bm X_q) + h(\bm X_q\bm Y_q\bm Z_q)\nonumber\\
                &= h(\bm X_q\bm Y_q\bm Z_q) - h(\bm Z_q|\bm X_q)
                \label{eq:reset:submodularity}
            \end{align}
            And now we inductively eliminate $h(\bm X_q\bm Y_q\bm Z_q)$.
            \label{case:reset:submodularity}
        \end{enumerate}
        By applying each one of the above three cases, the following integral quantity always
        decreases by at least one, thus proving that this process must terminate:
        \begin{align}
            \norm{\bm w}_1 + \norm{\bm m}_1 + 2\norm{\bm s}_1
            \label{eq:reset:potential}
        \end{align}
        In particular, Case~\ref{case:reset:composition} decreases $\norm{\bm w}_1$ by 1,
        Case~\ref{case:reset:monotonicity} decreases $\norm{\bm m}_1$ by 1, and Case~\ref{case:reset:submodularity} {\em increases} $\norm{\bm m}_1$ by 1 and decreases $\norm{\bm s}_1$ by 1.
    \end{enumerate}
\end{proof}

%%%%%%%%%%%%%%%%%%%%%%%%%%%%%%%%%%%%%%%%%%%%%%%%%%%%%%%%%%%%%%%%%%%%%%%%%%%%%%%%%%%%%%%%%%%
\subsection{The Proof Sequence Construction}
\label{subsec:proof-sequence}
The original $\panda$ algorithm is based on constructing a {\em proof sequence} for a given {\em Shannon-flow inequality} (Eq.~\eqref{eq:shannon-flow}).
We present here a highly non-trivial generalization of the proof sequence construction that works for $\omega$-Shannon-flow inequalities (Eq.~\eqref{eq:omega-shannon}).
Out of the box, the original proof sequence construction does {\em not} work for $\omega$-Shannon-flow inequalities
because, unlike Shannon-flow inequalities, the $\omega$-Shannon-flow ones can have {\em proper
conditionals}, i.e.,~terms $h(\bm X_j|\bm G_j)$ where $\bm G_j\neq\emptyset$, on the LHS.
The other challenge here is maintaining the $\omega$-dominance property of the triples
$\left(\frac{\alpha_j}{\kappa_j}, \frac{\beta_j}{\kappa_j}, \frac{\zeta_j}{\kappa_j}\right)$,
which is needed later to get a corresponding evaluation algorithm that uses matrix multiplication.

\begin{theorem}[Generalized Proof Sequence Construction]
    Given an integral $\omega$-Shannon-flow inequality of the form~\eqref{eq:omega-shannon},
    there exists a finite sequence of steps that transforms the terms on the RHS of the inequality
    into (a superset of) the terms on the LHS. Each step in the sequence replaces some terms on the RHS with
    smaller terms, thus the full sequence proves that the LHS is smaller than the RHS.
    Moreover, each step has either one of the following forms:
    \begin{itemize}[leftmargin=*]
        \item {\bf Decomposition Step:} $h(\bm X \bm Y) \to h(\bm X) + h(\bm Y|\bm X)$.
        \item {\bf Composition Step:} $h(\bm X) + h(\bm Y|\bm X) \to h(\bm X \bm Y)$.
        \item {\bf Monotonicity Step:} $h(\bm X\bm Y) \to h(\bm X)$.
        \item {\bf Submodularity Step:} $h(\bm Y|\bm X) \to h(\bm Y|\bm X\bm Z)$.
    \end{itemize}
    \label{thm:proof-sequence}
\end{theorem}
Note that by Shannon inequalities, each one of the above steps replaces one or two terms
with {\em smaller} terms. In particular,
\begin{align*}
    h(\bm X \bm Y) &= h(\bm X) + h(\bm Y|\bm X) &\text{(by Eq.~\eqref{eq:conditional})}\\
    h(\bm X\bm Y) &\geq h(\bm X) &\text{(by Eq.~\eqref{eq:monotone})}\\
    h(\bm Y|\bm X) &\geq h(\bm Y|\bm X\bm Z) &\text{(by Eq.~\eqref{eq:submod})}\\
\end{align*}

\subsubsection{Example of Theorem~\ref{thm:proof-sequence}}

\begin{table*}
\[
\small
\begin{array}{|l|l|}
    \hline
    \rowcolor{lightgray}\text{\bf Identity~\eqref{eq:omega-shannon:identity}} & \text{\bf Proof Steps}\\\hline
    \begin{aligned}[c]
        &\omega h(XYZ) + h(X) + h(Y) + (\omega-2)h(Z) \\
        &= 2\positiveterm{h(XY)} + (\omega-1) h(YZ) + (\omega-1)h(XZ)\nonumber\\
        &-[\negativeterm{h(XY)} + h(XZ) - h(X) - h(XYZ)]\nonumber\\
        &-[h(XY) + h(YZ) - h(Y) - h(XYZ)]\nonumber\\
        &-(\omega-2)[h(XZ) + h(YZ) - h(Z) - h(XYZ)]
    \end{aligned}&
    \begin{aligned}[c]
        \positiveterm{h(XY)} &\to h(X) + h(Y|X)\\
        h(Y|X) &\to h(Y|XZ)
    \end{aligned}\\\hline
    \begin{aligned}[c]
        &\omega h(XYZ) + h(X) + h(Y) + (\omega-2)h(Z) \\
        &= h(XY) + h(X) + \negativeterm{h(Y|XZ)} + (\omega-1) h(YZ) + (\omega-1)\positiveterm{h(XZ)}\nonumber\\
        &-[h(XY) + h(YZ) - h(Y) - h(XYZ)]\nonumber\\
        &-(\omega-2)[h(XZ) + h(YZ) - h(Z) - h(XYZ)]
    \end{aligned}&
    \begin{aligned}[c]
        \positiveterm{h(XZ)} + \negativeterm{h(Y|XZ)} &\to h(XYZ)
    \end{aligned}\\\hline
    \begin{aligned}[c]
        &\omega \negativeterm{h(XYZ)} + h(X) + h(Y) + (\omega-2)h(Z) \\
        &= h(XY) + h(X) + \positiveterm{h(XYZ)} + (\omega-1) h(YZ) + (\omega-2)h(XZ)\nonumber\\
        &-[h(XY) + h(YZ) - h(Y) - h(XYZ)]\nonumber\\
        &-(\omega-2)[h(XZ) + h(YZ) - h(Z) - h(XYZ)]
    \end{aligned}&
    \begin{aligned}[c]
        \text{Cancel $h(XYZ)$ on both sides}
    \end{aligned}\\\hline
    \begin{aligned}[c]
        &(\omega-1) h(XYZ) + h(X) + h(Y) + (\omega-2)h(Z) \\
        &= h(XY) + h(X) + (\omega-1) \positiveterm{h(YZ)} + (\omega-2)h(XZ)\nonumber\\
        &-[h(XY) + \negativeterm{h(YZ)} - h(Y) - h(XYZ)]\nonumber\\
        &-(\omega-2)[h(XZ) + h(YZ) - h(Z) - h(XYZ)]
    \end{aligned}&
    \begin{aligned}[c]
        \positiveterm{h(YZ)} &\to h(Y) + h(Z|Y)\\
        h(Z|Y) &\to h(Z|XY)
    \end{aligned}\\\hline
    \begin{aligned}[c]
        &(\omega-1) h(XYZ) + h(X) + h(Y) + (\omega-2)h(Z) \\
        &= \positiveterm{h(XY)} + h(X) + (\omega-2) h(YZ) + h(Y) + \negativeterm{h(Z|XY)}+ (\omega-2)h(XZ)\nonumber\\
        &-(\omega-2)[h(XZ) + h(YZ) - h(Z) - h(XYZ)]
    \end{aligned}&
    \begin{aligned}[c]
        \positiveterm{h(XY)} + \negativeterm{h(Z|XY)} &\to h(XYZ)
    \end{aligned}\\\hline
    \begin{aligned}[c]
        &(\omega-1) \negativeterm{h(XYZ)} + h(X) + h(Y) + (\omega-2)h(Z) \\
        &= \positiveterm{h(XYZ)} + h(X) + (\omega-2) h(YZ) + h(Y)+ (\omega-2)h(XZ)\nonumber\\
        &-(\omega-2)[h(XZ) + h(YZ) - h(Z) - h(XYZ)]
    \end{aligned}&
    \begin{aligned}[c]
       \text{Cancel $h(XYZ)$ on both sides}
    \end{aligned}\\\hline
    \begin{aligned}[c]
        &(\omega-2) h(XYZ) + h(X) + h(Y) + (\omega-2)h(Z) \\
        &= h(X) + (\omega-2) h(YZ) + h(Y)+ (\omega-2)\positiveterm{h(XZ)}\nonumber\\
        &-(\omega-2)[\negativeterm{h(XZ)} + h(YZ) - h(Z) - h(XYZ)]
    \end{aligned}&
    \begin{aligned}[c]
        \gamma\positiveterm{h(XZ)} &\to \gamma h(Z) + \gamma h(X|Z)\\
        \gamma h(X|Z) &\to \gamma h(X|YZ)
    \end{aligned}\\\hline
    \begin{aligned}[c]
        &(\omega-2) h(XYZ) + h(X) + h(Y) + (\omega-2)h(Z) \\
        &= h(X) + \positiveterm{(\omega-2) h(YZ)} + h(Y)+ (\omega-2)h(Z) + \negativeterm{(\omega-2)h(X|YZ)}
    \end{aligned}&
    \begin{aligned}[c]
       \positiveterm{\gamma h(YZ)} + \negativeterm{\gamma h(X|YZ)} &\to \gamma h(XYZ)
    \end{aligned}\\\hline
    \begin{aligned}[c]
        &\negativeterm{(\omega-2) h(XYZ)} + h(X) + h(Y) + (\omega-2)h(Z) \\
        &= h(X) + \positiveterm{(\omega-2) h(XYZ)} + h(Y)+ (\omega-2)h(Z)
    \end{aligned}&
    \begin{aligned}[c]
       \text{Cancel $(\omega-2)h(XYZ)$ on both sides}
    \end{aligned}\\\hline
    \begin{aligned}[c]
        &\negativeterm{h(X) + h(Y) + (\omega-2)h(Z)}\\
        &= \positiveterm{h(X) + h(Y)+ (\omega-2)h(Z)}
    \end{aligned}&
    \begin{aligned}[c]
       &\text{Cancel $h(X) + h(Y)+ (\omega-2)h(Z)$}\\
       &\text{on both sides}
    \end{aligned}\\\hline
    \begin{aligned}[c]
        0=0
    \end{aligned}&
    \begin{aligned}[c]
    \end{aligned}\\\hline
\end{array}
\]
\caption{Proof sequence construction for the Shannon inequality~\eqref{eq:intro:shannon:triangle}. We start from the identity form~\eqref{eq:intro:shannon:triangle:identity:explicit}.
We repeatedly look for an unconditional \positiveterm{\text{positive term}} on the RHS:
The existence of such term is guaranteed by Proposition~\ref{prop:omega-shannon:unconditional}.
If the positive term found also appears on the LHS, we cancel it from both sides.
Otherwise, it must be getting canceled with some \negativeterm{\text{negative term}}
on the RHS.
Based on the negative term found, we apply some proof steps to the RHS
of the identity to obtain a new identity, and we repeat the process until we end up with the trivial identity $0=0$.
The left column shows the sequence of identities obtained, and the right column shows the proof steps applied to get from one identity to the next.}
\label{table:ps}
\end{table*}

Before proving the theorem, we give an example.
Figure~\ref{fig:ps-algo:triangle} shows a compact form of the proof sequence for the $\omega$-Shannon-flow inequality in Eq.~\eqref{eq:intro:shannon:triangle}.
In particular, the proof sequence in Figure~\ref{fig:ps-algo:triangle} combines
two steps $h(Y|X) \to h(Y|XZ)$ and $h(XZ) + h(Y|XZ)\to h(XYZ)$
into a single step $h(XZ) + h(Y|X) \to h(XYZ)$.
Table~\ref{table:ps} shows the detailed proof steps for the same $\omega$-Shannon-flow inequality,
as well as how we come up with these steps in the first place.
In particular, Proposition~\ref{prop:omega-shannon:unconditional}
tells us that as long as $\norm{\bm \lambda}_1 + \norm{\bm\kappa}_1 > 0$,
there must exist some {\em unconditional} term $h(\bm Y_i|\emptyset)$ on the RHS of the $\omega$-Shannon-flow inequality.
We repeatedly rely on the existence of such unconditional term on the RHS.
In particular, we start from the identity form~\eqref{eq:intro:shannon:triangle:identity:explicit} of the $\omega$-Shannon-flow inequality in Eq.~\eqref{eq:intro:shannon:triangle}.
We pick such unconditional term on the RHS, say $h(XY)$.
Since~\eqref{eq:intro:shannon:triangle:identity:explicit} is an identity
and $h(XY)$ does not appear on the LHS, it must be getting canceled with some other term
on the RHS.
In this case, it is being canceled by $-h(XY)$ coming from $-h(Y;Z|X)$.
We apply the two proof steps $h(XY) \to h(X) + h(Y|X)$ and $h(Y|X) \to h(Y|XZ)$ to the RHS of the identity~\eqref{eq:intro:shannon:triangle:identity:explicit} and we also drop $-h(Y;Z|X)$ from the RHS, thus resulting in a new identity.
This new identity is depicted in the second row of Table~\ref{table:ps}.
We repeat the same process until we end up with the trivial identity $0=0$.

\subsubsection{Proof of Theorem~\ref{thm:proof-sequence}}
For the purpose of proving the Theorem~\ref{thm:proof-sequence} in general,
we will represent an $\omega$-Shannon-flow inequality slightly differently
from Eq.~\eqref{eq:omega-shannon}.
In particular, some terms on the LHS of Eq.~\eqref{eq:omega-shannon} might already occur on the RHS, therefore we could
decompose the RHS into two parts: One part which is a subset of terms on the LHS,
and another part containing the remaining RHS terms.
Formally, we can represent an $\omega$-Shannon-flow inequality as follows:
\begin{align}
    \sum_{\ell \in [L]} \lambda_\ell h(\bm U_{\ell})
    +& \sum_{j \in [J]} \left(\alpha_j h(\bm X_j|\bm G_j) +
        \beta_j h(\bm Y_j|\bm G_j) +
        \zeta_j h(\bm Z_j|\bm G_j) +
        \kappa_j h(\bm G_j)\right)
        \quad\leq\quad\nonumber\\
        &\sum_{j \in [J]} \left(\hat\alpha_j h(\bm X_j|\bm G_j) +
        \hat\beta_j h(\bm Y_j|\bm G_j) +
        \hat\zeta_j h(\bm Z_j|\bm G_j) +
        \hat\kappa_j h(\bm G_j)\right)
    +\sum_{i \in [I]} \hat w_i h(\bm Y_i |\bm X_i)
    \label{eq:omega-shannon-ps}
\end{align}
where in addition to the conditions from Definition~\ref{defn:omega-shannon},
we also require that
\begin{itemize}[leftmargin=*]
    \item For all $j \in [J]$, we have
    $0\leq\hat\alpha_j \leq \alpha_j$,
    $0\leq\hat\beta_j \leq \beta_j$,
    $0\leq\hat\zeta_j \leq \zeta_j$, and
    $0\leq\hat\kappa_j \leq \kappa_j$.
    \item For all $i \in [I]$, we have $0\leq \hat w_i \leq w_i$.
    \item Finally, we assume, for all $j \in [J]$,
    \begin{align}
        \kappa_j - \hat\kappa_j = 0 \quad\quad\Rightarrow\quad\quad
            \alpha_j - \hat\alpha_j= \beta_j - \hat\beta_j= \zeta_j - \hat\zeta_j = 0
        \label{eq:cond:well-behaved}
    \end{align}
\end{itemize}
By Proposition~\ref{prop:omega-shannon:identity},
for every inequality of the form~\eqref{eq:omega-shannon-ps},
there exist non-negative vectors $\bm m \defeq (m_p)_{p \in [P]}$ and $\bm s \defeq (s_q)_{q \in [Q]}$ such that
the following identity holds over symbolic variables $h(\bm X)$ for $\bm X \subseteq \calV$ where $h(\emptyset) = 0$:

\begin{align}
    \sum_{\ell \in [L]} \lambda_\ell h(\bm U_{\ell})
    +& \sum_{j \in [J]} \left(\alpha_j h(\bm X_j|\bm G_j) +
        \beta_j h(\bm Y_j|\bm G_j) +
        \zeta_j h(\bm Z_j|\bm G_j) +
        \kappa_j h(\bm G_j)\right)
        \quad=\quad\nonumber\\
        &\sum_{j \in [J]} \left(\hat\alpha_j h(\bm X_j|\bm G_j) +
        \hat\beta_j h(\bm Y_j|\bm G_j) +
        \hat\zeta_j h(\bm Z_j|\bm G_j) +
        \hat\kappa_j h(\bm G_j)\right)
    +\sum_{i \in [I]} \hat w_i h(\bm Y_i |\bm X_i)\nonumber\\
    - &\sum_{p \in [P]} m_p h(\bm Y_p |\bm X_p)
    - \sum_{q \in [Q]} s_q h(\bm Y_q;\bm Z_q |\bm X_q)
    \label{eq:omega-shannon-ps:identity}
\end{align}

\begin{proof}[Proof of Theorem~\ref{thm:proof-sequence}]
    An inequality of the form~\eqref{eq:omega-shannon} can be converted to the from~\eqref{eq:omega-shannon-ps} by initializing $\hat w_i = w_i$ for all $i \in [I]$ and
    $\hat\alpha_j = \hat\beta_j = \hat\zeta_j = \hat\kappa_j = 0$ for all $j \in [J]$.
    Consider an inequality of the form~\eqref{eq:omega-shannon-ps}
    and the corresponding identity~\eqref{eq:omega-shannon-ps:identity}.
    Define the quantity:
    \begin{align*}
        \Lambda \defeq \norm{\bm\lambda}_1 +
        \sum_{\substack{j \in [J]\\\bm G_j \neq \emptyset}} (\kappa_j - \hat\kappa_j)+
        \sum_{\substack{j \in [J]\\\bm G_j = \emptyset}} (\alpha_j - \hat\alpha_j) + (\beta_j - \hat\beta_j) + (\zeta_j - \hat\zeta_j)
    \end{align*}
    If $\Lambda = 0$, then
    by Condition~\eqref{eq:cond:well-behaved}, we have
    $\hat\alpha_j=\alpha_j$, $\hat\beta_j=\beta_j$, $\hat\zeta_j=\zeta_j$
    and $\hat\kappa_j=\kappa_j$
    for all $j$
    and we are done.
    Now assume that $\Lambda > 0$.
    Similar to the proof of Proposition~\ref{prop:omega-shannon:unconditional}, we have
    $$\sum_{i \in [I]\mid\bm X_i = \emptyset} \hat w_i \geq \Lambda > 0.$$
    Hence, there exists some $i_0 \in [I]$ such that $\bm X_{i_0} = \emptyset$ and $\hat w_{i_0} > 0$.
    Define $\bm W\defeq \bm Y_{i_0}$. We recognize the following cases:
    \begin{enumerate}[leftmargin=*]
        \item If there exists $\ell\in[L]$ where $\bm U_\ell = \bm W$ and $\lambda_\ell > 0$,
        then we reduce both $\lambda_\ell$ and $\hat w_{i_0}$ by one thus canceling out the term $h(\bm W)$ from both sides of the identity~\eqref{eq:omega-shannon-ps:identity}.
        \label{case:ps:cancel-lambda}
        \item If there  exists $j \in [J]$ where $\bm G_j\bm X_j = \bm W$ and $\alpha_j > \hat\alpha_j$,
        then we apply the following decomposition step to the term
        $h(\bm Y_{i_0}|\bm X_{i_0}) = h(\bm W |\emptyset)$ on the RHS of~\eqref{eq:omega-shannon-ps:identity}:
        \begin{align*}
            h(\bm W) \to h(\bm G_j) + h(\bm X_j|\bm G_j)
        \end{align*}
        and now we have $h(\bm X_j|\bm G_j)$ on both sides of the identity~\eqref{eq:omega-shannon-ps:identity}. We increase $\hat\alpha_j$ by one thus pairing
        the two terms $h(\bm X_j|\bm G_j)$ with one another.
        Condition~\eqref{eq:cond:well-behaved} continues to hold.
        \label{case:ps:cancel-alpha}
        \item If there exists $j \in [J]$ where either $\bm G_j\bm Y_j = \bm W$ and $\beta_j > \hat\beta_j$,
        or $\bm G_j\bm Z_j = \bm W$ and $\zeta_j > \hat\zeta_j$,
        then this is similar to the previous case.
        \label{case:ps:cancel-beta-or-zeta}
        \item If there exists $j \in [J]$ where $\bm G_j = \bm W$
        and $\kappa_j - \hat\kappa_j >
            (\alpha_j - \hat\alpha_j) + (\beta_j - \hat\beta_j) + (\zeta_j - \hat\zeta_j)$,
            then we have $h(\bm W)$ on both sides of identity~\eqref{eq:omega-shannon-ps:identity}.
            We increase $\hat\kappa_j$ by one thus pairing the two terms $h(\bm G_j)$ with one another. Condition~\eqref{eq:cond:well-behaved} continues to hold.
        \label{case:ps:cancel-kappa}
        \item If for every $j \in [J]$ where $\bm G_j = \bm W$ we have 
            $\kappa_j - \hat\kappa_j \leq
             (\alpha_j - \hat\alpha_j) + (\beta_j - \hat\beta_j) + (\zeta_j - \hat\zeta_j)$,
            then $h(\bm G_j)$ on the LHS already cancels out with one of the three terms
            $h(\bm X_j|\bm G_j)$, $h(\bm Y_j|\bm G_j)$, and $h(\bm Z_j|\bm G_j)$ on the LHS of~\eqref{eq:omega-shannon-ps:identity}.
            Therefore, $h(\bm W)$ on the RHS must be getting canceled with some other term, and this case falls under the next one.
        \label{case:ps:kappa-unchanged}
        \item In all other cases, the term $h(\bm W)$ must cancel out with some other term on
        the RHS of the identity~\eqref{eq:omega-shannon-ps:identity}.
        We recognize three cases similar to the original proof sequence construction
        from~\cite{DBLP:conf/pods/Khamis0S17,theoretics:13722}:
        \begin{enumerate}[leftmargin=*]
            \item $h(\bm W)$ cancels out with some other term $h(\bm Y_i|\bm X_i)$ for some $i \in [I]$ on the RHS where $\bm X_i = \bm W$.
            In this case, we can compose the two terms using the rewrite from Eq.~\eqref{eq:reset:composition}, which corresponds to a composition step:
            \begin{align*}
                h(\bm W) + h(\bm Y_i | \bm W) \to h(\bm Y_i \bm W) & &\text{(Composition Step)}
            \end{align*}
            \label{case:ps:composition}
            \item $h(\bm W)$ cancels out with some term $-h(\bm Y_p|\bm X_p)$ for some $p \in [P]$ where $\bm W = \bm X_p\bm Y_p$.
            In this case, we can apply the rewrite from Eq.~\eqref{eq:reset:monotonicity},
            which corresponds to a monotonicity step:
            \begin{align*}
                h(\bm W) \to h(\bm X_p) & &\text{(Monotonicity Step)}
            \end{align*}
            \label{case:ps:monotonicity}
            \item $h(\bm W)$ cancels out with some term $-h(\bm Y_q;\bm Z_q|\bm X_q)$ for some
            $q \in [Q]$ where $\bm W = \bm X_q\bm Y_q$.
            In this case, instead of applying the rewrite from Eq.~\eqref{eq:reset:submodularity}, we apply the following rewrite:
            \begin{align}
                h(\bm W) -h(\bm Y_q;\bm Z_q|\bm X_q) &=
                h(\bm W) -h(\bm X_q\bm Y_q) -h(\bm X_q\bm Z_q)+h(\bm X_q) + h(\bm X_q\bm Y_q\bm Z_q)\nonumber\\
                &= h(\bm X_q) + h(\bm Y_q|\bm X_q\bm Z_q)
                \label{eq:ps:submodularity}
            \end{align}
            The above rewrite corresponds to applying a decomposition step followed by a submodularity step:
            \begin{align*}
                h(\bm W) \to h(\bm X_q) + h(\bm Y_q|\bm X_q) & & \text{(Decomposition Step)}\\
                h(\bm Y_q|\bm X_q) \to h(\bm Y_q|\bm X_q\bm Z_q) & & \text{(Submodularity Step)}
            \end{align*}
            \label{case:ps:submodularity}
        \end{enumerate}
    \end{enumerate}
    In each one of the above cases, the following integral quantity always
    decreases by at least one, thus proving that this process must terminate:
    \begin{align}
        \norm{\bm\lambda}_1
        + \norm{\bm \alpha-\hat{\bm\alpha}}_1
        + \norm{\bm \beta-\hat{\bm\beta}}_1
        + \norm{\bm \zeta-\hat{\bm\zeta}}_1
        + \norm{\bm \kappa-\hat{\bm\kappa}}_1
        + \norm{\hat{\bm w}}_1 + \norm{\bm m}_1 + 2\norm{\bm s}_1
        \label{eq:ps:potential}
    \end{align}
    In particular,
    \begin{itemize}
        \item Case~\ref{case:ps:cancel-lambda} decreases $\norm{\bm \lambda}_1$ by 1.
        \item Cases~\ref{case:ps:cancel-alpha} and \ref{case:ps:cancel-beta-or-zeta}
        decrease $\norm{\bm \alpha-\hat{\bm\alpha}}_1
        + \norm{\bm \beta-\hat{\bm\beta}}_1
        + \norm{\bm \zeta-\hat{\bm\zeta}}_1$ by 1. Note that $\norm{\hat{\bm w}}_1$ remains unchanged here despite the decomposition step.
        This is because one of the two terms resulting from the decomposition step
        was paired with a term on the LHS of the identity~\eqref{eq:omega-shannon-ps:identity},
        hence it no longer contributes to $\norm{\hat{\bm w}}_1$.
        \item Case~\ref{case:ps:cancel-kappa} decreases both $\norm{\bm \kappa-\hat{\bm\kappa}}_1$ and $\norm{\hat{\bm w}}_1$ by 1.
        \item Case~\ref{case:ps:kappa-unchanged} falls under the remaining cases.
        \item Case~\ref{case:ps:composition} decreases $\norm{\hat{\bm w}}_1$ by 1.
        \item Case~\ref{case:ps:monotonicity} decreases $\norm{\bm m}_1$ by 1.
        \item Case~\ref{case:ps:submodularity} {\em increases} $\norm{\hat{\bm w}}_1$ by 1 and decreases $\norm{\bm s}_1$ by 1.
    \end{itemize}
\end{proof}

%%%%%%%%%%%%%%%%%%%%%%%%%%%%%%%%%%%%%%%%%%%%%%%%%%%%%%%%%%%%%%%%%%%%%%%%%%%%%%%%%%%%%%%%%%%
\subsection{$\omega$-Disjunctive Datalog Rules}
\label{subsec:omega-ddrs}

The original $\panda$ algorithm from~\cite{DBLP:conf/pods/Khamis0S17,theoretics:13722,panda-express} evaluates a BCQ in submodular width time by reducing the problem to evaluating a collection of Disjunctive Datalog rules (DDRs), which were reviewed in Section~\ref{subsec:prelims:ddr}.
For the purpose of designing a general algorithm for evaluating a BCQ in $\omega$-submodular width time, we will need to generalize the concept of DDRs.
In this section, we present our generalization of DDRs, called {\em $\omega$-DDRs}.
In Section~\ref{subsec:algo:panda:ddr}, we present an algorithm for evaluating $\omega$-DDRs.
Finally, in Section~\ref{subsec:algo:panda:osubw}, we show how to evaluate a BCQ $Q$
in $\omega$-submodular width time by reducing the problem to evaluating a collection of
$\omega$-DDRs, each of which is evaluated using the algorithm from Section~\ref{subsec:algo:panda:ddr}.

\begin{definition}[$\omega$-DDR and its model]
Let $\calV$ be a set of variables, and $\Sigma_{\inn}$ be an input schema.
Let $L$ and $J$ be two natural numbers, 
and $\Sigma_{\outt}$ be an output schema consisting of two types of atoms:
\begin{itemize}
    \item An atom $P_\ell(\bm U_\ell)$ for every $\ell \in [L]$, where $\bm U_\ell$ is a subset of $\calV$.
    \item Three atoms $S_j(\bm X_j\bm G_j), T_j(\bm Y_j\bm G_j), W_j(\bm Z_j\bm G_j)$
    for every $j \in [J]$, where $\bm G_j, \bm X_j, \bm Y_j, \bm Z_j$ are disjoint subsets of $\calV$.
\end{itemize}
    An {\em $\omega$-Disjunctive Datalog Rule ($\omega$-DDR)} with input and output schemas $\Sigma_{\inn}$ and $\Sigma_{\outt}$ respectively is an expression of the form:
    \begin{align}
        \bigvee_{\ell\in[L]} P_\ell(\bm U_\ell) \vee
        \bigvee_{j \in [J]}\left(S_j(\bm X_j\bm G_j)\wedge T_j(\bm Y_j\bm G_j)\wedge W_j(\bm Z_j\bm G_j)
        \right)
        \quad\cd\quad \bigwedge_{R(\bm Z) \in \Sigma_{\inn}} R(\bm Z)
        \label{eq:disjunctive-omega}
    \end{align}
Given a database instance $D$ over schema $\Sigma_{\inn}$,
{\em a model (or output instance)} for the $\omega$-DDR~\eqref{eq:disjunctive-omega}
is a collection of output tables $P_\ell(\bm U_\ell)$ for $\ell \in [L]$,
and tables $(S_j(\bm X_j \bm G_j), T_j(\bm Y_j \bm G_j), W_j(\bm Z_j \bm G_j))$
for $j \in [J]$ such that for every tuple $\bm t \in \bigjoin \Sigma_{\inn}$:
\begin{itemize}
    \item either there exists $\ell \in [L]$ where $\bm t_{\bm U_\ell}\in P_\ell$,
    \item or there exists $j \in [J]$ where $\bm t_{\bm X_j\bm G_j}\in S_j$,
    $\bm t_{\bm Y_j\bm G_j}\in T_j$,
    and $\bm t_{\bm Z_j\bm G_j}\in W_j$.
\end{itemize}
\label{defn:omega-ddr}
\end{definition}

\subsubsection{A new notion of model for an $\omega$-DDR}

For the purpose of evaluating an $\omega$-DDR in our target runtime bounds,
we need a different notion of a model for an $\omega$-DDR, where the model
is described through {\em sub-probability measures} rather than relations.
This follows the $\pandaexpress$ approach~\cite{panda-express}
in order to avoid unnecessary $\polylog$-factors in the runtime.
See Section~\ref{subsec:prelims:sub-probs} for some background on sub-probability measures.

\begin{definition}[Targets of an $\omega$-DDR]
    Consider an $\omega$-DDR of the form~\eqref{eq:disjunctive-omega}.
For each $\ell \in [L]$, we refer to the atom $P_\ell(\bm U_\ell)$ as a {\em join target} of the $\omega$-DDR, and its variable are $\vars(P_\ell(\bm U_\ell)) \defeq \bm U_\ell$.
For each $j \in [J]$, we think of the three conjuncts $S_j(\bm X_j\bm G_j)$, $T_j(\bm Y_j\bm G_j)$, and $W_j(\bm Z_j\bm G_j)$ as being {\em partially ordered},
where $S_j$ and $T_j$ are interchangeable with one another, whereas $W_j$ is not interchangeable with either one of them.
This partial order relationship is captured by the tuple $(\{S_j(\bm X_j \bm G_j), T_j(\bm Y_j \bm G_j)\}, W_j(\bm Z_j\bm G_j))$,
which we refer to as an {\em MM target} of the $\omega$-DDR, and its variables are
\[\vars(\{S_j(\bm X_j \bm G_j), T_j(\bm Y_j \bm G_j)\}, W_j(\bm Z_j\bm G_j)) \defeq \bm X_j \cup \bm Y_j \cup \bm Z_j \cup \bm G_j.\]
We refer to either kinds of targets as a {\em target} of the $\omega$-DDR.
\label{defn:omega-ddr:targets}
\end{definition}
Recall from Section~\ref{subsec:prelims:sub-probs} that we use $p_{\bm Y|\bm X}(\bm t)$ as a shorthand for
$p_{\bm Y|\bm X}(\bm t_{\bm Y}|\bm t_{\bm X})$.
\begin{definition}[A $B$-probabilistic model for an $\omega$-DDR]
    Consider an $\omega$-DDR of the form~\eqref{eq:disjunctive-omega}, and let $D$ be an input database instance over schema $\Sigma_{\inn}$, and $B \in \R_+$ be a constant.
    Given a join target $P_\ell(\bm U_\ell)$, a {\em probabilistic interpretation} of this target is a sub-probability measure $p_{\bm U_\ell}$ over $\dom^{\bm U_\ell}$.
    The {\em $B$-support} of $p_{\bm U_\ell}$ is the set of tuples $\bm t \in \dom^{\bm U_\ell}$ that satisfy:
    \begin{align}
        p_{\bm U_\ell}(\bm t) \geq \frac{1}{B}
    \end{align}
    Given an MM target $(\{S_j(\bm X_j \bm G_j), T_j(\bm Y_j \bm G_j)\}, W_j(\bm Z_j\bm G_j))$, a {\em probabilistic interpretation} of this target is a pair $((p_{\bm G_j}, p_{\bm X_j|\bm G_j}, p_{\bm Y_j|\bm G_j}, p_{\bm Z_j|\bm G_j}), (\alpha_j, \beta_j, \zeta_j))$ where $p_{\bm G_j}, p_{\bm X_j|\bm G_j}, p_{\bm Y_j|\bm G_j}, p_{\bm Z_j|\bm G_j}$ are sub-probability measures, and $(\alpha_j, \beta_j, \zeta_j)$ is an $\omega$-dominant triple.
    The {\em $B$-support} of this probabilistic interpretation is the set of tuples $\bm t \in \dom^{\bm X_j\bm Y_j\bm Z_j\bm G_j}$ that satisfy:
    \begin{align}
        p_{\bm G_j}(\bm t) \cdot
        \left(p_{\bm X_j|\bm G_j}(\bm t)\right)^{\alpha_j} \cdot
        \left(p_{\bm Y_j|\bm G_j}(\bm t)\right)^{\beta_j} \cdot
        \left(p_{\bm Z_j|\bm G_j}(\bm t)\right)^{\zeta_j} \geq \frac{1}{B}
        \label{eq:omega-ddr:mm-target-support}
    \end{align}
    A {\em $B$-probabilistic model} for an $\omega$-DDR over the input database $D$ is a collection of probabilistic interpretations, one for each target of the $\omega$-DDR, that satisfy
    the following:
    \begin{itemize}
        \item For every tuple $\bm t \in \bigjoin \Sigma_{\inn}$,
        there exists a target $A$ of the $\omega$-DDR (which could be either a join target or an MM target) such that $\bm t_{\vars(A)}$ is in the $B$-support of the probabilistic interpretation of $A$.
    \end{itemize}
    In other words, a $B$-probabilistic model for an $\omega$-DDR is valid
    iff every tuple $\bm t\in\bigjoin\Sigma$ is covered\footnote{Recall the notion of ``covered'' from Section~\ref{subsec:prelims:ddr}.} by the $B$-support of the probabilistic interpretation of some target of the $\omega$-DDR.
    \label{defn:prob:model:omega-ddr}
\end{definition}
\begin{proposition}
Every $B$-probabilistic model for an $\omega$-DDR is also a $B'$-probabilistic model for the same $\omega$-DDR for every $B' > B$.
\label{prop:model:omega-ddr:monotonicity}
\end{proposition}

\subsubsection{Performing matrix multiplication over models of $\omega$-DDRs}
In order to evaluate a BCQ $Q$ in $\omega$-submodular width time, we will describe in Section~\ref{subsec:algo:panda:osubw} a reduction from evaluating $Q$ to evaluating a collection of $\omega$-DDRs.
After we evaluate these $\omega$-DDRs, we will need to perform some fast matrix multiplication operations on the outputs (models) of these $\omega$-DDRs.
We describe below how to do that efficiently.
The extra $\log^2$ factor in our main theorem, Theorem~\ref{thm:panda:osubw}, is inherited from the following lemma.

\begin{lemma}[Performing fast matrix multiplication over $B$-probabilistic models of $\omega$-DDRs in time $O(B\cdot \log^2 N)$]
    Let $\bm X, \bm Y, \bm Z, \bm G$ be disjoint sets of variables.
    Suppose we have three MM targets (which could belong to different $\omega$-DDRs):
    \begin{align}
        &(\{S_1(\bm X\bm G), T_1(\bm Y\bm G)\}, W_1(\bm Z\bm G)),\nonumber\\
        &(\{S_2(\bm X\bm G), T_2(\bm Z\bm G)\}, W_2(\bm Y\bm G)),\nonumber\\
        &(\{S_3(\bm Y\bm G), T_3(\bm Z\bm G)\}, W_3(\bm X\bm G)).
    \end{align}
    Suppose that each of the above targets has a probabilistic interpretation.
    Let $B \in \R_+$ be a constant, and $E \subseteq \dom^{\bm X\bm Y\bm Z\bm G}$ be the intersection of the $B$-supports of the probabilistic interpretations of the three targets.
    Then, there exists a superset $\ov E \supseteq E$ such that for any two relations
    $M_1(\bm G\bm Y\bm X)$ and $M_2(\bm G\bm X\bm Z)$, the following query can be evaluated in time $O(B\cdot \log^2 N)$:
    \begin{align}
        Q(\bm G\bm Y\bm Z) \quad\cd\quad
            \ov{E}(\bm G\bm X\bm Y\bm Z) \wedge M_1(\bm G\bm Y\bm X) \wedge M_2(\bm G\bm X\bm Z)
            \label{eq:omega-join-matrix-mult:q}
    \end{align}
    Moreover, the relations $M_1$ and $M_2$ don't need to be materialized explicitly.
    Instead, it suffices to have lookup oracles for $M_1$ and $M_2$ that can answer membership queries in time $O(1)$.
    \label{lem:omega-join-matrix-mult}
\end{lemma}
\begin{proof}
    The three targets have three probabilistic interpretations, respectively:
    \begin{align*}
        &((p_{\bm G}^{(1)}, p_{\bm X|\bm G}^{(1)}, p_{\bm Y|\bm G}^{(1)}, p_{\bm Z|\bm G}^{(1)}), (\alpha_1, \beta_1, \zeta_1)),\\
        &((p_{\bm G}^{(2)}, p_{\bm X|\bm G}^{(2)}, p_{\bm Z|\bm G}^{(2)}, p_{\bm Y|\bm G}^{(2)}), (\alpha_2, \beta_2, \zeta_2)),\\
        &((p_{\bm G}^{(3)}, p_{\bm Y|\bm G}^{(3)}, p_{\bm Z|\bm G}^{(3)}, p_{\bm X|\bm G}^{(3)}), (\alpha_3, \beta_3, \zeta_3)).
    \end{align*}
    Let $\bm W \defeq \bm X \bm Y \bm Z \bm G$.
    The set $E$ is the set of tuples $\bm w \in \dom^{\bm W}$ that satisfy the following:
    \begin{align*}
        p_{\bm G}^{(1)}(\bm w) \cdot
        \left(p_{\bm X|\bm G}^{(1)}(\bm w)\right)^{\alpha_1} \cdot
        \left(p_{\bm Y|\bm G}^{(1)}(\bm w)\right)^{\beta_1} \cdot
        \left(p_{\bm Z|\bm G}^{(1)}(\bm w)\right)^{\zeta_1} &\geq \frac{1}{B}\\
        p_{\bm G}^{(2)}(\bm w) \cdot
        \left(p_{\bm X|\bm G}^{(2)}(\bm w)\right)^{\alpha_2} \cdot
        \left(p_{\bm Z|\bm G}^{(2)}(\bm w)\right)^{\beta_2} \cdot
        \left(p_{\bm Y|\bm G}^{(2)}(\bm w)\right)^{\zeta_2} &\geq \frac{1}{B}\\
        p_{\bm G}^{(3)}(\bm w) \cdot
        \left(p_{\bm Y|\bm G}^{(3)}(\bm w)\right)^{\alpha_3} \cdot
        \left(p_{\bm Z|\bm G}^{(3)}(\bm w)\right)^{\beta_3} \cdot
        \left(p_{\bm X|\bm G}^{(3)}(\bm w)\right)^{\zeta_3} &\geq \frac{1}{B}
    \end{align*}
    We define a sub-probability measure $p_{\bm G}$ as the mean of the three sub-probability measures $p_{\bm G}^{(1)}$, $p_{\bm G}^{(2)}$, and $p_{\bm G}^{(3)}$.
    Note that $p_{\bm G} \geq \frac{1}{3} p_{\bm G}^{(i)}$ for every $i \in [3]$.
    Similarly, we define $p_{\bm X|\bm G}$, $p_{\bm Y|\bm G}$, and $p_{\bm Z|\bm G}$ as the mean of the corresponding sub-probability measures from the three probabilistic interpretations.
    Let $E'$ be the set of tuples $\bm w \in \dom^{\bm W}$ that satisfy the following:
    \begin{align}
        p_{\bm G}(\bm w) \cdot
        \left(p_{\bm X|\bm G}(\bm w)\right)^{\alpha_1} \cdot
        \left(p_{\bm Y|\bm G}(\bm w)\right)^{\beta_1} \cdot
        \left(p_{\bm Z|\bm G}(\bm w)\right)^{\zeta_1} &\geq \frac{1}{81B}
        \label{eq:omega-join-matrix-mult-1}\\
        p_{\bm G}(\bm w) \cdot
        \left(p_{\bm X|\bm G}(\bm w)\right)^{\alpha_2} \cdot
        \left(p_{\bm Z|\bm G}(\bm w)\right)^{\beta_2} \cdot
        \left(p_{\bm Y|\bm G}(\bm w)\right)^{\zeta_2} &\geq \frac{1}{81B}
        \label{eq:omega-join-matrix-mult-2}\\
        p_{\bm G}(\bm w) \cdot
        \left(p_{\bm Y|\bm G}(\bm w)\right)^{\alpha_3} \cdot
        \left(p_{\bm Z|\bm G}(\bm w)\right)^{\beta_3} \cdot
        \left(p_{\bm X|\bm G}(\bm w)\right)^{\zeta_3} &\geq \frac{1}{81B}
        \label{eq:omega-join-matrix-mult-3}
    \end{align}
    Note that $E'$ is a super set of $E$.
    Later on, we will construct $\ov E$ to be a super set of $E'$.

    Fix a tuple $\bm g \in \dom^{\bm G}$.
    We will show that $\sigma_{\bm G = \bm g}Q$ can be computed in time $O(p_{\bm G}(\bm g)\cdot B\cdot \log^2 N)$. This in turn implies that the total time for computing $Q$ is $O(B\cdot \log^2 N)$, since $p_{\bm G}$ is a sub-probability measure.

    To that end, we partition $p_{\bm X|\bm G=\bm g}$ into $O(\log B) = O(\log N)$ buckets
    where bucket $i$ contains all values $\bm x \in \dom^{\bm X}$ that satisfy:
    \[\frac{1}{2^{i+1}} \quad<\quad p_{\bm X|\bm G=\bm g}(\bm x) \quad\leq\quad \frac{1}{2^i}\]
    Similarly, we partition $p_{\bm Y|\bm G=\bm g}$ into $O(\log N)$ buckets
    where bucket $j$ contains all values $\bm y \in \dom^{\bm Y}$ that satisfy:
    \[\frac{1}{2^{j+1}} \quad<\quad p_{\bm Y|\bm G=\bm g}(\bm y) \quad\leq\quad \frac{1}{2^j}\]
    There are $O(\log^2 N)$ pairs of buckets $(i, j)$.
    Fix one pair $(i, j)$ (in addition to the fixed $\bm g$ at the beginning).
    Let $E_{\bm g}^{ij}$ be the subset of $E'$ containing all tuples $\bm w\in\dom^{\bm W}$ that satisfy Eq~\eqref{eq:omega-join-matrix-mult-1}--\eqref{eq:omega-join-matrix-mult-3}
    and
    \begin{align}
        \bm w_{\bm G} = \bm g, \quad\quad
        \frac{1}{2^{i+1}} < p_{\bm X|\bm G=\bm g}(\bm w) \leq \frac{1}{2^i}, \quad\quad
        \frac{1}{2^{j+1}} < p_{\bm Y|\bm G=\bm g}(\bm w) \leq \frac{1}{2^j}
    \end{align}
    Let $m$ and $n$ be the total number of $\bm x$ and $\bm y$ values in the selected buckets $i$ and $j$ respectively.
    Note that $m \leq 2^{i+1}$ and $n \leq 2^{j+1}$.
    Let $\ov{p}$ be the smallest value of $p_{\bm Z|\bm G=\bm g}(\bm z)$
    over all $\bm z \in \dom^{\bm Z}$ that satisfies:
    \begin{align}
        p_{\bm G}(\bm g) \cdot
        \left(\frac{1}{2^{i}}\right)^{\alpha_1} \cdot
        \left(\frac{1}{2^{j}}\right)^{\beta_1} \cdot
        \left(\ov{p}\right)^{\zeta_1} &\quad\geq\quad \frac{1}{81B}\label{eq:omega-join-matrix-mult-4}\\
        p_{\bm G}(\bm g) \cdot
        \left(\frac{1}{2^{i}}\right)^{\alpha_2} \cdot
        \left(\frac{1}{2^{j}}\right)^{\zeta_2} \cdot
        \left(\ov{p}\right)^{\beta_2} &\quad\geq\quad \frac{1}{81B}\label{eq:omega-join-matrix-mult-5}\\
        p_{\bm G}(\bm g) \cdot
        \left(\frac{1}{2^{i}}\right)^{\zeta_3} \cdot
        \left(\frac{1}{2^{j}}\right)^{\alpha_3} \cdot
        \left(\ov{p}\right)^{\beta_3} &\quad\geq\quad \frac{1}{81B}\label{eq:omega-join-matrix-mult-6}
    \end{align}
    Let $\ov E_{\bm g}^{ij}$ be the set of all tuples $\bm w \in \dom^{\bm W}$ that satisfy
    the following:
    \begin{align}
        \bm w_{\bm G} = \bm g, \quad\quad
        \frac{1}{2^{i+1}} < p_{\bm X|\bm G=\bm g}(\bm w) \leq \frac{1}{2^i}, \quad\quad
        \frac{1}{2^{j+1}} < p_{\bm Y|\bm G=\bm g}(\bm w) \leq \frac{1}{2^j}, \quad\quad
        p_{\bm Z|\bm G=\bm g}(\bm w) \geq \ov{p}
    \end{align}
    Note that $\ov E_{\bm g}^{ij} \supseteq E_{\bm g}^{ij}$.
    Let $o$ be the number of different $\bm z$ values where $p_{\bm Z|\bm G=\bm g}(\bm z) \geq \ov{p}$. Note that $o \leq 1/\ov{p}$.
    From inequalities~\eqref{eq:omega-join-matrix-mult-4}--\eqref{eq:omega-join-matrix-mult-6} along with $m \leq 2^{i+1}$ and $n \leq 2^{j+1}$, there exists a constant $C$ that satisfies:\footnote{The constant $C$ can be assumed to be independent of the specific values of the $\omega$-dominant triples $(\alpha_i, \beta_i, \zeta_i)$.
    This is because every $\omega$-dominant triple $(\alpha_i, \beta_i, \zeta_i)$
    can be replaced with another $(\alpha_i', \beta_i', \zeta_i') \leq (\alpha_i, \beta_i, \zeta_i)$ where inequality~\eqref{eq:omega-dominant:3} is tight, i.e., where $\alpha_i' + \beta_i' + \zeta_i' = \omega$.
    Such replacement can only expand the $B$-support of the probabilistic interpretation of the MM target, and thus it can only expand the set $E$.}
    \begin{align}
        m^{\alpha_1} \cdot n^{\beta_1} \cdot o^{\zeta_1} &\quad\leq\quad C\cdot B\cdot p_{\bm G}(\bm g)\\
        m^{\alpha_2} \cdot n^{\zeta_2} \cdot o^{\beta_2} &\quad\leq\quad C\cdot B\cdot p_{\bm G}(\bm g)\\
        m^{\zeta_3} \cdot n^{\alpha_3} \cdot o^{\beta_3} &\quad\leq\quad C\cdot B\cdot p_{\bm G}(\bm g)
    \end{align}
    Let $Q_{\bm g}^{ij}$ be the same as $Q$ in Eq.~\eqref{eq:omega-join-matrix-mult:q} but where $\ov E$ is replaced by $\ov E_{\bm g}^{ij}$.
    In order to answer $Q_{\bm g}^{ij}$, we can perform a matrix multiplication between $M_1$ and $M_2$ after filtering them by $\ov E_{\bm g}^{ij}$.
    The filtered matrices have dimensions $(n\times m)$ and $(m\times o)$ respectively.
    The first matrix can be populated in time $O(n\cdot m)$ by iterating over all pairs
    of $(\bm y, \bm x)$ values, and then using the membership oracle for $M_1$.
    The same holds for the second matrix, which can be populated in time $O(m\cdot o)$.
    The multiplication of matrices of dimensions $(n\times m)$ and $(m\times o)$ takes the following time, which also upper bounds the sizes of both matrices, i.e., it also
    upper bounds $n\cdot m$ and $m\cdot o$:
    \begin{multline}
        \max(m\cdot n \cdot o^{\gamma},\quad
            m\cdot n^{\gamma} \cdot o,\quad
            m^{\gamma}\cdot n \cdot o) \leq \\
        \max(m^{\alpha_1} \cdot n^{\beta_1} \cdot o^{\zeta_1},\quad
            m^{\alpha_2} \cdot n^{\zeta_2} \cdot o^{\beta_2},\quad
            m^{\zeta_3} \cdot n^{\alpha_3} \cdot o^{\beta_3})
            \leq C\cdot B\cdot p_{\bm G}(\bm g)
    \end{multline}
    The first inequality above is proved similarly to the proof of Proposition~\ref{prop:mm:upper-bound}.
    This proves that $Q_{\bm g}^{ij}$ can be evaluated in time
    $O(B\cdot p_{\bm G}(\bm g))$.
    Since the total number of pairs of buckets $(i, j)$ is $O(\log^2 N)$, we conclude that $\sigma_{\bm G = \bm g}Q$ can be computed in time $O(B\cdot p_{\bm G}(\bm g)\cdot \log^2 N)$ as claimed.
\end{proof}

%%%%%%%%%%%%%%%%%%%%%%%%%%%%%%%%%%%%%%%%%%%%%%%%%%%%%%%%%%%%%%%%%%%%%%%%%%%%%%%%%%%%%%%%%%%
\subsection{Evaluating $\omega$-Disjunctive Datalog Rules}
\label{subsec:algo:panda:ddr}

In the previous section, we introduced a generalization of DDRs, called $\omega$-DDRs.
We now present an algorithm for evaluating these $\omega$-DDRs.
This algorithm will be used as a building block for our final algorithm for evaluating a BCQ $Q$
in $\omega$-submodular width time in Section~\ref{subsec:algo:panda:osubw}.

\begin{theorem}[Evaluating an $\omega$-DDR]
    Suppose we are given the following:
    \begin{itemize}[leftmargin=*]
        \item An $\omega$-DDR of the form~\eqref{eq:disjunctive-omega} with input schema
        $\Sigma_{\inn}$ and output schema $\Sigma_{\outt}$.
        \item An input database instance $D$ over schema $\Sigma_{\inn}$ of size $N$.
        \item An integral $\omega$-Shannon-flow inequality of the form~\eqref{eq:omega-shannon}
        that satisfies $\norm{\bm\lambda}_1 + \norm{\bm\kappa}_1 > 0$.
    \end{itemize}
    Define the quantity:
    \begin{align}
        B \defeq \left(\prod_{i \in [I]} \deg_{\Sigma_\inn}(\bm Y_i|\bm X_i)^{w_i}\right)^{\frac{1}{\norm{\bm\lambda}_1 + \norm{\bm\kappa}_1}}
        \label{eq:omega-DDR:B}
    \end{align}
    Then, \opandaexpress (Algorithm~\ref{alg:omega-panda-express}) computes a $B$-probabilistic model to the $\omega$-DDR (Definition~\ref{defn:prob:model:omega-ddr})
    in time $O(B\cdot \log N)$.
    \label{thm:panda:ddr}
\end{theorem}

\subsubsection{Example of the $\opandaexpress$ algorithm}
\label{subsubsec:algo:panda:ddr:example}
\begin{table*}
\[
\small
\begin{array}{|l|l|}
    \hline
    \rowcolor{lightgray}\text{\bf $\omega$-Shannon-flow Inequality and Steps} & \text{\bf Probabilistic Inequality and Steps}\\\hline
    \rowcolor{lightergray}
    \begin{aligned}[t]
        &\omega h(XYZ) + h(X) + h(Y) + (\omega-2)h(Z) \\
        &\leq {\color{blue}2h(XY) + (\omega-1) h(YZ) + (\omega-1)h(XZ)}\\
        & {\color{blue}\leq 2\omega}
    \end{aligned}&
    \begin{aligned}[t]
        {\color{red}p_{XY}(xy)^2\cdot p_{YZ}(yz)^{\omega-1} \cdot p_{XZ}(xz)^{\omega-1}
        \geq \frac{1}{N^{2\omega}}}
        = \frac{1}{B^{\omega+1}}\\
        p_{XY}(xy) \defeq \frac{1}{N},\quad p_{YZ}(yz) \defeq \frac{1}{N},\quad p_{XZ}(xz) \defeq \frac{1}{N}
    \end{aligned}\\\hline
    %%%%%%%%%%%%%%%%%%%%%%%%%%%%%%%%%%%%%%%%%%%%%%%%%%%%%%%%%%%%%%%%%%%%%%%%%%%%%%%%%%%%%%%%
    \begin{aligned}[t]
        &\color{blue}h(XY) \to h(X) + h(Y|X)
    \end{aligned}&
    \begin{aligned}[t]
        &{\color{red} p_{XY}(xy) \to p_X(x) \cdot p_{Y|X}(y|x)}\\
        &p_X(x) \defeq \sum_{y} p_{XY}(xy) = \frac{\deg_R(Y|X=x)}{N}\\
        &p_{Y|X}(y|x) \defeq \frac{p_{XY}(xy)}{p_X(x)} = \frac{1}{\deg_R(Y|X=x)}
    \end{aligned}\\\hline
    \begin{aligned}[t]
        &\color{blue}h(Y|X) \to h(Y|XZ)
    \end{aligned}&
    \begin{aligned}[t]
        &{\color{red} p_{Y|X}(y|x) \to p_{Y|XZ}(y|xz)}\\
        &p_{Y|XZ}(y|xz) \defeq p_{Y|X}(y|x) = \frac{1}{\deg_R(Y|X=x)}
    \end{aligned}\\\hline
    \begin{aligned}[t]
        &\color{blue}h(XZ) + h(Y|XZ) \to h(XYZ)
    \end{aligned}&
    \begin{aligned}[t]
        &{\color{red} p_{XZ}(xz) \cdot p_{Y|XZ}(y|xz) \to p'_{XYZ}(xyz)}\\
        &p'_{XYZ}(xyz) \defeq p_{XZ}(xz) \cdot p_{Y|XZ}(y|xz) = \frac{1}{N\cdot\deg_R(Y|X=x)}
    \end{aligned}\\\hline
    %%%%%%%%%%%%%%%%%%%%%%%%%%%%%%%%%%%%%%%%%%%%%%%%%%%%%%%%%%%%%%%%%%%%%%%%%%%%%%%%%%%%%%%%
    \begin{aligned}[t]
        &\color{blue}h(YZ) \to h(Y) + h(Z|Y)
    \end{aligned}&
    \begin{aligned}[t]
        &{\color{red} p_{YZ}(yz) \to p_Y(y) \cdot p_{Z|Y}(z|y)}\\
        &p_Y(y) \defeq \sum_{z} p_{YZ}(yz) = \frac{\deg_S(Z|Y=y)}{N}\\
        &p_{Z|Y}(z|y) \defeq \frac{p_{YZ}(yz)}{p_Y(y)} = \frac{1}{\deg_S(Z|Y=y)}
    \end{aligned}\\\hline
    \begin{aligned}[t]
        &\color{blue}h(Z|Y) \to h(Z|XY)
    \end{aligned}&
    \begin{aligned}[t]
        &{\color{red} p_{Z|Y}(z|y) \to p_{Z|XY}(z|xy)}\\
        &p_{Z|XY}(z|xy) \defeq p_{Z|Y}(z|y) = \frac{1}{\deg_S(Z|Y=y)}
    \end{aligned}\\\hline
    \begin{aligned}[t]
        &\color{blue}h(XY) + h(Z|XY) \to h(XYZ)
    \end{aligned}&
    \begin{aligned}[t]
        &{\color{red} p_{XY}(xy) \cdot p_{Z|XY}(z|xy) \to p''_{XYZ}(xyz)}\\
        &p''_{XYZ}(xyz) \defeq p_{XY}(xy) \cdot p_{Z|XY}(z|xy) = \frac{1}{N\cdot\deg_S(Z|Y=y)}
    \end{aligned}\\\hline
    %%%%%%%%%%%%%%%%%%%%%%%%%%%%%%%%%%%%%%%%%%%%%%%%%%%%%%%%%%%%%%%%%%%%%%%%%%%%%%%%%%%%%%%%
    \begin{aligned}[t]
        &\color{blue}\gamma h(XZ) \to \gamma h(Z) + \gamma h(X|Z)
    \end{aligned}&
    \begin{aligned}[t]
        &{\color{red}p_{XZ}(xz)^\gamma \to p_Z(z)^\gamma \cdot p_{X|Z}(x|z)^\gamma}\\
        &p_Z(z) \defeq \sum_{x} p_{XZ}(xz) = \frac{\deg_T(X|Z=z)}{N}\\
        &p_{X|Z}(x|z) \defeq \frac{p_{XZ}(xz)}{p_Z(z)} = \frac{1}{\deg_T(X|Z=z)}
    \end{aligned}\\\hline
    \begin{aligned}[t]
        &\color{blue}\gamma h(X|Z) \to \gamma h(X|YZ)
    \end{aligned}&
    \begin{aligned}[t]
        &{\color{red}p_{X|Z}(x|z)^\gamma \to p_{X|YZ}(x|yz)^\gamma}\\
        &p_{X|YZ}(x|yz) \defeq p_{X|Z}(x|z) = \frac{1}{\deg_T(X|Z=z)}
    \end{aligned}\\\hline
    \begin{aligned}[t]
        &\color{blue}\gamma h(YZ) + \gamma h(X|YZ) \to \gamma h(XYZ)
    \end{aligned}&
    \begin{aligned}[t]
        &{\color{red}p_{YZ}(yz)^\gamma \cdot p_{X|YZ}(x|yz)^\gamma \to p'''_{XYZ}(xyz)^\gamma}\\
        &p'''_{XYZ}(xyz) \defeq p_{YZ}(yz) \cdot p_{X|YZ}(x|yz) = \frac{1}{N\cdot\deg_T(X|Z=z)}
    \end{aligned}\\\hline
    %%%%%%%%%%%%%%%%%%%%%%%%%%%%%%%%%%%%%%%%%%%%%%%%%%%%%%%%%%%%%%%%%%%%%%%%%%%%%%%%%%%%%%%%
    \rowcolor{lightergray}
    \begin{aligned}[t]
        &\omega h(XYZ) + h(X) + h(Y) + (\omega-2)h(Z) \\
        &\leq {\color{blue}h(XYZ) + h(X) + h(Y) + (\omega-2)h(Z)} \\
        & {\color{blue}\leq  2\omega}
    \end{aligned}&
    \begin{aligned}[t]
        &\color{red}p'_{XYZ}(xyz) \cdot p''_{XYZ}(xyz) \cdot p'''_{XYZ}(xyz)^{\omega-2}\cdot p_X(x) \cdot p_Y(y) \cdot p_Z(z)^{\omega-2} \\
        & {\color{red}\geq \frac{1}{N^{2\omega}}} = \frac{1}{B^{\omega+1}}
    \end{aligned}\\\hline
\end{array}
\]
\caption{Illustration of the $\opandaexpress$ algorithm for evaluating the $\omega$-DDR from Eq.~\eqref{eq:pandaexpress:ddr}, guided by the $\omega$-Shannon-flow inequality from Eq.~\eqref{eq:intro:shannon:triangle}.
In particular, Table~\ref{table:ps} gives a proof sequence that transforms the RHS of inequality~\eqref{eq:intro:shannon:triangle} into the LHS.
Our aim is to mirror this process in the probability world, by transforming inequality~\eqref{eq:opnadaexpress:prob:ineq} into~\eqref{eq:opnadaexpress:prob:ineq:final}.
See Section~\ref{subsubsec:algo:panda:ddr:example} for more details.}
\label{table:opandaexpress}
\end{table*}
Before describing the general $\opandaexpress$ algorithm and proving the associated Theorem~\ref{thm:panda:ddr}, we start in this section with an example that illustrates the main ideas behind the algorithm and its analysis.
Consider the following $\omega$-DDR, whose body is the same as $Q_\triangle$ from Eq.~\eqref{eq:intro:triangle}, and whose head consists of one join target and one MM target:
\begin{align}
    P_1(X, Y, Z) \vee (S_1(X) \wedge T_1(Y) \wedge W_1(Z)) \quad\cd\quad
    R(X, Y) \wedge S(Y, Z) \wedge T(X, Z)
    \label{eq:pandaexpress:ddr}
\end{align}
Suppose that $|R| = |S| = |T| = N$.
Consider the $\omega$-Shannon-flow inequality~\eqref{eq:intro:shannon:triangle}
(or equivalently~\eqref{eq:intro:shannon:triangle:explicit}).
The quantity $B$ from Eq.~\eqref{eq:omega-DDR:B} becomes:
\begin{align}
    B \defeq N^{\frac{2\omega}{\omega+1}}
\end{align}
Note that $\frac{2\omega}{\omega+1}$ is the optimal objective value of the LP~\eqref{eq:intro:inner-lp:triangle}, which gave us the $\omega$-Shannon-flow inequality~\eqref{eq:intro:shannon:triangle} in the first place. This is not a coincidence,
as we will see later (Lemma~\ref{lmm:lp-to-omega-shannon}).
We explain here how we can utilize the $\omega$-Shannon-flow inequality~\eqref{eq:intro:shannon:triangle} in order to compute a $B$-probabilistic model to the $\omega$-DDR~\eqref{eq:pandaexpress:ddr} in time $O(B)$. (The extra $\log N$ is not needed in this example.)

The RHS of inequality~\eqref{eq:intro:shannon:triangle} is as follows:
\begin{align}
    2h(XY) + (\omega-1) h(YZ) + (\omega-1)h(XZ) \quad\leq\quad 2\omega
    \label{eq:opnadaexpress:shannon:ineq}
\end{align}
We initialize three probability distribution $p_{XY}(xy)$, $p_{YZ}(yz)$, and $p_{XZ}(xz)$
to be uniform, i.e., equal to $1/N$ over the relations $R$, $S$, and $T$, respectively.
Now, we have the following probabilistic inequality, which we think of as the {\em probabilistic interpretation} of inequality~\eqref{eq:opnadaexpress:shannon:ineq}:
\begin{align}
    p_{XY}(xy)^2\cdot p_{YZ}(yz)^{\omega-1} \cdot p_{XZ}(xz)^{\omega-1}
        \quad\geq\quad \frac{1}{N^{2\omega}}
        = \frac{1}{B^{\omega+1}}
    \label{eq:opnadaexpress:prob:ineq}
\end{align}
The proof sequence given earlier in Table~\ref{table:ps} transforms the RHS of the $\omega$-Shannon-flow inequality~\eqref{eq:intro:shannon:triangle} into the LHS.
Every time we apply a proof step, we mirror it in the probabilistic world by applying a corresponding transformation on inequality~\eqref{eq:opnadaexpress:prob:ineq}.
The process is illustrated in Table~\ref{table:opandaexpress}.
For example, to mirror the step $h(XY) \to h(X) + h(Y|X)$, we replace the term $p_{XY}(xy)$
in inequality~\eqref{eq:opnadaexpress:prob:ineq} with the product $p_X(x) \cdot p_{Y|X}(y|x)$, where $p_X(x)$ and $p_{Y|X}(y|x)$ are the marginal and conditional measures from Eq.~\eqref{eqn:marginal} and Eq.~\eqref{eqn:conditional}, respectively.
By Eq.~\eqref{eqn:marginal} and~\eqref{eqn:conditional}, this replacement preserves inequality~\eqref{eq:opnadaexpress:prob:ineq}.

At the end of the above process, we would have constructed distributions $p'_{XYZ}$, $p''_{XYZ}$, $p'''_{XYZ}$, $p_X$, $p_Y$, and $p_Z$ that satisfy inequality~\eqref{eq:opnadaexpress:prob:ineq} which now takes the following form:
\begin{align}
    p'_{XYZ}(xyz) \cdot p''_{XYZ}(xyz) \cdot p'''_{XYZ}(xyz)^{\omega-2}\cdot p_X(x) \cdot p_Y(y) \cdot p_Z(z)^{\omega-2}
    \geq \frac{1}{N^{2\omega}} = \frac{1}{B^{\omega+1}}
    \label{eq:opnadaexpress:prob:ineq:final}
\end{align}
The $\opandaexpress$ algorithm works by going through the sequence of probability distributions
defined in Table~\ref{table:opandaexpress} and actually computing them one by one (by materializing their supports along with the associated probabilities).
But there is a catch:
For the distributions $p'_{XYZ}$, $p''_{XYZ}$, and $p'''_{XYZ}$,
we cannot afford to materialize their full supports, since their sizes could exceed $B$.
Instead, we truncate them by keeping only the tuples that have probability at least $1/B$.
This guarantees that, after truncation, each one of these distributions will have support size of at most $B$.
But what about tuples that were truncated away in the above process?
This is were inequality~\eqref{eq:opnadaexpress:prob:ineq:final} comes to the rescue.
These tuples $xyz$ that were truncated away must have
$p'_{XYZ}(xyz)$, $p''_{XYZ}(xyz)$, $p'''_{XYZ}(xyz) < 1/B$.
Substituting this into inequality~\eqref{eq:opnadaexpress:prob:ineq:final}, we get
\begin{align}
    p_X(x)\cdot p_Y(y) \cdot p_Z(z)^{\omega-2} \quad\geq\quad
    \frac{B^\omega}{B^{\omega+1}} = \frac{1}{B}
\end{align}
Therefore, we can report the distributions $(p_X, p_Y, p_Z)$ along with the $\omega$-dominant
triple $(1, 1, \gamma)$ as a probabilistic interpretation of the MM target
$(\{S_1(X), T_1(Y)\}, W_1(Z))$ of the $\omega$-DDR~\eqref{eq:pandaexpress:ddr}.
By Eq.~\eqref{eq:omega-ddr:mm-target-support}, the $B$-support covers all tuples $xyz$ that were truncated away.
Therefore, if we return a $B$-probabilistic model that contains this probabilistic
interpretation of the MM target, along with the (geometric mean of the) truncated distributions $p'_{XYZ}$, $p''_{XYZ}$, and $p'''_{XYZ}$ as the probabilistic interpretation of the join target $P_1(X, Y, Z)$, then this is indeed a valid $B$-probabilistic model
of the $\omega$-DDR~\eqref{eq:pandaexpress:ddr} (Definition~\ref{defn:prob:model:omega-ddr}).

Earlier in Section~\ref{subsec:intro:algo}, we have seen a simpler algorithm, depicted in Figure~\ref{fig:ps-algo:triangle},
where there was no mentioning of probability distributions.
Instead, we just performed a  join, say $T(X,Z) \Join R_\ell(X, Y)$, over all $X$-values that
are {\em light},~i.e., have degree $\leq \Delta \defeq N^{\frac{\omega-1}{\omega+1}}$.
But how are the two algorithms related?
The answer is that they are one and the same!
In particular, in the above $\opandaexpress$ algorithm, the truncation condition $p'_{XYZ}(xyz) \geq 1/B$ corresponds to:
\begin{align}
    \frac{1}{N\cdot\deg_R(Y|X=x)} \geq \frac{1}{B} = \frac{1}{N^{\frac{2\omega}{\omega+1}}}
\end{align}
which is equivalent to $\deg_R(Y|X=x) \leq N^{\frac{\omega-1}{\omega+1}} = \Delta$.
In general, the $\opandaexpress$ algorithm is a generic way to come up with the right
partitioning strategies based on degrees, in order to achieve the desired runtime bounds.

\subsubsection{Description of the $\opandaexpress$ algorithm}

\begin{algorithm}[ht!]
    \caption{\opandaexpress($\calI, \calP$)}
    \begin{algorithmic}[1]
        \Statex {\bf Purpose:} Evaluating an $\omega$-DDR within the time bound from Theorem~\ref{thm:panda:ddr}.
        \vspace{-.25cm}\Statex\hrulefill
        \Statex {\bf Inputs:}
        \begin{itemize}
            \item $\calI$ is an integral $\omega$-Shannon-flow inequality of the form~\eqref{eq:omega-shannon}:
            \begin{align*}
                \sum_{\ell \in [L]} \lambda_\ell h(\bm U_{\ell}) +
                \sum_{j \in [J]} \left(\alpha_j h(\bm X_j|\bm G_j) +
                    \beta_j h(\bm Y_j|\bm G_j) +
                    \zeta_j h(\bm Z_j|\bm G_j) +
                    \kappa_j h(\bm G_j)\right)
                    \quad\leq\quad \sum_{i \in [I]} w_i h(\bm Y_i |\bm X_i)
            \end{align*}
            \item $\calP$ is a collection of sub-probability measures: $\{p_{\bm Y_i|\bm X_i}\}_{i \in [I]}$. See Remark~\ref{rmk:omega-panda-express}.
            \begin{itemize}
                \item If $\calP$ was not given, then it is initialized using Eq.~\eqref{eq:initial:measure}.
            \end{itemize}
        \end{itemize}
        \Statex {\bf Output:}
        \begin{itemize}
            \item A $B$-probabilistic model, where $B$ is given by Eq.~\eqref{eq:omega-DDR:B}.
        \end{itemize}
        \vspace{-.25cm}\Statex\hrulefill
        \If{$\norm{\bm \lambda}_1 + \norm{\bm \kappa}_1 = 0$} \algocomment{If there are no terms left on the LHS}\label{line:opanda-express:no-targets}
            \State \Return $\emptyset$\label{line:opanda-express:no-targets:return}
        \Statex \algocomment{If one term $\lambda_\ell h(\bm U_\ell)$ on the LHS appears on the RHS}
        \ElsIf {there are $\ell \in [L], i \in [I]$ where
        $(\bm U_\ell|\emptyset)=(\bm Y_i|\bm X_i)$ and $\lambda_\ell, w_i > 0$}\label{line:opanda-express:join-target}
            \State \Return $\{p_{\bm U_\ell}\}$\algocomment{Return a probabilistic interpretation of the join target $P_\ell(\bm U_\ell)$}\label{line:opanda-express:join-target:return}
        \Statex \algocomment{If one sum $\alpha_j h(\bm X_j|\bm G_j) +
                    \beta_j h(\bm Y_j|\bm G_j) +
                    \zeta_j h(\bm Z_j|\bm G_j) +
                    \kappa_j h(\bm G_j)$ on the LHS appears fully on the RHS}
        \ElsIf {there are $j \in [J]$ and pairwise-distinct $(i_1, i_2, i_3, i_4) \in [I]^4$ where\\
        \qquad\qquad$(\bm X_j|\bm G_j, \bm Y_j|\bm G_j, \bm Z_j|\bm G_j,\bm G_j|\emptyset) =
        (\bm Y_{i_1}|\bm X_{i_1}, \bm Y_{i_2}|\bm X_{i_2}, \bm Y_{i_3}|\bm X_{i_3}, \bm Y_{i_4}|\bm X_{i_4})$ and\\
        \qquad\qquad$(\alpha_j,\beta_j,\zeta_j,\kappa_j) \leq (w_{i_1}, w_{i_2}, w_{i_3}, w_{i_4})$}\label{line:opanda-express:mm-target}
            \State$\Sigma_\outt^l\gets\{((p_{\bm G_j}, p_{\bm X_j|\bm G_j}, p_{\bm Y_j|\bm G_j}, p_{\bm Z_j|\bm G_j}), (\alpha_j/\kappa_j, \beta_j/\kappa_j, \zeta_j/\kappa_j))\}$\label{line:opanda-express:mm-target:output}
            \Statex\algocomment{$\Sigma_\outt^l$ contains a probabilistic interpretation of the MM target $(\{S_j(\bm X_j \bm G_j), T_j(\bm Y_j \bm G_j)\}, W_j(\bm Z_j\bm G_j))$}
            \Statex\algocomment{We do NOT return immediately in this case!}
            \State Drop $\alpha_j h(\bm X_j|\bm G_j) +
                    \beta_j h(\bm Y_j|\bm G_j) +
                    \zeta_j h(\bm Z_j|\bm G_j) +
                    \kappa_j h(\bm G_j)$ from both sides of $\calI$ to obtain $\calI^h$
            \State Drop the corresponding measures from $\calP$ to obtain $\calP^h$
            \State $\Sigma_\outt^h \gets \opandaexpress(\calI^h, \calP^h)$\label{line:opanda-express:recursive:h:1}
            \State \Return $\Sigma_\outt^l \cup \Sigma_\outt^h$
        \EndIf
        \Statex\algocomment{Continue similar to the original \pandaexpress algorithm~\cite{panda-express}}
        \State Let $s$ be the first proof step in the proof sequence of $\calI$ (Theorem~\ref{thm:proof-sequence})
        \State $(\calI^l, \calP^l) \gets \applystep(s, \calI, \calP)$
        \State $\Sigma_{\outt}^l \gets \opandaexpress(\calI^l, \calP^l)$
        \algocomment{Light branch; continue on with the proof sequence}\label{line:opanda-express:recursive:l:2}
        \If {$s$ is a composition step $h(\bm X) + h(\bm Y | \bm X) \to h(\bm X \bm Y)$}
        \label{alg:tight:panda:reset}
            \State $(\calI^h, \calP^h) \gets \resetineq(\calI^l, \calP^l,\bm X\bm Y)$
            \algocomment{Apply Reset Lemma~\ref{lmm:reset} with $\bm Y_{i_0} = \bm X\bm Y$}
            \State $\Sigma_{\outt}^{h} \gets \opandaexpress(\calI^h, \calP^h)$
            \algocomment{Heavy branch}\label{line:opanda-express:recursive:h:2}
            \State \Return $\Sigma_{\outt}^l \cup \Sigma_{\outt}^{h}$
        \Else
            \State \Return $\Sigma_{\outt}^l$
        \EndIf
    \end{algorithmic}
    \label{alg:omega-panda-express}
\end{algorithm}

Now that we have seen an example of the $\opandaexpress$ algorithm for a specific $\omega$-DDR, we are ready to describe the general algorithm for any $\omega$-DDR.
    The \opandaexpress algorithm generalizes the $\pandaexpress$ algorithm from~\cite{panda-express}, and it is given in Algorithm~\ref{alg:omega-panda-express}.
    We describe it here in detail.
    It is a recursive algorithm that takes two input arguments:
    \begin{itemize}
        \item An integral $\omega$-Shannon-flow inequality $\calI$ of the form~\eqref{eq:omega-shannon}.
        \item A collection, $\calP$, of sub-probability measures: one measure $p_{\bm Y_i|\bm X_i}$ for each term $w_i h(\bm Y_i|\bm X_i)$ on the RHS
        of inequality $\calI$.
    \end{itemize}
    For the outer-most call to $\opandaexpress$, the inequality $\calI$ is set to be
    the integral $\omega$-Shannon-flow inequality from the statement of Theorem~\ref{thm:panda:ddr}.
    Also, for this outer-most call, the collection $\calP$ is initialized from the input database instance $D$ with schema $\Sigma_\inn$ from Theorem~\ref{thm:panda:ddr} as follows:
    For each term $w_i h(\bm Y_i|\bm X_i)$ on the RHS of inequality~\eqref{eq:omega-shannon},
    let $R_i\in\Sigma_\inn$ be a relation that minimizes the degree $\deg_{R_i}(\bm Y_i|\bm X_i)$, i.e., where
    $\deg_{R_i}(\bm Y_i|\bm X_i) = \deg_{\Sigma_\inn}(\bm Y_i|\bm X_i)$.
    Then, the associated sub-probability measure $p_{\bm Y_i|\bm X_i} \in \calP$ is initialized as follows: (Recall the notation ${R}_{|\bm X\cup\bm Y}$ from Section~\ref{subsec:prelims:degrees}.)
    \begin{align}
        p_{\bm Y_i|\bm X_i}(\bm y_i | \bm x_i) \defeq
        \begin{cases}
            \frac{1}{\deg_{R_i}(\bm Y_i|\bm X_i=\bm x_i)} &
                \text{if }(\bm x_i, \bm y_i) \in {R_i}_{|\bm X_i\cup\bm Y_i} \\
            0 & \text{otherwise}
        \end{cases}
        \label{eq:initial:measure}
    \end{align}

    \begin{remark}[Distinction from~\cite{panda-express} regarding how $\calP$ is defined]
        The original $\pandaexpress$ algorithm from~\cite{panda-express}
        assumes that each term $w_i h(\bm Y_i|\bm X_i)$
        on the RHS of the Shannon inequality~\eqref{eq:omega-shannon} (where $w_i$ is a natural number)
        has $w_i$ distinct sub-probability measures
        $p_{\bm Y_i|\bm X_i}^{(1)}, \ldots, p_{\bm Y_i|\bm X_i}^{(w_i)}$, one for each unit of $w_i$.
        In contrast, in our $\omega$-version of the algorithm, $\opandaexpress$,
        we assume that the term $w_i h(\bm Y_i|\bm X_i)$ has
        {\em only one} sub-probability measure
        $p_{\bm Y_i|\bm X_i}$. This is because whenever we have more than one, we can take their {\em geometric mean}; see Proposition~\ref{prop:geometric:mean}.
        We will show later that this is sufficient to maintain all invariants of the algorithm,
        most importantly inequality~\eqref{eq:B:invariant} below.
        \label{rmk:omega-panda-express}
    \end{remark}
    The algorithm proceed as follows:
    \begin{itemize}
        \item If there is no term left on the LHS of the inequality $\calI$, i.e., if we end up with $\norm{\bm \lambda}_1 + \norm{\bm \kappa}_1 = 0$ (line~\ref{line:opanda-express:no-targets} of Algorithm~\ref{alg:omega-panda-express}),
        then we return immediately with an empty $B$-probabilistic model to the $\omega$-DDR.
        \item If there is one term $h(\bm U_\ell)$ on the LHS of $\calI$ that also appears on the RHS as a term $h(\bm Y_i|\bm X_i)$ (which means $\bm X_i = \emptyset$ and $\bm Y_i=\bm U_\ell$),
        then we return a $B$-probabilistic model containing only the corresponding
        sub-probability measure $p_{\bm U_\ell}$.
        This measure $p_{\bm U_\ell}$ serves as a probabilistic interpretation of the join target $P_\ell(\bm U_\ell)$ of the $\omega$-DDR~\eqref{eq:disjunctive-omega}.
        \item If there is a sum $\alpha_j h(\bm X_j|\bm G_j) +
                    \beta_j h(\bm Y_j|\bm G_j) +
                    \zeta_j h(\bm Z_j|\bm G_j) +
                    \kappa_j h(\bm G_j)$ on the LHS of $\calI$ that appears in full on the RHS,
            then we construct a $B$-probabilistic model $\Sigma_\outt^l$
            containing only the tuple
            $((p_{\bm G_j}, p_{\bm X_j|\bm G_j}, p_{\bm Y_j|\bm G_j}, p_{\bm Z_j|\bm G_j}), (\alpha_j/\kappa_j, \beta_j/\kappa_j, \zeta_j/\kappa_j))$.
            This tuple serves as a probabilistic interpretation of the MM target
            $(\{S_j(\bm X_j \bm G_j), T_j(\bm Y_j \bm G_j)\}, W_j(\bm Z_j\bm G_j))$
            of the $\omega$-DDR~\eqref{eq:disjunctive-omega}.
            We do {\em not} immediately return $\Sigma_\outt^l$ in this case.
            Instead, we also drop $\alpha_j h(\bm X_j|\bm G_j) +
                    \beta_j h(\bm Y_j|\bm G_j) +
                    \zeta_j h(\bm Z_j|\bm G_j) +
                    \kappa_j h(\bm G_j)$
            from both sides of $\calI$,
            and recursively call the algorithm on the resulting $\calI$
            to compute another $B$-probabilistic model $\Sigma_\outt^h$.
            We return the union of the two models.
        \item Otherwise, we proceed very similarly to the original $\pandaexpress$
        algorithm from~\cite{panda-express}.
        In particular, by Theorem~\ref{thm:proof-sequence}, $\calI$
        must have a proof sequence.
        Let $s$ be the first step in the sequence.
        We use a sub-routine called $\applystep$ from~\cite{panda-express} (reviewed below) to apply the proof step $s$ to update the RHS of the inequality $\calI$
        and also mirror these changes into $\calP$.
        (Recall that each term on the RHS of $\calI$ is associated with a sub-probability measure in $\calP$, hence when we update the first, we have to update the second.)
        Let $\calI^l$ and $\calP^l$ be the new $\calI$ and $\calP$ after $\applystep$,
        respectively.
        We call the algorithm recursively on $\calI^l$ and $\calP^l$ to compute
        a $B$-probabilistic model $\Sigma_\outt^l$.
        Now we recognize two cases:
        \begin{itemize}
            \item If the proof step $s$ is {\em not} a composition step, then we return $\Sigma_\outt^l$.
            \item Otherwise, i.e., if $s$ has the form $h(\bm X) + h(\bm Y|\bm X)\to h(\bm X\bm Y)$, then applying this step to the RHS of $\calI$
            must have resulted in an inequality $\calI^l$ whose RHS contains $h(\bm X\bm Y)$.
            We use the reset lemma (Lemma~\ref{lmm:reset}) to drop
            this term $h(\bm X\bm Y)$ from the RHS of $\calI^l$
            and also drop the corresponding sub-probability measure from $\calP^l$,
            resulting in $\calI^h$ and $\calP^h$ respectively.
            Now, we recursively call the algorithm on $\calI^h$ and $\calP^h$
            to compute another $B$-probabilistic model $\Sigma_\outt^h$, and we return
            the union of $\Sigma_\outt^h$ and the previously computed $\Sigma_\outt^l$.
        \end{itemize}
    \end{itemize}
For sake of completeness, we describe here the $\applystep$ sub-routine,
which goes back to the original $\pandaexpress$ algorithm~\cite{panda-express}:
\begin{itemize}[leftmargin=*]
    \item If $s$ is a \emph{decomposition} step $h(\bm X \bm Y) \to h(\bm X) + h(\bm Y |
    \bm X)$, then let $p_{\bm X}$ be the marginal measure~\eqref{eqn:marginal} of $p_{\bm X
    \bm Y}$ on $\bm X$ and $p_{\bm Y | \bm X}$ be the conditional
    measure~\eqref{eqn:conditional} of $p_{\bm X \bm Y}$ on $\bm Y$ given $\bm X$.
    The new inequality $\calI^\ell$ results from $\calI$ by replacing the term $h(\bm X \bm Y)$ on the RHS of $\calI$ with the two terms $h(\bm X)$ and $h(\bm Y | \bm X)$.
    The new set of sub-probability measures $\calP^\ell$ results from $\calP$
    by replacing $p_{\bm X\bm Y}$ with the two measures $p_{\bm X}$ and $p_{\bm Y | \bm X}$.
    \item If $s$ is a {\em submodularity} step $h(\bm Y | \bm X) \to h(\bm Y | \bm X \bm Z)$,
    then define $p_{\bm Y | \bm X \bm Z} \defeq p_{\bm Y | \bm X}$.
    The new $\calI^\ell$ results from $\calI$ by replacing $h(\bm Y | \bm X)$ with $h(\bm Y | \bm X \bm Z)$ on the RHS, whereas the new $\calP^\ell$ results from $\calP$ by replacing $p_{\bm Y | \bm X}$ with $p_{\bm Y | \bm X \bm Z}$.
    \item If $s$ is a {\em monotonicity} step $h(\bm X \bm Y) \to h(\bm X)$, then define
    $p_{\bm X}$ as the marginal measure~\eqref{eqn:marginal} of $p_{\bm X \bm Y}$ on $\bm X$.
    The new $\calI^\ell$ results from $\calI$ by replacing $h(\bm X \bm Y)$ with $h(\bm X)$ on the RHS, whereas the new $\calP^\ell$ results from $\calP$ by replacing $p_{\bm X\bm Y}$ with $p_{\bm X}$.
    \item If $s$ is a {\em composition} step $h(\bm X) + h(\bm Y | \bm X) \to h(\bm X \bm Y)$,
    then let $p_{\bm X\bm Y}$ be the {\em truncated} product measure:
    \begin{align}
        p_{\bm X \bm Y}(\bm x, \bm y) \defeq
        \begin{cases}
            p_{\bm X}(\bm x) \cdot p_{\bm Y | \bm X}(\bm y | \bm x),
            & \text{if }
            p_{\bm X}(\bm x) \cdot p_{\bm Y | \bm X}(\bm y | \bm x) \geq 1/B \\
            0, & \text{otherwise}
        \end{cases}
        \label{eqn:truncated:product}
    \end{align}
    The new $\calI^\ell$ results from $\calI$ by replacing the two terms $h(\bm X)$ and $h(\bm Y | \bm X)$ on the RHS of $\calI$ with the single term $h(\bm X \bm Y)$, whereas the new $\calP^\ell$ results from $\calP$ by replacing the two measures $p_{\bm X}$ and $p_{\bm Y | \bm X}$ with the truncated product measure $p_{\bm X\bm Y}$.
\end{itemize}

\subsubsection{Correctness and runtime analysis of $\opandaexpress$}
We now prove Theorem~\ref{thm:panda:ddr} by proving both correctness and the runtime bound of $\opandaexpress$.
\begin{proof}[Proof of Theorem~\ref{thm:panda:ddr}]
    The algorithm is recursive where every recursive call is characterized by a pair $(\calI, \calP)$.
    Every recursive call makes at most two recursive calls.
    Hence, at any point in time, the execution of the algorithm can be modeled by a binary
    tree, where each node is a recursive call.
    The algorithm satisfies the following invariants:
    \begin{itemize}
        \item[(a)] Every inequality $\calI$ is a valid integral $\omega$-Shannon-flow inequality of the form~\eqref{eq:omega-shannon}.
        \item[(b)] Every $p_{\bm Y|\bm X}$ in $\calP$ is a valid sub-probability measure.
        \item[(c)] Every {\em unconditional} measure $p_{\bm Y}$ in $\calP$ satisfies:
        $p_{\bm Y}(\bm y) > 0 \Rightarrow p_{\bm Y}(\bm y) \geq 1/B$ for all $\bm y \in \dom^{\bm Y}$.
        \item[(d)] For every tuple $\bm t \in\Sigma_\inn$, there must exist:
        \begin{itemize}
            \item either an internal (i.e., non-leaf) node where the $B$-support of
            $((p_{\bm G_j}, p_{\bm X_j|\bm G_j}, p_{\bm Y_j|\bm G_j}, p_{\bm Z_j|\bm G_j}),$ $(\alpha_j/\kappa_j, \beta_j/\kappa_j, \zeta_j/\kappa_j))$ from line~\ref{line:opanda-express:mm-target:output}
            contains $\bm t_{\bm X_j\bm Y_j\bm Z_j\bm G_j}$,
            \item or a leaf node $(\calI, \calP)$ that satisfies the following two conditions:
            \begin{align}
                \norm{\bm \lambda}_1 + \norm{\bm \kappa}_1 &\quad>\quad 0\label{eq:B:invariant:norm}\\
                \prod_{i \in [I]} \left(p_{\bm Y_i|\bm X_i}(\bm t)\right)^{w_i}
                &\quad\geq\quad
                \frac{1}{B^{\norm{\bm \lambda}_1 + \norm{\bm \kappa}_1}}
                \label{eq:B:invariant}
            \end{align}
        \end{itemize}
    \end{itemize}
    Invariants (a) and (b) initially hold trivially.
    Similar to~\cite{panda-express}, we can assume that invariant (c) initially holds as well.
    It is straightforward to verify that the three invariants (a), (b), and (c)
    continue to be maintained throughout the execution of Algorithm~\ref{alg:omega-panda-express}.
    Hence, we focus on invariant (d).
    Initially, inequality~\eqref{eq:B:invariant:norm} holds by the statement of Theorem~\ref{thm:panda:ddr}.
    Moreover by definition of $B$
    from Eq.~\eqref{eq:omega-DDR:B} and by Eq.~\eqref{eq:initial:measure}, inequality~\eqref{eq:B:invariant} also holds initially as an equality.

    We now prove that invariant (d) holds inductively.
    Every time a recursive call makes a recursive call, a leaf node in the recursion tree
    becomes an internal node with one or two new children as leaves.
    Note that every internal node $(\calI, \calP)$ must necessarily satisfy 
    Inequality~\eqref{eq:B:invariant:norm}; otherwise, the algorithm would have already returned at line~\ref{line:opanda-express:no-targets:return}.
    For invariant (d) to continue to hold inductively, we need to prove the following two claims:
    \begin{claim}[Invariant (d) holds indutively for the recursive call in line~\ref{line:opanda-express:recursive:h:1} of Algorithm~\ref{alg:omega-panda-express}]
        Consider an arbitrary node $(\calI, \calP)$ that makes a recursive call to node
        $(\calI^h, \calP^h)$ in line~\ref{line:opanda-express:recursive:h:1} of Algorithm~\ref{alg:omega-panda-express}.
        Then, every tuple $\bm t \in \Sigma_\inn$
    that satisfies Inequalities~\eqref{eq:B:invariant:norm} and~\eqref{eq:B:invariant} for the node $(\calI, \calP)$ must satisfy
    either one of the following:
    \begin{itemize}
        \item either $\bm t$ must be covered by the $B$-support of $((p_{\bm G_j}, p_{\bm X_j|\bm G_j}, p_{\bm Y_j|\bm G_j}, p_{\bm Z_j|\bm G_j}),$ $(\alpha_j/\kappa_j, \beta_j/\kappa_j, \zeta_j/\kappa_j))$ from line~\ref{line:opanda-express:mm-target:output},
        \item or $\bm t$ must satisfy Inequalities~\eqref{eq:B:invariant:norm} and~\eqref{eq:B:invariant} for the new leaf $(\calI^h, \calP^h)$.
    \end{itemize}
    \label{clm:opanda-express:recrusive:h:1}
    \end{claim}
    \begin{proof}[Proof of Claim~\ref{clm:opanda-express:recrusive:h:1}]
        Consider a tuple $\bm t \in \Sigma_\inn$ that satisfies Inequalities~\eqref{eq:B:invariant:norm} and~\eqref{eq:B:invariant} for the node $(\calI, \calP)$, but is {\em not} covered by the $B$-support of
    $((p_{\bm G_j}, p_{\bm X_j|\bm G_j}, p_{\bm Y_j|\bm G_j}, p_{\bm Z_j|\bm G_j}),$ $(\alpha_j/\kappa_j, \beta_j/\kappa_j, \zeta_j/\kappa_j))$.
    By definition of $B$-support, this means that $\bm t$ must satisfy:
    \begin{align}
        \left(p_{\bm G_j}(\bm t)\right)^{\kappa_j} \cdot
        \left(p_{\bm X_j|\bm G_j}(\bm t)\right)^{\alpha_j} \cdot
        \left(p_{\bm Y_j|\bm G_j}(\bm t)\right)^{\beta_j} \cdot
        \left(p_{\bm Z_j|\bm G_j}(\bm t)\right)^{\zeta_j} \quad<\quad \frac{1}{B^{\kappa_j}}
        \label{eq:not-in-B-support}
    \end{align}
    We show that $\bm t$ must satisfy Inequalities~\eqref{eq:B:invariant:norm} and~\eqref{eq:B:invariant} for the new leaf $(\calI^h, \calP^h)$, which are as follows:
    \begin{align}
        \norm{\bm \lambda}_1 + \norm{\bm \kappa}_1-\kappa_j &\quad>\quad 0\label{eq:B:invariant:inductive:1:norm}\\
        \frac{\prod_{i \in [I]} \left(p_{\bm Y_i|\bm X_i}(\bm t)\right)^{w_i}}{\left(p_{\bm G_j}(\bm t)\right)^{\kappa_j} \cdot
        \left(p_{\bm X_j|\bm G_j}(\bm t)\right)^{\alpha_j} \cdot
        \left(p_{\bm Y_j|\bm G_j}(\bm t)\right)^{\beta_j} \cdot
        \left(p_{\bm Z_j|\bm G_j}(\bm t)\right)^{\zeta_j}}
        &\quad\geq\quad
        \frac{1}{B^{\norm{\bm \lambda}_1 + \norm{\bm \kappa}_1-\kappa_j}}
        \label{eq:B:invariant:inductive:1}
    \end{align}
    We prove Inequality~\eqref{eq:B:invariant:inductive:1} first as follows:
    \begin{align}
        1&\geq\frac{\prod_{i \in [I]} \left(p_{\bm Y_i|\bm X_i}(\bm t)\right)^{w_i}}{\left(p_{\bm G_j}(\bm t)\right)^{\kappa_j} \cdot
        \left(p_{\bm X_j|\bm G_j}(\bm t)\right)^{\alpha_j} \cdot
        \left(p_{\bm Y_j|\bm G_j}(\bm t)\right)^{\beta_j} \cdot
        \left(p_{\bm Z_j|\bm G_j}(\bm t)\right)^{\zeta_j}}\label{eq:B:invariant:step2}\\
        &\geq \frac{1}{B^{\norm{\bm \lambda}_1 + \norm{\bm \kappa}_1} \cdot \left(p_{\bm G_j}(\bm t)\right)^{\kappa_j} \cdot
        \left(p_{\bm X_j|\bm G_j}(\bm t)\right)^{\alpha_j} \cdot
        \left(p_{\bm Y_j|\bm G_j}(\bm t)\right)^{\beta_j} \cdot
        \left(p_{\bm Z_j|\bm G_j}(\bm t)\right)^{\zeta_j}}\label{eq:B:invariant:step3}\\
        &> \frac{1}{B^{\norm{\bm \lambda}_1 + \norm{\bm \kappa}_1-\kappa_j}}
        \label{eq:B:invariant:step4}
    \end{align}
    Inequality~\eqref{eq:B:invariant:step2} holds by definition of sub-probability measures.
    Inequality~\eqref{eq:B:invariant:step3} holds by the inductive hypothesis that
    node $(\calI, \calP)$ satisfies~\eqref{eq:B:invariant}.
    Finally, inequality~\eqref{eq:B:invariant:step4} holds by Inequality~\eqref{eq:not-in-B-support}.
    This proves Inequality~\eqref{eq:B:invariant:inductive:1}.
    Moreover, the above (strict) chain of inequalities also implies that Inequality~\eqref{eq:B:invariant:inductive:1:norm} holds.
    \end{proof}

    \begin{claim}[Invariant (d) holds inductively for the recursive calls in lines~\ref{line:opanda-express:recursive:l:2} and~\ref{line:opanda-express:recursive:h:2}]
        Consider an arbitrary node $(\calI, \calP)$ that makes two recursive calls to nodes
        $(\calI^l, \calP^l)$ and 
        $(\calI^h, \calP^h)$ in lines~\ref{line:opanda-express:recursive:l:2}
        and~\ref{line:opanda-express:recursive:h:2} of Algorithm~\ref{alg:omega-panda-express}.
        Then, every tuple $\bm t \in \Sigma_\inn$
    that satisfies Inequalities~\eqref{eq:B:invariant:norm} and~\eqref{eq:B:invariant} for the node $(\calI, \calP)$ must satisfy
    either one of the following:
    \begin{itemize}
        \item either $\bm t$ must satisfy Inequalities~\eqref{eq:B:invariant:norm} and~\eqref{eq:B:invariant} for the new leaf $(\calI^l, \calP^l)$,
        \item or $\bm t$ must satisfy Inequalities~\eqref{eq:B:invariant:norm} and~\eqref{eq:B:invariant} for the new leaf $(\calI^h, \calP^h)$.
    \end{itemize}
    \label{clm:opanda-express:recrusive:2}
    \end{claim}
    \begin{proof}[Proof of Claim~\ref{clm:opanda-express:recrusive:2}]
        Consider a tuple $\bm t \in \Sigma_\inn$ that satisfies Inequalities~\eqref{eq:B:invariant:norm} and~\eqref{eq:B:invariant} for the node $(\calI, \calP)$
        but {\em not} for the new leaf $(\calI^l, \calP^l)$.
        We show that $\bm t$ must satisfy Inequalities~\eqref{eq:B:invariant:norm} and~\eqref{eq:B:invariant} for the new leaf $(\calI^h, \calP^h)$.

        If $\bm t$ does not satisfy Inequality~\eqref{eq:B:invariant} for the new leaf $(\calI^l, \calP^l)$, then by Eq.~\eqref{eqn:truncated:product}, this means:
        \begin{align}
            p_{\bm X}(\bm t) \cdot p_{\bm Y|\bm X}(\bm t) < \frac{1}{B}
             \label{eq:not-in-B-support:composition}
        \end{align}
        Inequalities~\eqref{eq:B:invariant:norm} and~\eqref{eq:B:invariant} for the node $(\calI^h, \calP^h)$ are implied by the following two inequalities, which we prove next:
        \begin{align}
            \norm{\bm \lambda}_1 + \norm{\bm \kappa}_1-1 &\quad>\quad 0\label{eq:B:invariant:norm:recursive2}\\
            \frac{\prod_{i \in [I]} \left(p_{\bm Y_i|\bm X_i}(\bm t)\right)^{w_i}}{p_{\bm X}(\bm t) \cdot p_{\bm Y|\bm X}(\bm t)}
            &\quad\geq\quad
            \frac{1}{B^{\norm{\bm \lambda}_1 + \norm{\bm \kappa}_1-1}}
            \label{eq:B:invariant:recursive2}
        \end{align}
        First, we prove Inequality~\eqref{eq:B:invariant:recursive2} as follows:
        \begin{align}
            1 &\geq
            \frac{\prod_{i \in [I]} \left(p_{\bm Y_i|\bm X_i}(\bm t)\right)^{w_i}}{p_{\bm X}(\bm t) \cdot p_{\bm Y|\bm X}(\bm t)}
            \label{eq:B:invariant:step11}\\
            &\geq\frac{1}{B^{\norm{\bm \lambda}_1 + \norm{\bm \kappa}_1}\cdot p_{\bm X}(\bm t) \cdot p_{\bm Y|\bm X}(\bm t)}\label{eq:B:invariant:step22}\\
            &>\frac{1}{B^{\norm{\bm \lambda}_1 + \norm{\bm \kappa}_1-1}}
            \label{eq:B:invariant:step33}
        \end{align}
        Inequality~\eqref{eq:B:invariant:step11} holds by definition of sub-probability measures.
        Inequality~\eqref{eq:B:invariant:step22} holds by the inductive hypothesis that node $(\calI, \calP)$ satisfies~\eqref{eq:B:invariant}.
        Finally, inequality~\eqref{eq:B:invariant:step33} holds by Inequality~\eqref{eq:not-in-B-support:composition}.
        This proves Inequality~\eqref{eq:B:invariant:recursive2}.
        Moreover, the above (strict) chain of inequalities also implies that Inequality~\eqref{eq:B:invariant:norm:recursive2} holds.
    \end{proof}
    We now prove correctness and the runtime bound of the algorithm.
    \begin{claim}[Algorithm~\ref{alg:omega-panda-express} is correct]
        Algorithm~\ref{alg:omega-panda-express} returns a valid $B$-probabilistic model of the 
        input $\omega$-DDR over the input database instance $D$.
        \label{clm:opanda-express:correctness}
    \end{claim}
    \begin{proof}[Proof of Claim~\ref{clm:opanda-express:correctness}]
        This follows from invariant (d).
        In particular, consider the {\em completed} recursion tree of Algorithm~\ref{alg:omega-panda-express}, i.e., after all recursive calls have been made
        and the algorithm has terminated.
        The invariant says that for every tuple $\bm t\in\Sigma_\inn$:
        \begin{itemize}
            \item Either $\bm t$ is already covered by the $B$-support of some probabilistic interpretation
            of an MM target in line~\ref{line:opanda-express:mm-target:output}. This probabilistic interpretation is eventually returned
            by the outer most recursive call as part of the final $B$-probabilistic model.
            \item Or there exists some leaf that satisfies Inequalities~\eqref{eq:B:invariant:norm} and~\eqref{eq:B:invariant} for the leaf. By~\eqref{eq:B:invariant:norm},
            this leaf cannot return in line~\ref{line:opanda-express:no-targets:return}.
            Moreover, because the recursion tree is completed, this leaf cannot make any recursive calls. Hence, it must return in line~\ref{line:opanda-express:join-target:return}.
            By Eq.~\eqref{eq:B:invariant}, $\bm t$ must be covered by the support
            of $p_{\bm U_\ell}$ in line~\ref{line:opanda-express:join-target:return}.
            Moreover, by invariant (c), this implies that $\bm t$ is covered by the $B$-support of $p_{\bm U_\ell}$.
        \end{itemize}
    \end{proof}
    \begin{claim}[Runtime]
        Algorithm~\ref{alg:omega-panda-express} runs in time $O(B\cdot \log N)$
        in data complexity.
        \label{clm:opanda-express:runtime}
    \end{claim}
    \begin{proof}[Proof of Claim~\ref{clm:opanda-express:runtime}]
        Just like in the original \pandaexpress algorithm~\cite{panda-express}, the depth of the recursion tree does not depend on the data, hence it is a constant in data complexity.
        Since the branching factor is at most 2, the entire size of the recursion tree is also a constant.

        The runtime of the algorithm is dominated by the runtime of data dependent operations
        at every node.
        These data dependent operations are all inside the $\applystep$ sub-routine, and they are of the following types
        \begin{itemize}
            \item Given a sub-probability measure $p_{\bm X\bm Y}$, compute its marginal
            $p_{\bm X}$ and/or conditional $p_{\bm X|\bm Y}$ from Eq.~\eqref{eqn:marginal}
            and~\eqref{eqn:conditional} respectively.
            \item Given two sub-probability measures $p_{\bm X}$ and $p_{\bm Y|\bm X}$, compute their {\em truncated} product measure $p_{\bm X\bm Y}$ from Eq.~\eqref{eqn:truncated:product}.
        \end{itemize}
        Invariant (c) implies that every unconditional measure $p_{\bm X}$ and $p_{\bm X\bm Y}$
        above has a support whose size is at most $B$.
        Hence, marginals and conditionals can be computed in time $O(B)$.
        Moreover, by Eq.~\eqref{eqn:truncated:product}, the output of the truncated product
        also has support of size at most $B$, and can be computed in time $O(B \cdot \log N)$,
        where the extra $\log N$ factor is needed to sort the input measure $p_{\bm Y|\bm X}$
        for every $\bm x$.
    \end{proof}
\end{proof}
%%%%%%%%%%%%%%%%%%%%%%%%%%%%%%%%%%%%%%%%%%%%%%%%%%%%%%%%%%%%%%%%%%%%%%%%%%%%%%%%%%%%%%%%%%%

\subsection{Evaluating a BCQ $Q$ in $\omega$-submodular width time}
\label{subsec:algo:panda:osubw}
We are now ready to prove Theorem~\ref{thm:panda:osubw} about answering Boolean conjunctive queries in $\omega$-submodular width time.

\ThmMainResult*

In order to prove Theorem~\ref{thm:panda:osubw}, we will need some further preliminaries.
The following lemma is very similar to a lemma in \cite{DBLP:conf/pods/Khamis0S17,theoretics:13722}. Its proof relies on the observation that every feasible dual solution to the linear program~\eqref{eq:osubw:inner-lp} corresponds to an $\omega$-Shannon-flow inequality~\eqref{eq:omega-shannon} of a certain
structure. By strong duality, if we pick an optimal dual solution,
we can match the optimal objective value of the (primal) LP~\eqref{eq:osubw:inner-lp}.
\begin{lemma}[From LP~\eqref{eq:osubw:inner-lp} to an $\omega$-Shannon-flow Inequality~\eqref{eq:omega-shannon}]
    Given a hypergraph $\calH=(\calV,\calE)$, let $\ed\defeq \ed_{\calH}$, and consider a linear program of the following form
    (which is the same as LP~\eqref{eq:osubw:inner-lp} but where we drop the index $i\in [\calI]$ to reduce clutter):
    \begin{align}
        \max_{t, \bm h \in \Gamma \cap \ed}
        \bigl\{
            t \quad\mid\quad &\forall \ell \in [L],\quad
            t\leq h(\bm U_{\ell}),\nonumber\\
            &\forall j \in [J],\quad
            t \leq h(\bm X_{j}|\bm G_{j})+
        h(\bm Y_{j}|\bm G_{j}) +
        \gamma h(\bm Z_{j}|\bm G_{j})
        +h(\bm G_{j})
        \bigr\}
        \label{eq:osubw:inner-lp:simplified}
    \end{align}
    Let $\opt$ be the optimal objective value of the above LP.
    Then, there must exist
    an $\omega$-Shannon-flow inequality of the form~\eqref{eq:omega-shannon} that satisfies the following:
    \begin{itemize}[leftmargin=*]
        \item For each $i \in [I]$, we have $\bm X_i = \emptyset$ and $\bm Y_i \in \calE$.
        \item For each $j \in [J]$, we have $\alpha_j = \beta_j = \kappa_j$ and
        $\zeta_j = \kappa_j\cdot \gamma$.
        \item The coefficients of the inequality satisfy $\norm{\bm\lambda}_1+\norm{\bm\kappa}_1 > 0$ and:
        \begin{align}
            \frac{\norm{\bm w}_1}{\norm{\bm\lambda}_1+\norm{\bm\kappa}_1} = \opt.
        \end{align}
    \end{itemize}
    Moreover, if $\omega$ is rational, then the above $\omega$-Shannon-flow inequality can be chosen to be integral (Definition~\ref{defn:integral-omega-shannon}).
    \label{lmm:lp-to-omega-shannon}
\end{lemma}

\begin{example}
    As an example of Lemma~\ref{lmm:lp-to-omega-shannon},
    the LP~\eqref{eq:intro:inner-lp:triangle}
    has an optimal objective value of $\frac{2\omega}{\omega+1}$,
    and it results from the $\omega$-Shannon-flow inequality~\eqref{eq:intro:shannon:triangle}.
\end{example}

Given a hypergraph $\calH=(\calV,\calE)$ and a generalized variable elimination order $\ov{\bm\sigma} =$ $(\bm X_1, \bm X_2, \ldots, \bm X_{|\ov{\bm\sigma}|})$ $\in\ov{\pi}(\calV)$,
each set of vertices $\bm X_i$ can be eliminated by either a join algorithm or some matrix
multiplication. The concept of an {\em $\omega$-query plan}, that we define below,
amends a generalized variable elimination order with a mapping that specifies how to eliminate each variable.
\begin{definition}[$\omega$-Query Plan]
    Given a hypergraph $\calH=(\calV,\calE)$, an {\em $\omega$-query plan} is a pair
    $(\ov{\bm\sigma}, e)$ where:
    \begin{itemize}[leftmargin=*]
        \item $\ov{\bm\sigma}\in\ov{\pi}(\calV)$ is a generalized variable elimination order.
        \item $e$ is a function that maps every index $i \in [|\ov{\bm\sigma}|]$
        (that satisfies $U_i^{\ov{\bm\sigma}}\not\subseteq U_j^{\ov{\bm\sigma}}, \forall j \in [i-1]$)
        to a term $e(i)\in\{h(U^{\ov{\bm\sigma}}_i)\} \cup \args(\emm^{\ov{\bm\sigma}}_i)$.
    \end{itemize}
    \label{defn:omega-query-plan}
\end{definition}
Note that, for a given hypergraph $\calH$, the number of $\omega$-query plans is finite.
\begin{proof}[Proof of Theorem~\ref{thm:panda:osubw}]
    The algorithm is given in Algorithm~\ref{alg:omega-panda-express:outer}.
    Let $\calH=(\calV,\calE)$ be the query hypergraph.
    Consider the $\omega$-submodular width of $\calH$ written in the form of Eq.~\eqref{eq:osubw:distributed:swapped:raw}.
    Let $k$ be the number of vertices of $\calH$.
    Recall from Section~\ref{sec:computing-osubw} that $\calF$
    is the set of functions $f:\ov{\pi}(\calV)\to [k]$ that map every generalized
variable elimination order $\ov{\bm\sigma}\in\ov{\pi}(\calV)$ to a number $i$ between 1 and $|\ov{\bm\sigma}|$
that satisfies $U^{\ov{\bm\sigma}}_i \not\subseteq U^{\ov{\bm\sigma}}_j$ for all $j \in [i-1]$.
For a fixed function $f \in\calF$, let
$\calG_f$ be the set of functions $g_f$ that map every pair $(\ov{\bm\sigma}, \mm(\bm X;\bm Y;\bm Z|\bm G))$
where $\ov{\bm\sigma}\in\ov{\pi}(\calV)$ and
$\mm(\bm X;\bm Y;\bm Z|\bm G)\in\args(\emm^{\ov{\bm\sigma}}_{f(\ov{\bm\sigma})})$ to one
of the three terms in $\args(\mm(\bm X;\bm Y;\bm Z|\bm G))$ from Eq.~\eqref{eq:mm}.
In particular, $g_f(\ov{\bm\sigma}, \mm(\bm X;\bm Y;\bm Z|\bm G))$ is equal to either one of the following:
    \begin{align*}
        &h(\bm X|\bm G)+h(\bm Y|\bm G)+\gamma\cdot h(\bm Z|\bm G)+h(\bm G),\quad\\
        &h(\bm X|\bm G)+\gamma\cdot h(\bm Y|\bm G)+h(\bm Z|\bm G)+h(\bm G),\quad\\
        &\gamma\cdot h(\bm X|\bm G)+h(\bm Y|\bm G)+h(\bm Z|\bm G)+h(\bm G)
    \end{align*}
    For each $f \in \calF, g_f \in \calG_f$, we construct an $\omega$-DDR of the form~\eqref{eq:disjunctive-omega}:
    \begin{align}
        \bigvee_{\ov{\bm \sigma} \in \ov{\pi}(\calV)} P(U_{f(\ov{\bm\sigma})}^{\bm{\ov\sigma}}) \vee
        \bigvee_{\substack{\ov{\bm \sigma} \in \ov{\pi}(\calV)\\ \mm(\bm X;\bm Y;\bm Z|\bm G)\in\args(\emm^{\ov{\bm \sigma}}_{f(\bm {\ov \sigma})})\\
        (\{\bm X', \bm Y'\}, \bm Z'):\;\;g_f(\ov {\bm \sigma}, \mm(\bm X;\bm Y;\bm Z|\bm G))=\\ h(\bm X'|\bm G)+h(\bm Y'|\bm G)+\gamma\cdot h(\bm Z'|\bm G)+h(\bm G)}}\left(S(\bm X'\bm G)\wedge T(\bm Y'\bm G)\wedge W(\bm Z'\bm G)
        \right)
        \cd \bigwedge_{R(\bm Z) \in \Sigma_{\inn}} R(\bm Z)
        \label{eq:omega-ddr:subw:raw}
    \end{align}
    In particular, $(\{\bm X', \bm Y'\}, \bm Z')$ in Eq.~\eqref{eq:omega-ddr:subw:raw} is a partial order of $(\bm X, \bm Y, \bm Z)$ that is picked according to $g_f(\ov {\bm \sigma}, \mm(\bm X;\bm Y;\bm Z|\bm G))$.
    In this partial order, $\bm X'$ and $\bm Y'$ are interchangeable, but $\bm Z'$ is not interchangeable with either one of them.
    This mirrors the way MM targets of an $\omega$-DDR are defined in Definition~\ref{defn:omega-ddr:targets}, where we think of the conjunction $S(\bm X'\bm G) \wedge T(\bm Y'\bm G) \wedge W(\bm Z'\bm G)$ as being partially ordered $(\{S(\bm X'\bm G) \wedge T(\bm Y'\bm G)\}, W(\bm Z'\bm G))$ such that $S(\bm X'\bm G)$ and $T(\bm Y'\bm G)$ are interchangeable but $W(\bm Z'\bm G)$ is not interchangeable with either one of them.

    Just like going from Eq.~\eqref{eq:osubw:distributed:swapped:raw} to Eq.~\eqref{eq:osubw:distributed:swapped}, we can re-index the above $\omega$-DDRs so that
    we have one $\omega$-DDR for each $i \in [\calI]$:
    \begin{align}
        \bigvee_{\ell \in [L_i]} P(\bm U_{i\ell}) \vee
        \bigvee_{j \in [J_i]}\left(S(\bm X_{ij}\bm G_{ij})\wedge T(\bm Y_{ij}\bm G_{ij})\wedge W(\bm Z_{ij}\bm G_{ij})
        \right)
        \cd \bigwedge_{R(\bm Z) \in \Sigma_{\inn}} R(\bm Z)
        \label{eq:omega-ddr:subw}
    \end{align}
    By Eq.~\eqref{eq:osubw:distributed:swapped}, for each $i \in [\calI]$, the corresponding LP~\eqref{eq:osubw:inner-lp} has an optimal objective value that is upper bounded
    by $\osubw(\calH)$.
    Consider a fixed LP~\eqref{eq:osubw:inner-lp}, written in the form of Eq.~\eqref{eq:osubw:inner-lp:simplified} (where we drop the index $i$ to reduce clutter), and let $\opt$ be its optimal objective value.
    By Lemma~\ref{lmm:lp-to-omega-shannon},
    there exists an $\omega$-Shannon-flow inequality of the form~\eqref{eq:omega-shannon} that satisfies:
    \begin{align*}
        \frac{\norm{\bm w}_1}{\norm{\bm\lambda}_1+\norm{\bm\kappa}_1} = \opt\leq \osubw(\calH).
    \end{align*}
    Moreover, for each $i \in [I]$, the above inequality satisfies $\bm X_i = \emptyset$ and $\bm Y_i \in \calE$,
    hence $\deg_{\Sigma_\inn}(\bm Y_i|\bm X_i)\leq N$.
    Define $B \defeq N^{\opt}$.
    Then by Theorem~\ref{thm:panda:ddr},
    we can, in time $O(B\cdot \log N) = O(N^{\osubw(\calH)}\cdot \log N)$,
    compute a $B$-probabilistic model for the corresponding $\omega$-DDR~\eqref{eq:omega-ddr:subw}, or equivalently the $\omega$-DDR~\eqref{eq:omega-ddr:subw:raw}.
    We repeat the above for every $i \in [\calI]$, or equivalently
    for every $f \in \calF, g_f \in \calG_f$.
    The value of $B$ depends on the specific $\omega$-DDR, but it is always upper bounded by $\ov B$ defined as:
    \begin{align}
        \ov B \defeq N^{\osubw(\calH)}. 
    \end{align}
    This is because $\opt\leq \osubw(\calH)$.
    Therefore, we by Proposition~\ref{prop:model:omega-ddr:monotonicity}, we can assume
    that all the computed probabilistic models are $\ov B$-probabilistic models.

    \begin{example}
        For the triangle query $Q_\triangle$ from Eq.~\eqref{eq:intro:triangle},
        the submodular width formulation from Eq.~\eqref{eq:osubw:distributed:swapped:raw} takes the form~\eqref{eq:intro:osubw:triangle:distributed}, which is a maximum of three LPs.
        % The first LP is given more explicitly in Eq.~\eqref{eq:intro:inner-lp:triangle}.
        These three LPs correspond to the following three $\omega$-DDRs of the form~\eqref{eq:omega-ddr:subw:raw}, respectively:
        \begin{align}
            P_1(X, Y, Z) \vee (S_1(X) \wedge T_1(Y) \wedge W_1(Z)) \quad\cd\quad
                R(X, Y) \wedge S(Y, Z) \wedge T(Z, X)\label{eq:triangle:ddr1}\\
            P_2(X, Y, Z) \vee (S_2(X) \wedge T_2(Z) \wedge W_2(Y)) \quad\cd\quad
                R(X, Y) \wedge S(Y, Z) \wedge T(Z, X)\label{eq:triangle:ddr2}\\
            P_3(X, Y, Z) \vee (S_3(Y) \wedge T_3(Z) \wedge W_3(X)) \quad\cd\quad
                R(X, Y) \wedge S(Y, Z) \wedge T(Z, X)\label{eq:triangle:ddr3}
        \end{align}
        Consider the first $\omega$-DDR in Eq.~\eqref{eq:triangle:ddr1} (which is the same as Eq.~\eqref{eq:pandaexpress:ddr} from Section~\ref{subsubsec:algo:panda:ddr:example}).
        The corresponding LP is given more explicitly in Eq.~\eqref{eq:intro:inner-lp:triangle},
        and it has an optimal objective value of $\frac{2\omega}{\omega+1}$,
        which also happens to be the value of $\osubw(Q_\triangle)$.
        By Lemma~\ref{lmm:lp-to-omega-shannon}, we can derive the $\omega$-Shannon-flow inequality in Eq.~\eqref{eq:intro:shannon:triangle} from the LP in Eq.~\eqref{eq:intro:inner-lp:triangle}.
        Guided by this $\omega$-Shannon-flow inequality, we can use Theorem~\ref{thm:panda:ddr} to compute a $B$-probabilistic model for the $\omega$-DDR in Eq.~\eqref{eq:triangle:ddr1} in time $O(B\cdot \log N)$, where $B = N^{\frac{2\omega}{\omega+1}}$.
        This process is explained in detail in Section~\ref{subsubsec:algo:panda:ddr:example}.
        We can do the same for the other two $\omega$-DDRs.
    \end{example}

    Recall from Definition~\ref{defn:omega-ddr:targets} that every $\omega$-DDR of the form~\eqref{eq:omega-ddr:subw:raw} has two types of {\em targets}:
    \begin{itemize}
        \item {\em Join targets} of the form $P(\bm U_{f(\ov{\bm\sigma})}^{\bm{\ov\sigma}})$.
        \item {\em MM targets} of the form $(\{S(\bm X'\bm G), T(\bm Y'\bm G)\}, W(\bm Z'\bm G))$. In these targets, $S$ and $T$ are interchangeable, but $W$ is not interchangeable with either of them.
    \end{itemize}
    \begin{definition}[Target Selector]
        A {\em target selector} $\calT$ is a set of targets consisting of one target from each $\omega$-DDR of the form~\eqref{eq:omega-ddr:subw:raw}.
        \label{defn:target-selector}
    \end{definition}
    \begin{definition}[$\omega$-query plan supported by a target selector]
        Given an $\omega$-query plan $(\ov{\bm\sigma}, e)$,
        a target selector $\calT$ is said to {\em support} $(\ov{\bm\sigma}, e)$
        iff for every $i \in [|\ov{\bm\sigma}|]$ that satisfies $U_i^{\ov{\bm\sigma}}\not\subseteq U_j^{\ov{\bm\sigma}}, \forall j \in [i-1]$,
        we have:
        \begin{itemize}
            \item If $e(i) = h(U_i^{\ov{\bm \sigma}})$, then $\calT$ contains a join target $P(U_{i}^{\ov{\bm \sigma}})$.
            \item If $e(i) = \mm(\bm X;\bm Y;\bm Z|\bm G) \in \args(\emm_i^{\ov {\bm \sigma}})$, then $\calT$ contains three MM targets of the forms:
            \begin{align}
                &(\{S_1(\bm X\bm G), T_1(\bm Y\bm G)\}, W_1(\bm Z\bm G)),\nonumber\\
                &(\{S_2(\bm X\bm G), T_2(\bm Z\bm G)\}, W_2(\bm Y\bm G)),\nonumber\\
                &(\{S_3(\bm Y\bm G), T_3(\bm Z\bm G)\}, W_3(\bm X\bm G)).
                \label{eq:mm:targets}
            \end{align}
        \end{itemize}
        \label{defn:query-plan-supported}
    \end{definition}

    The following proposition is the {\em core} of this section, and it is the key to proving Theorem~\ref{thm:panda:osubw}.
    \begin{proposition}
        Every target selector $\calT$ must support at least one $\omega$-query plan $(\ov{\bm\sigma}, e)$.
        \label{prop:target}
    \end{proposition}
    \begin{example}
        Before proving Proposition~\ref{prop:target}, we give a quick example.
        Consider the $\omega$-DDRs in Eq.~\eqref{eq:triangle:ddr1},~\eqref{eq:triangle:ddr2}, and~\eqref{eq:triangle:ddr3} for the triangle query $Q_\triangle$.
        Each one of these $\omega$-DDRs has two targets: one join target and one MM target.
        Hence, there are $2^3=8$ possible target selectors:
        \begin{itemize}
            \item There are 7 target selectors that contain at least one join target $P_i(X, Y, Z)$.
            Consider any one of them.
            Such a target selector does indeed support an $\omega$-query plan $(\ov {\bm \sigma}, e)$.
            For example, we could pick $\ov{\bm\sigma} = (X, Y, Z)$, and define $e(X)$
            to be $h(X Y Z)$. We don't need to define $e(Y)$ and $e(Z)$ since $U_2^{\ov{\bm\sigma}}$ and $U_3^{\ov{\bm\sigma}}$ are subsets of $U_1^{\ov{\bm\sigma}}$.
            \item Finally, consider the target selector $\calT$ that contains no join target,
            hence contains all three MM targets. This target selector also supports an $\omega$-query plan $(\ov {\bm \sigma}, e)$.
            For example, we could pick $\ov{\bm\sigma} = (X, Y, Z)$, and define $e(X)$ to be $\mm(X; Y; Z)$.
            Just like above, we don't need to define $e(Y)$ and $e(Z)$.
        \end{itemize}
        Section~\ref{subsec:algo:4cycle} gives a more complicated example of Proposition~\ref{prop:target} for the 4-cycle query $Q_\square$ from Eq.~\eqref{eq:intro:4cycle}.
    \end{example}
    \begin{proof}[Proof of Proposition~\ref{prop:target}]
        Let $\calT$ be a target selector, and assume for the sake of contradiction that $\calT$ does not support any $\omega$-query plan $(\ov{\bm\sigma}, e)$.
        What this means is that for every generalized variable elimination order $\ov{\bm\sigma}\in\ov{\pi}(\calV)$, there is no function $e$ that makes
        this $\omega$-query plan $(\ov{\bm\sigma}, e)$ supported by $\calT$.
        In other words:
        \begin{claim}
            Assuming Proposition~\ref{prop:target} is false,
                for every generalized variable elimination order $\ov{\bm\sigma}\in\ov{\pi}(\calV)$, there must exist some $i \in [|\ov{\bm\sigma}|]$ that satisfies $U_i^{\ov{\bm\sigma}}\not\subseteq U_j^{\ov{\bm\sigma}}, \forall j \in [i-1]$ and:
            \begin{itemize}
                \item $\calT$ does not contain a join target $P(U_{i}^{\ov{\bm \sigma}})$, and
                \item for every $\mm(\bm X;\bm Y; \bm Z|\bm G)\in\args(\emm_i^{\ov{\bm\sigma}})$,
                $\calT$ is missing at least one of the three MM targets from Eq.~\eqref{eq:mm:targets}.
            \end{itemize}
            \label{clm:target:selector}
        \end{claim}
        We use the above claim to construct functions $f^*\in \calF$ and $g^*\in\calG_{f^*}$ as follows:
        (Recall the definitions of $\calF$ and $\calG_f$ for $f \in \calF$ from Section~\ref{sec:computing-osubw}.)
        For every $\ov{\bm\sigma}\in\ov{\pi}(\calV)$, let $f^*(\ov{\bm\sigma})$ be equal
        to $i \in [|\ov{\bm\sigma}|]$ that is assumed in Claim~\ref{clm:target:selector}, breaking ties arbitrarily if more than one choice for $i$ is possible.
        Moreover, for every $\mm(\bm X;\bm Y;\bm Z|\bm G)\in \args(\emm_{f^*(\ov{\bm \sigma})}^{\ov{\bm\sigma}})$,
        let $g^*(\ov{\bm\sigma}, \mm(\bm X;\bm Y;\bm Z|\bm G))$ be equal to one of the following three terms, depending on which of the three MM targets from Eq.~\eqref{eq:mm:targets} is missing from $\calT$:
        \begin{multline}
            g^*(\ov{\bm\sigma}, \mm(\bm X;\bm Y;\bm Z|\bm G)) \defeq\\
                \begin{cases}
                    h(\bm X|\bm G)+h(\bm Y|\bm G)+\gamma h(\bm Z|\bm G)+h(\bm G),& \text{if $(\{S_1(\bm X\bm G), T_1(\bm Y\bm G)\}, W_1(\bm Z\bm G))$ is missing from $\calT$}\\
                    h(\bm X|\bm G)+\gamma h(\bm Y|\bm G)+h(\bm Z|\bm G)+h(\bm G),& \text{else if $(\{S_2(\bm X\bm G), T_2(\bm Z\bm G)\}, W_2(\bm Y\bm G))$ is missing from $\calT$}\\
                    \gamma h(\bm X|\bm G)+h(\bm Y|\bm G)+h(\bm Z|\bm G)+h(\bm G),& \text{else if $(\{S_3(\bm Y\bm G), T_3(\bm Z\bm G)\}, W_3(\bm X\bm G))$ is missing from $\calT$}
                \end{cases}
        \end{multline}
        But then by construction, there must exist an $\omega$-DDR of the form~\eqref{eq:omega-ddr:subw:raw} that is associated with $f:=f^*$ and $g_f:=g^*$.
        Hence, the target selector $\calT$ must contain a target from this $\omega$-DDR.
        This contradicts Claim~\ref{clm:target:selector}.
    \end{proof}

    As explained before, we use Theorem~\ref{thm:panda:ddr} to compute a $B$-probabilistic model for each $\omega$-DDR of the form~\eqref{eq:omega-ddr:subw:raw},
    where the value of $B$ depends on the $\omega$-DDR.
    However, $B \leq (\ov B \defeq N^{\osubw(Q)})$,
    hence, by Proposition~\ref{prop:model:omega-ddr:monotonicity},
    we can assume that we have computed a $\ov B$-probabilistic model for each $\omega$-DDR of the form~\eqref{eq:omega-ddr:subw:raw}.
    As a result, every target in every $\omega$-DDR of the form~\eqref{eq:omega-ddr:subw:raw} now has a probabilistic interpretation.
    Hence, every target in every target selector also has a probabilistic interpretation.
    Those interpretations form a {\em probabilistic interpretation} for the target selector, as defined below.

    \begin{definition}[Probabilistic interpretation of a target selector]
        Given a target selector $\calT$, a {\em probabilistic interpretation} $\calP_\calT$ for $\calT$ consists of a probabilistic interpretation of each target in $\calT$ (which could be either a join target or an MM target).
        \label{defn:target-selector:probabilistic-interpretation}
    \end{definition}
    \begin{definition}[$\ov B$-support of a probabilistic interpretation of a target selector]
        Given a probabilistic interpretation $\calP_\calT$ of a target selector $\calT$
        and a number $\ov B \in \R_+$, we define the {\em $\ov B$-support} of $\calP_\calT$ to be the set of tuples $\bm t \in\bigjoin{\Sigma_\inn}$ that are covered\footnote{Recall what ``covered'' means from Section~\ref{subsec:prelims:ddr}.}
        by the $\ov B$-support of the probabilistic interpretation of {\em every} target in $\calT$.
        \label{defn:target-selector:support}
    \end{definition}
    \begin{proposition}
        Every tuple $\bm t \in \bigjoin\Sigma_\inn$ must occur in the $\ov B$-support of the probabilistic interpretation $\calP_\calT$ of some target selector $\calT$.
        \label{prop:target-selector:completeness}
    \end{proposition}
    \begin{example}
        Suppose we computed a $\ov B$-probabilistic model for each of the three $\omega$-DDRs in Eqs.~\eqref{eq:triangle:ddr1}, \eqref{eq:triangle:ddr2}, and~\eqref{eq:triangle:ddr3} for the triangle query $Q_\triangle$.
        By definition of a $\ov B$-probabilistic model, for each
        tuple $(x, y, z) \in R \Join S\Join T$ and for each one of the three $\omega$-DDRs,
        $(x, y, z)$ must be accounted for 
        by either the join target or the MM target of this $\omega$-DDR.
        Hence, $(x, y, z)$ must be accounted for by at least one of the eight target selectors
        that were mentioned before.
    \end{example}
    \begin{proof}[Proof of Proposition~\ref{prop:target-selector:completeness}]
        Fix an arbitrary tuple $\bm t \in \bigjoin\Sigma_\inn$.
        By definition of a $\ov B$-probabilistic model,
        for every $\omega$-DDR of the form~\eqref{eq:omega-ddr:subw:raw}, $\bm t$ must be covered by the $\ov B$-support of the probabilistic interpretation of at least one target of this $\omega$-DDR.
        If we pick those targets, we obtain a target selector $\calT$
        satisfying the proposition for $\bm t$.
    \end{proof}
    Proposition~\ref{prop:target} says that every target selector $\calT$ supports at least one $\omega$-query plan $(\ov{\bm\sigma}, e)$.
    The following proposition roughly says that we can use this $\omega$-query plan
    to evaluate the original query $Q$ filtered by (the $\ov B$-supports of the probabilistic
    interpretations of) the targets that appear in this particular query plan.
    These targets are a subset of the targets in $\calT$.
    \begin{proposition}[Evaluating $Q$ filtered by a target selector $\calT$]
        Given a target selector $\calT$ with a probabilistic interpretation $\calP_\calT$,
        let $\calS_\calT$ be the $\ov B$-support of $\calP_\calT$.
        Then, there exists a superset $\ov \calS_\calT \supseteq \calS_\calT$ where the following BCQ $Q_\calT$ can be computed in time $O(\ov B \cdot \log^2 N)$ in data complexity:
        \begin{align}
            Q_\calT() \quad\cd\quad \ov \calS_\calT(\calV) \wedge \bigwedge_{R(\bm X) \in \atoms(Q)} R(\bm X)
            \label{eq:bcq:target-selector}
        \end{align}
        \label{prop:bcq:target-selector}
    \end{proposition}
    \begin{proof}[Proof of Proposition~\ref{prop:bcq:target-selector}]
    By Proposition~\ref{prop:target}, there must exist an $\omega$-query plan
    $(\ov{\bm\sigma}, e)$ that is supported by $\calT$.
    Let $\ov{\bm\sigma} = (\bm X_1, \ldots, \bm X_{|\ov{\bm\sigma}|})$.
    Let $\calH=(\calV,\calE)$ be the hypergraph of the original query $Q$.
    Initially, each hyperedge $\bm Z\in\calE$ has a corresponding atom $R(\bm Z)\in\atoms(Q)$.
    We use the generalized variable elimination order $\ov{\bm\sigma}$
    to eliminate the variable sets $\bm X_1, \ldots, \bm X_{|\ov{\bm\sigma}|}$ and generate the corresponding hypergraph sequence
    $\calH^{\ov{\bm\sigma}}_1\defeq\calH, \ldots, \calH^{\ov{\bm\sigma}}_{|\ov{\bm\sigma}|+1}$,
    as described in Definition~\ref{defn:gve}.
    Each time we eliminate a variable set $\bm X_i$, we remove adjacent hyperedges
    $\partial^{\ov{\bm\sigma}}_i$ and add a new hyperedge
    $U^{\ov{\bm\sigma}}_i\setminus \bm X_i$. Our target below is to create a new
    corresponding atom $R(U^{\ov{\bm\sigma}}_i\setminus \bm X_i)$
    in time $O(\ov B\cdot \log^2 N)$,
    and use this new atom $R(U^{\ov{\bm\sigma}}_i\setminus \bm X_i)$ to replace old atoms
    $R(\bm Z)$ for $\bm Z \in \partial_i^{\ov{\bm\sigma}}$.
    Mirroring the hypergraph sequence $\calH^{\ov{\bm\sigma}}_1, \ldots, \calH^{\ov{\bm\sigma}}_{|\ov{\bm\sigma}|+1}$, this atom replacement process will create a sequence
    of {\em full} conjunctive queries $Q_1^{\ov{\bm\sigma}}, \ldots, Q_{|\ov{\bm\sigma}|+1}^{\ov{\bm\sigma}}$ (i.e.~ conjunctive queries where all variables are free)
    that are defined as follows, for every $i \in [|\ov{\bm\sigma}|+1]$:
    \begin{align}
        Q_i^{\ov{\bm\sigma}}(\calV_i^{\ov{\bm\sigma}}) \quad\cd\quad \bigwedge_{\bm Z \in \calE_i^{\ov{\bm\sigma}}} R(\bm Z)
    \end{align}
    Note that $Q_1^{\ov{\bm\sigma}}$ is the {\em full} version of the Boolean CQ $Q$.
    Our target is to inductively maintain the following invariant for every $i \in [|\ov{\bm\sigma}|+1]$:
    \begin{quote}
    {\bf Invariant:}
        There exists a superset $\calS_i^{\ov{\bm\sigma}} \supseteq \calS_\calT$ such that $Q_i^{\ov{\bm\sigma}}$ is equivalent to the following query:
        \begin{align}
            Q_i^{\ov{\bm\sigma}}(\calV_i^{\ov{\bm\sigma}}) \quad\cd\quad \calS_i^{\ov{\bm\sigma}}(\calV) \wedge \bigwedge_{R(\bm X) \in \atoms(Q)} R(\bm X)
            \label{eq:invariant:support-superset}
        \end{align}
    \end{quote}
    Assuming that the above invariant holds for $i = |\ov{\bm\sigma}|+1$,
    note that by choosing $\ov \calS_T$ to be $\calS_{|\ov{\bm\sigma}|+1}^{\ov{\bm\sigma}}$, $Q_\calT$ becomes equivalent to $Q_{|\ov{\bm\sigma}|+1}^{\ov{\bm\sigma}}$, thus proving the proposition. We now prove the invariant by induction on $i$.
    The invariant initially holds for $i=1$ because $\calS_1^{\ov{\bm\sigma}}$ can be chosen to be
    $\dom^{\calV}$.
    The elimination process proceeds as follows for $i=1, 2, \ldots, |\ov{\bm\sigma}|$:
    \begin{itemize}[leftmargin=*]
        \item If $U_i^{\ov{\bm\sigma}}\not\subseteq U_j^{\ov{\bm\sigma}}, \forall j \in [i-1]$, we recognize two cases:
        \begin{itemize}[leftmargin=*]
            \item If $e(i) = h(U_i^{\ov{\bm\sigma}})$, then we create the new atom
            $R(U^{\ov{\bm\sigma}}_i\setminus \bm X_i)$ as follows.
            By Definition~\ref{defn:query-plan-supported}, the target selector $\calT$
            must contain a join target $P(U_i^{\ov{\bm\sigma}})$.
            We initialize a corresponding table $P(U_i^{\ov{\bm\sigma}})$ by taking
            the $\ov B$-support of the probabilistic interpretation of $P(U_i^{\ov{\bm\sigma}})$.
            By definition of $\ov B$-support, the size of $P(U_i^{\ov{\bm\sigma}})$ is at most $\ov B$.
            We then compute the new atom $R(U^{\ov{\bm\sigma}}_i\setminus \bm X_i)$ as follows:
            \begin{align*}
                R(U^{\ov{\bm\sigma}}_i\setminus \bm X_i) \quad\cd\quad
                P(U_i^{\ov{\bm\sigma}}) \wedge \bigwedge_{\bm Z\in\partial_i^{\ov{\bm\sigma}}} R(\bm Z)
            \end{align*}
            Computing $R(U^{\ov{\bm\sigma}}_i\setminus \bm X_i)$ above takes
            time proportional to the size of $P(U_i^{\ov{\bm\sigma}})$, which is
            $O(\ov B)$.
            Note that $Q_{i+1}^{\ov{\bm\sigma}}$ is equivalent to:
            \begin{align*}
                Q_{i+1}^{\ov{\bm\sigma}}(\calV_{i+1}^{\ov{\bm\sigma}}) \quad\cd\quad
                Q_{i}^{\ov{\bm\sigma}}(\calV_{i}^{\ov{\bm\sigma}}) \wedge P(U_i^{\ov{\bm\sigma}})
            \end{align*}
            Therefore, invariant~\eqref{eq:invariant:support-superset} is inductively maintained by choosing $\calS_{i+1}^{\ov{\bm\sigma}}$ to be $\calS_i^{\ov{\bm\sigma}}$ filtered by $P(U_i^{\ov{\bm\sigma}})$:
            \begin{align*}
                \calS_{i+1}^{\ov{\bm\sigma}}(\calV) \quad\cd\quad
                \calS_i^{\ov{\bm\sigma}}(\calV) \wedge P(U_i^{\ov{\bm\sigma}})
            \end{align*}
            \item If $e(i) = \mm(\bm Y;\bm Z; \bm X_i|\bm G)$ where $\mm(\bm Y;\bm Z;\bm X_i|\bm G)\in \args(\emm^{\ov{\bm\sigma}}_i)$, then by Eq.~\eqref{eq:emm},
            there must exist $\calA, \calB\subseteq \partial^{\ov{\bm\sigma}}_i$
            such that $\calA \cup \calB = \partial^{\ov{\bm\sigma}}_i$ where $\bm A\defeq\cup \calA, \bm B\defeq \cup \calB, \bm A^{\bm G}\defeq \bm A \setminus \bm G, \bm B^{\bm G}\defeq \bm B \setminus \bm G$ satisfy:
            \begin{align*}
                \bm X_i = \bm A^{\bm G} \cap \bm B^{\bm G},\qquad
                \bm Y = \bm A^{\bm G}\setminus \bm X_i,\qquad
                \bm Z = \bm B^{\bm G}\setminus \bm X_i.
            \end{align*}
            Define two queries $M_1$ and $M_2$ {\em without} materializing their outputs:
            \begin{align*}
                M_1(\bm G \bm Y\bm X_i) &\cd
                \bigwedge_{\bm Z'\in\calA} R(\bm Z'),\\
                M_2(\bm G \bm X_i\bm Z) &\cd
                \bigwedge_{\bm Z'\in\calB} R(\bm Z').
            \end{align*}
            Note that even without materializing the outputs of $M_1$ and $M_2$, we can still
            support membership queries in $O(1)$ time by querying the underlying relations $R(\bm Z')$.
            By Definition~\ref{defn:query-plan-supported}, the target selector $\calT$ must contain three MM targets of the form:
            \begin{align*}
                &(\{S_1(\bm X_i\bm G), T_1(\bm Y\bm G)\}, W_1(\bm Z\bm G)),\\
                &(\{S_2(\bm X_i\bm G), T_2(\bm Z\bm G)\}, W_2(\bm Y\bm G)),\\
                &(\{S_3(\bm Y\bm G), T_3(\bm Z\bm G)\}, W_3(\bm X_i\bm G)).
            \end{align*}
            Let $E \subseteq \dom^{\bm X_i\bm Y\bm Z\bm G}$ be the intersection of the $\ov B$-supports of the probabilistic interpretations of the above three MM targets.
            Now we use Lemma~\ref{lem:omega-join-matrix-mult} to compute the following, for some superset $\ov E \supseteq E$:
            \begin{align}
                    R(\underbrace{\bm G\bm Y\bm Z}_{U^{\ov{\bm\sigma}}_i\setminus \bm X_i}) \quad\cd\quad
            \ov{E}(\bm G\bm X_i\bm Y\bm Z) \wedge M_1(\bm G\bm Y\bm X_i) \wedge M_2(\bm G\bm X_i\bm Z)
            \end{align}
            By Lemma~\ref{lem:omega-join-matrix-mult}, $R(U^{\ov{\bm\sigma}}_i\setminus \bm X_i)$ above can be computed in time $O(\ov B \cdot \log^2 N)$.
            Note that $Q_{i+1}^{\ov{\bm\sigma}}$ is equivalent to:
            \begin{align*}
                Q_{i+1}^{\ov{\bm\sigma}}(\calV_{i+1}^{\ov{\bm\sigma}}) \quad\cd\quad
                Q_{i}^{\ov{\bm\sigma}}(\calV_{i}^{\ov{\bm\sigma}}) \wedge \ov{E}(\bm G\bm X_i\bm Y\bm Z)
            \end{align*}
            Hence, invariant~\eqref{eq:invariant:support-superset} is inductively maintained by choosing $\calS_{i+1}^{\ov{\bm\sigma}}$ to be $\calS_i^{\ov{\bm\sigma}}$ filtered by $\ov E$:
            \begin{align*}
                \calS_{i+1}^{\ov{\bm\sigma}}(\calV) \quad\cd\quad
                \calS_i^{\ov{\bm\sigma}}(\calV) \wedge \ov{E}(\bm G\bm X_i\bm Y\bm Z)
            \end{align*}
        \end{itemize}
        \item If $U^{\ov{\bm\sigma}}_i \subseteq U^{\ov{\bm\sigma}}_j$ for some $j \in [i-1]$,
        then $U^{\ov{\bm\sigma}}_i \subseteq U^{\ov{\bm\sigma}}_j \setminus \bm X_j$.
        Hence, we can obtain $R(U^{\ov{\bm\sigma}}_i\setminus \bm X_i)$ by projecting
        $R(U^{\ov{\bm\sigma}}_j\setminus \bm X_j)$.
        The cost of this operation can be charged to computing $R(U^{\ov{\bm\sigma}}_j\setminus \bm X_j)$.
        Invariant~\eqref{eq:invariant:support-superset} is maintained trivially.
    \end{itemize}
    \end{proof}
    \begin{proposition}[Algorithm~\ref{alg:omega-panda-express:outer} is correct]
        The answer to $Q$ is equivalent to $\bigvee_\calT Q_\calT$.
        \label{prop:outer-algo:correct}
    \end{proposition}
    \begin{proof}[Proof of Proposition~\ref{prop:outer-algo:correct}]
        By definition of $Q_\calT$ from Eq.~\eqref{eq:bcq:target-selector},
        if $Q_\calT$ evaluates to \true for some $\calT$, then $Q$ must also evaluate to \true.
        Now, we prove the opposite direction.
        Suppose that $Q$ evaluates to \true.
        This means that there must exist a tuple $\bm t \in \bigjoin\Sigma_\inn$.
        By Proposition~\ref{prop:target-selector:completeness}, there must exist a target selector $\calT$ such that $\bm t$ belongs to the $\ov B$-support, $\calS_\calT$, of the probabilistic interpretation $\calP_\calT$ of $\calT$.
        This means $Q_\calT$ must evaluate to \true.
    \end{proof}
    \begin{proposition}[Algorithm~\ref{alg:omega-panda-express:outer} runs in the claimed time]
        Algorithm~\ref{alg:omega-panda-express:outer} runs in time $O(N^{\osubw(Q)}\cdot \log^2 N)$ in data complexity.
        \label{prop:outer-algo:time}
    \end{proposition}
    We will see below that the extra $\log^2 N$ factor in the runtime of Algorithm~\ref{alg:omega-panda-express:outer} is inherited from Proposition~\ref{prop:bcq:target-selector}, which in turn inherits it from Lemma~\ref{lem:omega-join-matrix-mult}.
    \begin{proof}[Proof of Proposition~\ref{prop:outer-algo:time}]
        The algorithm has two consecutive loops:
        \begin{itemize}
            \item The first loop (Line~\ref{alg:omega-panda-express:outer:loop1} of Algorithm~\ref{alg:omega-panda-express:outer:loop1}) iterates over every $f \in \calF$ and $g_f \in \calG_f$, where the sizes of $\calF$ and $\calG_f$ are constant in data complexity.
            The only data-dependant operation in this loop is computing a $B$-probabilistic model for an $\omega$-DDR of the form~\eqref{eq:omega-ddr:subw:raw} using Theorem~\ref{thm:panda:ddr}, which takes time $O(B \cdot \log N) = O(\ov B \cdot \log N)= O(N^{\osubw(Q)}\cdot \log N)$.
            \item The second loop (Line~\ref{alg:omega-panda-express:outer:loop2} of Algorithm~\ref{alg:omega-panda-express:outer}) iterates over every target selector $\calT$, where the number of target selectors is also a constant in data complexity.
            For each target selector $\calT$, we compute $Q_\calT$ from Eq.~\eqref{eq:bcq:target-selector} using Proposition~\ref{prop:bcq:target-selector} in time $O(\ov B \cdot \log^2 N) = O(N^{\osubw(Q)}\cdot \log^2 N)$.
        \end{itemize}
    \end{proof}
\end{proof}

%%%%%%%%%%%%%%%%%%%%%%%%%%%%%%%%%%%%%%%%%%%%%%%%%%%%%%%%%%%%%%%%%%%%%%%%%%%%%%%%%%%%%%%%%%%
\subsection{A more advanced example: The 4-cycle query}
\label{subsec:algo:4cycle}

\newcommand{\hltB}[1]{\colorbox{yellow}{$#1$}}

\begin{figure}
    {\small
\begin{align*}
    \osubw(\fourcycleH) = & \max \bigl(\\
&\textstyle{\max_{h \in \Gamma \cap \ed}}\quad\min(h(XYZ),\quad h(X) + h(Y) + \gamma h(Z),\quad h(YZW),\quad h(Y) + h(Z) + \gamma h(W)),\\
&\textstyle{\max_{h \in \Gamma \cap \ed}}\quad\min(h(XYZ),\quad h(X) + h(Y) + \gamma h(Z),\quad h(YZW),\quad h(Y) + \gamma h(Z) + h(W)),\\
&\textstyle{\max_{h \in \Gamma \cap \ed}}\quad\min(h(XYZ),\quad h(X) + h(Y) + \gamma h(Z),\quad h(YZW),\quad \gamma h(Y) + h(Z) + h(W)),\\
&\textstyle{\max_{h \in \Gamma \cap \ed}}\quad\min(h(XYZ),\quad h(X) + h(Y) + \gamma h(Z),\quad h(WXY),\quad h(W) + h(X) + \gamma h(Y)),\\
&\textstyle{\max_{h \in \Gamma \cap \ed}}\quad\min(h(XYZ),\quad h(X) + h(Y) + \gamma h(Z),\quad h(WXY),\quad h(W) + \gamma h(X) + h(Y)),\\
&\textstyle{\max_{h \in \Gamma \cap \ed}}\quad\min(h(XYZ),\quad h(X) + h(Y) + \gamma h(Z),\quad h(WXY),\quad \gamma h(W) + h(X) + h(Y)),\\
&\textstyle{\max_{h \in \Gamma \cap \ed}}\quad\min(h(XYZ),\quad h(X) + \gamma h(Y) + h(Z),\quad h(YZW),\quad h(Y) + h(Z) + \gamma h(W)),\\
&\hltB{\textstyle{\max_{h \in \Gamma \cap \ed}}\quad\min(h(XYZ),\quad h(X) + \gamma h(Y) + h(Z),\quad h(YZW),\quad h(Y) + \gamma h(Z) + h(W))},\\
&\textstyle{\max_{h \in \Gamma \cap \ed}}\quad\min(h(XYZ),\quad h(X) + \gamma h(Y) + h(Z),\quad h(YZW),\quad \gamma h(Y) + h(Z) + h(W)),\\
&\textstyle{\max_{h \in \Gamma \cap \ed}}\quad\min(h(XYZ),\quad h(X) + \gamma h(Y) + h(Z),\quad h(WXY),\quad h(W) + h(X) + \gamma h(Y)),\\
&\textstyle{\max_{h \in \Gamma \cap \ed}}\quad\min(h(XYZ),\quad h(X) + \gamma h(Y) + h(Z),\quad h(WXY),\quad h(W) + \gamma h(X) + h(Y)),\\
&\textstyle{\max_{h \in \Gamma \cap \ed}}\quad\min(h(XYZ),\quad h(X) + \gamma h(Y) + h(Z),\quad h(WXY),\quad \gamma h(W) + h(X) + h(Y)),\\
&\textstyle{\max_{h \in \Gamma \cap \ed}}\quad\min(h(XYZ),\quad \gamma h(X) + h(Y) + h(Z),\quad h(YZW),\quad h(Y) + h(Z) + \gamma h(W)),\\
&\textstyle{\max_{h \in \Gamma \cap \ed}}\quad\min(h(XYZ),\quad \gamma h(X) + h(Y) + h(Z),\quad h(YZW),\quad h(Y) + \gamma h(Z) + h(W)),\\
&\textstyle{\max_{h \in \Gamma \cap \ed}}\quad\min(h(XYZ),\quad \gamma h(X) + h(Y) + h(Z),\quad h(YZW),\quad \gamma h(Y) + h(Z) + h(W)),\\
&\textstyle{\max_{h \in \Gamma \cap \ed}}\quad\min(h(XYZ),\quad \gamma h(X) + h(Y) + h(Z),\quad h(WXY),\quad h(W) + h(X) + \gamma h(Y)),\\
&\textstyle{\max_{h \in \Gamma \cap \ed}}\quad\min(h(XYZ),\quad \gamma h(X) + h(Y) + h(Z),\quad h(WXY),\quad h(W) + \gamma h(X) + h(Y)),\\
&\textstyle{\max_{h \in \Gamma \cap \ed}}\quad\min(h(XYZ),\quad \gamma h(X) + h(Y) + h(Z),\quad h(WXY),\quad \gamma h(W) + h(X) + h(Y)),\\
&\textstyle{\max_{h \in \Gamma \cap \ed}}\quad\min(h(ZWX),\quad h(Z) + h(W) + \gamma h(X),\quad h(YZW),\quad h(Y) + h(Z) + \gamma h(W)),\\
&\textstyle{\max_{h \in \Gamma \cap \ed}}\quad\min(h(ZWX),\quad h(Z) + h(W) + \gamma h(X),\quad h(YZW),\quad h(Y) + \gamma h(Z) + h(W)),\\
&\textstyle{\max_{h \in \Gamma \cap \ed}}\quad\min(h(ZWX),\quad h(Z) + h(W) + \gamma h(X),\quad h(YZW),\quad \gamma h(Y) + h(Z) + h(W)),\\
&\textstyle{\max_{h \in \Gamma \cap \ed}}\quad\min(h(ZWX),\quad h(Z) + h(W) + \gamma h(X),\quad h(WXY),\quad h(W) + h(X) + \gamma h(Y)),\\
&\textstyle{\max_{h \in \Gamma \cap \ed}}\quad\min(h(ZWX),\quad h(Z) + h(W) + \gamma h(X),\quad h(WXY),\quad h(W) + \gamma h(X) + h(Y)),\\
&\textstyle{\max_{h \in \Gamma \cap \ed}}\quad\min(h(ZWX),\quad h(Z) + h(W) + \gamma h(X),\quad h(WXY),\quad \gamma h(W) + h(X) + h(Y)),\\
&\textstyle{\max_{h \in \Gamma \cap \ed}}\quad\min(h(ZWX),\quad h(Z) + \gamma h(W) + h(X),\quad h(YZW),\quad h(Y) + h(Z) + \gamma h(W)),\\
&\textstyle{\max_{h \in \Gamma \cap \ed}}\quad\min(h(ZWX),\quad h(Z) + \gamma h(W) + h(X),\quad h(YZW),\quad h(Y) + \gamma h(Z) + h(W)),\\
&\textstyle{\max_{h \in \Gamma \cap \ed}}\quad\min(h(ZWX),\quad h(Z) + \gamma h(W) + h(X),\quad h(YZW),\quad \gamma h(Y) + h(Z) + h(W)),\\
&\textstyle{\max_{h \in \Gamma \cap \ed}}\quad\min(h(ZWX),\quad h(Z) + \gamma h(W) + h(X),\quad h(WXY),\quad h(W) + h(X) + \gamma h(Y)),\\
&\textstyle{\max_{h \in \Gamma \cap \ed}}\quad\min(h(ZWX),\quad h(Z) + \gamma h(W) + h(X),\quad h(WXY),\quad h(W) + \gamma h(X) + h(Y)),\\
&\textstyle{\max_{h \in \Gamma \cap \ed}}\quad\min(h(ZWX),\quad h(Z) + \gamma h(W) + h(X),\quad h(WXY),\quad \gamma h(W) + h(X) + h(Y)),\\
&\textstyle{\max_{h \in \Gamma \cap \ed}}\quad\min(h(ZWX),\quad \gamma h(Z) + h(W) + h(X),\quad h(YZW),\quad h(Y) + h(Z) + \gamma h(W)),\\
&\textstyle{\max_{h \in \Gamma \cap \ed}}\quad\min(h(ZWX),\quad \gamma h(Z) + h(W) + h(X),\quad h(YZW),\quad h(Y) + \gamma h(Z) + h(W)),\\
&\textstyle{\max_{h \in \Gamma \cap \ed}}\quad\min(h(ZWX),\quad \gamma h(Z) + h(W) + h(X),\quad h(YZW),\quad \gamma h(Y) + h(Z) + h(W)),\\
&\textstyle{\max_{h \in \Gamma \cap \ed}}\quad\min(h(ZWX),\quad \gamma h(Z) + h(W) + h(X),\quad h(WXY),\quad h(W) + h(X) + \gamma h(Y)),\\
&\textstyle{\max_{h \in \Gamma \cap \ed}}\quad\min(h(ZWX),\quad \gamma h(Z) + h(W) + h(X),\quad h(WXY),\quad h(W) + \gamma h(X) + h(Y)),\\
&\textstyle{\max_{h \in \Gamma \cap \ed}}\quad\min(h(ZWX),\quad \gamma h(Z) + h(W) + h(X),\quad h(WXY),\quad \gamma h(W) + h(X) + h(Y))\bigr)
\end{align*}}
\caption{The $\omega$-submodular width of the 4-cycle query $\fourcycle$ (Eq.~\eqref{eq:intro:4cycle}), expanded by distributing $\min$ over $\max$ in Eq.~\eqref{eq:osubw:4cycle}.}
\label{fig:osubw:4cycle}
\end{figure}

\setlength{\fboxsep}{1pt}
\newcommand{\hltA}[1]{\fcolorbox{black}{gray!35}{$#1$}}

\begin{figure}
{\small
\begin{align*}
\hltA{P_{1}(XYZ)} \quad &\vee\quad (S_{1}(X) \wedge T_{1}(Y) \wedge W_{1}(Z)) \quad&\vee\quad P'_{1}(YZW) \quad &\vee\quad &(S'_{1}(Y) \wedge T'_{1}(Z) \wedge W'_{1}(W)) \quad &\cd\quad \cdots\\
P_{2}(XYZ) \quad &\vee\quad \hltA{(S_{2}(X) \wedge T_{2}(Y) \wedge W_{2}(Z))} \quad&\vee\quad P'_{2}(YZW) \quad &\vee\quad &(S'_{2}(Y) \wedge T'_{2}(W) \wedge W'_{2}(Z)) \quad &\cd\quad \cdots\\
P_{3}(XYZ) \quad &\vee\quad \hltA{(S_{3}(X) \wedge T_{3}(Y) \wedge W_{3}(Z))} \quad&\vee\quad P'_{3}(YZW) \quad &\vee\quad &(S'_{3}(Z) \wedge T'_{3}(W) \wedge W'_{3}(Y)) \quad &\cd\quad \cdots\\
P_{4}(XYZ) \quad &\vee\quad \hltA{(S_{4}(X) \wedge T_{4}(Y) \wedge W_{4}(Z))} \quad&\vee\quad P'_{4}(WXY) \quad &\vee\quad &(S'_{4}(W) \wedge T'_{4}(X) \wedge W'_{4}(Y)) \quad &\cd\quad \cdots\\
P_{5}(XYZ) \quad &\vee\quad \hltA{(S_{5}(X) \wedge T_{5}(Y) \wedge W_{5}(Z))} \quad&\vee\quad P'_{5}(WXY) \quad &\vee\quad &(S'_{5}(W) \wedge T'_{5}(Y) \wedge W'_{5}(X)) \quad &\cd\quad \cdots\\
P_{6}(XYZ) \quad &\vee\quad \hltA{(S_{6}(X) \wedge T_{6}(Y) \wedge W_{6}(Z))} \quad&\vee\quad P'_{6}(WXY) \quad &\vee\quad &(S'_{6}(X) \wedge T'_{6}(Y) \wedge W'_{6}(W)) \quad &\cd\quad \cdots\\
P_{7}(XYZ) \quad &\vee\quad \hltA{(S_{7}(X) \wedge T_{7}(Z) \wedge W_{7}(Y))} \quad&\vee\quad P'_{7}(YZW) \quad &\vee\quad &(S'_{7}(Y) \wedge T'_{7}(Z) \wedge W'_{7}(W)) \quad &\cd\quad \cdots\\
P_{8}(XYZ) \quad &\vee\quad \hltA{(S_{8}(X) \wedge T_{8}(Z) \wedge W_{8}(Y))} \quad&\vee\quad P'_{8}(YZW) \quad &\vee\quad &(S'_{8}(Y) \wedge T'_{8}(W) \wedge W'_{8}(Z)) \quad &\cd\quad \cdots\\
P_{9}(XYZ) \quad &\vee\quad \hltA{(S_{9}(X) \wedge T_{9}(Z) \wedge W_{9}(Y))} \quad&\vee\quad P'_{9}(YZW) \quad &\vee\quad &(S'_{9}(Z) \wedge T'_{9}(W) \wedge W'_{9}(Y)) \quad &\cd\quad \cdots\\
P_{10}(XYZ) \quad &\vee\quad \hltA{(S_{10}(X) \wedge T_{10}(Z) \wedge W_{10}(Y))} \quad&\vee\quad P'_{10}(WXY) \quad &\vee\quad &(S'_{10}(W) \wedge T'_{10}(X) \wedge W'_{10}(Y)) \quad &\cd\quad \cdots\\
P_{11}(XYZ) \quad &\vee\quad \hltA{(S_{11}(X) \wedge T_{11}(Z) \wedge W_{11}(Y))} \quad&\vee\quad P'_{11}(WXY) \quad &\vee\quad &(S'_{11}(W) \wedge T'_{11}(Y) \wedge W'_{11}(X)) \quad &\cd\quad \cdots\\
P_{12}(XYZ) \quad &\vee\quad \hltA{(S_{12}(X) \wedge T_{12}(Z) \wedge W_{12}(Y))} \quad&\vee\quad P'_{12}(WXY) \quad &\vee\quad &(S'_{12}(X) \wedge T'_{12}(Y) \wedge W'_{12}(W)) \quad &\cd\quad \cdots\\
P_{13}(XYZ) \quad &\vee\quad \hltA{(S_{13}(Y) \wedge T_{13}(Z) \wedge W_{13}(X))} \quad&\vee\quad P'_{13}(YZW) \quad &\vee\quad &(S'_{13}(Y) \wedge T'_{13}(Z) \wedge W'_{13}(W)) \quad &\cd\quad \cdots\\
P_{14}(XYZ) \quad &\vee\quad \hltA{(S_{14}(Y) \wedge T_{14}(Z) \wedge W_{14}(X))} \quad&\vee\quad P'_{14}(YZW) \quad &\vee\quad &(S'_{14}(Y) \wedge T'_{14}(W) \wedge W'_{14}(Z)) \quad &\cd\quad \cdots\\
P_{15}(XYZ) \quad &\vee\quad \hltA{(S_{15}(Y) \wedge T_{15}(Z) \wedge W_{15}(X))} \quad&\vee\quad P'_{15}(YZW) \quad &\vee\quad &(S'_{15}(Z) \wedge T'_{15}(W) \wedge W'_{15}(Y)) \quad &\cd\quad \cdots\\
\hltA{P_{16}(XYZ)} \quad &\vee\quad (S_{16}(Y) \wedge T_{16}(Z) \wedge W_{16}(X)) \quad&\vee\quad P'_{16}(WXY) \quad &\vee\quad &(S'_{16}(W) \wedge T'_{16}(X) \wedge W'_{16}(Y)) \quad &\cd\quad \cdots\\
\hltA{P_{17}(XYZ)} \quad &\vee\quad (S_{17}(Y) \wedge T_{17}(Z) \wedge W_{17}(X)) \quad&\vee\quad P'_{17}(WXY) \quad &\vee\quad &(S'_{17}(W) \wedge T'_{17}(Y) \wedge W'_{17}(X)) \quad &\cd\quad \cdots\\
\hltA{P_{18}(XYZ)} \quad &\vee\quad (S_{18}(Y) \wedge T_{18}(Z) \wedge W_{18}(X)) \quad&\vee\quad P'_{18}(WXY) \quad &\vee\quad &(S'_{18}(X) \wedge T'_{18}(Y) \wedge W'_{18}(W)) \quad &\cd\quad \cdots\\
P_{19}(ZWX) \quad &\vee\quad \hltA{(S_{19}(Z) \wedge T_{19}(W) \wedge W_{19}(X))} \quad&\vee\quad P'_{19}(YZW) \quad &\vee\quad &(S'_{19}(Y) \wedge T'_{19}(Z) \wedge W'_{19}(W)) \quad &\cd\quad \cdots\\
P_{20}(ZWX) \quad &\vee\quad \hltA{(S_{20}(Z) \wedge T_{20}(W) \wedge W_{20}(X))} \quad&\vee\quad P'_{20}(YZW) \quad &\vee\quad &(S'_{20}(Y) \wedge T'_{20}(W) \wedge W'_{20}(Z)) \quad &\cd\quad \cdots\\
P_{21}(ZWX) \quad &\vee\quad \hltA{(S_{21}(Z) \wedge T_{21}(W) \wedge W_{21}(X))} \quad&\vee\quad P'_{21}(YZW) \quad &\vee\quad &(S'_{21}(Z) \wedge T'_{21}(W) \wedge W'_{21}(Y)) \quad &\cd\quad \cdots\\
P_{22}(ZWX) \quad &\vee\quad \hltA{(S_{22}(Z) \wedge T_{22}(W) \wedge W_{22}(X))} \quad&\vee\quad P'_{22}(WXY) \quad &\vee\quad &(S'_{22}(W) \wedge T'_{22}(X) \wedge W'_{22}(Y)) \quad &\cd\quad \cdots\\
P_{23}(ZWX) \quad &\vee\quad \hltA{(S_{23}(Z) \wedge T_{23}(W) \wedge W_{23}(X))} \quad&\vee\quad P'_{23}(WXY) \quad &\vee\quad &(S'_{23}(W) \wedge T'_{23}(Y) \wedge W'_{23}(X)) \quad &\cd\quad \cdots\\
P_{24}(ZWX) \quad &\vee\quad \hltA{(S_{24}(Z) \wedge T_{24}(W) \wedge W_{24}(X))} \quad&\vee\quad P'_{24}(WXY) \quad &\vee\quad &(S'_{24}(X) \wedge T'_{24}(Y) \wedge W'_{24}(W)) \quad &\cd\quad \cdots\\
P_{25}(ZWX) \quad &\vee\quad (S_{25}(Z) \wedge T_{25}(X) \wedge W_{25}(W)) \quad&\vee\quad \hltA{P'_{25}(YZW)} \quad &\vee\quad &(S'_{25}(Y) \wedge T'_{25}(Z) \wedge W'_{25}(W)) \quad &\cd\quad \cdots\\
P_{26}(ZWX) \quad &\vee\quad (S_{26}(Z) \wedge T_{26}(X) \wedge W_{26}(W)) \quad&\vee\quad P'_{26}(YZW) \quad &\vee\quad &\hltA{(S'_{26}(Y) \wedge T'_{26}(W) \wedge W'_{26}(Z))} \quad &\cd\quad \cdots\\
P_{27}(ZWX) \quad &\vee\quad (S_{27}(Z) \wedge T_{27}(X) \wedge W_{27}(W)) \quad&\vee\quad \hltA{P'_{27}(YZW)} \quad &\vee\quad &(S'_{27}(Z) \wedge T'_{27}(W) \wedge W'_{27}(Y)) \quad &\cd\quad \cdots\\
P_{28}(ZWX) \quad &\vee\quad (S_{28}(Z) \wedge T_{28}(X) \wedge W_{28}(W)) \quad&\vee\quad P'_{28}(WXY) \quad &\vee\quad &\hltA{(S'_{28}(W) \wedge T'_{28}(X) \wedge W'_{28}(Y))} \quad &\cd\quad \cdots\\
P_{29}(ZWX) \quad &\vee\quad (S_{29}(Z) \wedge T_{29}(X) \wedge W_{29}(W)) \quad&\vee\quad P'_{29}(WXY) \quad &\vee\quad &\hltA{(S'_{29}(W) \wedge T'_{29}(Y) \wedge W'_{29}(X))} \quad &\cd\quad \cdots\\
P_{30}(ZWX) \quad &\vee\quad (S_{30}(Z) \wedge T_{30}(X) \wedge W_{30}(W)) \quad&\vee\quad P'_{30}(WXY) \quad &\vee\quad &\hltA{(S'_{30}(X) \wedge T'_{30}(Y) \wedge W'_{30}(W))} \quad &\cd\quad \cdots\\
P_{31}(ZWX) \quad &\vee\quad \hltA{(S_{31}(W) \wedge T_{31}(X) \wedge W_{31}(Z))} \quad&\vee\quad P'_{31}(YZW) \quad &\vee\quad &(S'_{31}(Y) \wedge T'_{31}(Z) \wedge W'_{31}(W)) \quad &\cd\quad \cdots\\
P_{32}(ZWX) \quad &\vee\quad \hltA{(S_{32}(W) \wedge T_{32}(X) \wedge W_{32}(Z))} \quad&\vee\quad P'_{32}(YZW) \quad &\vee\quad &(S'_{32}(Y) \wedge T'_{32}(W) \wedge W'_{32}(Z)) \quad &\cd\quad \cdots\\
P_{33}(ZWX) \quad &\vee\quad \hltA{(S_{33}(W) \wedge T_{33}(X) \wedge W_{33}(Z))} \quad&\vee\quad P'_{33}(YZW) \quad &\vee\quad &(S'_{33}(Z) \wedge T'_{33}(W) \wedge W'_{33}(Y)) \quad &\cd\quad \cdots\\
P_{34}(ZWX) \quad &\vee\quad \hltA{(S_{34}(W) \wedge T_{34}(X) \wedge W_{34}(Z))} \quad&\vee\quad P'_{34}(WXY) \quad &\vee\quad &(S'_{34}(W) \wedge T'_{34}(X) \wedge W'_{34}(Y)) \quad &\cd\quad \cdots\\
P_{35}(ZWX) \quad &\vee\quad \hltA{(S_{35}(W) \wedge T_{35}(X) \wedge W_{35}(Z))} \quad&\vee\quad P'_{35}(WXY) \quad &\vee\quad &(S'_{35}(W) \wedge T'_{35}(Y) \wedge W'_{35}(X)) \quad &\cd\quad \cdots\\
P_{36}(ZWX) \quad &\vee\quad \hltA{(S_{36}(W) \wedge T_{36}(X) \wedge W_{36}(Z))} \quad&\vee\quad P'_{36}(WXY) \quad &\vee\quad &(S'_{36}(X) \wedge T'_{36}(Y) \wedge W'_{36}(W)) \quad &\cd\quad \cdots
\end{align*}
}
\caption{The $\omega$-DDRs of the form~\eqref{eq:omega-ddr:subw:raw} for the 4-cycle query $\fourcycle$ from Eq.~\eqref{eq:intro:4cycle}.
All the $\omega$-DDRs have the same body $R(X, Y) \wedge S(Y, Z) \wedge T(Z, W) \wedge U(W, X)$,
which is abbreviated as $\cdots$.
Note that each one of the 36 $\omega$-DDRs mirrors one the 36 $\max$-terms in the expression for $\osubw(\fourcycleH)$ from Figure~\ref{fig:osubw:4cycle}.
There are $4^{36}$ target selectors, one of them is \hltA{\text{highlighted}}.
See the discussion in Section~\ref{subsec:algo:4cycle} for how this target selector supports an $\omega$-query plan.
}
\label{fig:omega-ddr:4cycle}
\end{figure}

Consider the 4-cycle BCQ $Q_\square$ from Eq.~\eqref{eq:intro:4cycle}.
As explained earlier in Table~\ref{tab:comparison},
the $\omega$-submodular width of $Q_\square$
is exactly $2 - \frac{3}{2\cdot \min\{\omega,\frac{5}{2}\}+1}$.
See Lemma~\ref{lem:4cycle} in the Appendix for a proof.
When $\omega \leq 5/2$, this expression simplifies to $\frac{4\omega-1}{2\omega+1}$,
thus matching the exponent in the runtime $O\left(N^{\frac{4\omega-1}{2\omega+1}}\right)$ from prior work~\cite{yuster2004detecting, dalirrooyfard2019graph},
as shown in Table~\ref{tab:intro:comparison}.
In contrast, when $\omega \geq 5/2$, the expression simplifies to $3/2$,
thus matching the exponent in the well-known runtime $O\left(N^{3/2}\right)$ from~\cite{DBLP:journals/algorithmica/AlonYZ97,panda-express}.

The hypergraph of $Q_\square$ is as follows:
\begin{align}
    \fourcycleH = (\{X,Y,Z,W\},\quad \{\{X,Y\}, \{Y,Z\}, \{Z,W\}, \{W,X\}\})
    \label{eq:4cycle:H}
\end{align}
The corresponding $\omega$-submodular width expression from Eq.~\eqref{eq:osubw:trimmed} is:
\begin{align}
    \osubw(\fourcycleH) = \max_{h \in \Gamma \cap \ed} \min \bigl(&
        \max(\min(h(XYZ), \mm(X; Y; Z)), \min(h(ZWX), \mm(Z; W; X))),\nonumber\\
        &\max(\min(h(YZW), \mm(Y; Z; W)), \min(h(WXY), \mm(W; X; Y)))
        \bigr)&
    \label{eq:osubw:4cycle}
\end{align}
In particular, Eq.~\eqref{eq:osubw:4cycle} suggests that the 4-cycle query $\fourcycle$ admits basically two kinds of $\omega$-query plans:
\begin{itemize}
    \item $\omega$-query plans that start by eliminating either $Y$ or $W$:
    \begin{itemize}
        \item Eliminating $Y$ involves the variables $X, Y, Z$ and can be done by either a join
        or an MM.
        \item Eliminating $W$ involves the variables $Z, W, X$ and can also be done by either a join or an MM.
    \end{itemize}
    \item $\omega$-query plans that start by eliminating either $X$ or $Z$:
    \begin{itemize}
        \item Eliminating $Z$ involves the variables $Y, Z, W$ and can be done by either a join
        or an MM.
        \item Eliminating $X$ involves the variables $W, X, Y$ and can also be done by either a join or an MM.
    \end{itemize}
\end{itemize}
By distributing every $\min$ over inner $\max$ in Eq.~\eqref{eq:osubw:4cycle} following Eq.~\eqref{eq:osubw:distributed:swapped:raw}, we get the expression for $\osubw(\fourcycleH)$
given in Figure~\ref{fig:osubw:4cycle}, which has 36 terms of the form $\max_{h \in \Gamma \cap \ed}\cdots$.
Each one of these 36 terms is upper bounded by $\osubw(\fourcycleH)$.
Moreover, each one of these terms corresponds to an $\omega$-DDR of the form~\eqref{eq:omega-ddr:subw:raw}.
The 36 $\omega$-DDRs are depicted in Figure~\ref{fig:omega-ddr:4cycle}.
Each $\omega$-DDR has four targets: two join targets and two MM targets.
Recall that in an MM target, e.g., $(S_1(X) \wedge T_1(Y) \wedge W_1(Z))$,
the relations $S_1$ and $T_1$ are interchangeable whereas $W_1$ plays a distinct role
in the definition of a $B$-probabilistic model for this MM target (Definition~\ref{defn:prob:model:omega-ddr}).
A target selector $\calT$ picks one target from each of the 36 $\omega$-DDRs,
hence the total number of target selectors is $4^{36}$.
Proposition~\ref{prop:target} guarantees that if we pick any of those target selectors,
then at least one of the above $\omega$-query plans will be supported.
For example, consider the highlighted target selector in Figure~\ref{fig:omega-ddr:4cycle}.
This target selector is chosen carefully to avoid supporting any $\omega$-query plan that starts with eliminating $Y$ or $W$.
However, while constructing such a target selector, eventually we run out of options and we end up being cornered into supporting some other $\omega$-query plan.
In this particular case, we end up supporting the $\omega$-query plan
that starts with eliminating $Z$ by a join $h(YZW)$ and eliminating $X$ by an MM $\mm(W; X; Y)$.
This is the only supported $\omega$-query plan for this target selector.

Suppose for now that we already evaluated the $\omega$-DDRs from Figure~\ref{fig:omega-ddr:4cycle}; we will come back shortly to explain how to evaluate them.
In order to evaluate the original BCQ $\fourcycle$,
the algorithm is to go over the target selectors, and for each one, to use
the supported $\omega$-query plan to evaluate $\fourcycle$ filtered by the targets in this
$\omega$-query plan, which are a subset of the targets in the target selector.
Completeness follows from
Proposition~\ref{prop:target-selector:completeness}, which guarantees that if we compute a
$\ov B$-probabilistic model for each of the 36 $\omega$-DDRs,
then every tuple $\bm t \in \bigjoin\Sigma_\inn$
is accounted for by some target selector.

We now explain briefly how to evaluate the $\omega$-DDRs from Figure~\ref{fig:omega-ddr:4cycle}.
Let's pick the eighth $\omega$-DDR as an example:
\begin{multline}
P_{8}(XYZ) \quad\vee\quad (S_{8}(X) \wedge T_{8}(Z) \wedge W_{8}(Y)) \quad\vee\quad P'_{8}(YZW) \quad\vee\quad (S'_{8}(Y) \wedge T'_{8}(W) \wedge W'_{8}(Z)) \cd\\
R(X, Y) \wedge S(Y, Z) \wedge T(Z, W) \wedge U(W, X)
\label{eq:omega-ddr:4cycle:example}
\end{multline}
The corresponding term from Figure~\ref{fig:osubw:4cycle} is \hltB{\text{highlighted}}.
We repeat it below:
\begin{align}
    \max_{h \in \Gamma \cap \ed}\quad\min(h(XYZ),\quad h(X) + \gamma h(Y) + h(Z),\quad h(YZW),\quad h(Y) + \gamma h(Z) + h(W))
    \label{eq:osubw:4cycle:example}
\end{align}
By the equation from Figure~\ref{fig:osubw:4cycle}, we know that the quantity in Eq.~\eqref{eq:osubw:4cycle:example} is upper bounded by $\osubw(\fourcycleH)$, which is
$\frac{4\omega-1}{2\omega+1}$ {\em assuming} $\omega \leq 5/2$.
In this particular case, the quantity from Eq.~\eqref{eq:osubw:4cycle:example} is {\em exactly} the $\omega$-submodular width.
This quantity is equivalent to the optimal objective value of the following linear program of the form~\eqref{eq:osubw:inner-lp}:
\begin{align}
    \max_{t, \bm h \in \Gamma \cap \ed}
    \bigl\{
        t \quad\mid\quad &t \leq h(XYZ),\nonumber\\
        &t \leq h(X) + \gamma h(Y) + h(Z),\nonumber\\
        &t \leq h(YZW),\nonumber\\
        &t \leq h(Y) + \gamma h(Z) + h(W)
    \bigr\}
    \label{eq:osubw:4cycle:lp}
\end{align}
Lemma~\ref{lmm:lp-to-omega-shannon} suggests that the optimal dual solution to the above LP
corresponds to the following $\omega$-Shannon-flow inequality:
\begin{multline}
(\omega-1) h(XYZ) + \omega h(YZW) + (h(X) + \gamma h(Y) + h(Z)) + (h(Y) +  \gamma h(Z) + h(W)) \leq\\
        \omega h(XY) + (2\omega-2) h(YZ) + \omega h(ZW) + h(WX)
        \label{eq:osubw:4cycle:shannon-flow}
\end{multline}
Note that the RHS of the above inequality is upper bounded by $4\omega-1$ because $\bm h\in\ed$.
Moreover, the LHS is lower bounded by $(2\omega+1)t$.
Hence, this inequality proves that $t$ is upper bounded by $\frac{4\omega-1}{2\omega+1}$,
which is the optimal objective value of the LP from~\eqref{eq:osubw:4cycle:lp}.
Inequality~\eqref{eq:osubw:4cycle:shannon-flow} is a Shannon inequality specifically because it is a sum of the following basic Shannon inequalities:
\begin{align*}
&h(XY) + h(WX) - h(X) - h(XYW) \geq 0\\
&h(XYW) - h(YW) \geq 0\\
&h(YW) + h(ZW) - h(W) - h(YZW) \geq 0\\
&(\omega-1)h(XY) + (\omega-1)h(YZ) - (\omega-1)h(Y) - (\omega-1)h(XYZ) \geq 0\\
&(\omega-1)h(YZ) + (\omega-1)h(ZW) - (\omega-1)h(Z) - (\omega-1)h(YZW) \geq 0
\end{align*}
By Theorem~\ref{thm:panda:ddr}, the $\opandaexpress$ algorithm (Algorithm~\ref{alg:omega-panda-express}) takes inequality~\eqref{eq:osubw:4cycle:shannon-flow}
along with the collection $\calP$ of sub-probability measures from Eq.~\eqref{eq:initial:measure} which correspond to the RHS of~\eqref{eq:osubw:4cycle:shannon-flow}, and computes a $B$-probabilistic model for the $\omega$-DDR from Eq.~\eqref{eq:omega-ddr:4cycle:example}, where $B = N^{\frac{4\omega-1}{2\omega+1}}=N^{\osubw(\fourcycle)}$, in time $O(B \cdot \log N)$.
The other $\omega$-DDRs from Figure~\ref{fig:omega-ddr:4cycle} are evaluated similarly.

\section{Related Work}
  Fast matrix multiplication has recently attracted a lot of attention in the database community for processing join-project queries.
Let $\OUT$ be the number of query results. Note that for Boolean conjunctive queries that we study in this paper, the output is a Boolean value, hence $\out = 1$.
Let $\rectOmega(a,b,c)$ be the smallest exponent for multiplying two rectangular matrices of dimensions $n^a \times n^b$ and $n^b \times n^c$ within $O\left(n^{\rectOmega(a,b,c)}\right)$ time.
For the purpose of analyzing the complexity of rectangular matrix multiplication,
several constants have been defined,
 namely, $\alpha \le 1$ defined as the largest constant such that $\rectOmega(1,\alpha,1)=2$, and $\mu$ defined as the (unique) solution to the equation $\rectOmega(\mu, 1, 1) = 2\mu + 1$.
Note that $\alpha = 1$ if and only if $\omega=2$, and the current best bounds on $\alpha$ are ${0.321334 <} \alpha \le 1$~{\cite{williams2024new}}.
Moreover, $\mu = \frac{1}{2}$ if $\omega=2$, and the current best bounds on $\mu$ are $\frac{1}{2} \le \mu  {< 0.527661}$~{\cite{williams2024new}}.
Amossen and Pagh~\cite{amossen2009faster} first proposed an algorithm for the Boolean matrix multiplication by using fast matrix multiplication techniques, which runs in $\tilde O\left(N^{\frac{2}{3}}\cdot \OUT^{\frac{2}{3}} + N^{\frac{(2-\alpha)\omega -2}{(1+\omega)(1-\alpha)}} \cdot \OUT^{\frac{2-\alpha \omega}{(1+\omega)(1-\alpha)}} +\OUT\right)$ time when $\OUT \ge N$, and $\tilde O\left(N \cdot \OUT^{\frac{\omega-1}{\omega+1}}\right)$ time when $\OUT < N$, where $\tilde O$ hides a factor of $O(N^{o(1)})$.
If $\omega = 2$, this result degenerates to $\tilde O\left(N^{\frac{2}{3}}\cdot \OUT^{\frac{2}{3}} + N\cdot \OUT^{\frac{1}{3}}\right)$.
Very recently, Abboud et al.~\cite{abboud2023time} have completely improved this to $O\left(N \cdot \OUT^{\frac{\mu}{1+\mu}} + \OUT + N^{\frac{(2+\alpha)\mu}{1+\mu}} \cdot \OUT^{\frac{1-\alpha \mu}{1+\mu}}\right)$.
If $\omega = 2$, this complexity can be simplified to $O\left(N \cdot \OUT^{\frac{1}{3}} + \OUT\right)$.
On the other hand, improving this result for any value of $\OUT$ is rather difficult, assuming the hardness of the {\em all-edge-triangle} problem~\cite{abboud2023time}.
Many algorithms have been proposed for Boolean matrix multiplication, whose complexity is also measured by the domain size of attributes; we refer readers to~\cite{abboud2023time} for details.
Deep et al.~\cite{deep2020fast} and Huang et al. \cite{huang2022density} also investigated the efficient implementation of these algorithms on sparse matrix multiplication in practice.
Very recently, Hu \cite{hu2024fast} first applied fast matrix multiplication to speed up acyclic join-project queries and showed a large polynomial improvement over the combinatorial Yannakakis algorithm~\cite{DBLP:conf/vldb/Yannakakis81}.
As an independent line of research, fast matrix multiplication has been extensively used in the algorithm community to speed up detecting, counting, and listing subgraph patterns, such as cliques~\cite{DBLP:journals/algorithmica/AlonYZ97, patrascu2010towards, dalirrooyfard2024towards} and cycles~\cite{jin2024listing}.
In addition, this technique has been used to approximately count triangles~\cite{tvetek2022approximate} and cycles~\cite{censor2024fast}, k-centering in graphs~\cite{jin2025beyond}, etc.

\section{Conclusion}
\label{sec:conclusion}
We proposed a general framework for evaluating Boolean conjunctive queries that naturally
subsumes both combinatorial and non-combinatorial techniques under the umbrella of information
theory.
Our framework generalizes the notion of submodular width by incorporating matrix multiplication,
and provides a matching algorithm for evaluating queries in the corresponding time complexity.
Using this framework, we show how to recover the best known complexity for various classes of
queries as well as improve the complexity for some of them.

Our framework can be straightforwardly extended from Boolean conjunctive queries to count and sum queries, albeit {\em without} free variables.
In particular, in the combinatorial world, the transition from Boolean to count/sum queries
(and more generally to aggregate queries over any semiring)
is done by adapting the submodular width to become the {\em Sharp-submodular width}~\cite{10.1145/3426865}.
The latter is a variant of the submodular width that is obtained by relaxing the notion of polymatroids.
In the presence of matrix multiplication, a similar transition from Boolean to count/sum queries is possible by adapting 
the $\omega$-submodular width to become the {\em Sharp-$\omega$-submodular width}.
Nevertheless, this transition only applies to count/sum queries, i.e., aggregate queries over the semiring
of real numbers $(\R, +, \times)$, and does {\em not} extend to aggregate queries over other semirings like the tropical semiring
$(\R_+, \min, +)$, since fast matrix multiplication does not apply in the latter.

It is not clear how to extend our framework to full conjunctive queries due to the use of fast matrix multiplication.
For example, consider the full version of the triangle query from Eq.~\eqref{eq:intro:triangle}
where we replace the head $Q_\triangle()$ with $Q_\triangle(X, Y, Z)$.
While we could use matrix multiplication between $S$ and $T$ to compute the number of triangles
containing every edge $(X, Y)$ in $R$, the multiplication result cannot help us recover a list of
those triangles.
Nevertheless, we could still use this multiplication result to answer the version of the query where the head is $Q_\triangle(X, Y)$.
Prior results from Table~\ref{tab:intro:comparison} are not known to apply to the corresponding full conjunctive queries with an additional $+O(\out)$
factor in the runtime where $\out$ is the output size.
We leave characterizing (non-full) conjunctive queries that can benefit from fast matrix multiplication as an open problem.

\bibliographystyle{siam}
\bibliography{bib}

\appendix
\section{Computing the Submodular Width}
\label{app:computing-subw}
We summarize here how
to compute the submodular width of a given hypergraph $\calH=(\calV,\calE)$, as defined in Eq.~\eqref{eq:subw}.
To that end, we need to apply a couple of tricks.
Let $\calT(\calH) =$ $\{(T_1, \chi_1), (T_2, \chi_2), \ldots,$ $(T_K, \chi_K)\}$
be the list of all tree decompositions of $\calH$.
Let $\bm t = (t_1, t_2, \ldots, t_K) \in \nodes(T_1)\times\ldots\times \nodes(T_K)$ denote a tuple of nodes, one from each tree decomposition. (Note that $t_1, \ldots, t_K$ are {\em not} nodes from the same tree decomposition.)
By distributing\footnote{Note that $(\R, \max, \min)$ is a {\em commutative semiring}.} the min over the max in Eq.~\eqref{eq:subw}, we can equivalently write (where $[K]$ denotes $\{1, \ldots, K\}$):
\begin{align}
    \subw(\calH) &\quad\defeq\quad
            \max_{\bm h \in \Gamma \cap \ed}\quad
            \min_{(T,\chi) \in \calT(\calH)}\quad
            \max_{t \in \nodes(T)}\quad
            h(\chi(t))\nonumber\\
        &\quad=\quad
            \max_{\bm h \in \Gamma \cap \ed}\quad
            \max_{\substack{(t_1, \ldots, t_K)\in\\\nodes(T_1)\times\ldots\times\nodes(T_K)}}\quad
            \min_{i \in [K]}\quad
            h(\chi_i(t_i))\label{eq:subw:min-over-max}\\
        &\quad=\quad
            \max_{\substack{(t_1, \ldots, t_K)\in\\\nodes(T_1)\times\ldots\times\nodes(T_K)}}\quad
            \underbrace{\max_{\bm h \in \Gamma \cap \ed}\quad
            \min_{i \in [K]}\quad
            h(\chi_i(t_i))}_{\text{Equivalent to LP~\eqref{eq:subw:lp}}}\label{eq:subw:swap-max}
\end{align}
Eq.~\eqref{eq:subw:swap-max} follows by swapping the two max operators in Eq.~\eqref{eq:subw:min-over-max}.
Now for a {\em fixed} tuple $(t_1, \ldots, t_K)$, consider the inner optimization problem:
\begin{align}
    \max_{\bm h \in \Gamma \cap \ed}\quad
            \min_{i \in [K]}\quad
            h(\chi_i(t_i))
    \label{eq:subw:pre-lp}
\end{align}
Although the constraints $\bm h \in \Gamma \cap \ed$ are linear,
Eq.~\eqref{eq:subw:pre-lp} is still not a linear program (LP) because the objective is a minimum of linear functions. However, we can convert~\eqref{eq:subw:pre-lp} into an LP
by introducing a new variable $w$ that is upper bounded by each term in the minimum:
\begin{align}
    \max_{w, \bm h \in \Gamma \cap \ed}\{
        w\mid w \leq h(\chi_1(t_1)), \ldots, w \leq h(\chi_K(t_K))
        \}
    \label{eq:subw:lp}
\end{align}
As a result, computing the submodular width can be reduced to taking the maximum solution of a finite number of LPs, one for each tuple $(t_1, \ldots, t_K) \in \nodes(T_1)\times\ldots\times\nodes(T_K)$.

\begin{example}
    \label{ex:4-cycle:computing-subw}
    Consider the 4-cycle hypergraph from Example~\ref{ex:4cycle:td}. The submodular width of this hypergraph can be computed by considering the two tree decompositions $(T_1, \chi_1)$ and $(T_2,\chi_2)$ from Example~\ref{ex:4cycle:td}.
    Applying Eq.~\eqref{eq:subw}, we get:
    \begin{align*}
        \subw(\calH) =
            \max_{\bm h \in \Gamma \cap \ed}
            \min(
                \max(h(ABC), h(CDA)),
                \max(h(BCD), h(DAB))
            )
    \end{align*}
    By distributing the min over the inner max and then swapping the two max operators, we get:
    \begin{align*}
        \subw(\calH) = \max\biggl(
            &\max_{\bm h \in \Gamma \cap \ed}
            \min(h(ABC), h(BCD)),\quad
            \max_{\bm h \in \Gamma \cap \ed}
            \min(h(ABC), h(DAB)),\\
            &\max_{\bm h \in \Gamma \cap \ed}
            \min(h(CDA), h(BCD)),\quad
            \max_{\bm h \in \Gamma \cap \ed}
            \min(h(CDA), h(DAB))\biggr)
    \end{align*}
    Let's take the first term inside the max as an example $\max_{\bm h \in \Gamma \cap \ed}
            \min(h(ABC), h(BCD))$. This term is equivalent to the following LP:
    \begin{align*}
        \max_{w, \bm h \in \Gamma \cap \ed}\{
            w\mid w \leq h(ABC), w \leq h(BCD)
        \}
    \end{align*}
    The above LP has an optimal objective value of $3/2$. The other three terms inside the max are similar. Therefore, the submodular width of the 4-cycle hypergraph is $3/2$.
\end{example}

\section{Missing details from Section~\ref{sec:upper}}
\label{app:upper}
We show here how to compute the $\omega$-submodular width for the classes of queries that
are given by Table~\ref{tab:comparison} in Section~\ref{sec:upper}.
We start with some preliminaries: 
%\begin{proposition}[upper bound on $\osubw$]
%%\label{prop:upper-bound}
%    Let $c >0$ be a positive parameter. If for each polymatroid $\bm h \in \Gamma \cap \ed$, there exists a generalized elimination ordering $\ov{\bm \sigma} \in \pi(\calV)$ such that 
%    \begin{align}
%        \max_{\substack{i \in [|\ov{\bm\sigma}|]\\
%            {\color{red}\forall j < i,  U_i^{\ov{\bm\sigma}} \not\subseteq U_j^{\ov{\bm\sigma}}}}}\quad
%        \min(h(U^{\ov{\bm\sigma}}_i), \emm^{\ov{\bm\sigma}}_i) \le c
%    \label{eq:subw:ub}
%\end{align}
%    then $\osubw(\calH) \le c$. 
%\end{proposition}
%\begin{proposition}[lower bound on $\osubw$]
%\label{prop:lower-bound}
%If there exists a polymatroid $\bm h\in \Gamma \cap \ed$, such that for any generalized elimination ordering $\ov{\bm \sigma} \in \pi(\calV)$ such that 
%\begin{align}
%        \max_{\substack{i \in [|\ov{\bm\sigma}|]\\
%            {\color{red}\forall j < i,  U_i^{\ov{\bm\sigma}} \not\subseteq U_j^{\ov{\bm\sigma}}}}}\quad 
%        \min(h(U^{\ov{\bm\sigma}}_i), \emm^{\ov{\bm\sigma}}_i) \ge c
%    \label{eq:subw:lb}
%\end{align}
%then $\osubw(\calH) \ge c$.
%\end{proposition}

\begin{definition}[Fractional Edge Covering Number]
    For a hypergraph $\calH = (\calV, \calE)$, a fractional edge covering is a function $W:\calE \to [0,1]$ such that $\sum_{e \in \calE: X \in e} W(e) \ge 1$ holds for each vertex $X \in \calV$. The fractional edge covering number of $\calH$, denoted as $\rho^*(\calH)$, is defined as the minimum weight assigned to hyperedges over all possible fractional edge covering, i.e., $\rho^*(\calH) = \min_{W} \sum_{e \in \calE} W(e)$.
\end{definition}
\begin{proposition}
    \label{prop:agm-bound}
        In any hypergraph $\calH=(\calV,\calE)$, for any polymatroid $\bm h \in \Gamma \cap \ed$, $h(\calV) \le \rho^*(\calH)$.
\end{proposition}
    
\subsection{Clique Hypergraphs}

\begin{proposition}
\label{prop:clique-dominated}
    In a $k$-clique hypergraph $\calH = (\{X_1,X_2,\cdots,X_k\}, \{\{X_i,X_j\}: i,j \in [k], i \neq j\})$, for any generalized elimination ordering $\ov{\bm \sigma} \in \pi(\calV)$ and any $i \in[\ov{\bm \sigma}]$, $U^{\ov{\bm \sigma}}_i \subseteq \calV = U^{\ov{\bm \sigma}}_1$.
\end{proposition}

From Proposition~\ref{prop:clique-dominated} and Definition~\ref{defn:osubw}, we can simplify $\osubw$ for clique hypergraphs $\calH$ as below:
    \begin{align}
        {\osubw}(\calH) \quad = \quad
            \max_{\bm h \in \Gamma \cap \ed}\quad
            {\min\left(h(\calV), \min_{\ov{\bm\sigma} \in \ov{\pi}(\calV)} \emm^{\ov{\bm\sigma}}_1\right)}
            \label{eq:osubw:trimmed:clique}
    \end{align}

\begin{proposition}
    \label{prop:random-group}
    In a clique hypergraph $\calH = (\calV, \calE)$, with a polymatroid $\bm h \in \Gamma \cap \ed$ such that $h(\bm A | \bm B) = h(\bm A)$ for arbitrary $\bm A, \bm B\subseteq \calV$ with $\bm A \cap \bm B = \emptyset$, for any generalized elimination ordering $\ov{\bm \sigma} \in \pi(\calV)$,
    \[{\emm^{\ov{\bm \sigma}}_1 = \min_{(\bm A, \bm B, \bm C) \textsf{ is a partition of } \calV} \mm(\bm A, \bm B, \bm C)}\]
\end{proposition}

\begin{proof} The direction $\displaystyle{\emm^{\ov{\bm \sigma}}_1 \ge \min_{(\bm A, \bm B, \bm C) \textsf{ is a partition of } \calV} \mm(\bm A, \bm B, \bm C;\emptyset)}$ follows (\ref{defn:emm}). Below, we focus on:
    \begin{align*}
        \emm^{\ov{\bm \sigma}}_1 & = \min_{(\bm A, \bm B, \bm C, \bm D) \textsf{ is a partition of } \calV} \mm(\bm A, \bm B, \bm C|\bm D)\\
        & = \min_{(\bm A, \bm B, \bm C, \bm D) \textsf{ is a partition of } \calV} \min \begin{cases}
             \gamma h(\bm A|\bm D) + h(\bm B|\bm D) + h(\bm C|\bm D) + h(\bm D) & \\
             h(\bm A|\bm D) + \gamma  h(\bm B|\bm D) + h(\bm C|\bm D) + h(\bm D) & \\
             h(\bm A|\bm D) + h(\bm B|\bm D) + \gamma h(\bm C|\bm D) + h(\bm D) & \\
        \end{cases} \\
        & = \min_{(\bm A, \bm B, \bm C, \bm D) \textsf{ is a partition of } \calV} \min \begin{cases}
             \gamma h(\bm A) + h(\bm B) + h(\bm C) + h(\bm D) & \\
             h(\bm A) + \gamma  h(\bm B) + h(\bm C) + h(\bm D) & \\
             h(\bm A) + h(\bm B) + \gamma h(\bm C) + h(\bm D) & \\
        \end{cases} \\
        &\le \min_{(\bm A, \bm B, \bm C, \bm D) \textsf{ is a partition of } \calV}  \min \begin{cases}
             \gamma h(\bm A) + h(\bm B) + h(\bm C\bm D)& \\
             h(\bm A) + \gamma  h(\bm B) + h(\bm C \bm D) & \\
             h(\bm A) + h(\bm B) + \gamma h(\bm C \bm D) & \\
        \end{cases} \\
        & = \min_{(\bm A, \bm B, \bm C) \textsf{ is a partition of } \calV} \mm(\bm A, \bm B, \bm C)
    \end{align*}
    Together, we have completed the proof.
\end{proof}

\begin{lemma}
\label{lem:clique-3}
    For the 3-clique hypergraph as defined in (\ref{eq:intro:triangle}), $\osubw(\calH) =\frac{2\omega}{\omega+1}$. 
\end{lemma}

\begin{proof}
  \underline{\em Direction $\osubw(\calH) \le \frac{2\omega}{\omega+1}$.} We partition all polymatroids $h \in  \Gamma \cap \ed$ into the following two cases:
    \begin{itemize}[leftmargin=*]
        \item Case 1: there exists a pair of vertices with their conditional entropy is smaller than $\frac{\omega-1}{\omega+1}$, say $h(Y|X)\le \frac{\omega-1}{\omega+1}$. %Wlog, assume $i=2$ and $j=1$. 
        In this case, we have $h(XYZ) = h(ZX) + h(Y|ZX) \le h(ZX) + h(Y|X) \le 1 + \frac{\omega-1}{\omega+1} = \frac{2\omega}{\omega+1}$, where $h(Y|ZX) \le h(Y|X)$ follows the sub-modularity of $h$ and $h(ZX) \le 1$ follows the edge-domination of $h$. 
        \item Case 2: None of the above, i.e., for each pair of vertices, their conditional entropy is larger than $\frac{\omega-1}{\omega+1}$. As $h(Y|X) > \frac{\omega-1}{\omega+1}$ and $h(XY) \le 1$, we have $h(X) = h(XY) - h(Y|X) \le 1 - \frac{\omega-1}{\omega+1} = \frac{2}{\omega+1}$. Similarly, we have $h(Y), h(Z) \le \frac{2}{\omega+1}$. Consider a generalized elimination ordering $\ov{\bm \sigma}= (\{X\},\{Y\},\{Z\})$. We have $\emm^{\ov{\bm \sigma}}_1 \le \mm(Y;Z;X) = \frac{2\omega}{\omega+1}$ by choosing $\calA = \{\{X,Y\}\}$ and $\calB = \{\{X,Z\}\}$.
    \end{itemize}
    
    \underline{\em Direction $\osubw(\calH) \ge \frac{2\omega}{\omega+1}$.} We identify a polymatroid $\bm h$ below. Let $a,b,c,d$ be independently random variables with $h(a)=h(b)=h(c) =\frac{\omega-1}{\omega+1}$ and $h(d)=\frac{3-\omega}{\omega+1}$. Let $X = (ad), Y =(bd), Z=(cd)$.
        \begin{itemize}[leftmargin=*]
        \item $h(X) = h(Y) = h(Z) =\frac{2}{\omega+1}$;
        %\item $h(X|Y) = h(Y|X) = h(X|Z) = h(Z|X) = h(Y|Z) = h(Z|Y) =\frac{\omega-1}{\omega+1}$;
        \item $h(XY) = h(XZ) = h(YZ)= 1$;
        \item $h(XYZ) = \frac{2\omega}{\omega+1}$.
        \end{itemize}

    \begin{figure}
        \centering
    \includegraphics[scale=1.0]{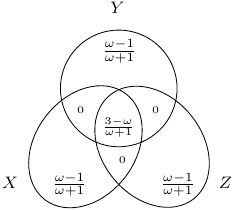}
        \vspace{-1em}
        \caption{Diagram of the polymatroid for the clique-3 hypergraph.}
        \label{fig:enter-label}
    \end{figure}
    Consider a generalized elimination ordering $\ov{\bm \sigma} = (\{X\}, \{Y\}, \{Z\})$, $\emm^{\ov{\bm\sigma}}_1 = \mm(Y,Z,X) = \frac{2\omega}{\omega+1}$. Hence, $\displaystyle{{\min\left(h(\calV), \min_{\ov{\bm\sigma} \in \ov{\pi}(\calV)} \emm^{\ov{\bm\sigma}}_1\right)}}=$ $ \frac{2\omega}{\omega+1}$ holds for $\bm h$. From (\ref{eq:osubw:trimmed:clique}), we conclude $\osubw(\calH) \ge \frac{2\omega}{\omega+1}$. The other generalized elimination orderings are the same.
\end{proof}

\begin{lemma}
\label{lem:clique-4}
    For the 4-clique hypergraph as defined in (\ref{eq:H:4clique}), $\osubw(\calH) = \frac{\omega+1}{2}$. 
\end{lemma}

\begin{proof}
    \underline{\em Direction $\osubw(\calH) \le \frac{\omega+1}{2}$.} We partition all polymatroids $\bm h \in \Gamma \cap \ed$ into the following two cases:
    
    {\bf Case 1:} there exists a vertex such that its all conditional entropies are smaller than $\frac{1}{2}$, say, $W$ is such a vertex with $h(X|W), h(Y|W), h(Z|W) \le \frac{1}{2}$. Consider a generalized elimination ordering $\ov{\bm \sigma}= (\{X\},\{Y\},\{Z\},\{W\})$. By choosing $\calA = \{\{X,W\}, \{Y,W\}\}$ and $\calB = \{\{X,W\}, \{Z,W\}\}$, 

        \begin{align*} \emm^{\ov{\bm \sigma}}_1 \le \mm(Y; Z; X|W) &= \max\begin{cases}
        \gamma h(Y|W) + h(X|W) + h(Z|W) +h(W)&\\
        h(Y|W) + \gamma h(Z|W) + h(X|W) +h(W)&\\
        h(Y|W) + h(Z|W) + \gamma h(X|W) + h(W)&\end{cases}
        \le \frac{\gamma}{2} + \frac{1}{2} +1 = \frac{\omega+1}{2}.
        \end{align*}
    {\bf  Case 2:} None of the above, i.e., each vertex has at least one conditional entropy larger than $\frac{1}{2}$. Wlog, suppose $h(Y|X) > \frac{1}{2}$. As $h(XY) \le 1$, we have $h(X) = h(XY) - h(Y|X) < 1 - \frac{1}{2} = \frac{1}{2}$. Similarly, we have $h(Y), h(Z), h(W) < \frac{1}{2}$. Consider a generalized elimination ordering $\ov{\bm \sigma}=(\{X\},\{Y\},\{Z\},\{W\})$. We have $\emm^{\ov{\bm \sigma}}_1 \le \mm(Y; Z; X|W) =\frac{\omega+1}{2}$. 

    \underline{\em Direction $\osubw(\calH) \ge \frac{\omega+1}{2}$.}
    We identify a polymatroid $\bm h\in \Gamma \cap \ed$ as follows: let $X,Y,Z,W$ be independently random variables with $h(X) = h(Y) = h(Z) = h(W) =\frac{1}{2}$. Then, we have $h(\calV) = 2$. From Proposition~\ref{prop:random-group}, for an arbitrary generalized elimination ordering $\ov{\bm \sigma} \in \pi(\calV)$, $\emm^{\ov{\bm \sigma}}_1 = \frac{\omega+1}{2}$. From (\ref{eq:osubw:trimmed:clique}), we conclude $\osubw(\calH) \ge \frac{\omega+1}{2}$.
\end{proof}
\begin{lemma}
\label{lem:clique-5}
    For the following 5-clique hypergraph $\calH$, $\osubw(\calH) =\frac{\omega}{2} + 1$.
    \begin{align}
     \calH = (\{X, Y, Z, W, L\}, \{\{X, Y\}, \{X, Z\}, \{X, W\}, \{X, L\}, \{Y, Z\}, \{Y, W\}, \{Y, L\}, \{Z, W\}, \{Z, L\}, \{W, L\}\})
     \label{eq:H:5clique}
    \end{align} 
\end{lemma}

\begin{proof}
     \underline{\em Direction $\osubw(\calH) \le \frac{\omega}{2}+1$.} We partition all polymatroids $\bm h \in \Gamma \cap \ed$ into the following five cases:
    \begin{itemize}[leftmargin=*]
        \item Case 1: there exists a vertex such that all its conditional entropies are smaller than $\frac{1}{2}$, say, $L$ is such a vertex with $h(X|L), h(Y|L), h(Z|L), h(W|L) \le \frac{1}{2}$. Consider a generalized elimination ordering $\ov{\bm \sigma}= (\{X\},\{YZ\},\{W\},\{L\})$.
        By choosing $\calA = \{\{X,W\}, \{Y,W\}, \{X,L\}\}$ and $\calB = \{\{X,W\}, \{Z,W\}\}$, 
        \begin{align*}    
        \emm^{\ov{\bm \sigma}}_1 \le \mm(YZ, W, X|L) & = 
        h(YZ) + \gamma \cdot \min\{h(W|L), h(X|L)\} + \max\{h(W|L), h(X|L)\} + h(L)\\
        & \le 1 +\frac{\gamma}{2} + 1 = \frac{\omega}{2}+1.
        \end{align*}     
        \item Case 2: None of the above, i.e., each vertex has at least one conditional entropy larger than $\frac{1}{2}$.  Suppose $h(Y|X) > \frac{1}{2}$ for vertex $x$. As $h(XY) \le 1$, we have $h(X) = h(XY) - h(Y|X) < 1 - \frac{1}{2} = \frac{1}{2}$. Similarly, we have $h(Y), h(Z), h(W), h(L) < \frac{1}{2}$.  Consider a generalized elimination ordering $\ov{\bm \sigma}= (\{X\},\{Y\},\{Z\},\{W\},\{L\})$. We have $\emm^{\ov{\bm \sigma}}_1  \le \mm(YZ; WL; X) \le \frac{\omega}{2}+1$. 
    \end{itemize}
    
    \underline{\em Direction $\osubw(\calH) \ge \frac{\omega}{2}+1$.} We identify a polymatroid $\bm h\in \Gamma \cap \ed$ as follows: let $X,Y,Z,W,L$ be independently random variables with $h(X) = h(Y) = h(Z) = h(W) = h(L) = \frac{1}{2}$. Then, we have $h(\calV) = \frac{5}{2}$. From Proposition~\ref{prop:random-group}, for an arbitrary generalized elimination ordering $\ov{\bm \sigma} \in \pi(\calV)$, $\emm^{\ov{\bm \sigma}}_1 = \frac{\omega}{2} + 1$.
    \end{proof}

    \begin{lemma}
    \label{lem:clique}
    For a $k$-clique hypergraph $\calH =(\calV, \calE)$ for $k \ge 6$, where $\calV = \{X_1,X_2,\cdots, X_k\}$ and $\calE = \{\{X_i, X_j\}: i,j \in [k], i \neq j\}$, $\osubw(\calH) = \frac{1}{2} \cdot \lceil \frac{k}{3}\rceil +  \frac{1}{2} \cdot \lceil \frac{k-1}{3}\rceil + \frac{1}{2} \cdot \lfloor \frac{k}{3}\rfloor \cdot (\omega-2)$.
    \end{lemma}

    \begin{proof}
    \underline{\em Direction $\osubw(\calH) \le  \frac{1}{2} \cdot \lceil \frac{k}{3}\rceil +  \frac{1}{2} \cdot \lceil \frac{k-1}{3}\rceil + \frac{1}{2} \cdot \lfloor \frac{k}{3}\rfloor \cdot (\omega-2)$}. Let $i= \lfloor \frac{k}{3}\rfloor$ and $j = \lceil \frac{k-1}{3}\rceil $. Let $\bm X = \{X_1,\cdots,X_i\}$, $\bm Y = \{X_{i+1}, \cdots, X_{i+j}\}$ and $\bm Z = \{X_{i+j+1},\cdots, X_k\}$. we consider a generalized elimination ordering $\ov{\bm \sigma}=(\bm X, \bm Y, \bm Z)$.  Implied by Proposition~\ref{prop:agm-bound}, $h(\bm X) \le \frac{1}{2} \cdot \lfloor \frac{k}{3} \rfloor$, $h(\bm Y) \le \frac{1}{2} \cdot \lceil \frac{k-1}{3} \rceil$, and $h(\bm Z) \le \frac{1}{2} \cdot \lceil \frac{k}{3} \rceil$. 
    \begin{align*}    
        \emm^{\ov{\bm \sigma}}_1 \le \mm(\bm Y; \bm Z; \bm X) =  h(\bm Y) + h(\bm Z) + \gamma h(\bm X) = \frac{1}{2} \cdot \lceil \frac{k}{3}\rceil +  \frac{1}{2} \cdot \lceil \frac{k-1}{3}\rceil + \frac{1}{2} \cdot \lfloor \frac{k}{3}\rfloor \cdot (\omega-2)
    \end{align*}
    
    \underline{\em Direction $\osubw(\calH) \ge  \frac{1}{2} \cdot \lceil \frac{k}{3}\rceil +  \frac{1}{2} \cdot \lceil \frac{k-1}{3}\rceil + \frac{1}{2} \cdot \lfloor \frac{k}{3}\rfloor \cdot (\omega-2)$}.   We identify a polymatroid $\bm h$ as follows: let $X_1,X_2,\cdots, X_k$ be independently random variables with $h(X_i) = \frac{1}{2}$ for $i \in [k]$. Then, we have $h(\calV) = \frac{k}{2}$. From Proposition~\ref{prop:random-group}, for an arbitrary generalized elimination ordering $\ov{\bm \sigma} \in \pi(\calV)$, $\emm^{\ov{\bm \sigma}}_1 = \frac{1}{2} \cdot \lceil \frac{k}{3}\rceil +  \frac{1}{2} \cdot \lceil \frac{k-1}{3}\rceil + \frac{1}{2} \cdot \lfloor \frac{k}{3}\rfloor \cdot (\omega-2)$. From (\ref{eq:osubw:trimmed:clique}), we conclude $\osubw(\calH) \ge \frac{\omega}{2} + 1$.
\end{proof}

    \subsection{Cycle Hypergraphs}

    \begin{lemma}
    \label{lem:4cycle}
        For the following hypergraph representing a $4$-cycle $\calH$, $\osubw(\calH) =
        2 -\frac{3}{2 \cdot \min\{\omega, \frac{5}{2}\} + 1}$.
        \begin{align}
        \calH = (\{X,Y,Z,W\},\quad \{\{X,Y\}, \{Y,Z\}, \{Z,W\}, \{W,X\}\})
        \label{eq:H:4cycle}
        \end{align}
        
    \end{lemma}

    \begin{proof}                          \underline{\em Direction $\osubw(\calH) \le 2 -\frac{3}{2 \cdot \min\{\omega, \frac{5}{2}\} + 1}$.} %Let $\Delta = \min\left\{\frac{1}{2}, \frac{2(\omega-2)}{2\omega-3}\right\}$.
    For simplicity, for vertex $X$, we denote $h(Y|X)$, $h(W|X)$ as the {\em neighboring conditional entropy} of $X$. Similar denotations apply to $Y,Z,W$. We partition all polymatroids $\bm h \in \Gamma \cap \ed$ into the following four cases:
    \begin{itemize}[leftmargin=*]
        \item Case 1: each vertex has some neighboring conditional entropy larger than $\Delta$. Wlog, assume vertex $X$ has $h(Y|X) > \Delta$. As $h(XY) \le 1$, we have $h(X) = h(XY) - h(Y|X) \le 1 - \Delta$. Similarly, we can show $h(Y), h(Z), h(W) < 1- \Delta$. Consider a generalized elimination ordering $\ov{\bm \sigma} = (\{X\},\{Z\}, \{Y\}, \{W\})$. By choosing $\calA = \{\{X,Y\}\}$ and $\calB= \{\{X,W\}\}$, we have $\emm^{\ov{\bm \sigma}}_1 \le \mm(Y;W;X) \le \omega(1-\Delta)$. Meanwhile, $h\left(U^{\ov{\bm \sigma}}_1\right) = h(XYW) \le h(XY)+ h(W) \le 2 - \Delta$. By choosing $\calA = \{\{Z,Y\}\}$ and $\calB= \{\{Z,W\}\}$, we have $\emm^{\ov{\bm \sigma}}_2 \le \mm(Y;W;Z) \le \omega(1-\Delta)$.
        Meanwhile, $h\left(U^{\ov{\bm \sigma}}_2\right) = h(YZW) \le h(YZ)+ h(W) \le 2 - \Delta$.  %Note that $U^{\ov{\bm \sigma}}_3, U^{\ov{\bm \sigma}}_4 \subseteq U^{\ov{\bm \sigma}}_1$. Hence, $\displaystyle{\min_{i \in [|\ov{\bm \sigma}|]} \emm_i^{\ov{\bm \sigma}} \le \omega(1-\Delta)}$.
        As $U^{\ov{\bm \sigma}}_3, U^{\ov{\bm \sigma}}_4\subseteq U^{\ov{\bm \sigma}}_1$, we have $\displaystyle{\max_{i \in [|\ov{\bm \sigma}|]} \min\left(h(U^{\ov{\bm \sigma}}_i) ,\emm_i^{\ov{\bm \sigma}}\right) \le \min\left\{\omega(1-\Delta), 2 - \Delta\right\}}$.
        \item Case 2: one vertex has all neighboring conditional entropies smaller than $\Delta$, and the other three vertices have some neighboring conditional entropies larger than $\Delta$. Wlog, assume $Z$ is the vertex with $h(Y|Z), h(W|Z) \le \Delta$.  Similar to Case 1, we can have $h(X), h(Y), h(W) < 1- \Delta$. Consider a generalized elimination ordering $\ov{\bm \sigma} = (\{X\},\{Z\},\{Y\},\{W\})$. By choosing $\calA = \{\{X,Y\}\}$ and $\calB= \{\{X,W\}\}$, we have $\emm^{\ov{\bm \sigma}}_1 \le \mm(Y;W;X) \le \omega(1-\Delta)$. Meanwhile, $h\left(U^{\ov{\bm \sigma}}_1\right) = h(XYW) \le h(XY)+ h(W) \le 2 - \Delta$. We have $h(U^{\ov{\bm \sigma}}_2) = h(YZW) \le h(Y|ZW) + h(ZW) \le h(Y|Z) + h(ZW) \le 1+ \Delta$. As $U^{\ov{\bm \sigma}}_3, U^{\ov{\bm \sigma}}_4\subseteq U^{\ov{\bm \sigma}}_1$, we have:
        $$\max_{i \in [|\ov{\bm \sigma}|]} \min\left(h(U^{\ov{\bm \sigma}}_i) ,\emm_i^{\ov{\bm \sigma}}\right) \le \max\left\{\min\left\{\omega(1-\Delta), 2 - \Delta\right\}, 1+ \Delta\right\}.$$
        \item Case 3: two non-tangent vertices have all their neighboring conditional entropies smaller than $\Delta$. Suppose $X, Z$ are two non-tangent vertices with all their neighboring conditional entropies smaller than $\Delta$. Consider a generalized elimination ordering $\ov{\bm \sigma} = (\{X\},\{Z\},\{Y\},\{W\})$. We have $h(U^{\ov{\bm \sigma}}_1) = h(XYW) \le h(W|XY) + h(XY) \le h(W|X) + h(XY) \le 1 + \Delta$. Similarly, $h(U^{\ov{\bm \sigma}}_2) = h(ZYW) \le h(W|ZY) + h(ZY) \le h(W|Z) + h(ZY) \le 1 + \Delta$. As $U^{\ov{\bm \sigma}}_3, U^{\ov{\bm \sigma}}_4\subseteq U^{\ov{\bm \sigma}}_1$,  we have $\displaystyle{\max_{i \in [|\ov{\bm \sigma}|]} \min\left(h(U^{\ov{\bm \sigma}}_i) ,\emm_i^{\ov{\bm \sigma}}\right) \le 1+ \Delta}$.
        \item Case 4: two tangent vertices have all their neighboring conditional entropies smaller than $\Delta$, and the other two have some neighboring conditional entropy larger than $\Delta$. Wlog, assume $Z,W$ are two tangent vertices with all their neighboring conditional entropies smaller than $\Delta$. Similar as above, we have $h(X), h(Y) \le  1-\Delta$. %We have two different observations:
        %\begin{itemize}[leftmargin=*]
            %\item Consider a generalized elimination ordering $\ov{\bm \sigma}=(\{Y\},\{W\},\{X\},\{Z\})$. We have $h(U^{\ov{\bm \sigma}}_1) = h(XYZ) \le h(X) + h(YZ) \le 2 - \Delta$ and $h(U^{\ov{\bm \sigma}}_2) = h(XZW) \le h(X) + h(ZW) \le 2 - \Delta$. As $U^{\ov{\bm \sigma}}_3, U^{\ov{\bm \sigma}}_4\subseteq U^{\ov{\bm \sigma}}_1$, we have $\displaystyle{\max_{i \in [|\ov{\bm \sigma}|]} \min\left(h(U^{\ov{\bm \sigma}}_i) ,\emm_i^{\ov{\bm \sigma}}\right) \le 2 -\Delta}$.
            %\item 
            We also further distinguish two more cases:
            \begin{itemize}[leftmargin=*]
                \item If $h(Y|Z) \le \frac{\Delta}{2}$ and $h(X|W) \le \frac{\Delta}{2}$, we have $h(XYZW) \le h(X|YZW) + h(Y|ZW) + h(ZW) \le  h(X|W) + h(Y|Z) + h(ZW) \le 1 + \Delta$.
                \item Otherwise, either $h(Y|Z) > \frac{\Delta}{2}$ or $h(X|W) > \frac{\Delta}{2}$. Wlog, assume $h(Y|Z) > \frac{\Delta}{2}$. Similar as above, we have $h(Z)\ge 1 -\frac{\Delta}{2}$. Consider a generalized elimination ordering $\ov{\bm \sigma}= (\{Y\},\{W\},\{Z\},\{X\})$. By choosing $\calA = \{\{X,Y\}\}$ and $\calB = \{\{Y,Z\}\}$, we have $\emm^{\ov{\bm \sigma}}_1 \le \mm(X;Z;Y|\emptyset) = (\omega-1)(1-\Delta) + (1-\frac{\Delta}{2})$ and $h(U^{\ov{\bm \sigma}}_1) = h(XYZ)\le h(X)+ h(YZ) \le 2-\Delta$. Moreover, $h(U^{\ov{\bm \sigma}}_2) = h(ZWX) \le h(Z|WX)+ h(WX) \le h(Z|W) + h(WX) \le 1 + \Delta$. As $U^{\ov{\bm \sigma}}_3, U^{\ov{\bm \sigma}}_4\subseteq U^{\ov{\bm \sigma}}_1$, we have:
                $$\max_{i \in [|\ov{\bm \sigma}|]} \min\left(h(U^{\ov{\bm \sigma}}_i) ,\emm_i^{\ov{\bm \sigma}}\right) \le \max\left\{1+\Delta, \min\left\{ (\omega-1)(1-\Delta) + \left(1-\frac{\Delta}{2}\right), 2-\Delta\right\}\right\}.$$
            \end{itemize}
        \end{itemize}
    %\end{itemize} 
    Putting all cases together, we obtain the upper bound: 
    \[\max\left\{1 + \Delta, \min\left\{(\omega-1)(1-\Delta) + 1-\frac{\Delta}{2}, 2-\Delta\right\}\right\}\] If $\frac{5}{2} \le \omega \le 3$, we set $\Delta = \frac{1}{2}$ to obtain the upper bound as $\frac{3}{2}$. If $2\le \omega \le \frac{5}{2}$, we set $\Delta = \frac{2(\omega-1)}{2\omega+1}$ to obtain the upper bound as $\frac{4\omega-1}{2\omega+1}$. Putting two cases together, we obtain the desired upper bound.

    \underline{\em Direction $\osubw(\calH) \ge 2 -\frac{3}{2 \cdot \min\{\omega, \frac{5}{2}\} + 1}$.} Correspondingly, we identify a polymatroid $\bm h \in \Gamma \cap \ed$ by distinguishing the following two cases:
   
    {\bf Case 1: $\frac{5}{2} \le \omega \le 3$.} Let $a,b,c,d,e$ be independently random variables with $h(a)=h(b) =h(c) =h(d) =\frac{1}{4}$ and $h(e) = \frac{1}{2}$. Let $h(X) = (ab)$, $h(Y) = (cd)$, $h(Z) = (de)$ and $h(W) =(ae)$. We have
                \begin{itemize}[leftmargin=*]
                \item $h(X) = h(Y) = \frac{1}{2}$, $h(Z) = h(W) = \frac{3}{4}$;
                \item $h(XY) = h(YZ) = h(ZW) = h(WX) = 1$ and $h(XZ) = h(YW) = \frac{5}{4}$;
                \item $h(XZW) =h(YZW) = \frac{5}{4}$ and $h(XYZ) = h(XYW) = \frac{3}{2}$;
                \item $h(XYZW) = \frac{3}{2}$;
            \end{itemize}
            \noindent Consider any generalized elimination ordering $\ov{\bm \sigma} \in \pi(\calV)$. We distinguish the following cases:
            \begin{itemize}[leftmargin=*]
                \item $\ov{\bm \sigma}[1] = \{X\}$. $h(U^{\ov{\bm \sigma}}_1) = h(XYW)=\frac{3}{2}$ and $\emm^{\ov{\bm \sigma}}_1 = \mm(Y,W,X) = \frac{\omega}{2} + \frac{1}{4} \ge \frac{3}{2}$ when $\omega \ge \frac{5}{2}$. The case with $\ov{\bm \sigma}[1] = \{W\}$ is the same.
                \item $\ov{\bm \sigma}[1] = \{W\}$. $h(U^{\ov{\bm \sigma}}_1) = h(XZW)=\frac{5}{4}$ and $\emm^{\ov{\bm \sigma}}_1 = \mm(X,Z,W) = \frac{\omega}{2} + \frac{1}{4} \ge \frac{3}{2}$ when $\omega \ge \frac{5}{2}$. %$h(U^{\ov{\bm \sigma}}_1) = h(XZW)=\frac{3}{2}$ and $\emm^{\ov{\bm \sigma}}_2 =\mm(X,W,Z) = \frac{\omega}{2} + \frac{1}{4} \ge \frac{3}{2}$. 
                The case with $\ov{\bm \sigma}[1] = \{Z\}$ is the same.
                \item $|\bm{\bm \sigma}[1]|=2$. $h(U^{\ov{\bm \sigma}}_1) = h(\calV) = \frac{3}{2}$ and $\emm^{\ov{\bm \sigma}}_1 \ge \mm(X,W,Z) = \frac{\omega}{2} + \frac{1}{4} \ge \frac{3}{2}$.
            \end{itemize}
            \ \
           \begin{figure}[t]
                \centering                %%\includegraphics[scale=0.85]{4-cycle-1.pdf}
                %\includegraphics[scale=0.8]{4-cycle-2.pdf}
                %\ 
                \includegraphics[scale=0.8]{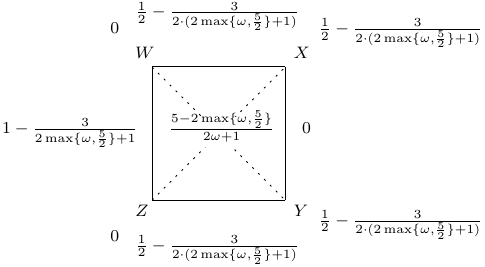}
                \caption{Diagram of the polymatroid for cycle-4 hypergraph.}
        \end{figure}
        
        {\bf Case 2:} $\omega < \frac{5}{2}$. Let $a,b,c,d,e,f$ be independently random variables with $h(a) = \frac{2(\omega-1)}{2\omega+1}$, $h(b)=h(c) = h(d) = h(e) =\frac{\omega-1}{2\omega+1}$ and $h(f) = \frac{5-2\omega}{2\omega+1}$. Let $X=(bcf), Y=(def), Z =(aef), W=(abf)$. We have
        \begin{itemize}[leftmargin=*]
            \item $h(W)= h(Z) = \frac{\omega+2}{2\omega+1}$ and $h(X) = h(Y) = \frac{3}{2\omega+1}$;
            \item $h(WX) = h(XY) = h(YZ) = h(ZW) = 1$, $h(WY) = h(XZ) = \frac{3\omega}{2\omega+1}$;
            \item $h(WXY) = h(XYZ)= \frac{4\omega-1}{2\omega+1}$, $h(XZW) = h(YZW) = \frac{3\omega}{2\omega+1} $;
            \item $h(XYZW) = \frac{4\omega-1}{2\omega+1}$;
        \end{itemize}
       Consider any generalized elimination ordering $\ov{\bm \sigma} \in \pi(\calV)$. We distinguish the following cases:
       \begin{itemize}[leftmargin=*]
            \item $\ov{\bm \sigma}[1] = \{X\}$. $h(U^{\ov{\bm \sigma}}_1) = h(XYW)=\frac{4\omega-1}{2\omega+1}$ and $\emm^{\ov{\bm \sigma}}_1 = \mm(Y,W,X) = \frac{\omega+2}{2\omega+1} + \frac{3(\omega-1)}{2\omega+1} = \frac{4\omega-1}{2\omega+1}$. The case with $\ov{\bm \sigma}[1] = \{Y\}$ is the same.
            \item $\ov{\bm \sigma}[1] = \{W\}$. $h(U^{\ov{\bm \sigma}}_1) = h(XZW)=\frac{3\omega}{2\omega+1}$ and $\emm^{\ov{\bm \sigma}}_1 = \mm(X,Z,W) = \frac{2(\omega+2)}{2\omega+1} + \frac{3(\omega-2)}{2\omega+1} = \frac{5\omega-2}{2\omega+1}> \frac{4\omega-1}{2\omega+1}$. The case with $\ov{\bm \sigma}[1] = \{Z\}$ is the same.
            \item $|\ov{\bm \sigma}[1]| = 2$. $h(U^{\ov{\bm \sigma}}_1) = h(\calV)=\frac{4\omega-1}{2\omega+1}$ and $\emm^{\ov{\bm \sigma}}_1 \ge \mm(Y,W,X) = \frac{4\omega-1}{2\omega+1}$.
       \end{itemize}
    \end{proof}

    \cite{dalirrooyfard2019graph} defines an exponent $\rectC_k$ for detecting a $k$-cycle in a graph, and the definition is based on {\em rectangular} matrix multiplication. We define below an upper bound,
    $\squareC_k$, on $\rectC_k$ that uses only {\em square} matrix multiplication and show that
    $\osubw(\calH) \le \squareC_k$ for any $k$-cycle graph $\calH$.
    In particular, $\rectC_k \leq \squareC_k$ and this becomes an equality when 
    $\omega = 2$.
    Moreover, \cite{dalirrooyfard2019graph} shows that $\rectC_k = \squareC_k$ when $k$ is odd as well as $k=4$ or $6$. More details can be found in \cite{dalirrooyfard2019graph}.

    First, we recall the definition of $\rectC_k$ from~\cite{dalirrooyfard2019graph}.
    Recall from Table~\ref{tab:intro:comparison} that $\rectOmega(a,b,c)$ is the smallest exponent for multiplying two rectangular matrices of sizes $n^a \cdot n^b$ and $n^b \cdot n^c$ within $O\left(n^{\rectOmega(a,b,c)}\right)$ time,
    whereas $\squareOmega(a,b,c)$ is the smallest upper bound on $\rectOmega(a,b,c)$
    that is obtained through {\em square} matrix multiplication, i.e.~ $\squareOmega(a, b, c) \defeq \max\{a + b+(\omega-2)c, a+(\omega-2)b + c, (\omega-2)a + b + c\}$.
    In particular, $\rectOmega(a,b,c) \leq \squareOmega(a, b, c)$, and this becomes an equality when $\omega = 2$ or when $a=b=c$.
    Given a graph of size $N$, let $\mathbf{D} = \{0,\log_N 2, 2\log_N 2, \cdots, 1\}^{2k}$.
    For each vector $\vec{d}$, \cite{dalirrooyfard2019graph} defines a function $\rectP^{\vec{d}}_{i,j}$ (for $i < j$) in a recursive way as follows:
    \begin{equation}
        \label{eq:rect-ck-pij}
        \rectP^{\vec{d}}_{i,j} = \min\left\{\rectP^{\vec{d}}_{i,j-1} + d^+_{j-1}, \rectP^{\vec{d}}_{i+1,j} + d^{-}_{i+1}, \min_{i<r < j}\left\{\rectP^{\vec{d}}_{i,r}, \rectP^{\vec{d}}_{r,j}, \rectOmega(1-d_i, 1-d_r,  1-d_j)\right\}\right\}
    \end{equation}
    with $\rectP^{\vec{d}}_{i,i+1} = 1$. 
    Finally, \cite{dalirrooyfard2019graph} defines $\rectC_k$ as follow:
    \begin{equation}
        \label{eq:rect-ck}
        \rectC_k = \max_{\vec{d} =\left(d^-_1, d^+_1, d^-_2, d^+_2, \cdots, d^-_k, d^+_k\right) \in \mathbf{D}}\min\left\{\min_{i \in [k]} 2 - d_i, \min_{i,j \in [k]: i<j} \max\left\{\rectP^{\vec{d}}_{i,j}, \rectP^{\vec{d}}_{j,i}\right\}\right\}
    \end{equation}

    In contrast, by replacing $\rectOmega(1-d_i, 1-d_r,  1-d_j)$ with $\squareOmega(1-d_i, 1-d_r,  1-d_j)$ in Eq.~\eqref{eq:rect-ck-pij}, we obtain our variants $P^{\vec{d}}_{i,j} $ and $\squareC_k$ defined below (with $P^{\vec{d}}_{i,i+1} = 1$):

    \begin{equation}
        \label{eq:square-ck-pij}
        P^{\vec{d}}_{i,j} = \min\left\{P^{\vec{d}}_{i,j-1} + d^+_{j-1}, P^{\vec{d}}_{i+1,j} + d^{-}_{i+1}, \min_{i<r < j}\left\{P^{\vec{d}}_{i,r}, P^{\vec{d}}_{r,j}, \squareOmega(1-d_i, 1-d_r,  1-d_j)\right\}\right\}
    \end{equation}
    \begin{equation}
        \label{eq:square-ck}
        \squareC_k = \max_{\vec{d} =\left(d^-_1, d^+_1, d^-_2, d^+_2, \cdots, d^-_k, d^+_k\right) \in \mathbf{D}}\min\left\{\min_{i \in [k]} 2 - d_i, \min_{i,j \in [k]: i<j} \max\left\{P^{\vec{d}}_{i,j}, P^{\vec{d}}_{j,i}\right\}\right\}
    \end{equation}
    
    \begin{lemma}
    \label{lem:cycle}
        For the following hypergraph representing a $k$-cycle:
        \begin{align}
        \calH = (\{X_1,X_2,\cdots,X_k\},\quad \{\{X_1,X_2\}, \{X_2,X_3\}, \cdots, \{X_{k-1},X_k\}, \{X_k,X_1\}\})
        \label{eq:H:cycle}
        \end{align}
        $\osubw(\calH) \le \squareC_k$ as defined in Eq.~\eqref{eq:square-ck}. %where $\rectC_k$ is the best known exponent for detecting $k$-cycle in graphs \cite{dalirrooyfard2019graph}.\footnote{This is a new analysis of the classic algorithm in \cite{yuster2004detecting}.}
    \end{lemma}

    \begin{proof}[Proof of Lemma~\ref{lem:cycle}]
    We partition all polymatroids $\bm h$ into the following $(1+ \log N)^{|\calV|}$ cases. For simplicity, let $X_1 = X_{k+1}$. Each polymatroid $\bm h \in \Gamma \cap \ed$ is associated with a $k$-tuple vector $\vec{d} = (d^-_1, d^+_1, d^-_2, d^+_2, \cdots, d^-_k, d^+_k) \in \mathbf{D}$ such that $d^{-}_i \le h(X_{i-1}|X_i)< d^-_{i} + \log_N 2$ and $d^{+}_i \le h(X_{i+1}|X_i) < d^+_{i}+ \log_N 2$. We have the following two observations:
    \begin{itemize}[leftmargin=*]
        \item For each $i \in [k]$, consider a generalized elimination ordering $\ov{\bm \sigma} =(\{X_{i-1}\}, \cdots, \{X_1\},$ $\{X_k\}, \cdots, \{X_{i}\}).$ As $h(X_{i-1}|X_i) \ge d^-_i$ and $h(X_{i-1}X_i) \le 1$, we have $h(X_i)= h(X_{i-1}X_i) - h(X_{i-1}|X_i) < 1- d^{-}_i$. Similarly, we have $h(X_i) < 1 - d^{+}_i$. For simplicity, we define $d_i = \max\{d^{-}_i, d^{+}_i\}$. For each $i \in [k]$, we have $h\left(X_iX_{j}X_{j+1}\right) \le h(X_i) + h(X_jX_{j+1}) < \min\{1 - d^-_{i}, 1 - d^+_{i}\} + 1 = 2 - d_i$. So, $h\left(U^{\ov{\bm \sigma}}_{j}\right) \le 2 - d_i$ for any $j \in [|\ov{\bm \sigma}|]$, i.e., $\displaystyle{\max_{j \in [|\ov{\bm \sigma}|]} h\left(U^{\ov{\bm \sigma}}_{j}\right) \le 2 - d_i}$.
        \item For each pair of distinct values $i,j \in [k]$ (wlog assume $i < j$), we define $\ov{\bm \sigma}^{\vec{d}}_{i,j}$ as an ordering of $\{X_{i+1},X_{i+2},\cdots,X_{j-1}\}$ as follows:
        \[\ov{\bm \sigma}^{\vec{d}}_{i,j} = \left\{
        \begin{array}{rl}
        \left(\ov{\bm \sigma}^{\vec{d}}_{i,j-1}, \ X_{j-1}\right) & \text{if } P^{\vec{d}}_{i,j} = P^{\vec{d}}_{i,j-1} + d^+_{j-1},\\
        \left(\ov{\bm \sigma}^{\vec{d}}_{i+1,j}, \ X_{i+1}\right) & \text{if } P^{\vec{d}}_{i,j} = P^{\vec{d}}_{i+1,j} + d^-_{i+1},\\
        \left(\ov{\bm \sigma}^{\vec{d}}_{i,r}, \ \ov{\bm \sigma}^{\vec{d}}_{r,j}, \ X_{r}\right) & \text{if } P^{\vec{d}}_{i,j} = \min\left\{P^{\vec{d}}_{i,r}, P^{\vec{d}}f_{r,j}, \squareOmega(1-d_i, 1-d_r,  1-d_j)\right\}, 
        \end{array} \right. \]
        and $\ov{\bm \sigma}^{\vec{d}}_{i,i+1} = \emptyset$. At last, we construct a generalized elimination order $\ov{\bm \sigma}^{\vec{d}} = \left(\ov{\bm \sigma}^{\vec{d}}_{i,j}, \ov{\bm \sigma}^{\vec{d}}_{j,i}, X_i, X_j\right)$. It can be proved that $\displaystyle{\max_{\ell \in \left[\left|\ov{\bm \sigma}_{i,j}^{\vec{d}}\right|\right]} \min\left\{h\left({U^{\ov{\bm \sigma}_{i,j}^{\vec{d}}}_\ell}\right), \emm^{\ov{\bm \sigma}_{i,j}^{\vec{d}}}_\ell\right\} \le P^{\vec{d}}_{i,j}}$. In the base case when $j = i+1$, $h(\emptyset) = 0 \le 1$ trivially holds. In general, we distinguish three cases:
        \begin{itemize}[leftmargin=*]
            \item If $\ov{\bm \sigma}^{\vec{d}}_{i,j} = \left(\ov{\bm \sigma}^{\vec{d}}_{i,j-1}, \ X_{j-1}\right)$: by hypothesis, assume this claim holds for $\ov{\bm \sigma}^{\vec{d}}_{i,{j-1}}$, i.e., \[\displaystyle{\max_{\ell \in \left[\left|\ov{\bm \sigma}_{i,j-1}^{\vec{d}}\right|\right]} \min\left\{h\left({U^{\ov{\bm \sigma}_{i,j-1}^{\vec{d}}}_\ell}\right), \emm^{\ov{\bm \sigma}_{i,j-1}^{\vec{d}}}_\ell\right\} \le P^{\vec{d}}_{i,j-1}}\]
            For any $\ell \in \left[\left|\ov{\bm \sigma}_{i,j}^{\vec{d}}\right|-1\right]$, $U^{\ov{\bm \sigma}_{i,j}^{\vec{d}}}_{\ell} = U^{\ov{\bm \sigma}_{i,j-1}^{\vec{d}}}_{\ell}$, and $\emm^{\ov{\bm \sigma}_{i,j-1}^{\vec{d}}}_\ell = \emm^{\ov{\bm \sigma}_{i,j}^{\vec{d}}}_\ell$. 
            For $\ell = \left|\ov{\bm \sigma}_{i,j}^{\vec{d}}\right|$, $h(U_\ell^{\ov{\bm \sigma}_{i,j}^{\vec{d}}}) = h(X_iX_{j-1}X_j) \le h(X_iX_{j-1}) + h(X_j|X_{j_1}) \le P^{\vec{d}}_{i,j-1} + d^+_{j-1}$. % and $\emm_\ell^{\ov{\bm \sigma}_{i,j}^{\vec{d}}} = \mm(X_i, X_j, X_{j-1}) \le P^{\vec{d}}_{i,j-1} + d^+_{j-1}$.
            Together, we have:
            \[\displaystyle{\max_{\ell \in \left[\left|\ov{\bm \sigma}_{i,j}^{\vec{d}}\right|\right]} \min\left\{h\left({U^{\ov{\bm \sigma}_{i,j}^{\vec{d}}}_\ell}\right), \emm^{\ov{\bm \sigma}_{i,j}^{\vec{d}}}_\ell\right\} \le P^{\vec{d}}_{i,j-1}} + d^+_{j-1}\]
            \item If $\ov{\bm \sigma}^{\vec{d}}_{i,j} = \left(\ov{\bm \sigma}^{\vec{d}}_{i+1,j}, \ X_{j-1}\right)$: the case is similar as above.
            \item If $\ov{\bm \sigma}^{\vec{d}}_{i,j} =\left(\ov{\bm \sigma}^{\vec{d}}_{i,r}, \ \ov{\bm \sigma}^{\vec{d}}_{r,j}, \ X_{r}\right)$: by hypothesis, assume this claim holds for $\ov{\bm \sigma}^{\vec{d}}_{i,r}$ and $\ov{\bm \sigma}^{\vec{d}}_{r,j}$, 
            i.e., \[\displaystyle{\max_{\ell \in \left[\left|\ov{\bm \sigma}_{i,r}^{\vec{d}}\right|\right]} \min\left\{h\left({U^{\ov{\bm \sigma}_{i,r}^{\vec{d}}}_\ell}\right), \emm^{\ov{\bm \sigma}_{i,r}^{\vec{d}}}_\ell\right\} \le P^{\vec{d}}_{i,r}}\]
            \[\displaystyle{\max_{\ell \in \left[\left|\ov{\bm \sigma}_{r,j}^{\vec{d}}\right|\right]} \min\left\{h\left({U^{\ov{\bm \sigma}_{r,j}^{\vec{d}}}_\ell}\right), \emm^{\ov{\bm \sigma}_{r,j}^{\vec{d}}}_\ell\right\} \le P^{\vec{d}}_{r,j}}\]
            For any $\ell \in \left[\left|\ov{\bm \sigma}_{i,r}^{\vec{d}}\right|\right]$, $U^{\ov{\bm \sigma}_{i,r}^{\vec{d}}}_{\ell} = U^{\ov{\bm \sigma}_{i,j}^{\vec{d}}}_{\ell}$, and $\emm^{\ov{\bm \sigma}_{i,r}^{\vec{d}}}_\ell = \emm^{\ov{\bm \sigma}_{i,j}^{\vec{d}}}_\ell$. 
            For any $\ell \in \left[\left|\ov{\bm \sigma}_{r,j}^{\vec{d}}\right|\right]$, $U^{\ov{\bm \sigma}_{r,j}^{\vec{d}}}_{\ell} = U^{\ov{\bm \sigma}_{i,j}^{\vec{d}}}_{\ell+ \left|\ov{\bm \sigma}_{i,r}^{\vec{d}}\right|}$, and $\emm^{\ov{\bm \sigma}_{r,j}^{\vec{d}}}_\ell = \emm^{\ov{\bm \sigma}_{i,j}^{\vec{d}}}_{\ell+ \left|\ov{\bm \sigma}_{i,r}^{\vec{d}}\right|}$.
            For $\ell = \left|\ov{\bm \sigma}_{i,j}^{\vec{d}}\right|$, $\emm_\ell^{\ov{\bm \sigma}_{i,j}^{\vec{d}}} = \mm(X_i, X_r, X_j) \le \squareOmega(1-d_i, 1-d_r, 1-d_j)$, since $h(X_i) \le 1- d_i$, $h(X_r) \le 1- d_r$ and $h(X_j) \le 1- d_j$. Together, we have:
            \[\displaystyle{\max_{\ell \in \left[\left|\ov{\bm \sigma}_{i,j}^{\vec{d}}\right|\right]} \min\left\{h\left({U^{\ov{\bm \sigma}_{i,j}^{\vec{d}}}_\ell}\right), \emm^{\ov{\bm \sigma}_{i,j}^{\vec{d}}}_\ell\right\} \le \min\left\{P^{\vec{d}}_{i,r}, P^{\vec{d}}_{r,j}, \squareOmega(1-d_i, 1-d_r, 1-d_j)\right\}}\]
        \end{itemize}
        We define a generalized elimination ordering $\ov{\bm \sigma}^{\vec{d}} = \left(\ov{\bm \sigma}^{\vec{d}}_{i,j},\ov{\bm \sigma}^{\vec{d}}_{j,i}, X_i, X_j\right)$. Hence, \[\displaystyle{\max_{\ell \in \left[\left|\ov{\bm \sigma}^{\vec{d}}\right|\right]} \min\left\{h\left({U^{\ov{\bm \sigma}^{\vec{d}}}_\ell}\right), \emm^{\ov{\bm \sigma}^{\vec{d}}}_\ell\right\} \le \max\left\{P^{\vec{d}}_{i,j}, P^{\vec{d}}_{j,i}\right\}}.\]
    \end{itemize}    
    Combining these two cases and applying these arguments to all possible $\vec{d}$, we have
    \[\osubw(\calH) \le \min_{\vec{d} = (d^-_1, d^+_1, d^-_2, d^+_2, \cdots, d^-_k, d^+_k)}\min\left\{\min_{i \in [k]} 2 - d_i, \min_{i,j \in [k]: i<j} \max\left\{P^{\vec{d}}_{i,j}, P^{\vec{d}}_{j,i}\right\}\right\} = \squareC_k.\]
    \end{proof}

    \subsection{Clustered Hypergraphs}

    \begin{definition}[Clustered Hypergraph]
    A hypergraph $\calH = (\calV, \calE)$ is clustered if for any pair of vertices $X_i,X_j\in \calV$, there exists some hypergraph $e\in \calE$ such that $X_i, X_j \in e$.    
    \end{definition}
    
    \begin{lemma}
    \label{lem:clustered}
        A clustered hypergraph $\calH = (\calV, \calE)$ satisfies the following properties:
        \begin{itemize}[leftmargin=*]
            \item $\subw(\calH) = \rho^*(\calH)$.
            \item for any generalized elimination ordering $\ov{\bm \sigma} \in \pi(\calV)$, $U^{\ov{\bm \sigma}}_1 = \calV$; and therefore $U^{\ov{\bm \sigma}}_i \subseteq U^{\ov{\bm \sigma}}_1$ for any $i \in [|\ov{\bm \sigma}|]$. 
        \end{itemize}
    \end{lemma}

    \begin{lemma}
    \label{lmm:3-pyramid}
        For the following hypergraph representing a 3-pyramid: \begin{align}
            \calH = \left(\{Y, X_1, X_2,X_3\},\quad \{\{Y,X_1\}, \{Y,X_2\}, \{Y, X_3\}, \{X_1,X_2,X_3\}\}\right)
            \label{eq:H:pyramid-3}
        \end{align}
        $\osubw(\calH) = 2 - \frac{1}{\omega}$. %\displaystyle{2 - \displaystyle{\frac{1}{(\omega-2) \cdot \left\lfloor \frac{k-1}{2}\right\rfloor + \lceil \frac{k-1}{2}\rceil + 1}}}$ 
    \end{lemma}

    \begin{proof}[Proof of Lemma~\ref{lmm:3-pyramid}]
        \underline{\em Direction $\osubw(\calH) \le 2 - \frac{1}{\omega}$.} We partition all polymatroids $\bm h$ into the following cases:
        \item Case 1: there exists some $i \in [3]$ such that $h(Y|X_i) \le \Delta$. In this case, $h(\calV) \le h(Y|X_1X_2X_3) + h(X_1X_2X_3) \le h(Y|X_i) + h(X_1X_2X_3) \le \Delta + 1$;
        \item Case 2: $h(X_i|Y) \le \frac{\Delta}{2}$ for each $i \in [3]$. In this case, $h(\calV) \le h(X_1|Y) + h(X_2|Y) + h(X_3Y) \le \Delta + 1$.
        \item Case 3: none of the cases above, i.e., $h(Y|X_i) > \Delta$ for each $i \in [3]$. As $h(X_iY) \le 1$ and $h(Y|X_i) > \Delta$, we have $h(X_i) = h(X_iY) - h(Y|X_i) \le 1-\Delta$. As $h(X_iY) \le 1$ and $h(X_i|Y) > \frac{\Delta}{k-1}$, we have $h(Y) = h(X_iY) - h(X_i|Y) \le 1- \frac{\Delta}{2}$. Let $\bm A = \left\{X_2\right\}$ and $\bm B = \left\{X_3\right\}$. Note that $h(\bm A) \le (1-\Delta)$ and $h(\bm B) \le (1-\Delta)$. We assume that $\Delta$ is a parameter such that $\Delta \ge (1-\Delta)$ holds, i.e., $\Delta \ge \frac{1}{2}$. Consider a generalized elimination ordering $\ov{\bm \sigma} = (\{Y\},\bm A, \bm B, \{X_1\})$. We have 
            \begin{align*}
                \emm^{\ov{\bm \sigma}}_1  \le \mm(\bm A; \bm B; Y|X_1)  & \le %\mm(\bm A; \bm B; X_1Y) & = 
                %\gamma h(\bm A|X_1) + h(\bm B|X_1) + h(Y|X_1) + h(X_1)\\
                %& \le 
                \gamma h(\bm A) + h(\bm B) + h(X_1Y) \le \gamma (1-\Delta)  + (1-\Delta) + 1
            \end{align*}
            Combining all the cases, we obtain the desired upper bound $\displaystyle{1+ \displaystyle{\max\left\{\Delta,\gamma(1-\Delta)+ (1-\Delta)\right\}} = 2 - \displaystyle{\frac{1}{\omega}}}$, 
            by setting $\Delta = 1 - \frac{1}{\omega}$ (note that $\Delta \ge \frac{1}{2}$).

            \smallskip
            
            \underline{\em Direction $\osubw(\calH) \ge 2 - \frac{1}{\omega}$.} We identify the following polymatroid $\bm h \in \Gamma \cap \ed$:
            \begin{itemize}[leftmargin=*]
                \item$h(X_1) = h(X_2) = h(X_3) =\frac{1}{\omega}$;
                \item$h(Y) = 1- \frac{1}{\omega}$;
                \item$h(X_1X_2) = h(X_1X_3) = h(X_2X_3) = \frac{2}{\omega}$;
                \item $h(X_1Y) = h(X_2Y) = h(X_3Y)=1$;
                \item$h(X_1X_2X_3) = 1$;
                \item$h(X_1X_2Y) = h(X_1X_3Y) = h(X_2X_3Y) =  \frac{\omega+1}{\omega}$;
                \item$h(X_1X_2X_3Y) = 2 - \frac{1}{\omega}$;
            \end{itemize}

     \begin{figure}
            \centering
        \includegraphics[width=0.3\linewidth]{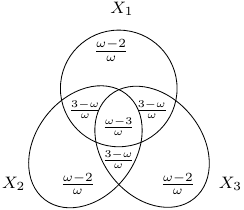}
        \caption{Diagram of the polymatroid for 3-Pyramid hypergraph.}
        \end{figure}     
    
     We distinguish an arbitrary generalized elimination ordering $\ov{\bm \sigma}$ into the following cases:
    \begin{itemize}[leftmargin=*]
        \item $\ov{\bm \sigma}[1] = \{Y\}$. $h(U^{\ov{\bm \sigma}}_1) = h(\calV) = 2 - \frac{1}{\omega}$.  $\emm^{\ov{\bm \sigma}}_1 = \min\left\{\mm(X_1, X_2, Y; X_3), \mm(X_1,X_2X_3,Y)\right\} = 2 - \frac{1}{\omega}$;
        \item $\ov{\bm \sigma}[1] = \{X_1\}$. $h(U^{\ov{\bm \sigma}}_1) = h(\calV) = 2 - \frac{1}{\omega}$. $\emm^{\ov{\bm \sigma}}_1 = \min\{\mm(X_2,Y,X_1|X_3),$ $\mm(X_2X_3, Y ,X_1),$ $\mm(X_2,X_3Y,X_1)\} = $
        
        The case with $\ov{\bm \sigma}[1] = \{X_2\}$ or $\ov{\bm \sigma}[1] = \{X_3\}$ is the same.
        \item $\ov{\bm \sigma}[1] = \{X_1, Y\}$. $h(U^{\ov{\bm \sigma}}_1) = h(\calV) = 2 - \frac{1}{\omega}$.  $\emm^{\ov{\bm \sigma}}_1 = \mm(X_2,X_3,X_1Y) = 2 - \frac{1}{\omega}$. The case with $\ov{\bm \sigma}[1] = \{X_2, Y\}$ or $\ov{\bm \sigma}[1] = \{X_3, Y\}$ is the same.
        \item $\ov{\bm \sigma}[1] = \{X_1, X_2\}$. $h(U^{\ov{\bm \sigma}}_1) = h(\calV) = 2 - \frac{1}{\omega}$.  $\emm^{\ov{\bm \sigma}}_1 = \mm(X_3,Y,X_1X_2) = 2 - \frac{1}{\omega}$. The case with $\ov{\bm \sigma}[1] = \{X_1, X_3\}$ or $\ov{\bm \sigma}[1] = \{X_2, X_3\}$ is the same.
    \end{itemize}
    Hence, for such $\bm h$, we have $\displaystyle{\max_{i \in [|\ov{\bm \sigma}|]} \min\left(h\left(U^{\ov{\bm \sigma}}_i\right), \emm^{\ov{\bm \sigma}}_i\right)} \ge 2 - \frac{1}{\omega}$.
    \end{proof}

    \begin{lemma}
        \label{lmm:k-pyramid}
        For the following hypergraph representing a $k$-pyramid: \begin{align}
            \calH = \left(\{Y, X_1, X_2,\cdots, X_k\}, \ \{\{Y,X_1\}, \{Y,X_2\}, \cdots, \{Y, X_k\}, \{X_1,X_2,\cdots,X_k\}\}\right)
        \label{eq:H:pyramid-k}
        \end{align}
        for $k \ge 3$, $\osubw(\calH) \le \displaystyle{2 - \frac{2}{\omega \cdot (k-1) -k + 3}}$. 
        %$\osubw(\calH) \le \displaystyle{2 - \displaystyle{\frac{1}{(\omega-2) \cdot \left\lfloor \frac{k-1}{2}\right\rfloor + \lceil \frac{k-1}{2}\rceil + 1}}}$.
    \end{lemma}

    \begin{proof}[Proof of Lemma~\ref{lmm:k-pyramid}]We partition all polymatroids $\bm h$ into the following two cases:
    \begin{itemize}[leftmargin=*]
        \item Case 1: there exists some $i \in [k]$ such that $h(Y|X_i) \le \Delta$. In this case, $h(\calV) \le h(Y|X_1X_2\cdots X_k) + h(X_1X_2\cdots X_k) \le h(Y|X_i) + h(X_1X_2\cdots X_k) \le \Delta + 1$;
        \item Case 2: $h(X_i|Y) \le \frac{\Delta}{k-1}$ for each $i \in [k]$. In this case, $h(\calV) \le \sum_{i\in [k-1]}h(X_1|Y) + h(X_kY) \le \Delta + 1$.
        \item Case 3: none of the cases above, i.e., $h(Y|X_i) > \Delta$ for each $i \in [k]$. As $h(X_iY) \le 1$ and $h(Y|X_i) > \Delta$, we have $h(X_i) = h(X_iY) - h(Y|X_i) \le 1-\Delta$. As $h(X_iY) \le 1$ and $h(X_i|Y) > \frac{\Delta}{k-1}$, we have $h(Y) = h(X_iY) - h(X_i|Y) \le 1- \frac{\Delta}{k-1}$.

        We further distinguish two more cases depending on whether $k$ is even or odd. 
        \begin{itemize}[leftmargin=*]
            \item {\bf When $k$ is odd.} 
            Let $\bm A = \left\{X_2, X_3, \cdots, X_{\frac{k+1}{2}}\right\}$ and $\bm B = \left\{X_{\frac{k+1}{2}+1}, X_{\frac{k+1}{2}+2}, \cdots, X_k \right\}$. Note that $h(\bm A) \le (1-\Delta) \cdot \frac{k-1}{2}$ and $h(\bm B) \le (1-\Delta) \cdot \frac{k-1}{2}$. In this case, we assume that $\Delta > (1-\Delta) \cdot\frac{k-1}{2}$. Consider a generalized elimination ordering $\ov{\bm \sigma} = (\{Y\},\bm A, \bm B, \{X_1\})$. We have 
            \begin{align*}
                \emm^{\ov{\bm \sigma}}_1  \le \mm(\bm A; \bm B; Y|X_1)  & \le %\mm(\bm A; \bm B; X_1Y) & = 
                %\gamma h(\bm A|X_1) + h(\bm B|X_1) + h(Y|X_1) + h(X_1)\\
                %& \le 
                \gamma h(\bm A) + h(\bm B) + h(X_1Y) \le (\gamma+1)(1-\Delta) \cdot \frac{k-1}{2} + 1
            \end{align*}
            Combining all the cases, we obtain the desired upper bound: \[1+ \displaystyle{\max\left\{\Delta, (\gamma+1)(1-\Delta) \cdot \frac{k-1}{2} \right\}} = \displaystyle{2 - \displaystyle{\frac{2}{\omega \cdot (k-1) - k + 3}}}.\]
            
            \item {\bf When $k$ is even.} We partition all polymatroids $\bm h$ into the following two cases:  Let $\bm A = \left\{X_1, X_2, X_3, \cdots, X_{\frac{k}{2}}\right\}$ and $\bm B = \left\{X_{\frac{k}{2}+1}, X_{\frac{k}{2}+1}, \cdots, X_k \right\}$. Note that $h(\bm A) \le (1-\Delta) \cdot \frac{k}{2}$ and $h(\bm B) \le (1-\Delta) \cdot \frac{k}{2}$. We assume that $\Delta$ is a parameter such that $1-\frac{\Delta}{k-1} > (1-\Delta) \cdot\frac{k}{2}$. Consider a generalized elimination ordering $\ov{\bm \sigma} = (\{Y\}, \bm A, \bm B)$. We have \begin{align*} \emm^{\ov{\bm \sigma}}_1  \le \mm(\bm A; \bm B; Y)  & \le \gamma h(\bm A) + h(\bm B) + h(Y) \le \gamma (1-\Delta) \cdot \frac{k}{2} + (1-\Delta) \cdot \frac{k}{2} + 1 -\frac{\Delta}{k-1}\end{align*} Combining all the cases, we obtain the desired upper bound: \[1+ \displaystyle{\max\left\{\Delta,\gamma (1-\Delta) \cdot \frac{k}{2} + (1-\Delta) \cdot \frac{k}{2} - \frac{\Delta}{k-1}\right\}} \ge \displaystyle{2 - \displaystyle{\frac{2}{\omega \cdot (k-1) - k + 3}}}\]
         \end{itemize}

    \end{itemize} 
    \end{proof}

    \begin{lemma}
    \label{lem:example-2}
    For the following hypergraph 
    \[\calH = \left(\{X,Y,Z,W,L\},\quad \{\{X,Y,W\},\{X,Y, L\},\{X,Z\}, \{Y,Z\}, \{Z,W,L\}\}\right)\] 
    $\osubw(\calH) \le 2- \frac{1}{2(\omega -2) +3}$.
    \end{lemma}
    \paragraph{Remark.} For the hypergraph $\calH$ in Lemma~\ref{lem:example-2}, $\subw(\calH) = \frac{9}{5}$.  If $\omega =3$, $\osubw(\calH) \le 2 - \frac{1}{2 \omega -1} = \subw(\calH) = \frac{9}{5}$. If $\omega < 3$, $\osubw(\calH) < \subw(\calH)$.
    \begin{proof}
     Let $\Delta = \frac{2(\omega-1)}{2\omega -1}$.  We partition all polymatroids $\bm h$ into the following cases:
        \begin{itemize}[leftmargin=*]
            \item Case 1: $h(XY|W) \le \Delta$ or $h(ZL|W) \le \Delta$. wlog, suppose $h(XY|W) \le \Delta$. We have $h(XYZWL) \le h(XY|ZWL) + h(ZWL) \le h(XY|W) + h(ZWL) \le 1 + \Delta$.
            \item Case 2: $h(XY|L) \le \Delta$ or $h(ZW|L) \le \Delta$. This case is similar to Case 1.
            \item Case 3: $h(X|Z) \le \frac{\Delta}{2}$ and $h(Y|Z) \le \frac{\Delta}{2}$. We have  $h(XYZWL) \le h(X|ZWL) + h(Y|ZWL) + h(ZWL) \le h(X|Z) + h(Y|Z) + h(ZWL) \le 1 + \Delta$.
            \item Case 4: either $h(L|XY) \le \frac{2\Delta-1}{3}$ or $h(W|XY) \le \frac{2\Delta-1}{3}$, and either $h(Z|X) \le \frac{1+\Delta}{3}$ or $h(Z|Y) \le \frac{1+\Delta}{3}$. Wlog, suppose $h(L|XY) \le \frac{2\Delta-1}{3}$ and $h(Z|X) \le \frac{1+\Delta}{3}$. In this case, we have $h(XYZWL)\le h(L|XYW) + h(Z|XYW) + h(XYW) \le 1 + \Delta$.  
            \item Case 5: none of the cases above. As $h(XY|W) > \Delta$ and $h(XYW) \le 1$, we have $h(W) = h(XYW) -  h(XY|W) < 1- \Delta$. As $h(XY|L) > \Delta$ and $h(XYL) \le 1$, we have $h(L) < 1 -\Delta$. As $h(X|Z) > \frac{\Delta}{2}$ and $h(XZ) \le 1$, we have $h(Z) < 1- \frac{\Delta}{2}$. If either $h(L|XY) \le \frac{2\Delta-1}{3}$ or $h(W|XY) \le \frac{2\Delta-1}{3}$, but $h(Z|X) \le \frac{1+\Delta}{3}$ and $h(Z|Y) \le \frac{1+\Delta}{3}$, we have $h(X) \le 1 -\frac{1+\Delta}{3} = \frac{2-\Delta}{3}$ and $h(Y) \le 1 -\frac{1+\Delta}{3} = \frac{2-\Delta}{3}$. Hence, $h(XY)\le \frac{2(2-\Delta)}{3}$. If both $h(L|XY) > \frac{2\Delta-1}{3}$ and $h(W|XY) > \frac{2\Delta-1}{3}$, we have $h(XY) \le 1 -\frac{2\Delta-1}{3} = \frac{2(2-\Delta)}{3}$. Consider a generalized elimination ordering $\ov{\bm \sigma} = (\{XY\}, \{Z\}, \{W\}, \{L\})$. 

            \begin{align*}    
        \emm^{\ov{\bm \sigma}}_1 \le \mm(Z;W;XY|L)
        &\le\max\begin{cases}
        \gamma h(Z) + h(W) + h(XYL)&\\
        h(Z) + \gamma h(W) + h(XYL) &\\
        h(Z) + h(W) +  \gamma h(XYL)&\end{cases} \le (\omega-2) \cdot (1-\Delta) + 1 - \frac{\Delta}{2} + 1
        \end{align*}
        \end{itemize}
        Putting everything together, we obtain the desired upper bound 
        \[\max\left\{1+\Delta, (\omega-2) \cdot (1-\Delta) + 2 - \frac{\Delta}{2}\right\} = 2- \frac{1}{2(\omega -2) +3}.\] 
    \end{proof}

\end{document}